\documentclass[12pt,a4paper,twoside]{book}

\usepackage{vuwthesis} 
\usepackage{pseudocode}
\usepackage[font=small,labelfont=bf]{caption}
\usepackage{bibunits}
\def\year{{2005}}

\usepackage{palatino} 
\usepackage{url} 

\usepackage{amsmath} 
\usepackage{amsfonts} 
\usepackage{amssymb}
\usepackage{amsthm}
\usepackage{times}
\usepackage{eucal}
\usepackage{mathrsfs}
\usepackage{exscale}
\usepackage{calc}
\usepackage{float}
\usepackage{graphics}
\usepackage{graphicx}
\usepackage{verbatim}
\usepackage{enumerate}
\usepackage{epsfig}
\usepackage{tabularx}
\usepackage{multirow}
\usepackage{longtable}
\usepackage{color}
\newlength{\defoddsidemargin}
\newlength{\defevensidemargin}
\newlength{\deftextwidth}
\newlength{\defheadheight}
\newlength{\deftopmargin}
\newlength{\deftextheight}

\def\indred{{\rightrightarrows}}
\def\projred{{\Rightarrow}}

\def\F{{\mathcal{F}}}
\def\B{{\mathbb B}}
\def\R{{\mathbb R}}
\def\Q{{\mathbb Q}}

\def\N{{\mathbb N}}

\def\I{{\mathbb I}}
\def\U{{\mathbb U}}
\def\V{{\mathbb V}}

\def\argmin{{\operatorname{argmin}\,}}

\def\time{{\mathrm{time}}}
\def\e{{\varepsilon}}
\def\S{{\mathbb S}}

\def\Top{{\mathcal{T}}}
\def\Unif{{\mathcal{U}}}
\def\cj#1{{#1}^\ast}
\def\qam#1{{#1}^\mathfrak{s}}
\def\qsum#1{{#1}^\mathfrak{u}}

\newcommand{\KompE}[2]{\mathcal{K}({#1},{#2})} 
\newcommand{\Komp}[2]{\mathcal{K}_0({#1},{#2})} 
\newcommand{\ClsdE}[2]{\mathscr{C}({#1},{#2})} 
\newcommand{\Clsd}[2]{\mathscr{C}_0({#1},{#2})} 
\newcommand{\PowE}[1]{\mathcal{P}({#1})} 
\newcommand{\Pow}[1]{\mathcal{P}_0({#1})} 
\newcommand{\PowF}[1]{\mathcal{P}_{\omega}({#1})} 

\newcommand{\diam}{\mathrm{diam}}
\newcommand{\ball}[2]{\mathfrak{B}_{#2}(#1)}
\newcommand{\lball}[2]{\mathfrak{B}_{#2}^L(#1)}
\newcommand{\rball}[2]{\mathfrak{B}_{#2}^R(#1)}
\newcommand{\lnbhd}[2]{{#1}_{#2}^L}
\newcommand{\rnbhd}[2]{{#1}_{#2}^R}
\newcommand{\nbhd}[2]{{#1}_{#2}}
\newcommand{\cball}[2]{\overline{\mathfrak{B}_{#2}}(#1)}
\newcommand{\clball}[2]{\overline{\mathfrak{B}_{#2}^L}(#1)}
\newcommand{\crball}[2]{\overline{\mathfrak{B}_{#2}^R}(#1)}

\newcommand{\NW}{\mathbf{NW}}
\newcommand{\SW}{\mathbf{SW}}

\def\SL#1{\mathcal{SL}_0({#1})}

\newcommand{\abs}[1]{\left\vert#1\right\vert}

\newcommand{\ceil}[1]{\left\lceil#1\right\rceil}
\def\norm#1{\left\Vert#1\right\Vert}

\newcommand{\gh}[1]{\colorbox{gr}{\makebox[1.5mm]{#1}}}
\definecolor{gr}{rgb}{0.85,0.85,0.85}

\newtheorem{thm}{Theorem}[section]
\newtheorem{corol}[thm]{Corollary}
\newtheorem{algo}[thm]{Algorithm}
\newtheorem{lemma}[thm]{Lemma}
\newtheorem{prop}[thm]{Proposition}
\theoremstyle{definition}
\newtheorem{defn}[thm]{Definition}

\newtheorem{example}[thm]{Example}

\theoremstyle{remark}
\newtheorem{remark}[thm]{Remark}

\newenvironment{defin}{\begin{defn}}{\hspace*{0pt}\hfill $\blacktriangle$ \end{defn}}

\newfloat{algo}{h}{alg}[section]
\floatname{algo}{Algorithm}
\def\algocapt#1{\addcontentsline{alg}{subsection}%
		{\protect\numberline{\csname thepseudocode\endcsname}%
		{\ignorespaces #1}}}

\begin{document}
\frontmatter


\title{Quasi-metrics, Similarities and Searches: aspects of geometry of protein datasets}
\author{Aleksandar Stojmirovi\'c}
\date{\today}
\subject{Mathematics}

\abstract{
\addcontentsline{toc}{chapter}{Abstract}
A \emph{quasi-metric} is a distance function which satisfies the triangle inequality but is not symmetric: it can be thought of as an asymmetric metric. Quasi-metrics were first introduced in 1930s and are a subject of intensive research in the context of topology and theoretical computer science.

The central result of this thesis, developed in Chapter 3, is that a natural correspondence exists between similarity measures between biological (nucleotide or protein) sequences and quasi-metrics. As sequence similarity search is one of the most important techniques of modern bioinformatics, this motivates a new direction of research: development of geometric aspects of the theory of quasi-metric spaces and its applications to similarity search in general and large protein datasets in particular.

The thesis starts by presenting basic concepts of the theory of quasi-metric spaces illustrated by numerous examples, some previously known, some novel. In particular, the universal countable rational quasi-metric space and its bicompletion, the universal bicomplete separable quasi-metric space are constructed. Sets of biological sequences with some commonly used similarity measures provide a further and the most important example.

Chapter 4 is dedicated to development of a notion of the quasi-metric space with Borel probability measure, or \emph{pq-space}. The concept of a $pq$-space is a generalisation of a notion of an $mm$-space from the asymptotic geometric analysis: an $mm$-space is a metric space with Borel measure that provides the framework for study of the \emph{phenomenon of concentration of measure on high dimensional structures}. While some concepts and results are direct extensions of results about $mm$-spaces, some are intrinsic to the quasi-metric case. One of the main results of this chapter indicates that `a high dimensional quasi-metric space is close to being a metric space'. 

Chapter 5 investigates the geometric aspects of the theory of database similarity search. It extends the existing concepts of a \emph{workload} and an \emph{indexing scheme} in order to cover more general cases and introduces the concept of a \emph{quasi-metric tree} as an analogue to a \emph{metric tree}, a popular class of access methods for metric datasets. The results about $pq$-spaces are used to produce some new theoretical bounds on performance of indexing schemes.

Finally, the thesis presents some biological applications. Chapter 6 introduces FSIndex, an indexing scheme that significantly accelerates similarity searches of short protein fragment datasets. The performance of FSIndex turns out to be very good in comparison with existing access methods. Chapter 7 presents the prototype of the system for discovery of short functional protein motifs called PFMFind, which relies on FSIndex for similarity searches. 
}

\phd

\maketitle
\addcontentsline{toc}{chapter}{Acknowledgements}
\chapter*{Acknowledgements}

I am indebted to many people and institutions who have helped me to survive and even enjoy the four years it took to produce this thesis.

First of all I wish to offer my sincerest thanks to my supervisors, Dr. Vladimir Pestov, who was a Reader in Mathematics at Victoria University of Wellington when I started my PhD studies and is now a Professor of Mathematics at the University of Ottawa, and Dr. Bill Jordan, Reader in Biochemistry at Victoria University of Wellington, who have supported me and guided me in all imaginable ways during the course of the study. Dr. Mike Boland from the Fonterra Research Centre was principal in getting my study off the ground by introducing me to the problem of short peptide fragments.
 
My scholarship stipend was provided through a Bright Future Enterprise Scholarship jointly funded by the The Foundation for Research, Science and Technology and Fonterra Research Centre (formerly The New Zealand Dairy Research Institute).

I have enjoyed a generous and consistent support from the Faculty of Science, the School of Mathematical and Computing Sciences and the School of Biological Sciences at the Victoria University of Wellington. Not only have they contributed significant funds towards my travels to conferences and to Canada to visit my supervisor as well as towards a part of tuition fees, but have provided an excellent environment to work in. I would particularly like to thank Dr. Peter Donelan, who was the head of the School of Mathematical and Computing Sciences for most of the time I was doing my thesis and who signed my progress reports instead of my principal supervisor. I am grateful to Professor Estate Khmaladze and Dr. Peter Andreae for being willing to listen to my numerous questions in their respective areas. I also wish to acknowledge the system programmers Mark Davis and Duncan McEwan for maintaining our systems and being always available to answer my questions about C programming, UNIX, networks etc. I wish to thank the Department of Mathematics and Statistics of the University of Ottawa, which has accepted me as a visitor on two occasions for four months in total. 

I thank my colleagues Azat Arslanov and Todd Rangiwhetu who at times shared office with me for encouraging me and proofreading some of my manuscripts.

I would like to thank Professor Vitali Milman who, while being a visitor in Wellington, offered a lot of encouragement and some very helpful advice on how to approach mathematics. A very special thanks goes to Dr. Markus Hegland for convincing me to learn the Python programming language and ease my programming burden. Markus was also one of the supervisors (the other being Vladimir Pestov) for my summer 1999 project at the Australian National University that is presented as Appendix \ref{app:distexp}. Professor Paolo Ciaccia and Dr. Marco Patella have generously made the source code for their M-tree publicly available on the web and have agreed to send me a copy of the code for mvp-tree.

My mother Ljiljana has supported me throughout my studies and sacrificed a lot to see me where I am now. No words can ever be sufficient to express my gratitude.

\markboth{Contents}{Contents}
\addcontentsline{toc}{chapter}{Contents}
\tableofcontents

\addcontentsline{toc}{chapter}{List of Figures}
\listoffigures

\addcontentsline{toc}{chapter}{List of Tables}
\listoftables

\addcontentsline{toc}{chapter}{List of Algorithms}
\listof{algo}{List of Algorithms}

\mainmatter

\chapter{Introduction}\label{ch:intro}

The main focus of this thesis is on application of concepts of modern mathematics not previously used in biological context to problems of biological sequence similarity search as well as to the general theory of indexability of databases for fast similarity search. The biological applications are concentrated to investigations of short protein fragments using a novel tool, called FSIndex, which allows very fast retrieval of similarity based queries of datasets of short protein fragments. 

Clearly, this work stands at an intersection of several disciplines. The approach is mostly mathematical and rigorous where possible but also touches some aspects of the database theory and computational biology. The main result, presented in Chapter \ref{ch:bioseq_qm}, shows that deep connections exist between \emph{quasi-metrics} (asymmetric distance functions), and similarity measures on biological sequences. This motivates an effort to generalise the concepts and techniques from asymptotic geometric analysis and database indexing that apply to metric spaces to their quasi-metric counterparts, and to apply the resulting structures to biological questions.

The present chapter introduces the biological background associated with proteins and their short fragments and outlines the remainder of the thesis. It is assumed that general concepts related to biological macromolecules are well known and only those particularly relevant will be emphasised. Many important concepts will only be mentioned briefly and their detailed explanation left for the subsequent chapters.

\section{Proteins}

\subsection{Basic concepts} 

\emph{Proteins} are organic macromolecules consisting of \emph{amino acids} joined by \emph{peptide bonds}, essential for functioning of a living cell. They are involved in all major cellular processes, playing a variety of roles, such as catalytic (enzymes), structural, signalling, transport etc.

Structurally, proteins are linear chains (\emph{polypeptides}) composed of the twenty standard amino acids which can be classified according to their chemical properties (Table \ref{tbl:amino_acids}). A protein in the living cell is produced through the processes of \emph{transcription} and \emph{translation}. Simply stated, the information encoded by a \emph{gene} on DNA is transcribed into a mRNA molecule which is then translated into a protein on \emph{ribosomes} by putting an amino acid for every \emph{codon} triplet of nucleotides on mRNA. Constituent amino acids of a protein can be post-translationally modified, for example by attaching a sugar or a phosphate group on their side chains.
 
Four distinct aspects of protein structure are generally recognised. The \emph{primary structure} of a protein is the sequence of its constituent amino acids. The \emph{secondary structure} refers to the local sub-structures such as \emph{$\alpha$-helix}, \emph{$\beta$-sheet} or \emph{random coil}. The \emph{tertiary structure} is the spatial arrangement of a single polypeptide chain while the \emph{quaternary structure} refers to the arrangements of multiple polypeptides (\emph{protein subunits}) forming a \emph{protein complex}. We refer to the tertiary and quaternary structures as \emph{conformations}.

\begin{table}[!ht]
{\small
\begin{tabular}{|l|>{\centering}m{1.1cm}|>{\centering}m{1.1cm}|>{\centering}m{1.2cm}|>{\centering}m{1.8cm}|m{3.5cm}|}
\hline
Name & Three Letter Code & One Letter Code & Residue Mass (Da) & Abundance (\%) & Properties \\ \hline\hline
Glycine & Gly & G & 57.0 & 6.93  &  no side chain\\ \hline
Alanine & Ala & A & 71.1 & 7.80  & \multirow{5}{2.8cm}{non-polar aliphatic} \\ 
Valine & Val & V & 99.1 & 6.69  & \\ 
Isoleucine & Ile & I & 113.2 & 5.91 & \\
Leucine & Leu & L & 113.2 & 9.62 & \\
Methionine & Met & M & 131.2 & 2.37 & \\ \hline
Phenylalanine & Phe & F & 147.2 & 4.02  & \multirow{2}{2.8cm}{non-polar aromatic} \\
Tryptophan & Trp & W & 186.2 & 1.16 & \\ \hline
Serine & Ser & S & 87.1 & 6.89 & \multirow{4}{2.8cm}{polar aliphatic} \\
Threonine & Thr & T & 101.1 & 5.46 & \\
Asparagine & Asn & N & 114.1 & 4.22 & \\
Glutamine & Gln & Q & 128.1 & 3.93 & \\ \hline
Tyrosine & Tyr & Y & 162.2 & 3.09 & polar aromatic \\ \hline
Lysine & Lys & K & 128.2 & 5.93 & \multirow{3}{2.8cm}{charged, basic}\\
Arginine & Arg & R & 156.2 & 5.29 & \\
Histidine & His & H & 137.1 & 2.27 & \\ \hline
Aspartic acid & Asp & D & 115.1 & 5.30 & \multirow{2}{2.8cm}{charged, acidic} \\
Glutamic acid & Glu & E & 129.1 & 6.59 & \\ \hline
Cysteine & Cys & C & 103.1 & 1.57 & forms disulphide bridges\\ \hline
Proline & Pro & P & 97.1 & 4.85 & cyclic, disrupts structure\\ \hline 
\end{tabular}
}
\caption[The standard amino acids.]{The standard amino acids. Residue mass is the mass of amino acid minus the mass of a molecule of water (18.0 Da). Relative abundances are taken from the Release 44.0 of SwissProt sequence database \cite{Boeckmann2003}.} \label{tbl:amino_acids}
\end{table}

Protein function in general is determined by the conformation but it is strongly believed that  secondary, tertiary and quaternary structure are all determined by the amino acid sequence. So far, there has been no solution to the \emph{folding problem}, which is to determine the conformation solely from the amino acid sequence by computational means. All presently known structures have been determined either experimentally, by using crystallographic or NMR (Nuclear Magnetic Resonance) techniques, or by homology modelling from closely related sequences with experimentally derived structures. 

While the number of possible amino acid sequences is very large, known proteins take a relatively small amount of conformations \cite{Murzin:1995,Sander:1996}. There is an ongoing effort to determine all possible conformations proteins can take, that is, to produce a map of the conformation space \cite{Sander:1996,Holm:1997-dali,HSZK03}. Such a map would enable modelling of all the structures which have not been experimentally determined using the existing structures of the similar proteins.

A \emph{structural motif} is a three-dimensional structural element or \emph{fold} consisting of consecutive secondary structures, for example, the $\beta$-barell motif. Structural motifs can but need not be associated with biological function. A \emph{structural domain} is a unit of structure having a specific function which combines several motifs and which can fold independently. A protein \emph{sequence motif} is a amino-acid pattern associated with a biological function. It may, but need not, be associated with a structural motif.

\subsection{Protein sequence alignment}

Sequence alignment is presently one of the cornerstones of computational biology and bioinformatics \cite{Sterky10784293}. As mentioned before, all elements of protein structure and function ultimately depend on the sequence and in addition, sequence data is most readily available, mostly originating from the translations of the sequences of genes and transcripts obtained through large scale sequencing projects \cite{Venter12610531,Winslow12750305} such as the recently completed Human Genome Project \cite{Lander11237011}. Raw sequences produced by the sequencing projects need to be \emph{annotated}, that is, functional descriptions attached to each sequence and/or its constituent parts \cite{Stein11433356}. The most widely used (but not always adequate \cite{Rost14685688,Gerlt11178260}) technique for annotation is \emph{homology} or \emph{similarity} search where the unannotated sequences are annotated according to their similarity to previously annotated sequences \cite{Bork9537411} resulting in great savings of time and effort required for experimental analysis of each sequence. 

Much of the sequence data is easily accessible from public repositories \cite{Galperin14681349}, the best known being the database collection at the National Center for Biotechnology Information (NCBI -- \url{http://www.ncbi.nlm.nih.gov}) in the\\ \mbox{United States} \cite{Wheeler12519941}. The NCBI repository contains among many others the \emph{GenBank} \cite{BKLOW04} DNA sequence database, a part of the international collaboration involving its European (\emph{EMBL}) \cite{Kulikova14681351} and Japanese (DDBJ) \cite{Miyazaki14681352} counterparts and the \emph{RefSeq} \cite{Pruitt11125071}, the set of reference gene, transcript and protein sequences for a variety of organisms. The major source of protein related resources is the ExPASy site \cite{Gasteiger12824418} at the Swiss Institute of Bioinformatics (\url{http://www.expasy.org}), the home of \emph{SwissProt}, a human curated database of annotated protein sequences, and its companion \emph{TrEMBL}, a database of machine-annotated translated coding sequences from EMBL \cite{Boeckmann2003}. SwissProt and TrEMBL together form the \emph{Uniprot} \cite{BAWBB05} universal protein resource. Uniprot has sequence composition similar to the NCBI RefSeq protein dataset. 

The principal technique for general pairwise biological sequence comparison is known as \emph{alignment}\footnote{The term `alignment' is used to denote both the method of sequence comparison and a particular transformation of one sequence into another.}. We distinguish a \emph{global alignment} where the whole extent of both sequences is aligned and \emph{local alignment} where only substrings (contiguous subsequences) are aligned. The foundations of the algorithms for sequence alignment have been developed in the 1970s and early 1980s \cite{NW70,Se74,WSB76,SWF81} culminating with the famous \emph{Smith-Waterman} \cite{SW81} algorithm for local sequence alignments.

Pairwise sequence alignment is based on transformations of one sequence into other which is broken into transformations of substrings one sequence into substrings of other. Ultimately two types of transformations are used: \emph{substitutions} where one residue (amino acid in proteins) is substituted for another and \emph{indels} or \emph{insertions} and \emph{deletions} where a residue or a sequence fragment is inserted (in one sequence) or deleted (in the other). Indels are often called \emph{gaps} and alignments without gaps are called \emph{ungapped}. Each of the basic transformations is assigned a numerical \emph{score} or \emph{weight} and the transformation with the optimal score is reported as the `best' alignment of the two sequences. All algorithms for computation of pairwise alignments use the \emph{dynamic programming} \cite{BHK59} technique. 

Alignment scores can be \emph{distances} in which case all scores are positive and identity transformations (no changes) have the score $0$. Distances are often required to have additional properties such as to satisfy the \emph{triangle inequality}. Alternatively, transformation scores may be given as \emph{similarities} which are large and positive for \emph{matches} (identity transformations) and some (`close') \emph{mismatches} while other mismatches and gaps have a negative score. The choice of whether to use similarities or distances is influenced by available computational algorithms: similarities are preferred in sequence comparisons because they are more suitable for local alignments while distances are often used in phylogenetics \cite{Gusfield97}. Furthermore, similarity scores are, at least in some cases, amenable for statistical and information-theoretic interpretations \cite{KaA90,Alt91,KA93}.  

According to the `basic' alignment model, the transformation scores only depend on the residues being substituted in the case of substitutions, and lengths of the gaps in the case of indels. There is no dependence on the position of the transformation within the two sequences being compared nor on the previous or subsequent transformations. In this model, substitution scores come from \emph{score matrices}, the best known being the PAM \cite{Dayhoff:1978} and BLOSUM \cite{Henikoff:1992} families of amino acid matrices. Both PAM and BLOSUM matrices were derived from multiple alignments (alignments of more than two sequences) of related proteins.

The most widely used tool for sequence similarity search is BLAST (Basic Local Alignment Search Tool) \cite{altschul97gapped} developed at the NCBI. BLAST is a based on heuristic search algorithm which uses dynamic programming on only a relatively small part of the sequence database searched while retrieving most of the \emph{hits} or \emph{neighbours}. The importance of BLAST cannot be overestimated -- its applications range from day-to-day use by biologists to find sequences similar to the sequences of their interest to high throughput automated annotation, sequence clustering and many others. Finding efficient algorithms which would improve on BLAST in accuracy and/or speed remains one of the areas of very active development \cite{Kent:2002,GiWaWaVo00,MXSM03,Hu04}. 

While BLAST is quite fast and accurate, it cannot always retrieve all biologically significant homologs due to limitations of the basic alignment model. Improvements to the basic alignment model involve the use of \emph{Position Specific Score Matrices} or PSSMs, also known as \emph{profiles} \cite{Gribskov:1987}, which assign different substitution scores at different positions. PSI-BLAST \cite{altschul97gapped} uses PSSMs through an iterative technique where the results of each search are used to compute a PSSM for a subsequent iteration -- the first search is performed using the basic model. This method is known to retrieve more `distant' homologues which would be missed using the basic model. More sophisticated sequence and alignment models such as \emph{Hidden Markov Models} (HMMs) \cite{Durbin:1998,Eddy98,Karplus:1998,Hargbo:1999} can be used with even more accuracy if there is sufficient data for their training. In most common cases, a substantial body of statistical theory for interpretation of the results exists \cite{Durbin:1998,EG01}.

\subsection{Short peptide fragments}

While most of the works relating to protein sequence analysis concentrate on either full sequences, or fragments of medium length (50 amino acids -- e.g. \cite{LLTY97}), the main biological focus of this thesis is on short peptide fragments of lengths 6 to 15. 

While short peptide fragments can be interesting as being parts of larger functional domains, they often have important physiological function on their own. To mention one of many examples, a large variety of peptides are generated in the gut lumen during normal digestion of dietary proteins and absorbed through the gut mucosa. Smaller fragments, that is dipeptides and tripeptides, are the primary source of dietary nitrogen. Larger peptides, many of which have been shown to have physiological activity may also be absorbed. These peptides may modulate neural, endocrine, and immune function \cite{ZaSi04,kitts03}. Short peptide motifs may also have a role in disease. For example, it was discovered that one of the proteins encoded by HIV-1 and Ebola viruses contains a conserved short peptide motif which, due to its interaction with host cell proteins involved in protein sorting, plays a significant role in progress of the disease \cite{martinserano01}. 

The biological part of this thesis aims to develop tools for identifying conserved fragment motifs among possibly otherwise unrelated protein sequences. Such tools may produce the results that would enable determination of the origin of fragments with no obvious function. The investigation is not restricted solely to bioactive peptides but considers all possible fragments (of given lengths) of full sequences available from the databases. 

The main paradigm can be expressed as follows:
\begin{quote}
\emph{A sequence fragment that recurs in a non random and unexpected pattern indicates a possible structural motif that has a biological function.}
\end{quote}

The approach taken here mirrors that of full sequence analysis -- the principal technique used is similarity search using substitution matrices and profiles. However, the sequence comparison model uses a global ungapped similarity measure comparing the fragments of the same length. This can be justified by computational advantages -- it leads to sequence comparisons of linear instead of quadratic complexity, and also by the specific nature of the problem. 

One issue which is not so problematical with longer sequences is that of statistical significance. According to the model of Karlin and Altschul \cite{KaA90} used (in a slightly modified form) in BLAST, short alignments are not statistically significant at the levels routinely used for full sequence analysis -- there are too few possible alignments between two short fragments . In other words, high scoring alignments of two short fragments are not unlikely to occur by chance and hence the results of searches cannot be immediately assumed to have a biological significance. The current attempt towards overcoming this problem is based on using the iterative approach to refine the sequence profile and insistence on strong conservation among the search results.

Reliance on similarity search and the vast scale of existing sequence databases puts a premium on fast query retrieval that cannot be obtained using existing tools such as BLAST, which, at significance levels necessary to retrieve sufficient numbers of hits, essentially reduces to sequential scan of all fragments. Hence it is necessary to first develop an \emph{index} that would speed up the search and to do so it is necessary to explore the geometry of the space of peptide fragments. This leads to the other central concepts of the thesis: \emph{indexing schemes} and \emph{quasi-metrics}.

\section{Indexing for Similarity Search}

Indexing a dataset means imposing a structure on it which facilitates query retrieval. Most common uses of databases require indexing for exact queries, where all records matching a given key are retrieved. On the other hand, many kinds of databases such as multimedia, spatial and indeed biological, need to support query retrieval by similarity -- then need to fetch not only the objects that match the query key exactly but also those that are `close' according to some similarity measure. Hence, substantial amount of research is directed towards efficient algorithms and data structures for indexing of datasets for similarity search \cite{MaThTs99}.

It is not surprising that geometric as well as purely computational aspects such as I/O costs
are heavily represented in the existing works on indexing for similarity search. Indeed, most publications concentrate on the algorithms and data structures which can be applied to the datasets which can be represented as vector or metric (distance) spaces \cite{CNBYM,HjSa03}. In many cases, the so-called \emph{Curse of Dimensionality} \cite{Friedman97} is encountered: performance of indexing schemes deteriorates as the dimension of datasets grow so that at some stage sequential scan outperforms any indexing scheme \cite{BeyerGRS99,HinneburgAK00}. This manifestation has been linked by Pestov \cite{Pe00} to the phenomenon of \emph{concentration of measure on high-dimensional structures}, well known from the asymptotic geometric analysis \cite{MS86,Le01}.

In their influential paper \cite{H-K-P}, Hellerstein, Koutsoupias and Papadimitriou stressed the need for a general theory of \emph{indexability} in order to provide a unified approach to a great variety of schemes used to index into datasets for similarity search and provided a simple model of an \emph{indexing scheme}. The aim of this thesis is to extend their model so that it corresponds more closely to the existing indexing schemes for similarity search and to apply the methods from the asymptotic geometric analysis for performance prediction. Sharing the philosophy espoused in \cite{Papa95}, that theoretical developments and massive amounts of computational work must proceed in parallel, we apply some of the theoretical concepts to concrete datasets of short peptide fragments. In that way we both demonstrate important theoretical and practical techniques and obtain an efficient indexing scheme which can be used to answer biological questions.

\section{Quasi-metrics}

One of the fundamental concepts of modern mathematics is the notion of a \emph{metric space}: a set together with a distance function which separates points (i.e. the distance between two points $0$ if and only if they are identical), is symmetric and satisfies the \emph{triangle inequality}. The theory of metric spaces is very well developed and provides the foundation of many branches of mathematics such as geometry, analysis and topology as well as more applied areas. In many practical applications, it is to a great advantage if the distance function is a metric and this is often achived by symmetrising or otherwise manipulating other distance functions.

A \emph{quasi-metric} is a distance function which satisfies the triangle inequality but is not symmetric. There are two versions of the separation axiom: either it remains the same as in the case of metric, that is, for a distance between two points to be $0$ they must be the same, or, it is allowed that one distance between two different points be $0$ but not both. In all cases the distance between two identical points has to be $0$. Hence, for any pair of points in a quasi-metric space there are two distances which need not be the same. Quasi-metrics were first introduced in 1930s \cite{W31} and are a subject of intensive research in the context of topology and theoretical computer science \cite{Ku01}.

While much of the results from the theory of metric spaces transfer directly to the quasi-metric case, there are some concepts which are unique to the quasi-metrics, the most important being the concept of \emph{duality}. Every quasi-metric has its \emph{conjugate} quasi-metric which is obtained by reversing the order of each pair of points before computing the distance. Existence of two quasi-metrics, the original one and its conjugate leads to other dual structures depending on which quasi-metric is used: balls, neighbourhoods, contractive functions etc. We distinguish them by calling the structures obtained using the original quasi-metric the \emph{left} structures while the structures obtained using the conjugate quasi-metric are called the \emph{right} structures. The \emph{join} or symmetrisation of the left and right structures produces a corresponding metric structure.

Another important concept which has no metric counterpart is that of an associated partial order. Every quasi-metric space can be associated with a partial order and every partial order can be shown to arise from a quasi-metric. Hence, quasi-metrics are not only generalised metrics, but also generalised partial orders. This fact has been important for the theoretical computer science applications and also has significance in the context of sequence based biology.

While the topological properties of quasi-metric and related structures have been extensively investigated \cite{Ku01}, much less is known about the geometric aspects. We therefore aim to extend the concepts from the asymptotic geometric analysis to quasi-metric spaces in order to have results analogous to those involving metric spaces as well as to investigate the phenomena specific to the asymmetric case. Such results can then be applied to the theory of indexing for similarity search and its applications to sequence based biology.

\section{Overview of the Chapters}

Chapter \ref{ch:1} introduces quasi-metric spaces and related concepts. The emphasis is on the notions used in the subsequent chapters as well as on examples. In the last section, we construct examples of universal quasi-metric spaces of some classes. A universal quasi-metric space of a given class contains a copy of every quasi-metric space of that class and satisfies in addition the \emph{ultrahomogeneity} property. This notion is a generalisation of a well known concept of a universal metric space first constructed by Urysohn \cite{Ur27}. While there are no direct applications of universal quasi-metric spaces in this thesis, our construction serves two purposes: it provides examples of quasi-metric spaces not previously known and sets the foundations for possible further research mirroring the investigations \cite{Usp98,Ver02,Pe02a} relating to the universal metric spaces and their groups of isometries.

Chapter \ref{ch:bioseq_qm} explores in detail the connections between biological sequence similarities and quasi-metrics. The main result is the Theorem \ref{thm:locsim2qm} which shows that local similarity measures on biological sequences can be, under some assumptions frequently fullfilled in the real applications, naturally converted into equivalent quasi-metrics. While it was long known that global similarities can be converted to metrics or quasi-metrics, it was believed \cite{SWF81} that no such conversion exists for the local case, at least with respect to metrics. 

Chapter \ref{ch:2} introduces the central mathematical object of this study: the quasi-metric space with measure, or \emph{pq-space}. This is a generalisation of a metric space with measure or an \emph{mm-space} which provides the framework for study of the phenomenon of concentration of measure on high dimensional structures. We extend these concepts to pq-spaces and point out the similarities and differences to the metric case. In particular we study the interplay between asymmetry and concentration -- the Theorem \ref{thm:qpclosemm} indicates that `a high dimensional quasi-metric space is close to being a metric space'. The results from Chapter \ref{ch:2} as well as an alternative formulation of the main results from Chapter \ref{ch:bioseq_qm} are published in a paper to appear in Topology Proceedings \cite{AS2004}.

Chapter \ref{ch:3}, partially based on the joint preprint with Pestov \cite{PeSt02}, is dedicated to applications of the mathematical concepts and results of previous chapters to indexing for similarity search. We extend, among others, the concepts of \emph{workload} and  \emph{indexing scheme} first introduced by Hellerstein, Koutsoupias and Papadimitriou \cite{H-K-P} in order to make them more suitable for analysis of similarity search and apply them to numerous existing published examples. We only consider \emph{consistent} indexing schemes -- those that are guaranteed to always retrieve all query results. Most existing indexing schemes for similarity search can only be applied to metric workloads and while quasi-metrics are mentioned in the literature (e.g. in \cite{CiPa02}), no general quasi-metric indexing scheme exists. We therefore introduced a concept of a \emph{quasi-metric tree} and dedicated a separate section to it. Chapter \ref{ch:3} also contains a proposal for a general framework for analysis of indexing schemes and an application of the concepts developed in Chapter \ref{ch:2} to the analysis of performance of range queries. 

Chapter \ref{ch:4}, building on a second joint preprint with Pestov \cite{StPe03}, examines some aspects of geometry of workloads over datasets of short peptide fragments and introduces FSIndex, an indexing scheme for such workloads. FSIndex is based on partitioning of amino acid alphabet and combinatorial generation of neighbouring fragments. Experimental results provide an illustration of many concepts from Chapter \ref{ch:3} and show that FSIndex strongly outperformes some established indexing schemes while not using significantly more space. It also has an advantage that a single instance of FSIndex can be used for searches using multiple similarity measures.

Chapter \ref{ch:5} introduces the prototype of the \emph{PFMFind} method for identifying potential short motifs within protein sequences that uses FSIndex to query datasets of protein fragments. Preliminary experimental evaluations, involving six selected protein sequences, show that PFMFind is capable of finding highly conserved and functionally important domains but needs improvemement with respect to fragments having unusual amino acid compositions.   

Appendix \ref{app:distexp} presents previously unpublished results on estimation of dimension of datasets that the thesis author obtained as a summer student at the Australian National University in summer 1999/2000. It takes the concept of \emph{distance exponent} introduced by Traina \textit{et al.} \cite{TrainaTF99} and provides it with more rigourous foundations. Several computational techniques for computing distance exponent are proposed and tested on artificially generated datasets. The best performing method is applied in Chapter \ref{ch:4} to estimate the dimensions of two datasets of short peptide fragments.

\chapter{Quasi-metric Spaces}\label{ch:1}\index{Quasi-metric!space|(} \index{Quasi-metric|(}

In this chapter we introduce the concept of a quasi-metric space\index{Quasi-metric!space} with related notions. A quasi-metric\index{Quasi-metric} can be thought of as an ``asymmetric metric''\index{Metric}; indeed by removing the symmetry\index{Symmetry axiom} axiom from the definition of metric\index{Metric} one obtains a quasi-metric\index{Quasi-metric}. However, we shall adopt a more general definition which has the advantage of naturally inducing a partial order\index{Partial order}. Thus, a notion of a quasi-metric\index{Quasi-metric} generalises both distances\index{Distance} and partial orders\index{Partial order}. 

There is substantial amount of publications about topological\index{Structure!topological} and uniform\index{Structure!uniform} structures related to quasi-metric spaces\index{Quasi-metric!space} -- the major review by K{\"u}nzi \cite{Ku01} contains 589 references. In contrast, there is a relative scarcity of works on geometric and analytic aspects which is partially being addressed by the recent papers on quasi-normed \index{Quasi-normed space} and biBanach\index{biBanach space} spaces \cite{GRRoSP01,GRRoSP02,RpSPVa03,GaRoSP03,GRRoSP03a}. While most known applications of quasi-metrics\index{Quasi-metric} come from theoretical computer science\index{Theoretical computer science}, the aim for this thesis is to show that there is a fundamental connection to sequence based biology\index{Sequence based biology}.

Duality\index{Duality} is a very important phenomenon often associated with asymmetric\index{Structure!asymmetric} structures. The topological aspects of duality are investigated in great detail in the paper by Kopperman\index{Kopperman, R.} \cite{Kop95}. In the case of quasi-metrics\index{Quasi-metric}, duality\index{Duality} is manifested by having two structures, which we call left\index{Structure!left} and right\index{Structure!right}, associated with notions generalised from metric spaces\index{Metric!space}. The symmetrisation\index{symmetrisation} (or a `join') of these two structures corresponds to a metric structure\index{Structure!metric}.  

The present chapter consists mostly of the review of the literature and basic concepts illustrated by examples. Our main new contribution is contained in Section \ref{sec:universal_qm}, which introduces universal quasi-metric spaces\index{Quasi-metric!space!universal} analogous to the Urysohn\index{Urysohn space|see{Metric space, universal}} universal metric spaces\index{Metric!space!universal} first introduced by Urysohn\index{Urysohn, P.} \cite{Ur27}. 

\section{Basic Definitions}

\begin{defin}
Let $X$ be a set. Consider a mapping $d:X\times X\to\R_+$ and the following axioms for all $x,y,z\in X$:
\begin{enumerate}[(i)]
\item $d(x,x) = 0$.
\item $d(x,z) \leq d(x,y)+ d(y,z)$.
\item $d(x,y)=d(y,x)=0\implies x=y$.
\item $d(x,y)=d(y,x)$.
\end{enumerate}
The axiom (ii) is known as the \emph{triangle inequality}\index{Triangle inequality|textbf}, the axiom (iii) is called the \emph{separation axiom}\index{Separation axiom|textbf} and the axiom (iv) is called the \emph{symmetry axiom}\index{Symmetry axiom|textbf}.

A function $d$ satisfying axioms (i),(ii) and (iii) is called a \emph{Quasi-metric}\index{Quasi-metric|textbf} and if it also satisfies (iv) it is a \emph{metric}\index{Metric|textbf}.  A pair $(X,d)$, where $X$ is a set and $d$ a (quasi-) metric, is called a (quasi-) metric space\index{Quasi-metric!space|textbf} \index{Metric!space|textbf}. 

For a quasi-metric $d$, its \emph{conjugate}\index{Quasi-metric!conjugate|textbf} (or \emph{dual}\index{Quasi-metric!dual|see{Quasi-metric, conjugate}}) quasi-metric $\cj{d}$ is defined for all $x,y \in X$ by \[\cj{d}(x,y) = d(y,x),\] and its \emph{associated metric}\index{Quasi-metric!associated metric|textbf} $\qam{d}$ by \[\qam{d}(x,y)=\max\{d(x,y), d(y,x)\}.\]  The associated metric\index{Quasi-metric!associated metric|textbf} is is the smallest metric\index{Metric} majorising $d$. 
\end{defin}

A quasi-metric\index{Quasi-metric} $d$ is a metric if and only if it coincides with its conjugate quasi-metric\index{Quasi-metric!conjugate}.

\begin{remark}
\label{rem:dist}
A function satisfying axioms (i),(ii) above but not necessarily satisfying the separation axiom (axiom (iii)) is called a \emph{pseudo-quasi-metric}\index{Pseudo-quasi-metric|textbf} and if it also satisfies the axiom (iv) it is called a \emph{pseudo-metric}\index{Pseudo-metric|textbf}. We use the generic term \emph{distance}\index{Distance|textbf} to denote any of the pseudo-quasi-metrics.
  
If a distance is allowed to take values in $\R_+ \cup\{\infty\}$ (the extended half-reals\index{Extended half-real line|textbf}), it is called an \emph{extended distance}\index{Distance!extended|textbf}\index{Extended distance|see{Distance, extended}} depending on the other axioms satisfied (e.g. extended pseudo-quasi-metric).
\end{remark}

Another often used symmetrisation of a quasi-metric\index{Quasi-metric} is the `sum' metric\index{Quasi-metric!sum metric|textbf} \index{Metric!sum|see{Quasi-metric, sum metric}} $\qsum{d}$ where for each $x,y \in X$ \[\qsum{d}(x,y) = d(x,y)+d(y,x).\]

We now summarise some standard notation.
\begin{defin}
Let $(X,d)$ be a quasi-metric space, $x\in X$, $A, B \subseteq X$ and
$\e > 0$. Denote by
\begin{description}
\item[$\bullet$] $\diam(A) := \sup \{d(x,y): \ x,y \in A\}$, the \emph{diameter}\index{Diameter}
of set $A$;
\item[$\bullet$] $\lball{x}{\e} := \{ y \in X: \ d(x,y) < \e\}$, the \emph{left open ball}\index{Open ball!left!|textbf} of radius $\e$ centred at $x$;
\item[$\bullet$] $\rball{x}{\e} := \{ y \in X: \ d(y,x) < \e\}$, the \emph{right open ball}\index{Open ball!right!|textbf} of radius $\e$ centred at $x$;
\item[$\bullet$] $\ball{x}{\e} := \{y \in X: \ \qam{d}(x,y) < \e\}$, the \emph{associated metric open ball}\index{Open ball!associated metric|textbf} of radius $\e$ centred at $x$;
\item [$\bullet$]$d(x,A) := \inf \{d(x,y): \ y \in A\}$, the \emph{left distance}\index{Distance!to a set!left|textbf} from $x$ to
$A$;
\item[$\bullet$] $d(A,x) := \inf \{d(y,x): \ y \in A\}$, the \emph{right distance}\index{Distance!to a set!right|textbf} from $x$ to $A$;
\item[$\bullet$] $\qam{d}(A,x) := \inf \{\qam{d}(x,y): \ y \in A\}$, the \emph{associated metric distance}\index{Distance!to a set!associated metric|textbf} from $x$ to $A$;
\item [$\bullet$]$\lnbhd{A}{\e} := \{x\in X: \ d(A,x) < \e\}$, the \emph{left $\e$-neighbourhood}\index{Neighbourhood!of a set!left|textbf} of $A$;
\item[$\bullet$] $\rnbhd{A}{\e} := \{x\in X: \ d(x,A) < \e\}$, the \emph{right $\e$-neighbourhood}\index{Neighbourhood!of a set!right|textbf} of $A$;
\item[$\bullet$] $\nbhd{A}{\e} := \{x\in X: \ \qam{d}(A,x) < \e\}$, the \emph{associated metric $\e$-neighbourhood}\index{Neighbourhood!of a set!associated metric|textbf} of $A$.
\item[$\bullet$] $d(A,B) := \inf\{d(x,y): \ x \in A, \ y \in B\}$, the distance between $A$ and $B$\index{Distance!between sets|textbf}.
\end{description}
\end{defin}

The left balls\index{Open ball!left!} , distances\index{Distance!to a set!left}, and neighbourhoods\index{Neighbourhood!of a set!left} coincide with the right\index{Open ball!right}\index{Distance!to a set!right} \index{Neighbourhood!of a set!right} versions in the case of metric spaces. 

\begin{remark}
Our notation in some cases slightly differs from that adopted in the literature. We use $\qam{d}$ to denote the associated metric \index{Quasi-metric!associated metric}  (and later the norm associated to a quasi-norm) in order to avoid any confusion that can arise from the more usual symbols $d^s$ or $d^S$. Also note that we denote the open balls\index{Open ball} by $\mathfrak{B}$ while we shall use $\mathcal{B}$ to denote a Borel $\sigma$-algebra\index{Borel!sigma@$\sigma$-algebra} of measurable sets and $\mathscr{B}$ to denote the set of blocks of an indexing scheme\index{Blocks of an indexing scheme}\index{Indexing scheme}. The notation $\qsum{d}$ is our own -- `u' is the second letter of the word `sum' and `s' was already used.
\end{remark}
\begin{remark}
We shall often (but not always) use $x\vee y$ to denote $\max\{x,y\}$ and $x\wedge y$ to denote $\min\{x,y\}$.
\end{remark}

The following result generalises the triangle inequality\index{Triangle inequality} to the distances from points to sets\index{Distance!to a set}.
\begin{lemma}
\label{lemma:set_tr_eq}
Let $(X,d)$ be a pseudo-quasi-metric space\index{Pseudo-quasi-metric!space}. Then for all $x,y\in X$ and $A\subset X$, \[d(x,A)\leq d(x,y) + d(y,A).\] 
\begin{proof}
By the triangle inequality, for all $z\in A$, $d(x,z) \leq d(x,y)+d(y,z)$. Taking infimum over all $z\in A$ of both sides of the inequality produces the desired result.
\end{proof}
\end{lemma}

\begin{defin}
Let $(X,d_X)$ and $(Y,d_Y)$ be two quasi-metric spaces. A map $\varphi:X\to Y$ is called a (\emph{quasi-metric}) \emph{isometry}\index{Quasi-metric!isometry|textbf} if $\varphi$ is a bijection and for all $x,y\in X$, \[d_Y(\varphi(x),\varphi(y)) = d_X(x,y).\]
\end{defin}

\begin{lemma}\label{lemma:metr_isometry}
Let $\varphi:X\to Y$ be an isometry between quasi-metric spaces\index{Isometry!between quasi-metric spaces} $(X,d_X)$ and $(Y,d_Y)$. Then $\varphi$ is also an isometry between metric spaces\index{Isometry!between metric spaces} $(X,d^\mathfrak{s}_X)$ and $(Y,d^\mathfrak{s}_Y)$. \qed
\end{lemma}

\section{Topologies and quasi-uniformities}\index{Quasi-metric!topology}

Each quasi-metric\index{Quasi-metric} $d$ naturally induces a topology\index{Topology} $\Top(d)$ whose base\index{Base!of topology}\index{Topology!base} consists of all open left balls\index{Open ball!left} $\lball{x}{\e}$, centred at any $x\in X$, of radius $\e>0$. This is a base indeed. Take any $x,y\in X$ and $\e,\delta>0$ such that $\lball{x}{\e}\cap\lball{y}{\delta}\neq \emptyset$. For any $z\in\lball{x}{\e}\cap\lball{y}{\delta}$ set $\zeta=\min\{\e-d(x,z),\delta - d(y,z)\}$ and observe that $\lball{z}{\zeta}\subseteq \lball{x}{\e}\cap\lball{x}{\delta}$.  

\begin{figure}[h!]
\begin{center}
  \input{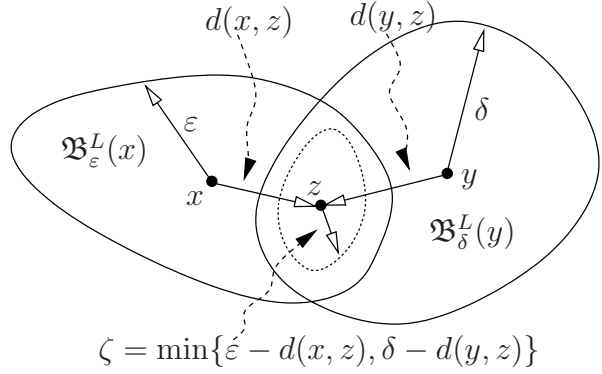}
  \caption{Left open balls form a base for a quasi-metric topology.}
  \label{fig:qmbase}
\end{center}
\end{figure}

Thus, a set $U$ is open\index{Open ball} if for each $x\in U$ there is an $\e>0$ such that $\lball{x}{\e}\subseteq U$.\index{Open ball!left} The topology $\Top(\cj{d})$ is defined in similar way: its base consists of all open right balls $\rball{x}{\e}$\index{Open ball!right} of radius $\e>0$. Hence, one can naturally associate a \emph{bitopological space}\index{Bitopological space} $(X,\Top(d),\Top(\cj{d}))$ to a quasi-metric space $(X,d)$. The relationships between quasi-metric and bitopological spaces are well researched \cite{Ku01}. 

\begin{defin}
A topological space is \emph{quasi-metrisable}\index{Topology!quasi-metrisable|textbf} \index{Quasi-metrisable!topology|see{Topology, quasi-metrisable}} if there exists a quasi-metric $d$ such that $\Top=\Top(d)$.
\end{defin}

\begin{remark}\label{rem:qam_base}
Note that for any quasi-metric space $(X,d)$, $\ball{x}{\e}=\lball{x}{\e} \cap \rball{x}{\e}$ and hence the base of the metric topology $\Top(\qam{d})$ consists exactly of intersections of left and right open balls of the same radius, centred at any point. Therefore, $\Top(\qam{d})$ is the supremum of $\Top(d)$ and $\Top(\cj{d})$: \[\Top(\qam{d})=\Top(d)\vee \Top(\cj{d}).\]
\end{remark}

Not every topology\index{Topology} is induced by a quasi-metric\index{Quasi-metric}, however Kopperman\index{Kopperman, R.} \cite{Kop88} showed that every topology on a space $X$ is generated by a \emph{continuity function}\index{Continuity function}; that is, an analogue of a quasi-metric\index{Quasi-metric} which takes values in a semigroup\index{Semigroup} of a special kind called a \emph{value semigroup}\index{Value semigroup}. The question of which topologies are quasi-metrisable\index{Topology!quasi-metrisable} (i.e. can be induced from a quasi-metric\index{Quasi-metric}) has been long open. We mention the characterisations by Kopperman \cite{Kop93} in terms of bitopological spaces and by Vitolo\index{Vitolo, P.} \cite{Vi95} (see Corollary \ref{corol:Vitolo_embed}) in terms of hyperspaces\index{Hyperspace} of metric spaces\index{Metric!space}.  

The topology $\Top(d)$ induced by a quasi-metric\index{Quasi-metric!topology} $d$ clearly satisfies the $T_0$ separation axiom\index{Separation axiom!T0@$T_0$}. The induced topology is $T_1$\index{Separation axiom!T1@$T_1$} if and only if $d$ also satisfies the property $d(x,y)=0 \implies x=y$ for all $x,y\in X$. Often in the literature, the $T_0$ quasi-metric\index{Quasi-metric!T0@$T_0$} is called the \emph{pseudo-quasi-metric}\index{Pseudo-quasi-metric} while the name \emph{quasi-metric} is reserved only for the $T_1$\index{Quasi-metric!T1@$T_1$} case \cite{DePa00, Ku01}. The definition presented here is also widely used \cite{RoSa00,Vi99} and comes mostly from computer science applications where the association with partial orders justifies consideration of the $T_0$ quasi-metrics. Partial orders also arise naturally in the context of biological sequences\index{Biological sequence} which are the main objects of study of this thesis. 

\begin{defin}
A \emph{partial order}\index{Partial order|textbf} on a set $X$ is a binary relation\index{Binary relation} $\leq\subseteq X\times X$ which is reflexive\index{Reflexive binary relation|textbf}, antisymmetric\index{Antisymmetric binary relation|textbf} and transitive\index{Transitive binary relation|textbf}, that is,
\begin{enumerate}[(i)]
\item for all $x\in X$, $x\leq x$.
\item for all $x,y\in X$, $x\leq y\ \wedge\ y\leq x\ \implies x=y$.
\item for all $x,y,z\in X$, $x\leq y\ \wedge\ y\leq z\ \implies x\leq z$.
\end{enumerate}
\end{defin}

\begin{defin}
Let $(X,d)$ be a quasi-metric space. The \emph{associated partial order}\index{Quasi-metric!associated partial order|textbf} $\leq_d$ is defined by \[ x\leq_d y\iff d(x,y) = 0.\]
\end{defin}

It is easy to see that $\leq_d$ is indeed a partial order\index{Partial order} and hence one can associate a partial order\index{Quasi-metric!associated partial order} to every quasi-metric\index{Quasi-metric}. The converse is also true.
 
\begin{example}[\cite{KuVa94}]
\label{ex:Alexandroff_topology}
Let $(X,\leq) $ be a partially ordered set\index{Partialy ordered set} and for any $x,y\in X$, set $d(x,y)=0$ if $x\leq y$ and $d(x,y) = 1$ otherwise. It is clear that $d$ is a quasi-metric\index{Quasi-metric} and that $\leq_d$ coincides with $\leq$. The topology $\Top(d)$ induced by $d$ is called the \emph{Alexandroff topology}\index{Alexandroff topology}. The metric associated to $d$ is the discrete, that is $\{0,1\}$-valued, metric\index{Metric!discrete} (c.f. the Example \ref{ex:discrete_metric} below).
\end{example}

Quasi-metrics\index{Quasi-metric} also generate the so-called \emph{quasi-uniformities}\index{Quasi-uniformity|textbf} which are uniformities\index{Uniformity} but for the lack of symmetry \cite{FlLi82}. More formally, a \emph{quasi-uniformity} $\Unif$ on a set $X$ is a non-empty collection of subsets of $X\times X$, called \emph{entourages (of the diagonal)}\index{Entourage|textbf}, satisfying
\begin{enumerate}
\item Every subset of $ X\times X$ containing a set of $\Unif$ belongs to $\Unif$;
\item Every finite intersection of sets of $\Unif$ belongs to $\Unif$;
\item Every set in $\Unif$ contains the diagonal (the set $\{(x,x) \ |\ x\in X\}$);
\item If $U$ belongs to $\Unif$, then exists $V$ in $\Unif$ such that, whenever $(x,y),\ (y,z)\in V$, then $(x,z)\in U$.
\end{enumerate}

Axioms 1 and 2 mean that $\Unif$ is a \emph{filter}\index{Filter}. Any collection $\B$ of entourages satisfying 3, 4 and which is a \emph{prefilter}\index{Prefilter} (that is, for each $A,B\in\B$ there is a $C\in\B$ with $C\subseteq A\cap B$) generates a quasi-uniformity $\Unif$ which is the smallest filter on $X\times X$ containing $\B$. In this case, $\B$ is called a \emph{basis} of $\Unif$.

\begin{defin}
A pair of the form $(X,\Unif)$ where $X$ is a set and $\Unif$ is quasi-uniformity on $X$ is called a \emph{quasi-uniform} space\index{Quasi-uniform space|textbf}.
\end{defin}

Let $(X,\Unif)$ and $(Y,\mathcal{V})$ be quasi-uniform spaces. A function $f:X\to Y$ is called \emph{quasi-uniformly continuous}\index{Function!continuous!quasi-uniformly|textbf}\index{Quasi-uniformly continuous!see{Function, continuous, quasi-uniformly}} iff for each $V\in\mathcal{V}$, $f^{-1}(V)\in\Unif$. This exactly mirrors the notion of uniformly continuous\index{Function!continuous!uniformly} function between uniform spaces.

Let $(X,d)$ be a quasi-metric space. Denote by $N_r=\{(x,y)\ | \ d(x,y)\leq r\}$ the entourage\index{Entourage} of radius $r>0$. The \emph{quasi-metric quasi-uniformity}\index{Quasi-metric!quasi-uniformity}\index{Quasi-uniformity} $\Unif$ on $X$ has as a base\index{Quasi-uniformity!base} the set all entourages of radius $r>0$, that is, $U\in\Unif\iff\exists r\in\R_+: N_r\subseteq U$. The dual (conjugate) quasi-uniformity $\cj\Unif$ is generated by the entourages $\cj{N}_r=\{(x,y)\ | \ d(y,x)\leq r\}$ and the symmetrisation $\qam{\Unif} =\Unif\vee\cj\Unif$ produces a uniformity. It is easy to see that for any quasi-metric\index{Quasi-metric}, the uniformity $\qam{\Unif}$ is equivalent to the uniformity generated by the associated metric $\qam{d}$.

We now recall parts of the basic theory of completions of quasi-metric spaces. All statements are particular cases of corresponding statements for quasi-uniformities. 

Recall that a sequence $x_1,x_2,\ldots$ of points in a metric space $(X,\rho)$ is \emph{Cauchy} if for every $\e>0$ there exists $N\in\N$ such that for all $i,j> N$, $\rho(x_i,x_j)<\e$. A metric space $(X,\rho)$ is \emph{complete} if every Cauchy sequence is convergent in $X$.
\begin{defin}
A quasi-metric space $(X,d)$ is called \emph{bicomplete} if the associated metric space $(X,\qam{d})$ is complete.   
\end{defin}

The theory of bicomplete quasi-uniformities was developed in \cite{Csa60} and \cite{LiFl78}. It is well known that every quasi-metric space $(X,d)$ has a unique (up to a quasi-metric isometry) bicompletion $(\tilde{X},\tilde{d})$ such that $(\tilde{X},\tilde{d})$ is a bicomplete extension of $(X,d)$ in which $(X,d)$ is $\Top(\tilde{d})$-dense. The associated metrics $\qam{(\tilde{d})}$ and $\tilde{\qam{d}}$ coincide so $(X,d)$ is also $\Top(\tilde{\qam{d}})$-dense in $\tilde{X}$. Furthermore, if $D$ is a $\Top(\tilde{d})$-dense subspace of a quasi-metric space $(X,d)$ and $f:(D,d\vert_D)\to (Y,\rho)$ is a quasi-uniformly continuous map where $(Y,\rho)$ is a bicomplete quasi-metric space, then there exists a (unique) quasi-uniformly continuous extension $\tilde{f}:\tilde{X}\to Y$ of $f$.

Apart from the above definition there are in existence more restricted notions of completeness of quasi-metric and quasi-uniform spaces developed by Doitchinov \cite{Doi88,Doi91,Doi91a}, which we will not use in this work.

We now present some well-known examples of quasi-metric spaces.

\begin{example}\label{ex:discrete_metric}
Let $X$ be any set and set $d: X\times X\to\R$ by:
\[ d(x,y) =
        \begin{cases}
        0, & \text{if} \ x = y \\ 1, & \text{if} \ x \neq y.
        \end{cases}
\]
It can be easily checked that $d$ is a metric and such metric is called the \emph{discrete} metric. The topology induced by $d$ is discrete: every singleton is open.
\end{example}

Next we define the quasi-metrics on $\R$ generating the so-called upper and lower topology. 
\begin{defin}\label{defn:uLuR}
The \emph{left quasi-metric} $u^L: \R\times\R\to\R_+$ is given by
\[ u^L(x,y) = \max\{x-y,\ 0\}.\]
Similarly, define the \emph{right quasi-metric} $u^R: \R\times\R\to\R_+$ by
\[ u^R(x,y) = \max\{y-x,\ 0\}.\]
\end{defin}
It is trivial to show that $u^L$ and $u^R$ are quasi-metrics which are conjugate to each other.
The associated metric $u=\max\{u^L, u^R\}$ is the canonical absolute value metric on $\R$ given by $u(x,y) = \abs{x-y}$. The base for the left topology $\Top(u^L)$ consists of all sets of the form $(\xi,\infty)$ and the base for the right topology $\Top(u^R)$ of all sets of the form $(-\infty,\xi)$, where $\xi\in\R$. Hence $\Top(u^L)$ and $\Top(u^R)$ are $T_0$ but not $T_1$ separated. The partial order associated with $u^L$ (in this case a linear order) is the usual order on reals, while $u^R$ induces the reverse order.

For any topological space $(X,\Top)$, a continuous function $(X,\Top)\to (\R,u^L)$ is often called \emph{lower semicontinuous} and a continuous function $(X,\Top)\to (\R,u^R)$ is \emph{upper semi-continuous}. In accordance with this terminology, $\Top(u^L)$ is often called the \emph{topology of lower semicontinuity} on reals while $\Top(u^R)$ is called the \emph{topology of upper semicontinuity}.

\begin{remark}
It is worth noting that for any quasi-metric space $(X,d)$, the quasi-metric $d$, taken as a function $X\times X\to\R$ is lower semicontinuous with respect to the product topology $\Top(\cj{d})\times\Top(d)$ and upper semicontinuous with respect to the product topology $\Top(d)\times\Top(\cj{d})$. Indeed, let $U=\{(x,y):\ d(x,y)<\delta\}$ and let $V=\{(x,y):\ d(x,y)>\delta\}$. One can show using the triangle inequality that \[U= \bigcup_{(x,y)\in U} \left(\rball{(x,y)}{\frac12 \left(d(x,y)-\delta\right)}\times\lball{(x,y)}{\frac12 \left(d(x,y)-\delta\right)} \right),\] and \[V= \bigcup_{(x,y)\in V} \left(\lball{(x,y)}{\frac12 \left(\delta-d(x,y)\right)}\times\rball{(x,y)}{\frac12 \left(\delta-d(x,y)\right)} \right),\] and hence $U$ is open in $\Top(\cj{d})\times\Top(d)$ and $V$ is open in $\Top(d)\times\Top(\cj{d})$. However, $d$ is not in general lower or upper semicontinuous with respect to the product topologies $\Top(d)\times\Top(d)$ or $\Top(\cj{d})\times\Top(\cj{d})$. For the counter example, set $d=u^L$ and consider neighbourhoods of $(0,0)$.
\end{remark}

\begin{example}[\cite{KuVa94,DePa00}]
\label{ex:sorgenfrey_line}
Another quasi-metric on $\R_+$ is given by
\[ d(x,y) =
        \begin{cases}
        \min(1,y-x), & \text{if} \ x\leq y \\ 1, & \text{otherwise.}
        \end{cases}
\]
In this case $d$ induces a $T_1$ topology $\Top$ on $\R$ whose base consists of all left balls centred at $x\in\R$ of the form $\lball{x}{r}=[x,x+r)$, where $0<r<1$ (for any $x\in\R$, and $r\geq 1$, $\lball{x}{r}=\R$). The topological space $(\R,\Top)$ is called the \emph{Sorgenfrey line}, a well known object in topology and a source of many counter-examples. The associated metric $\qam{d}$ is the discrete metric. 
\end{example}

Any unbounded quasi-metric can be converted to a bounded quasi-metric while preserving the topology in the following way.
\begin{example}
Let $(X,d)$ be an extended quasi-metric space. Then $\rho: X\times X\to\R_+$ defined by \[\rho(x,y)=\min\{1, d(x,y)\},\] is a quasi-metric such that $\Top(\rho) = \Top(d)$. The proof of quasi-metric axioms is trivial and the fact that topologies coincide follows from the fact that all open balls of radius not greater than $1$ coincide.
\end{example}

\begin{defin}\label{def:setsofsubsets}
Let $(X,\Top)$ be a topological space. Denote by 
\begin{description}
\item[$\bullet$] $\PowE{X}$,\quad the set of all subsets of $X$;
\item[$\bullet$] $\Pow{X}$,\quad the set of all non-empty subsets of $X$;
\item[$\bullet$] $\PowF{X}$,\quad the set of all finite subsets of $X$;
\item[$\bullet$] $\KompE{X}{\Top}$,\quad the set of all compact subsets of $X$;
\item[$\bullet$] $\Komp{X}{\Top}$,\quad the set of all non-empty compact subsets of $X$;
\item[$\bullet$] $\ClsdE{X}{\Top}$,\quad the set of all closed subsets of $X$;
\item[$\bullet$] $\Clsd{X}{\Top}$,\quad the set of all non-empty closed subsets of $X$.
\end{description}
If the topology $\Top$ is generated by a quasi-metric $d$ we will often replace $\Top$ in the above expressions by $d$, for example obtaining $\KompE{X}{d}$ for the set of all compact subsets of $X$.

The set $\PowE{X}$ (or restrictions as above) with some (topological) structure is often called a \emph{hyperspace}.
\end{defin}

\begin{example}[\cite{DePa00}]
\label{ex:set_difference}
Let $X$ be a set and let $\mathcal{N}=\PowF{X}$. Define $\rho: \mathcal{N}\times\mathcal{N}\to\R$ by $\rho(A,B)=\abs{A\setminus B} = \abs{A}-\abs{A\cap B}$.

It is easy to see that $A\subseteq B\iff \rho(A,B)=0$. The triangle inequality can be verified by noting that $A\setminus C =(A\setminus (B\cup C))\cup ((A\cap B)\setminus C) \subseteq (A\setminus B) \cup (B\setminus C)$ and hence $\rho$ is a quasi-metric with the associated order corresponding to the set inclusion. The symmetrisation $\qsum{\rho}(A,B)= \abs{A\bigtriangleup B} = \abs{A}+\abs{B}-2\abs{A\cap B}$ produces the well-known symmetric difference metric.
\begin{figure}[!ht] \label{fig:setdiff}
\begin{center}
  \input{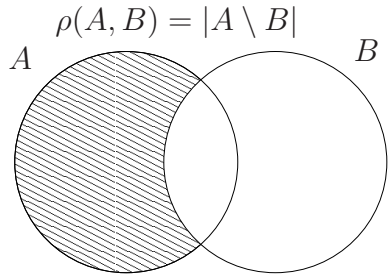}
  \caption{Set difference quasi-metric.}
\end{center}
\end{figure}
\end{example}

\begin{example}
More generally, let $(X,\Sigma,\mu)$ be a measure space and $\mathcal{N}=\Sigma_{\text{fin}}/\mu$, the set of equivalence classes of measurable subsets of finite measure, that is, for any $A,B\in\Sigma$ such that $\mu(A)<\infty$ and $\mu(B)<\infty$, $A\sim B\iff \mu(A\setminus B)=\mu(B\setminus A)=0$. Then, by the same argument as above, the function $\rho: \mathcal{N}\times\mathcal{N}\to\R$ where $\rho(A,B)=\mu(A\setminus B)$, is a $T_0$ quasi-metric.
\end{example}

\begin{example}
\label{ex:prodqm}
Let $(X_i, d_i)$, $i=1,2\ldots n$ be quasi-metric spaces and suppose $X = X_1 \times X_2 \ldots \times X_n$, that is, for each $x\in X$, $x=(x_1,x_2\ldots x_n)$, $x_i \in X_i$. Define $d:X\times X\to\R$ by
\[d(x,y) = \sum_{i=1}^n d_i(x_i, y_i).\]
Then it is easy to show that $(X,d)$ is a quasi-metric space. We will call the product spaces of this kind the \emph{$\ell_1$-type quasi-metric spaces}. They will feature extensively later on.
\end{example}

\begin{example}\label{ex:hammingdist}
Let $X$ be an $\ell_1$-type product space as above. The \emph{Hamming metric} is a metric obtained by setting each $d_i$ above to be the discrete metric. In other words, \[d(x,y)=\abs{\{i: x_i\neq y_i\}}.\]
\end{example}

\section{Quasi-normed Spaces}

Important examples of quasi-metrics are induced by quasi-norms, the asymmetric versions of norms. The research area of quasi-normed spaces has seen a significant development in recent years both in theory \cite{GRRoSP01,GRRoSP02,RpSPVa03,GaRoSP03,GRRoSP03a} and applications \cite{RoSa00, RoSch02a}. We survey here some of the main definitions and examples.

Recall that a \emph{semigroup} $(X,\star)$ is a set $X$ with a binary operation $\star$ satisfying
\begin{enumerate}
\item $\forall x,y\in X,\qquad x\star y \in X$\qquad (closure), 
\item $\forall x,y,z\in X,\qquad x\star (y\star z) = (x\star y)\star z$\qquad (associativity). 
\end{enumerate}
A \emph{monoid} or a \emph{semigroup with identity} is a semigroup $(X,\star)$ containing a unique element $e\in X$ (also called a \emph{neutral element}) such that $\forall x\in X$, $x\star e=e\star x=x$, and a \emph{group} $(X,\star)$ is a monoid where each element has an inverse, that is, $\forall x\in X$,$\exists x^{-1}\in X$: $x\star x^{-1} = x^{-1}\star x = e$. A \emph{homomorphism} from a semigroup $(X,\star)$ to a semigroup $(Y,\ast)$ is map $\phi:X\to Y$ such that $\forall x,y\in X$, $\phi(x)\ast\phi(y) = \phi(x\star y)$. An \emph{isomorphism} is a homomorphism which is a bijection such that its inverse is also a homomorphism.

\begin{defin}
A \emph{semilinear} (or \emph{semivector}) \emph{space} on $\R_+$ is a triple $(X, +, \cdot)$ such that $(X, +)$ is an Abelian semigroup with neutral element $0\in X$ and $\cdot$ is a function
$\R_+ \times X \to X$ which satisfies for all $x,y\in X$ and $a,b\in\R_+$:
\begin{enumerate}[(i)]
\item $a \cdot (b \cdot x) = (ab) \cdot x$,
\item $(a+b) \cdot x = (a \cdot x) + (b \cdot x)$,
\item $a \cdot (x+y) = (a \cdot x) + (a \cdot y)$, and
\item $1 \cdot x = x$.
\end{enumerate}
Whenever an element $x \in X$ admits an inverse it can be shown to be unique and is denoted $-x$. If we replace in the above definition $\R_+$ with $\R$ and ``semigroup'' with ``group'' we obtain an ordinary vector (or linear) space.
\end{defin}

\begin{defin}[\cite{RoSch02a}]
Let $(E,+,\cdot)$ be a linear space over $\R$ where $e$ is the neutral element of  $(E,+)$. A \emph{quasi-norm} on $E$ is a is a function $\norm{\cdot}:E\to\R_+$ such that for all $x,y\in E$ and $a\in\R_+$:
\begin{enumerate}[(i)]
\item $\norm{x}=\norm{-x}=0 \ \iff x=e$,
\item $\norm{a\cdot x} = a\norm{x}$, and
\item $\norm{x+y}\leq\norm{x}+\norm{y}$.
\end{enumerate}
The pair $(E,\norm{\cdot})$ is called a \emph{quasi-normed space}.
\end{defin}
It is easy to verify that the function $\qam{\norm{\cdot}}$ defined on $E$ by $\qam{\norm{x}} = \max\{\norm{x},\norm{-x}\}$ is a norm on $E$. 

The quasi-norm $\norm{\cdot}$ induces a quasi-metric $d_{\norm{\cdot}}$ in a natural way.
\begin{lemma}
Let $(E,\norm{\cdot})$ be a quasi-normed space. Then $d_{\norm{\cdot}}$ defined for all $x,y\in E$ by \[d_{\norm{\cdot}}(x,y)=\norm{y-x}\] is a quasi-metric whose conjugate $d^*_{\norm{\cdot}}$ is given by $d^*_{\norm{\cdot}}(x,y) = \norm{x-y}$.
\begin{proof}
Let $x,y,z\in E$. We have $d_{\norm{\cdot}}(x,x)=\norm{x-x}=\norm{e}=0$. Also if $d_{\norm{\cdot}}(x,y)=d_{\norm{\cdot}}(y,x)=0$ it follows by the first axiom that $\norm{y-x}=\norm{x-y}=0$ and hence $x-y=e$, that is $x=y$.

For the triangle inequality we have
\begin{align*}
d_{\norm{\cdot}}(x,y)+d_{\norm{\cdot}}(y,z) & = \norm{y-x} + \norm{z-y}\\
&\geq \norm{y-x+z-y}\\
&\geq \norm{z-x}\\
&= d_{\norm{\cdot}}(x,z) \quad \text{as required.}
\end{align*}
The statement about the conjugate is obvious.
\end{proof}
\end{lemma}

\begin{defin}[\cite{RoSch02a}]
A quasi-normed space $(E,\norm{\cdot})$ where the induced quasi-metric $d_{\norm{\cdot}}$ is bicomplete is called a \emph{biBanach} space.
\end{defin}

\begin{example}
\label{ex:half_line_qnorm}
A quasi-norm on $\R$ is given for all $x\in\R$ by $\norm{x}=\max\{x, \ 0\}$. It is easy to show that $u^R$ (Definition \ref{defn:uLuR}) is induced by the above quasi-norm.
\end{example}

\begin{example}[\cite{RoSch02a}]
\label{ex:B_star}
Let $(E, \norm{\cdot})$ be a quasi-normed space. Define \[\mathcal{B}_E^* = \{f:\N\to E\ |\ \sum_{n=1}^{\infty} 2^{-n}\qam{\norm{f(n)}}<\infty\}.\] The set $\mathcal{B}_E^*$ can be made into a linear space using standard addition and scalar multiplication of functions. Set the quasi norm for each $f\in \mathcal{B}_E^*$ by \[\norm{f}_{\mathcal{B}^*} = \sum_{n=1}^{\infty} 2^{-n}\norm{f(n)}.\] Then, the space $(\mathcal{B}_E^*, \norm{\cdot}_{\mathcal{B}^*})$ is a quasi-normed space and is a biBanach space if $E$ is a biBanach space.
\end{example}

We conclude this section by considering quasi-normed semilinear spaces and the dual complexity space.

\begin{defin}[\cite{RoSch02a}]
A \emph{quasi-normed semilinear space} is a pair $(F, \norm{\cdot}_F)$ such that $F$ is a non-empty subset of a quasi-normed space $(E,\norm{\cdot})$ with the properties that $(F, +|_F, \cdot|_F)$ is semilinear space on $\R_+$ and $\norm{\cdot}_F$ is a restriction of the quasi-norm $\norm{\cdot}$ to $F$.

The space $(F, \norm{\cdot}_F)$ is called a \emph{biBanach semilinear space} if $(E,\norm{\cdot})$ is a biBanach space and $F$ is closed in the Banach space $(E,\qam{\norm{\cdot}})$.
\end{defin}

The complexity space and its dual have been introduced and extensively studied in the papers by Schellekens \cite{Sch95} and Romaguera and  Schellekens \cite{RoSch98,RoSch02a} respectively, in order to study the complexity of programs. The example below presents the dual complexity space as an example of a quasi-normed semilinear space.

\begin{example}[\cite{RoSch02a}]
\label{ex:dual_complexity}
Let $(F, \norm{\cdot}_F)$ be a quasi-normed semilinear space where $F$ is a non-empty subset of a quasi-normed space $(E,\norm{\cdot})$. Let \[\mathcal{C}^*= \{f:\N\to F\ | \sum_{n=1}^{\infty} 2^{-n}\qam{\norm{f(n)}}<\infty\}.\] It is apparent that $\mathcal{C}^*$ is a semilinear space and that $\mathcal{C}^*\subset\mathcal{B}_E^*$ (Example \ref{ex:B_star}). Define for each $f\in\mathcal{C}^*$ \[\norm{f}_{\mathcal{C}^*} = \sum_{n=1}^{\infty} 2^{-n}\norm{f(n)}_F\] so that  $(\mathcal{C}^*,\norm{\cdot}_{\mathcal{C}^*})$ becomes a quasi-normed semilinear space. It associated quasi-metric space $(\mathcal{C}^*,d_{\norm{\cdot}_{\mathcal{C}^*}})$ is called the \emph{dual complexity space}.
\end{example}

Section \ref{sec:Lip_funcs} will present a further example of a quasi-normed semilinear space.

\section{Lipschitz Functions}
\label{sec:Lip_funcs}

While the quasi-metric spaces have been extensively studied from a topological point of view, the properties of the non-contracting maps between them, also called 1-Lipschitz functions, have not received the same attention. The only widely available reference solely on this topic is the paper by Romaguera and Sanchis \cite{RoSa00}. In this section we will define left- and right- Lipschitz maps, present a few basic results and examples, as well as survey some of the results by Romaguera and Sanchis. Lipschitz maps will be extensively used in subsequent chapters and new structures will be introduced where needed.

\begin{defin}\label{def:leftLip}
Let $(X,d)$ and $(Y,\rho)$ be quasi-metric spaces. A map $f: X\to Y$ is called \emph{left $K$-Lipschitz} if there exists $K\in\R_+$ such that for all $x,y \in X$ \[\rho(f(x), f(y)) \leq Kd(x,y).\]
The constant $K$ is called a \emph{left Lipschitz constant}.  Similarly, $f$ is \emph{right $K$-Lipschitz} if $\rho(f(y),f(x)) \leq Kd(x,y)$. 

Maps that are both left and right $K$-Lipschitz are called $K$-Lipschitz.
\end{defin}

Left-Lipschitz functions are commonly called \emph{semi-Lipschitz} \cite{RoSa00} but we use the above nomenclature in order to be consistent with the other ``one-sided'' (left- or right-) structures we introduced. Indeed, it is easy to note that every left $K$-Lipschitz map $(X,d)\to(Y,\rho)$ is right $K$-Lipschitz as a mapping $(X,\cj{d})\to(Y,\rho)$.

\begin{lemma}\label{lemma:Lip_cts}
Let $(X,d)$ and $(Y,\rho)$ be quasi-metric spaces and let $f:X\to Y$ be a left 1-Lipschitz map. Then $f$ is continuous with respect to the left topologies on both spaces.
\begin{proof}
Take any $\e>0$. We need to show that there is $\delta>0$ such that for any $y\in Y$ and $x\in X$, $f^{-1}(\lball{y}{\e})\supseteq\lball{x}{\delta}$. Pick $\delta=\e-\rho(y,f(x))$. It follows that for any $z\in\lball{x}{\delta}$,
\begin{align*}
\rho(y, f(z)) & \leq \rho(y,f(x)) + \rho(f(x),f(z)) \\
& \leq \rho(y,f(x))+\rho(x,z) \\
& < \rho(y,f(x))+\delta = \e.
\tag*{\hspace{-1em}\qedhere}
\end{align*}
\end{proof}
\end{lemma}

\subsection{Examples}

From now on we will concentrate on the maps from a quasi-metric space $(X,d)$ to $(\R,u^L)$. Recall that the quasi-metric $u^L$ is given by $u^L(x,y)=\max\{x-y,\ 0\}=x-y\vee\ 0$. The following is an obvious fact.
\begin{lemma}\label{lemma:flipLip}
Let $(X,d)$ be a quasi-metric space and $f:(X,d)\to(\R,u^L)$ a left $K$-Lipschitz function. Then, $g:(X,d)\to(\R,u^L)$ where $g=-f$ is a right $K$-Lipschitz function. \qed
\end{lemma}

Unless stated otherwise, we will consider $u^L$ as the \emph{canonical quasi-metric} on $\R$.  The main examples of Lipschitz functions are, as in the metric case, distance functions from points or sets, as well as sums of such functions. For each example both a left- and a right- 1-Lipschitz function will be produced but the proofs will be presented only for the left case since the right case would be follow by duality.

\begin{lemma}\label{lem:ptLips}
Let $(X,d)$ be a quasi-metric space and $y\in X$. Then the function $d_{y}: X\to\R$, where \[d_{y}(x) = d(x, y),\] is left 1-Lipschitz and the function $\cj{d}_{y}: X\to\R$, where \[\cj{d}_{y}(x) = d(y,x),\] is right 1-Lipschitz.
\end{lemma}
\begin{proof}
Let $x,z \in X$. Then $d_{y}(x)-d_{y}(z)=d(x,y)-d(z,y)\leq d(x,z)$ by the triangle inequality. Similarly, $\cj{d}_{y}(z)-d_{y}(x)=d(y,z)-d(y,x) \leq d(x,z)$.
\end{proof}

\begin{lemma}\label{lemma:setdists}
Let $(X,d)$ be a quasi-metric space and $A\subseteq X$. Then $d_{A}:X\to\R$, where
\[d_{A}(x)=d(x, A),\] is left 1-Lipschitz and $\cj{d}_{A}: X\to\R$, where \[\cj{d}_{A}(x)=d(A, x),\] is right 1-Lipschitz.
\end{lemma}
\begin{proof}
Let $x,y \in X$. Then
\begin{align*}
d(x,y) + d_{A}(y) & = d(x,y) + \inf_{w\in A}\{d(y,w)\} \\ 
& = \inf_{w\in A}\{d(x,y) + d(y,w)\} \\ & \geq \inf_{w\in A}\{d(x,w)\}
& \text{by the triangle inequality}\\ & = d_{A}(x).
\tag*{\hspace{-1em}\qedhere}
\end{align*}
\end{proof}

\begin{lemma}\label{lemma:Lip_convex}
Let $(X,d)$ be a quasi-metric space, $\left\{f_i\right\}_{i=1}^n$ a finite collection of left (right) 1-Lipschitz functions $X\to\R$ and $\left\{\lambda_i\right\}_{i=1}^n$ a collection of coefficients such that $\lambda_i\geq 0$ for all $i=1,2\ldots n$ and $\sum_{i=1}^n\lambda_i=1$. Then, \[f=\sum_{i=1}^n \lambda_if_i \] is left (right) 1-Lipschitz.
\begin{proof}
We prove the left case only.
\begin{align*}
f(x)-f(y) &= \sum_{i=1}^{n} \lambda_if_i(x) - \sum_{i=1}^{n} \lambda_if_i(y)\\
& = \sum_{i=1}^{n} \lambda_i (f_i(x)-f_i(y)) \\
& \leq \sum_{i=1}^{n} \lambda_i\ d(x,y)\\
& = d(x,y).
\tag*{\hspace{-1em}\qedhere}
\end{align*}
\end{proof}
\end{lemma}

In particular, for any collection $\left\{f_i\right\}_{i=1}^n$ of left 1-Lipschitz functions, the normalised sum $f= \frac{1}{n}\sum_{i=1}^{n} f_i$ is also left 1-Lipschitz.

\subsection{Quasi-normed spaces of left-Lipschitz functions and best approximation}

Another example of a semilinear quasi-normed space was produced by Romaguera and Sanchis \cite{RoSa00} who constructed a quasi-normed semilinear space of left Lipschitz functions.

Denote by $\SL{d}$ the set of all left Lipschitz functions on a quasi-metric space $(X,d)$ that vanish at some fixed point $x_0$. We can define for all $f,g \in \SL{d}$ and $a\in\R_+$ the sum $f+g$ and scalar multiple $a \cdot f $ in the usual way, producing a semilinear space $(\SL{d},+,\cdot)$ on $\R_+$.

Also, the function $\norm{.}_d: \SL{d}\to\R_+$ defined by \[\norm{f}_d = \sup_{d(x,y) \neq 0} \frac{(f(x) - f(y)) \vee 0}{d(x,y)} < \infty\] is a quasi-norm on $\SL{d}$ and hence $(\SL{d}, \norm{.}_d)$ forms a quasi-normed semilinear space.

\begin{thm}[\cite{RoSa00}]
The function $\rho_d:\SL{d}\times\SL{d}$ where \[\rho_d(f,g) = \sup_{d(x,y) \neq 0} \frac{((f-g)(x) - (f-g)(y)) \vee 0}{d(x,y)} \] is a bicomplete extended quasi-metric on $\SL{d}$. \qed
\end{thm}

Recall that a set $S$ in a linear space $E$ is \emph{convex} if and only if for any collection $x_1, x_2\ldots x_n\in S$ and $\lambda_1, \lambda_2,\ldots\lambda_n\in\R_+$ such that $\sum_{i=1}^n\lambda_i=1$, we have $\sum_{i=1}^n \lambda_i\ x_i \in S$. This definition can be extended to semilinear spaces and hence, by the Lemma \ref{lemma:Lip_convex}, the set of 1-Lipschitz functions vanishing at a fixed point is a convex subset of $\SL{d}$. 

\subsubsection{Best approximation}

From now on to the end of this section let $(X,d)$ be, as before, a quasi-metric space and denote by $cl_X\{y\}$ the closure $\{x:\ d(x,y)=0\}$ of the subset $\{y\}$ in the topology $\Top(d)$. Let $Y \subset X$, $p\in X$ and denote by $P_Y(p)$ \emph{the set of points of best approximation to $p$ by elements of Y}, that is: \[P_Y(p) = \{y_0 \in Y: \ d(p, Y) = d(p, y_0)\}\]
\begin{figure}[!ht]
\begin{center}
  \input{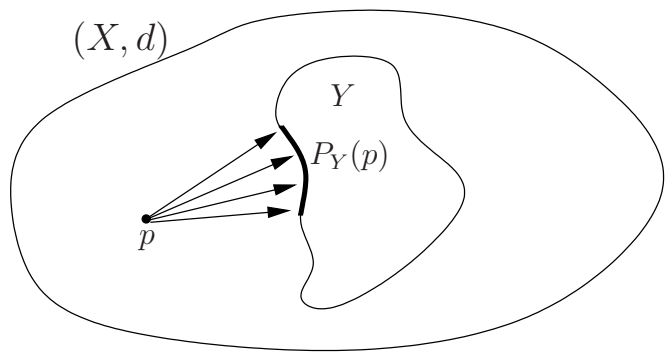}
  \caption{Set of points of best approximation.}
  \label{fig:bestapprox}
\end{center}
\end{figure}

\begin{thm}[\cite{RoSa00}]
Let $p \notin \bigcup\{cl_X\{y\}\ |\ y \in Y\}$ and let $M \subset Y$. Then $M \subset P_Y(p)$ if and only if there exists $f \in \SL{d}$ such that
\begin{enumerate}
\item $\norm{f}_d = 1$,
\item $f_{|Y}=0$, and
\item $d(p, y) = f(p) - f(y)$ for all $y \in M$. \qed
\end{enumerate}
\end{thm}

Furthermore, define $Y_0 = \{f \in \SL{d} \ \text{and} \ f_{|Y}=0\}$, and for each $x,y\in X$ such that $d(x,y) \neq 0$ set \[d_{Y_0}(x,y) = \sup_{\norm{f}_d \neq 0}\left\{f\in Y_0: \frac{(f(x) - f(y)) \vee 0}{\norm{f}_d}\right\}.\]

\begin{thm}[\cite{RoSa00}]
Let $p \notin Y$ and let $M \subset Y$. Then $M \subset P_Y(p)$ if and
only if $d_{Y_0}(p,y) = d(p,y)$ for all $y\in M$. \qed
\end{thm}



\section{Hausdorff quasi-metric}
\label{ssect:Hausdorff_qm}
Asymmetric variants of the Hausdorff metric provide further examples of quasi-metrics.
 
\begin{defin}
Let $(X,\rho)$ be a metric space. A map $\rho_H:\Komp{X}{\rho}\times\Komp{X}{\rho}\to\R_+$ 
defined by 
\[ \rho_H(A,B) = \max\{\sup_{a\in A} \rho(a,B),\ \sup_{b\in B} \rho(b,A)\}, \]
is called the \emph{Hausdorff metric}.
\end{defin}

\begin{figure}[!ht] \label{fig:Hausdorff_qm}
\begin{center}
  \input{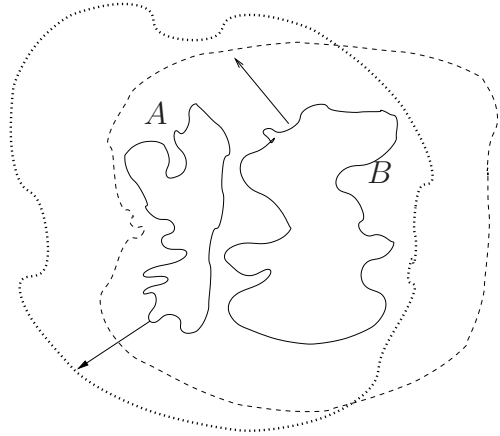}
  \caption{Hausdorff distance between two sets.}
\end{center}
\end{figure}

\begin{remark}
An equivalent, more geometric way would be to define \[\rho_H(A,B)=\inf\{\e > 0: A \subseteq \nbhd{B}{\e}\wedge B\subseteq\nbhd{A}{\e}\}.\] In other words, $\rho_H(A,B)$ is the infimal $\e\geq 0$ such that for every $\delta>0$, $A$ is contained in the $(\e+\delta)$-neighbourhood of $B$ and $B$ is contained in the $(\e+\delta)$-neighbourhood of $A$ (Fig. \ref{fig:Hausdorff_qm}).
\end{remark}

At this stage we omit the proof that Hausdorff metric is indeed a metric on $\Komp{X}{\rho}$ since it follows from the properties of the Hausdorff quasi-metric defined below. 

\begin{defin}
Let $(X,d)$ be a pseudo-quasi-metric space. Denote by $d_H^+$, $d_H^-$, and $d_H$, the maps $\Pow{X}\times \Pow{X}\to\R_+\cup\{\infty\}$ where for all $A,B\in\Pow{X}$,
\begin{align*}
d_H^+(A,B) & = \sup_{a\in A} d(a,B),\\
d_H^-(A,B) & = \sup_{b\in B} d(A,b), &\text{and}\\
d_H(A,B) & = \max\{d_H^+(A,B),\ d_H^-(A,B)\}.
\end{align*}
\end{defin}

\begin{lemma}\label{lemma:pqm_sets}
Let $(X,d)$ be a pseudo-quasi-metric space. Then $d_H^+$, $d_H^-$, and $d_H$ are extended pseudo-quasi-metrics.
\begin{proof}
It is obvious that for any $A\in\Pow{X}$, $d_H^+(A,A)=d_H^-(A,A)=d_H(A,A)=0$ as $d$ is a pseudo-quasi-metric. To prove the triangle inequality let $A,B,C\in\Pow{X}$. Take any $a\in A, b\in B$. By the Lemma \ref{lemma:set_tr_eq}, we have
\begin{align*}
d(a,C) & \leq d(a,b) + d(b,C)\\
&\leq d(a,b) + d_H^+(B,C), \quad \text{by the definition of $d_H^+$.}
\end{align*}
Hence, $d(a,C)\leq d(a,B) + d_H^+(B,C)$ and by taking supremum over $a\in A$ on both sides we get $d_H^+(A,C) \leq d_H(A,B)+d_H^+(B,C)$ as required. 

The statement for $d_H^-$ follows by the same argument once we note that $d_H^-(A,B)=\sup_{b\in B}d(A,b)= \sup_{b\in B}\cj{d}(b,A)$. It is obvious that if both $d_H^+$ and $d_H^-$ satisfy the triangle inequality then $d_H$ does as well.
\end{proof}
\end{lemma}

\begin{lemma}\label{lemma:Haus_max}
Let $(X,d)$ be a quasi-metric space with $\rho=\qam{d}$, the associated metric. Then for any $A,B\in\Pow{X}$
\begin{align*}
\rho_H^+(A,B) &= \max\{d_H^+(A,B),\ d_H^-(B,A)\} \quad \text{and}\\
\rho_H^-(A,B) &= \max\{d_H^-(A,B),\ d_H^+(B,A)\}
\end{align*}
\begin{proof}
The result follows straight from the definition. 
\begin{align*}
\max\{d_H^+(A,B), d_H^-(B,A)\} & = \sup_{a\in A}\max\{d(a,B),\ d(B,a)\}\\
& = \sup_{a\in A} \rho(a,B)\\ 
& = \rho_H^+(A,B)
\end{align*}
Similarly, $\max\{d_H^-(A,B),\ d_H^+(B,A)\} = \sup_{b\in B} \rho(A,b) = \rho_H^-(A,B)$.
\end{proof}
\end{lemma}

\begin{lemma}\label{lemma:Haus_qm}
Let $(X,d)$ be a quasi-metric space. Then $d_H$ restricted to $\Clsd{X}{d}$ is an extended quasi-metric and restricted to $\Komp{X}{d}$ is a quasi-metric.
\begin{proof}
To show $d_H$ is an extended quasi-metric, only the separation axiom needs to be proven as the rest follows by the Lemma \ref{lemma:pqm_sets}.

Suppose $A,B\in\Clsd{X}{d}$ and $d_H(A,B)=d_H(B,A)=0$. Let $\rho=\qam{d}$. By the Lemma \ref{lemma:Haus_max}, we have $\rho_H^+(A,B)=\rho_H^-(A,B)=0$. Now, if $\rho_H^+(A,B)=0$, then for all $a\in A$ there exists a $b\in B$ such that $\rho(a,b)=0$ as $B$ is closed, implying $a=b$ since $\rho$ is a metric. Hence, $\rho_H^+(A,B)=0\implies A\subseteq B$. Similarly, $\rho_H^-(A,B)=0\implies B\subseteq A$ as $\rho_H^-(A,B)=d_H^+(B,A)$. Therefore, $d_H(A,B)=d_H(B,A)=0$ implies $A=B$.

If $A,B\in \Komp{X}{d}$, for any $a\in A$, the function $a\mapsto d(a,B)$ is left 1-Lipschitz (Lemma \ref{lemma:setdists}), hence continuous (Lemma \ref{lemma:Lip_cts}) and bounded since $A$ is compact. Hence $d_H(A,B)<\infty$ and thus $d_H$ is a quasi-metric. 
\end{proof}
\end{lemma}

We are therefore justified to state the following

\begin{defin}
Let $(X,d)$ be a quasi-metric space. The map $d_H$ restricted to $\Clsd{X}{d}$ is called a \emph{Hausdorff extended quasi-metric} and restricted to $\Komp{X}{d}$ is called a \emph{Hausdorff quasi-metric}.
\end{defin}

\begin{corol}\label{corol:Haus_assoc}
Let $(X,d)$ be a quasi-metric space. The Hausdorff metric over $\Komp{X}{\qam{d}}$ restricted to $\Komp{X}{d}$ is the metric associated to the Hausdorff quasi-metric over $\Komp{X}{d}$.
\begin{proof}
Follows from the Lemmas \ref{lemma:Haus_max} and \ref{lemma:Haus_qm}.
\end{proof}
\end{corol}

A stronger statement for $d_H^+$ and $d_H^-$ is possible if the underlying space is $T_1$-separated.

\begin{lemma}
Let $(X,d)$ be a $T_1$ quasi-metric space. Then $q_H^+$ and $q_H^-$, restricted to $\Clsd{X}{d}$, are extended quasi-metrics whose associated orders correspond to set inclusion. They are quasi-metrics if they are restricted to $\Komp{X}{d}$.
\begin{proof}
As in Lemma \ref{lemma:Haus_qm}, we only need to prove separation -- the rest follows by the Lemma \ref{lemma:pqm_sets}.
Take any $A,B\in\Clsd{X}{d}$ and suppose $q_H^+(A,B)=0$. Then, for all $a\in A$ and for all $\e>0$, there is a $b\in B$ such that $d(a,b)<\e$. Since $B$ is closed, there exists a $b_0\in B$ such that $d(a,b_0)=0$ and therefore $a=b_0$ as $d$ satisfies the $T_1$ separation axiom. Thus $A\subseteq B \iff d_H^+(A,B)=0$ and it immediately follows that the associated order is set inclusion and that $d_H(A,B)=d_H(B,A)=0\iff A=B$.

If $A,B\in \Komp{X}{d}$, for any $a\in A$, the function $a\mapsto d(a,B)$ is left 1-Lipschitz (Lemma \ref{lemma:setdists}), hence continuous (Lemma \ref{lemma:Lip_cts}) and bounded since $B$ is compact. Hence $d_H^+(A,B)<\infty$. 

The statements for $d_H^-$ follow by duality.
\end{proof}
\end{lemma}

\begin{remark}\label{rem:hqm}
The assumption that $d$ satisfies the $T_1$ separation axiom is indeed necessary for separation. Consider the following example of a general quasi-metric space where the $q_H^+(A,B)=q_H^+(B,A)=0$ no longer implies $A=B$.

Let $X = \{a,b,c\}$ and define a quasi-metric $q$ by $q(a,a) = q(b,b) = q(c,c) = q(a,b) = q(c,b) = 0$ and $q(a,c) = q(b,a) = q(b,c) = q(c,a) = 1$. Let $A = \{a,b\}$ and $B = \{b,c\}$. It can be easily verified (Figure \ref{fig:hqm}) that $q$ is indeed a quasi-metric on $X$ and that $q_H^+(A,B) = q_H^+(B,A) = 0$ but $A\neq B$.
\end{remark}

\begin{figure}[!ht] \label{fig:hqm}
\begin{center}
  \input{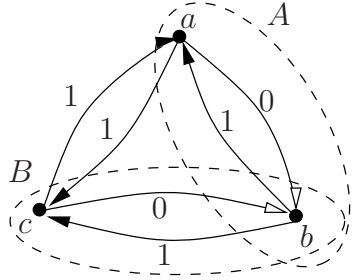}
  \caption{Illustration of Remark \ref{rem:hqm}.}
\end{center}
\end{figure}

The construction above was observed by Berthiaume \cite{Ber77} in a more general context of quasi-uniformities over hyperspaces of quasi-uniform spaces. There exist alternative definitions of Hausdorff quasi-metric. Vitolo \cite{Vi95} defines an (extended) Hausdorff quasi-metric $e_d$ over the collection of all nonempty closed subsets of a \emph{metric} space $(X,d)$ by \[e_d(A,B) = \sup_{a\in A}d(a,B),\] that is, in our notation, his quasi-metric corresponds to $d_H^+$. We now briefly survey his application of this quasi-metric to quasi-metrisability of topological spaces.

\begin{thm}[Vitolo \cite{Vi95}]
Every (extended) quasi-metric space embeds into the quasi-metric space of the form $(\Clsd{Y}{\rho},\rho_H^+)$, where $(Y,\rho)$ is a metric space. \qed
\end{thm}

Let $(X,d)$ be a quasi-metric space. The proof involves construction of the space $Y=X\times\R_+$ with the metric $\rho$ where \[ \rho((s,\alpha),(t,\beta)) = \qam{d}(s,t) + \abs{\alpha - \beta} \]
for all $(s,\alpha)$, $(t,\beta) \in Y$. The mapping $E: X\to\Clsd{Y}{\rho}$ where \[E(z)=\{(y,\eta)\in X: d(y,z)\leq\eta\} \] produces the required embedding.

\begin{corol}[Vitolo \cite{Vi95}]\label{corol:Vitolo_embed}
A topological space is quasi-metrisable if and only if it admits a topological embedding into a hyperspace. \qed
\end{corol}

\section{Weighted quasi-metrics and partial metrics}
\label{sec:weighted_qm}
Our main example of a quasi-metric comes from biological sequence analysis. It turns out that the similarity scores between biological sequences can often be mapped to a more restricted class of quasi-metrics, the \emph{weighted quasi-metrics} \cite{KuVa94, Vi99}, or equivalently, the \emph{partial metrics} \cite{Ma94}. Chapter \ref{ch:bioseq_qm} presents the full development of the biological application while the present section surveys the mathematical theory that was originally developed in the context of theoretical computer science.

\subsection{Weighted quasi-metrics}

\begin{defin}[\cite{KuVa94,Vi99}]
Let $(X,d)$ be a quasi-metric space. The quasi-metric $d$ is called a \emph{weightable quasi-metric} if there exists a function $w: X\to\R_+$, called the \emph{weight function} or simply the \emph{weight}, satisfying for every $x,y \in X$ \[d(x,y) + w(x) = d(y,x) + w(y).\] In this case we call $d$ \emph{weightable} by $w$.

A quasi-metric $d$ is \emph{co-weightable} if its conjugate quasi-metric $\cj{d}$ is weightable. The weight function $w$ by which $\cj{d}$ is weightable is called the \emph{co-weight} of $d$ and $d$ is \emph{co-weightable} by $w$. 

A triple $(X,d,w)$ where $(X,d)$ is a quasi-metric space and $w$ a function $X\to\R_+$ is called a \emph{weighted quasi-metric space} if $(X,d)$ is weightable by $w$ and a \emph{co-weighted quasi-metric space} if $(X,d)$ is co-weightable by $w$. 

In all the above, if the weight function $w$ takes values in $\R$ instead of $\R_+$, the prefix \emph{generalised} is added to the definitions.
\end{defin}

Not every quasi-metric space is weightable \cite{Ma94} but each metric space is obviously weightable, admitting constant weight functions. If $(X,d,w)$ is a weighted quasi-metric space then so is $(X,d,w+C)$ where $C\geq 0$. 

\begin{defin}[\cite{Sch00}]
Let $X$ be a set. A function $f: X\to\R_+$ is \emph{fading} if $\inf_{x\in X}f(x) = 0$. A weighted quasi metric space $(X,d,w)$ is of fading weight if its weight function is fading.
\end{defin}
\begin{lemma}[\cite{KuVa94}, \cite{Sch00}]
The weight functions of a weightable quasi-metric space are strictly decreasing (with respect to the associated partial order). These are exactly the functions of the form $f+C$, where $C\geq 0$ and where $f$ is the unique fading weight of the space.
\end{lemma}

\begin{example}
The set-difference quasi-metric on finite sets (Example \ref{ex:set_difference}) is co-weightable with a co-weight assigning to each set $A$ its cardinality $\abs{A}$.
\end{example}

\begin{example}[\cite{KuVa94}]\label{ex:weighted_reals}
Let $X=\R_+$ and set $d=u^R\vert_{\R_+}$, the restriction of $u^R$ to positive reals (i.e. for any $x,y\in\R_+$ $d(x,y)= y-x$ if $x\leq y$ and $d(x,y)=0$ if $y<x$). Set $w(x) = x$ for all $x\in X$. It is easy to verify that $(X,d,w)$ is a weighted quasi-metric space and that $w$ is its unique fading weight function.
\end{example}

Example \ref{ex:weighted_reals} shows that a weightable quasi-metric space need not be co-weightable -- in that case its weight is unbounded. Further examples are provided in 
\cite{KuVa94}. It is easy to see that a generalised weightable quasi-metric space is exactly a space which is weightable or co-weightable. The following result can be used to distinguish between weighted and co-weighted quasi-metric spaces.

\begin{lemma}[\cite{KuVa94}, \cite{Vi99}]\label{lem:weightable}
Let $(X, d, w)$ be a generalised weighted quasi-metric space.
\begin{itemize}
\item If $w > m$ for all $x\in X$, $(X,d, w-m)$ is a weighted quasi-metric space;
\item If $w < M$ for all $x\in X$, $(X,\cj{d}, M-w)$ is a weighted quasi-metric space;
\item If $(X,\cj{d},u)$ is a generalised weighted quasi-metric space then $w+u$ is constant on $X$. \qed
\end{itemize}
\end{lemma}

\begin{lemma}\label{lemma:weightLip}
Let $(X,d,w)$ be a weighted quasi-metric space. Then $w$ is a right-1-Lipschitz function.
\begin{proof}
Let $x,y\in X$. Then $w(x)-w(y) = d(y,x) - d(x,y) \leq d(y,x)$.
\end{proof}
\end{lemma}
Hence it follows that a weight function $w$ for a weightable quasi-metric space $(X,d,w)$ is continuous function $X\to\R_+$ with regard to the quasi-metric $u^R$ (i.e. it is upper semicontinuous).

Partial topological characterisation of weighted quasi-metric spaces was obtained by K\"unzi and Vajner \cite{KuVa94}. For example, they show that Sorgenfrey line is not weightable. The full results of their investigation are out of scope of this thesis and we only present a theorem about weightability of Alexandroff topologies.

\begin{thm}[\cite{KuVa94}]
Let $\leq$ be a partial order on a set $X$ and $\Top$ be the full Alexandroff topology on $X$.

Then $(X,\Top)$ admits a weightable quasi-metric if and only if there is a function $w:X\to\R_+$ such that for each $x\in X$ there exists $l_x > 0$ such that for any $y,z\in X$ with $x\leq y$, $z < y$ and $x\nleq z$ we have $w(z)-w(y)\geq l_x$. \qed
\end{thm}

\subsection{Bundles over metric spaces}

Vitolo \cite{Vi99} characterised weighted quasi-metric spaces as bundles over a metric space.

\begin{defin}
Let $(X,\rho)$ be a metric space. A \emph{bundle over $(X,\rho)$} \cite{Vi99} is the weighted quasi-metric space $(X \times \R_+, d, w)$ where \[ d((x,\xi), (y,\eta)) = \rho(x,y) + \xi - \eta \]
and \[ w((x,\xi)) = 2 \xi.\]
\end{defin}

\begin{thm}[\cite{Vi99}]
\label{thm:weighted_qm_decomposition}
Every weighted quasi-metric space embeds into the bundle over a metric space. \qed
\end{thm}

In fact, every weighted quasi-metric space can be constructed from a metric space and a non-distance-increasing (1-Lipschitz) positive real-valued function on it. If a generalised weighted quasi-metric space is desired, such function can take values over the whole real line.

\begin{thm}[\cite{Vi99}]
Given a metric space $(Y,\rho)$ and a 1-Lipschitz function $f:Y\to\R_+$, let $G = \{ (s,f(s)):s \in Y\}$ be the graph of $f$. If $d: Y\to\R$ is defined by
\[ ((s,f(s)),(t,f(t))) \mapsto \rho(s,t) + f(t) - f(s)\] then $(G,d,2f)$ is a weighted quasi-metric space. Moreover, every weighted quasi-metric space can be constructed in this way.

The quasi-metric space $(G,d)$ is $T_1$-separated if and only if the function $f$ above also satisfies \[\forall s,t\in Y: s\neq t, \quad \abs{f(s)-f(t)} < \rho(s,t). \qed \]
\end{thm}

\begin{thm}[\cite{Vi99}]
A quasi-metric space $(X,d)$ admits a generalised weight if and only if \[\forall x,y,z\in X \quad d(x,y)+d(y,z)+d(z,x) = d(x,z)+d(z,y)+d(y,x).\] Furthermore, $(X,d)$ is weightable if and only if it admits a generalised weight and for some (equivalently for each) $a\in X$, the set \[T_a=\{d(a,x)-d(x,a)\ |\ x\in X\}\] is bounded below.
\end{thm}

The generalised weight function above is given by $\gamma_a(x)=q(a,x)-q(x,a)$, $a\in X$. The statement can be dualised to the co-weightable case and used to distinguish weightable and co-weightable quasi-metric spaces.

\subsection{Partial metrics}\label{subsec:partmetr}

Matthews \cite{Ma94} proposed the concept of a partial metric, a generalisation of metrics which allows distances of points from themselves to be non-zero. He then showed that partial metrics correspond to weighted quasi-metrics. Partial metrics were further developed with a view to the applications in theoretical computer science \cite{O'N97, BuSc97, BuSh98, RoSch99, Sch00}. The greatest relevance of partial metrics in the context of this thesis is that similarity scores between biological sequences very often correspond exactly to partial metrics.

\begin{defin}[Matthews \cite{Ma94}]
Let $X$ be a set. A map $p: X\times X\to\R$ is called a \emph{partial metric} if for any $x,y,z\in X$:
\begin{enumerate}
\item $p(x,y)\geq p(x,x)$;
\item $x=y\iff p(x,x)=p(y,y)=p(x,y)$;
\item $p(x,y)=p(y,x)$;
\item $p(x,z)\leq p(x,y)+p(y,z)-p(y,y)$.
\end{enumerate}

For a partial metric $p$ its \emph{associated partial order} $\leq_p$ is defined so that for all $x,y\in X$, \[ x\leq_p y \iff p(x,x) = p(x,y).\]
\end{defin}

A partial metric $p$ induces a topology $\Top(p)$ whose base are the open balls of radius $\e>0$ of the form $\{y\in X: p(x,y) < p(x,x)+\e\}$ (\cite{O'N97}).

\begin{example}[\cite{Ma94}]
\label{ex:Baire_pmetric}
Let $X$ be any set and $Y=X^\N$, the set of all infinite sequences of elements of $X$. The \emph{Baire metric} is a distance $d$ on $Y$ defined for all $x,y \in Y$ by:
\[d(x,y) = 2^{-\sup\{i\in\N:\ x_j=y_j\ \forall j<i\}}.\]

Denote by $X^*$ the set of all finite and infinite sequences over $X$ and for each finite sequence $y\in X^*$ denote by $\abs{y}$ its length (we agree that for all $y\in X^\N$, $\abs{y}=\infty$). The map $p:X^*\times X^*\to\R$, where for all $x,y\in X^*\times X^*$ \[p(x,y) = 2^{-\sup\{i\in\N:\ i\leq\abs{x}\wedge i\leq\abs{y}\wedge x_j=y_j\ \forall j<i \}}\] is called the \emph{Baire partial metric}. It follows that $p(x,x)=2^{-\abs{x}}$.
\end{example}

\begin{thm}[\cite{Ma94}]\label{thm:partmetr2qm}
Let $X$ be a set.
\begin{enumerate}
\item For any partial metric $p$ on $X$, the map $q:X\times X\to\R$ where for all $x,y\in X$ \[ q(x,y) =p(x,y)-p(x,x)\] is a generalised weighted quasi-metric with weight function $w: x\mapsto p(x,x)$ such that $\Top(p)=\Top(q)$ and $\leq_p=\leq_q$.
\item For any (generalised) weighted quasi-metric $q$ over $X$ with weight function $w$, the map $p:X\times X\to\R$ where for all $x,y\in X$ \[ p(x,y)=q(x,y)+w(x)\] is a partial metric such that $\Top(q)=\Top(p)$ and $\leq_q=\leq_p$. \qed
\end{enumerate}
\end{thm}

\subsection{Semilattices, semivaluations and semigroups}

In this subsection we review the results of Schellekens \cite{Sch00} and Romaguera and Schellekens \cite{RoSch02} about the weightable quasi-metrics on semilattices and semigroups. These are, in the context of lattices, also mentioned in \cite{O'N97, BuSc97, BuSh98}. Again, the motivation comes from biological sequences, which are also instances of semigroups.

\begin{defin}
Let $(X,\leq)$ be a partial order. Then $(X,\leq)$ is called a \emph{join semilattice} if for every $x,y\in X$ there exists a supremum, denoted $x\sqcup y$ and a \emph{meet semilattice} if for every $x,y\in X$ there exists an infimum, denoted $x\sqcap y$. A \emph{lattice} is a partial order which is both a join and a meet semilattice. 
\end{defin}
\begin{defin}
If $(X,\preceq)$ is a join semilattice then a function $f:(X,\preceq)\to\R_+$ is a \emph{join valuation} iff for all $x,y,z\in X$ \[ f(x\sqcup z)\leq f(x\sqcup y)+ f(y\sqcup z) - f(y)\] and $f$ is a \emph{join co-valuation} iff for all $x,y,z\in X$ \[ f(x\sqcup z)\geq f(x\sqcup y)+ f(y\sqcup z) - f(y).\]

If $(X,\preceq)$ is a meet semilattice then a function $f:(X,\preceq)\to\R_+$ is a \emph{meet valuation} iff for all $x,y,z\in X$ \[ f(x\sqcap z)\geq f(x\sqcap y)+ f(y\sqcap z) - f(y)\] and $f$ is a \emph{meet co-valuation} iff for all $x,y,z\in X$ \[ f(x\sqcap z)\leq f(x\sqcap y)+ f(y\sqcap z) - f(y).\]

A function is a \emph{semivaluation} if it is either a join valuation or a meet valuation. A \emph{semivaluation space} is a semilattice equipped with a semivaluation.
\end{defin}
\begin{defin}
A quasi-metric space $(X,d)$ is called a join (meet) semilattice quasi-metric space if its associated partial order is a join (meet) semilattice.
\end{defin}

Equivalently, a quasi-metric space $(X,d)$ is a join semilattice if for all $x,y\in X$ there exists a $z\in X$ such that $d(x,z)=0$ and $d(y,z)=0$ and a meet semilattice if for all $x,y\in X$ there exists a $z\in X$ such that $d(z,x)=0$ and $d(z,y)=0$.

\begin{defin}\label{def:invmeetsemilattice}
A join semilattice quasi-metric space $(X,d)$ is called \emph{invariant} if for all $x,y,z\in X$ $d(x\sqcup z, y\sqcup z) \leq d(x,y)$. Similarly, a meet semilattice quasi-metric space $(X,d)$ is invariant if for all $x,y,z\in X$ $d(x\sqcap z, y\sqcap z) \leq d(x,y)$.
\end{defin}

We are now able to state the main theorem of \cite{Sch00}, associating invariant weighted quasi-metrics and monotone semivaluations on meet semilattices. There is also a dual of this theorem for join semilattices that is not presented here.

\begin{thm}[\cite{Sch00}]
\label{thm:semilattice_weighted}
For every meet semilattice $(X,\preceq)$ there exists a bijection between invariant co-weightable quasi-metrics $d$ on $X$ with $\leq_d=\preceq$ and fading strictly increasing meet valuations $f:(X,\preceq)\to(\R_+,\leq)$. The map $f\mapsto d_f$ is defined by $d_f(x,y)= f(x)-f(x\sqcap y)$. The inverse is the function which to each weightable space $(X,d)$ assigns its unique fading co-weight.

Similarly, one can show that for every meet semilattice $(X,\preceq)$ there exists a bijection between invariant weightable quasi-metrics $d$ on $X$ with $\leq_d=\preceq$ and fading strictly decreasing meet valuations $f:(X,\preceq)\to(\R_+,\leq)$. The map $f\mapsto d_f$ is defined by $d_f(x,y)= f(x\sqcap y)-f(x)$. The inverse is the function which to each weightable space $(X,d)$ assigns its unique fading weight. \qed
\end{thm}

The connection of the above result to the quasi-metric semigroups was explored in \cite{RoSch02}.

\begin{defin}\label{def:qmsemigroup}
A \emph{quasi-metric semigroup} is a triple $(X,d,\star)$ such that $(X,d)$ is a quasi-metric space and $(X,\star)$ is a semigroup such that $d$ is \emph{$\star$-invariant}, that is, for all $x,y,z\in X$ \[d(x\star z,y\star z)\leq d(x,y) \quad \text{and}\quad d(z\star x,z\star y)\leq d(x,y).\]
\end{defin}

\begin{defin}
We call the triple $(X,\preceq,\star)$ an \emph{ordered semigroup} if $(X,\preceq)$ is a partial order and $(X,\star)$ a semigroup and for all $x,y,z\in X$, \[x\preceq y\implies \left(x\star z\preceq y\star z \quad\wedge\quad  z\star x\preceq z\star y\right).\] Furthermore, if $(X,\preceq)$ is a meet semilattice, $(X,\preceq,\star)$ is called an \emph{ordered meet semigroup} or just meet semigroup. 
\end{defin}
It is obvious that a quasi-metric semigroup $(X,d,\star)$ corresponds to an ordered semigroup $(X,\leq_q,\star)$. Romaguera and Schellekens obtained the following extension of the Theorem \ref{thm:semilattice_weighted}.

\begin{thm}[\cite{RoSch02}]
Let $(X,\preceq,\star)$ be a meet semigroup, $d$ an invariant weighted quasi-metric  with $\leq_d=\preceq$ and $f$ the corresponding strictly decreasing meet valuation $f:(X,\preceq)\to(\R_+,\leq)$ as per Theorem \ref{thm:semilattice_weighted}. Then $(X,d,\star)$ is a meet semigroup if and only if for all $x,y,a,b\in X$ \[f(a\star b\sqcap x\star y) -f(a\star b)\leq f(a\sqcap x) +f(b\sqcap y) -f(a) -f(b).\qed\]
\end{thm}

We now survey some of the examples from \cite{RoSch02} and \cite{Sch00}. More examples will be provided by the biological sequences.

\begin{example}
Recall the Baire partial metric from Example \ref{ex:Baire_pmetric} on the set $\Sigma^*$, of all finite and infinite sequences of elements of an alphabet $\Sigma$. We also include $\varnothing$, the empty sequence in $\Sigma^*$. The corresponding weighted quasi-metric given by $b(x,y)=p(x,y)-p(x,x)$ is an invariant meet semilattice quasi-metric. The corresponding partial order corresponds to prefix ordering: $b(x,y)=0$ if and only if $x$ is a prefix of $y$.
\end{example}

\begin{example}[\cite{O'N98,RoSch02}]
Denote by $I(\R)$ the set of all closed intervals of $\R$ and equip it with a partial metric $p$ defined by \[p([a,b],[c,d]) = \max\{b,d\} - \min\{a,c\}.\] The associated weighted quasi-metric space is a join semilattice with the partial order being the reverse inclusion.
\end{example}

\begin{example}
Consider the dual complexity space $(\mathcal{C}^*,d_{\mathcal{C}^*})$ (Example \ref{ex:dual_complexity}) over the quasi-normed semilinear space $(\R_+,\norm{\cdot}_{\R_+})$ where $\norm{x}_{\R_+}=x$ (this is a restriction of the quasi-norm on $\R$ from Example \ref{ex:half_line_qnorm}), that is 
\begin{align*}
\mathcal{C}^*= \{f:\N\to\R_+\ | \sum_{n=1}^{\infty} 2^{-n}\ f(n)<\infty\}\quad\text{and}\\
d_{\mathcal{C}^*}(f,g) = \sum_{n=1}^{\infty} 2^{-n}\left(g(n)-f(n)\vee\ 0\right)\quad\forall f,g\in \mathcal{C}^*.
\end{align*}

Then $(\mathcal{C}^*,d_{\mathcal{C}^*})$ is a weighted quasi-metric with the weight being the quasi-norm on $\mathcal{C}^*$ (i.e. $w(f)=\sum_{n=1}^{\infty} 2^{-n}\ f(n)$), inducing an invariant meet semilattice. As it is also a semigroup with respect to the addition, it is an example of a weightable invariant meet semigroup.
\end{example}

\section{Weighted Directed Graphs}\label{sec:wght_dur_graph}

A further important class of examples of quasi-metrics is provided by directed graphs.

\begin{defin}
A \emph{directed graph}, or \emph{digraph} is a pair $(V,E)$, where $V$ is a set of \emph{vertices} or \emph{nodes} and $E\subseteq V\times V$ a set of \emph{edges}.

A \emph{weighted directed graph} or \emph{weighted digraph} is a triple $(V,E,\gamma)$ where $(V,E)$ is a directed graph and $\gamma:E\to\R$ is a function associating a \emph{weight} assigned to each edge.
\end{defin}

\begin{defin}
Let $\Gamma=(V,E)$ be a directed graph and let $u,v\in V$. A (directed) \emph{path} connecting $u$ and $v$ is a finite sequence of vertices $v_0,v_1,\ldots v_n$, such that $v_0=u$, $v_n=v$ and for all $i=1,2,\ldots, n$, $(v_{i-1},v_i)\in E$. 

For each $u,v\in V$, denote by $\mathscr{P}(u,v)$ the set of all paths connecting $u$ and $v$ and by $\ell(p)=n$ the length of a path $p$.

A (directed) \emph{cycle} is a path connecting a point with itself.

A directed graph $\Gamma=(V,E)$ is \emph{connected} if for every pair of vertices $u$ and $v$ there exists a path connecting them.
\end{defin}

\begin{remark}\label{rem:zeropath}
A one element sequence $x_0$ is also a path. Indeed, in that case the condition that for all $1\leq i\leq n$, $(v_{i-1},v_i)\in E$, is trivially true. The length of such path is obviously $0$.
\end{remark}

A connected weighted directed graph with positive weights on all edges can be turned into a quasi-metric space by using the weight of the shortest path between two vertices as a distance.

\begin{defin}
Let $\Gamma=(V,E,\gamma)$ be a connected weighted directed graph and let $p$ be a path in $\Gamma$. Define \emph{the weight} of $p$, denoted $\gamma(p)$ by \[\gamma(p)=\sum_{i=1}^{\ell(p)}\gamma(p_{i-1},p_i).\]

If in addition the weight $\gamma(e)$, of any edge $e\in E$, is non-negative, we call the map $d_\Gamma:V\times V\to\R$, defined by \[d_\Gamma(u,v) = \inf_{p\in\mathscr{P}(u,v)} \gamma(p),\]
the \emph{path distance} on $\Gamma$.
\end{defin}

\begin{lemma}\label{lemma:gr2qm}
Let $\Gamma=(V,E,\gamma)$ be a connected weighted directed graph with non-negative weights such that for all $u,v\in V$ and for all paths $p$ and $q$ such that $p\in\mathscr{P}(u,v)$ and $q\in\mathscr{P}(v,u)$,
\begin{equation}\label{eq:nozerocycles}
\gamma(p)=\gamma(q)=0\implies u=v.
\end{equation}

Then the path distance $d_\Gamma$ is a quasi-metric on $V$.

\begin{proof}
Let $u\in V$. The path $p=u$ has length $\ell(p)=0$ (c.f. the Remark \ref{rem:zeropath}) and the set $\{i\in\N:1\leq u\leq\ell(p)\}$ is empty. Since a sum over an empty set must be $0$, and $\gamma$ is a non-negative function, we have $d_\Gamma(u,u)=0$. The separation axiom follows directly from (\ref{eq:nozerocycles}). For the triangle inequality, it is sufficient to observe that for any three points $u,v,w\in V$ and any paths $p\in\mathscr{P}(u,v)$ and $q\in\mathscr{P}(v,w)$, there exists a path $r\in\mathscr{P}(u,w)$, where $r=p_0,p_1,\ldots p_{\ell(p)}q_1q_2\ldots q_{\ell(q)}$ such that $\gamma(r)=\gamma(p)+\gamma(q)$.
\end{proof}
\end{lemma}

\begin{remark}
The condition (\ref{eq:nozerocycles}) is equivalent to the property that no cycle of positive length can have a zero weight. 
\end{remark}

We call the above metric on graphs a \emph{path quasi-metric}. The above construction is natural and well known (there is a full book devoted to distances in graphs \cite{BuHa90}), especially in the form of path metric which is the metric associated to the path quasi-metric of the above Lemma. It naturally leads to consideration of geometric properties of digraphs, as in \cite{CJTW93}. The converse is also true: every quasi-metric space can be turned into a weighted directed graph such that the quasi-metric corresponds to a path metric.

\begin{lemma}\label{lemma:qm2gr}
Let $(X,\rho)$ be a quasi-metric space. Then there exists a weighted directed graph $\Gamma=(V,E,\gamma)$ with non-negative weights such that $d_\Gamma = \rho$.
\begin{proof}
Set $V=X$ and $E$ the set of all pairs $(x,y)$ where $x,y\in X$. For any pair $(x,y)\in X$, set $\gamma(x,y)=\rho(x,y)$ so that $\Gamma=(V,E,\gamma)$ is a weighted directed graph. It is now straightforward to observe that $d_\Gamma = \rho$. 
\end{proof}
\end{lemma}

We now review other published work connecting quasi-metrics and graphs.

Jawhari, Misane and Pouzet \cite{JaMiPo86} consider graphs and ordered sets as a kind of quasi-metric space where the values of the distance function belong to an ordered semigroup equipped with an involution. In this framework, the graph- or order- preserving maps are exactly the `Lipschitz' maps. They generalise various results on retraction and fixed point property for classical metric spaces to such spaces.

Deza and Panteleeva \cite{DePa00} introduce polyhedral cones and polytopes associated with quasi-metrics on finite sets. A \emph{cone} $C$ generated by a set $X \subseteq \R^n$ is the set $\{\sum_{x\in X} \lambda_x x\ |\ \lambda_x \in \R_+ \ \text{for all} \ x \in X\}$. They compute generators and facets of these polyhedra for small values of $n$ and study their graphs. This paper generalises some ideas presented in the book by Deza and Laurent \cite{DeLa97}. Unfortunately, analogues of $\ell_1$ embedability and other interesting issues developed in the book are not touched.

\section{Universal Quasi-metric Spaces}\label{sec:universal_qm}

Universal metric spaces were introduced by Pavel Urysohn (an alternative spelling is Uryson) in the 1920's -- his paper \cite{Ur27} was published posthumously in 1927. He showed that there exists a unique universal countable rational metric space $\U^\Q$ and that its completion is the universal complete separable metric space $\U$, also called the \emph{Urysohn space}. The spaces $\U$ and $\U^\Q$ are not only universal in the usual sense that they contain an isometric copy of every complete separable or countable rational metric space respectively -- they are also \emph{ultrahomogeneous}, that is, every isometry between finite subspaces of $\U$ or $\U^\Q$ extends to a global isometry.

Urysohn spaces and their groups of isometries have recently received considerable attention \cite{Usp90,Usp98,Ver01,Ver02,Pe02a,KPT04,Usp04,Ver04}. We construct the universal countable rational quasi-metric space, which we shall denote $\V^\Q$ and the universal bicomplete separable quasi-metric space $\V$ using a construction similar to Urysohn's and note that the associated metric spaces are exactly the spaces $\U^\Q$ and $\U$ respectively. 

\begin{defin}
A quasi-metric $(X,d)$ where the quasi-metric $d$ takes only rational values is called a \emph{rational quasi-metric space}.
\end{defin}

\begin{defin}
Let $\varphi$ be a class of quasi-metric spaces. A quasi-metric space $\V=(\V,d_{\V})$ of class $\varphi$ is called \emph{universal} or \emph{Urysohn} if it satisfies the following properties:
\begin{enumerate}[(i)]
\item For every quasi-metric space $X=(X,d_X)$ of class $\varphi$ there exists an isometric embedding $X\hookrightarrow\V$; \hfill (\emph{Universality})
\item For every two isometric finite quasi-metric subspaces $F,F'$ of $\V$, the isometry $F\leftrightarrow F'$ extends to a global isometry $\V\leftrightarrow\V$; \hfill (\emph{Ultrahomogeneity})
\end{enumerate}
\end{defin}


We make use of the following definition.

\begin{defin}
Let $X=(X,d_X)$ be a (rational) quasi-metric space, $F$ a finite quasi-metric subspace of $X$ and $Y=(Y,d_Y)$ a (rational) quasi-metric space such that $Y=F\cup\{y\}$, a one point quasi-metric extension of $X$. A (rational) quasi-metric space $W=(W,d_W)$ is called a \emph{$U$-extension (respectively $U^\Q$-extension) of $X$ with respect to $F$ and $Y$} if there exists an isometric embedding $X\hookrightarrow W$ and a point $w\in W$ such that the embedding $F\hookrightarrow X$ extends to an isometric embedding $Y\hookrightarrow W$ sending $y$ to $w$.

A quasi-metric space which is a $U$-extension ($U^\Q$-extension) of $X$ with respect to all finite subsets of $X$ and their one point extensions is called a \emph{universal $U$-extension ($U^\Q$-extension) of $X$}.

A quasi-metric space which is a $U$-extension ($U^\Q$-extension) of all of its finite subsets is called \emph{$U$-universal ($U^\Q$-universal)}.
\end{defin}

We now characterise the universal countable rational quasi-metric space as a countable $U^\Q$-universal quasi-metric space and the universal bicomplete separable quasi-metric space as a bicomplete separable $U$-universal quasi-metric space and show they are unique up to an isometry. Existence of these spaces is proven in Subsections \ref{subsec:UQ} and \ref{subsec:UBS}.

\begin{lemma}\label{lemma:univ_qm2}
Let $U$ and $U'$ be countable $U^\Q$-universal quasi-metric spaces and $F$ and $F'$ finite quasi-metric subspaces of $U$ and $U'$ respectively. Then an isometry $F\leftrightarrow F'$ extends to a global isometry $U\leftrightarrow U'$. 
\begin{proof}
We prove the statement using the so-called shuttle or back-and-forth argument. Let $x_0,x_1\ldots x_n$ be an enumeration of $U\setminus F$ and $y_0,y_1\ldots y_n$ an enumeration of $U'\setminus F'$. Let $X_0=F$ and $Y_0=F'$. By our assumption, there exists an isometry $F\leftrightarrow F'$. Now for each $n\in\N$, 
\begin{itemize}
\item If $x_n\notin X_n$, set $X'_{n+1} = X_n\cup\{x_n\}$. Clearly $X'_{n+1}$ is finite and by the $U^\Q$-universality of $U'$ there exists $y\in U'\setminus Y_n$ such that the isometric embedding $X_n\hookrightarrow Y_n$ extends to an isometric embedding $X'_{n+1}\hookrightarrow Y_n\cup\{y\}$. Set $Y'_{n+1} = Y_n\cup\{y\}$.

If $x_n\in X_n$, set $X'_{n+1} = X_n$ and $Y'_{n+1} = Y_n$. 
\item If $y_n\notin Y'_{n+1}$, set $Y_{n+1} = Y'_{n+1}\cup\{y_n\}$. By the $U^\Q$-universality of $U$, there exists $x\in U\setminus X'_{n+1}$ such that the isometric embedding $Y'_{n+1}\hookrightarrow X'_{n+1}$ extends to an isometric embedding $Y_{n+1}\hookrightarrow X'_{n+1}\cup\{x\}$. Set $X_{n+1} = X'_{n+1}\cup\{x\}$.

If $y_n\in Y'_{n+1}$, set $Y_{n+1} = Y'_{n+1}$ and $X_{n+1} = X'_{n+1}$. 
\end{itemize}

It is clear by the recursive construction that for each $n\in\N$, $X_n\subset X_{n+1}$, $Y_n\subset Y_{n+1}$,  there exists an isometry $X_n\leftrightarrow Y_n$ and for all $m\leq n$, $x_m\in X_{n+1}$ and $y_m\in Y_{n+1}$. It is now sufficient to observe that $U=\bigcup_{n\in\N}X_n$ and $U'=\bigcup_{n\in\N}X'_n$ to establish existence of a global isometry $U\leftrightarrow U'$.
\end{proof}
\end{lemma}

\begin{lemma}\label{lemma:countable_ext}
Let $U=(U,d_U)$ be a $U$- ($U^\Q$-) universal quasi-metric space, $X=(X,d_X)$ a countable (rational) quasi-metric space and $F$ a finite subspace of $X$. Then an isometric embedding 
$F\hookrightarrow U$ extends to an isometric embedding $X\hookrightarrow U$.
\begin{proof}
Let $x_1,x_2,\ldots$ be an enumeration of $X\setminus F$ and set $F_0=F$ and $F_{n+1}=F_n\cup \{x_{n+1}\}$ for all $n\in\N$. By the $U$- (or $U^\Q$-) universality of $U$, $F_0\hookrightarrow U$ extends to an isometric embedding $F_1=F_0\cup\{x_1\}\hookrightarrow U$. Assume that for all $i\leq k$, an isometric embedding $F_i\hookrightarrow U$ extends to an isometric embedding $F_{i+1}\hookrightarrow U$. Since $F_{k+1}$ is finite subset of $X$ and $F_k$ embeds isometrically in $U$ by our assumption, it follows by the $U$- (or $U^\Q$-) universality of $U$ that an isometric embedding $F_{k+1}\hookrightarrow U$ extends to an isometric embedding $F_{k+2}\hookrightarrow U$. Hence, by induction, for all $i\in\N$, an isometric embedding $F_i\hookrightarrow U$ extends to an isometric embedding $F_{i+1}\hookrightarrow U$ and therefore there exists an isometric embedding $X=\bigcup_{i=0}^\infty F_i\hookrightarrow U$.
\end{proof}
\end{lemma}

\begin{prop}
A countable $U^\Q$-universal quasi-metric space is the universal countable rational quasi-metric space. Such space is unique up to an isometry.
\begin{proof}
Universality follows by $U^\Q$-universality and the Lemma \ref{lemma:countable_ext} while ultrahomogeneity is a consequence of the Lemma \ref{lemma:univ_qm2}. Suppose $\V^\Q$ and $\V_1^{\Q}$ are two universal countable rational quasi-metric spaces. Take any finite rational quasi-metric space $F$. By universality, $F$ embeds isometrically into $\V^\Q$ and $\V_1^\Q$ and by the Lemma \ref{lemma:univ_qm2} the isometry between images of $F$ in $\V^\Q$ and $\V_1^{\Q}$ extends to a global isometry. Hence any two universal countable rational quasi-metric spaces are isometric.
\end{proof}
\end{prop}

\begin{remark}
In fact, $U^\Q$-universality is equivalent to the universality for a countable rational quasi-metric space since obviously universality implies $U^\Q$-universality.
\end{remark}

\begin{prop}
A bicomplete separable $U$-universal quasi-metric space is the universal bicomplete separable quasi-metric space. Such space is unique up to an isometry.
\begin{proof}
Let $X$ be a bicomplete separable $U$-universal quasi-metric space. Every bicomplete separable quasi-metric space $Y$ contains a countable dense subset $Y'$ which, by the Lemma \ref{lemma:countable_ext} embeds into a dense subspace of a $U$-universal space. This embedding obviously extends to all Cauchy (with respect to the associated metric) sequences of points in $Y'$ whose limits are all in $X$. Therefore, $X$ satisfies universality. On the other hand, the Lemma \ref{lemma:univ_qm2} can be used to extend the isometric embedding $F'\hookrightarrow X$ of any finite subset of a countable dense subset $Y'$ of $Y$ to the isometric embedding $Y'\hookrightarrow X$ which can then be extended to a global embedding since $Y$ and $X$ are bicomplete. 

The Lemma \ref{lemma:univ_qm2} also implies uniqueness. Suppose $\V$ and $\V_1$ are two universal bicomplete separable quasi-metric spaces. Any finite rational quasi-metric space $F$ embeds isometrically into $\V$ and $\V_1$ by universality and by the Lemma \ref{lemma:univ_qm2} the isometry between images of $F$ in $\V$ and $\V_1$ extends to a global isometry between countable dense subsets of $\V$ and $\V_1$. Since $\V$ and $\V_1$ are bicomplete, such isometry extends to an isometry $\V\leftrightarrow\V_1$.
\end{proof}
\end{prop}

\begin{remark}
The metric space associated to a universal quasi-metric space is also universal since every isometry between quasi-metric spaces is an isometry between their associated metric spaces (Lemma \ref{lemma:metr_isometry}). Therefore, $\qam{(\V^\Q)}=\U^\Q$ and $\qam{\V}=\U$. 
\end{remark}


\subsection{Universal countable rational quasi-metric space}\label{subsec:UQ}

\begin{lemma}\label{lemma:W_ext}
Let $X=(X,d_X)$ be a quasi-metric space and $F$ a finite quasi-metric subspace of $X$. Let $Y=(Y,d_Y)$, where $Y=F\cup\{y\}$, be a (rational) quasi-metric space containing $F$ as a quasi-metric subspace plus an extra point $\{y\}$. Then, there exists a $U$-extension of $X$ with respect to $F$ and $Y$. If all $X$ and $Y$ are rational quasi-metric spaces, there exists a $U^\Q$-extension of $X$ with respect to $F$ and $Y$.
\begin{proof}
Let $X,F$ and $Y$ be as above and $\Gamma_X=(X,E,\gamma)$ the weighted directed graph from the Lemma \ref{lemma:qm2gr} such that the path quasi-metric on $\Gamma_X$ coincides with $d_X$. Add another point to $\Gamma_X$, that is, let $\Gamma_W=(W,E',\gamma')$ be a weighted directed graph such that $W=X\cup\{w\}$, $E'=E\cup\{(x,w)\ |\ x\in F\}\cup\{(w,x)\ |\ x\in F\}$ and 
\begin{equation}\label{eq:gr_ext}
 \gamma'(u,v) = \begin{cases}
\gamma(u,v) & \text{if $u\in X$ and $v\in X$,}\\
d_Y(u,w) & \text{if $u\in X$ and $v=w$, and}\\
d_Y(w,v) & \text{if $u=w$ and $v\in X$}.
\end{cases}
\end{equation}
It is clear that $\Gamma_W$ is connected and hence the path quasi-metric $d_{\Gamma_W}$ is well-defined (Lemma \ref{lemma:gr2qm}). Let $d_W=d_{\Gamma_W}$ and $Y'=F\cup\{w\}$. To complete the proof we verify that $d_W\vert F=d_X\vert F$ and $d_W\vert Y'=d_Y$. Let $u,v\in W$. Denote by $\mathscr{P}(u,v)$ the set of all paths in $W$ linking $u$ and $v$.

Since $F$ embeds isometrically in $X$, and $X$ embeds isometrically in $W$ it is clear that $d_W\vert F\leq d_X\vert F$. Let $u,v\in F$ and suppose that there exists a path $p\in\mathscr{P}(u,v)$ such that $d_W(u,v)=\gamma'(p) < d_X(u,v)$. Then $p$ must pass through $w$ implying that  $d_W(u,v) = d_W(u,w) + d_W(w,v) = d_Y(u,w) + d_Y(w,v)\geq d_Y(u,v)$ by the triangle inequality. As $Y$ is an extension of $F$, we have $d_Y(u,v)=d_X(u,v)$, implying $d_W(u,v)\geq d_X(u,v)$ and contradicting our premise. Therefore, $d_W\vert F=d_X\vert F=d_Y\vert F$.

Let $u\in F$. It is clear from the Equation \ref{eq:gr_ext} that $d_W(u,w) \leq d_Y(u,w)$ and $d_W(w,u) \leq d_Y(w,u)$. Suppose there exists a path $p\in\mathscr{P}(u,w)$ such that $d_W(u,w)=\gamma'(p) < d_Y(u,w)$. As there is no edge $(x,w)$ in $E'$ for any $x\in X\setminus F$, such $p$ cannot pass through any point in $x\in X\setminus F$, nor can it pass through $w$ except as a last point. On the other hand, for any $v\in F$, $d_W(u,v)+d_W(v,w) = d_Y(u,v)+d_Y(v,w)\geq d_W(u,w)$ by the triangle inequality. This contradicts our supposition and hence $d_W(u,w) = d_Y(u,w)$. In the same way it can be shown that $d_W(w,u) = d_Y(w,u)$ and therefore $d_W\vert Y'=d_Y$.

It is obvious that $(W,d_W)$ is a rational quasi-metric space if $d_X$ and $d_Y$ take values in rationals. 
\end{proof}
\end{lemma}

Denote by $W(X,(F,Y))$ the $U$- (or $U^\Q$-) extension of $X$ with respect to $F$ and $Y$ constructed in the Lemma \ref{lemma:W_ext}.

\begin{lemma}\label{lemma:Z_ext}
Let $(X,d_X)$ be a countable rational quasi-metric space. Then there exists a countable
$U^\Q$-universal extension of $X$. 

\begin{proof}
Let $\mathcal{N}(X)$ be the set of all pairs $(F,Y)$ where $F$ is a finite subspace of $X$ and $Y$ is a rational quasi-metric space $Y=F\cup\{y\}$ containing $F$ as a quasi-metric subspace plus an extra point $\{y\}$. Since $X$ is countable and $d_X$ takes values in $\Q$, $\mathcal{N}(X)$ is countable. Let $N_0,N_1,\ldots$ be an enumeration of $\mathcal{N}(X)$. We now construct the required space recursively. 

Let $Z_0 = W(X, N_0)$ and $Z_{i+1}=W(Z_i, N_{i+1})$ for all $i\in\N$. We claim that for each $i\in\N$, $X\subset Z_i$ and $Z_i$ is a $U^\Q$ extension of $X$ with respect to $N_i$. Indeed, $X\subset Z_0$ and $Z_0$ is a $U^\Q$ extension of $X$ with respect to $N_0$. Assuming for all $k\in\N$ that $X\subset Z_k$ and denoting $N_{k+1}=(F',Y')$, it follows that $F'$ is a finite subset of $Z_k$ and hence $Z_{k+1}$ is well-defined. By the Lemma \ref{lemma:W_ext}, $X\subset Z_k\subset Z_{k+1}$ and $Z_i$ is a $U^\Q$ extension of $X$ with respect to $N_{k+1}$. Our claim therefore follows by induction and the union $\bigcup_{i\in\N}Z_i$ is the required countable $U^\Q$-universal extension of $X$.
\end{proof}
\end{lemma}

Denote by $Z(X)$ the $U^\Q$-universal extension of a rational quasi-metric space constructed in the Lemma \ref{lemma:Z_ext}.

\begin{corol}\label{lemma:univ_qm1}
There exists a countable $U^\Q$-universal quasi-metric space $\V^\Q$.
\begin{proof}
We again employ recursion. Set $U_0 = \{*\}$, a one-point quasi-metric space, $U_{n+1}=Z(U_n)$ for all $i\in\N$ and $U = \bigcup_{n\in\N}U_n$. We claim that for every finite rational quasi-metric space $F=(F,d_F)$ of cardinality $n\geq 1$
\begin{enumerate}[(i)]
\item there exists an isometric embedding $F\hookrightarrow U_{n-1}$, and
\item $U_n$ is a $U^\Q$-universal extension of $F$.
\end{enumerate}
It is clear by the above construction that this is indeed the case for the one-point quasi-metric space. Assume our claim holds for some $k\in N$ and let $F'$ be a finite quasi-metric space of cardinality $k+1$. Let $F''$ be a $k$-point restriction of $F'$. By our claim (ii), $U_n$ is a $U^\Q$-universal extension of $F''$ and hence contains an isometric copy of $F'$. By the Lemma \ref{lemma:Z_ext}, $U_{k+1}$ is a $U^\Q$-universal extension of $F'$ and we have proven our claim by induction. Each of sets $U_n$ is countable and therefore $\V=U$ is a countable $U^\Q$-universal quasi-metric space.
\end{proof}
\end{corol}

\subsection{Universal bicomplete separable quasi-metric space}\label{subsec:UBS}

To show that the bicompletion of the universal countable rational quasi-metric space is the universal bicomplete separable quasi-metric space we extend the argument of Gromov (\cite{Gr99}, pp.80--81) for the universal metric spaces. 

\begin{lemma}\label{lemma:approx_univ}
Let $X=(X,d_X)$ be a quasi-metric space admitting an everywhere dense $U^\Q$-universal quasi-metric subspace $Z=(Z,d_Z)$. Then for each finite subset $F\subset X$, every $\delta>0$ and any one point quasi-metric extension $(Y,d_Y)$ of $F$, where $Y=F\cup\{y\}$, there exists $x\in X$ such that for all $f\in F$ \[\abs{d_X(x,f)-d_Y(y,f)}\leq \delta\] and \[\abs{d_X(f,x)-d_Y(f,y)}\leq \delta.\]
\begin{proof}
Let $X$, $Y$, $Z$ and $F=\{f_1,f_2,\ldots,f_n\}$ be as above and let $\delta>0$ and $\e=\frac{\delta}{4}$. Since $Z$ is everywhere dense in $X$ we can approximate $F$ by the set $F'=\{f'_1, f'_2,\ldots,f'_n\}\subset Z$ such that for all $i=1,2,\ldots n$, $d_{X}(f_i,f'_i)\leq\e$ and $d_{X}(f'_i,f_i)\leq\e$. Let $\Gamma_{F'}=(F',E,\gamma)$ be the weighted directed graph from the Lemma \ref{lemma:qm2gr} such that the path quasi-metric on $\Gamma_{F'}$ coincides with $d_X\vert F'$. Construct a one point extension $\Gamma_{Y'}=(Y',E',\gamma')$ such that $Y'=F'\cup\{y'\}$ and $E'=E\cup\{(y',f'_i),(f'_i,y') \ |\ i=1,2\ldots,n\}\cup\{(y',y')\}$. Set $\gamma'(y',y')=0$ and for each $i$, let $\gamma(f'_i,y')$ be any rational such that \[d_Y(y,f_i)-\e\leq \gamma'(y',f'_i)\leq d_Y(y,f_i)+\e, \] and $\gamma(y',f'_i,)$ a rational such that \[d_Y(f_i,y)-\e\leq \gamma'(f'_i,y')\leq d_Y(f_i,y)+\e.\] 
By the Lemma \ref{lemma:gr2qm}, $Y'=(Y,d_{\Gamma_{Y'}})$ forms a rational quasi-metric space which is a one point extension of $F'\subset Z$. By the $U^\Q$-universality of $Z$, there exists $x\in Z$ such that for each $i=1,2,\ldots n$, $d_X(x,f'_i) = d_Z(x,f'_i) = d_{\Gamma_{Y'}}(y',f'_i)$ and $d_X(f'_i,x) = d_Z(f'_i,x) = d_{\Gamma_{Y'}}(f'_i,y')$. It remains to verify the required inequalities.

Clearly, for each $i$, $d_{\Gamma_{Y'}}(f'_i,y')\leq \gamma'(f'_i,y')$ and hence
\begin{align*}
d_X(x,f_i) & \leq d_X(x,f'_i) + d_X(f'_i,f_i)\\
& \leq d_{\Gamma_{Y'}}(y',f'_i) + \e\\
& \leq \gamma'(y',f'_i) + \e\\
& \leq d_Y(y,f_i) + 2\e.
\end{align*}
On the other hand, since $d_{\Gamma_{Y'}}$ is a path quasi-metric, there exists $1\leq j\leq n$ such that $\displaystyle d_{\Gamma_{Y'}}(y',f'_i) = \gamma'(y',f'_j) + d_X(f'_j,f'_i)$ (this includes the case $j=i$) and therefore
\begin{align*}
d_X(x,f_i) & \geq d_X(x,f'_i) - d_X(f_i,f'_i)\\
& \geq d_{\Gamma_{Y'}}(y',f'_i) - \e\\
& \geq \gamma'(y',f'_j) + d_X(f'_j,f'_i) - \e\\
& \geq d_Y(y,f_j) + d_X(f_j,f_i) - d_X(f'_i,f_i) - d_X(f_j,f'_j) - 2\e\\
& \geq d_Y(y,f_i) + d_Y(f_j,f_i) - 4\e\\
& \geq d_Y(y,f_i) - 4\e.
\end{align*}
Thus, for all $f\in F$, $\abs{d_X(x,f)-d_Y(y,f)}\leq 4\e=\delta$. The other inequality is verified in the same way.
\end{proof}
\end{lemma}

\begin{lemma}\label{lemma:univ_completion}
Let $X=(X,d_X)$ be a bicomplete quasi-metric space admitting an everywhere dense $U^\Q$-universal quasi-metric subspace. Then $X$ is a $U$-universal quasi-metric space.
\begin{proof}
Let $X$ be a as above, $F$ a finite subset of $X$ and $(F\cup\{y\}, d_Y)$ a one-point quasi-metric extension of $F$. We must show that there exists a point $x\in X$ such that for each $f\in F$, $d_X(x,f)=d_Y(y,f)$ and $d_X(f,x)=d_Y(f,y)$.

Assume without loss of generality that for all $f\in F$, $\qam{d}_Y(y,f)\geq\delta>0$,
that is, one of the distances $d_Y(y,f)$ and $d_Y(f,y)$ is bounded below by $\delta$ while the other can be $0$. We find by induction a sequence of points $x_0,x_1,\ldots x_i,\ldots\in X$ such that for all $f\in F$ and all $i=1,2\ldots$
\begin{enumerate}[(i)]
\item $\abs{d_X(f,x_i)-d_Y(f,y)}\leq \delta 2^{-i}$,
\item $\abs{d_X(x_i,f)-d_Y(y,f)}\leq \delta 2^{-i}$,
\item $\qam{d}_X(x_j,x_{j+1})\leq \delta 2^{-j+2}$\ for all $j=2,3,\ldots i$, and
\item $\min\{d_X(f,x_i),d_X(x_i,f)\}\geq 3\delta 2^{-i}$. 
\end{enumerate}
Indeed, assume such elements $x_i$ exist for all $i=1,2,\ldots k$. Let $F_k=F\cup\{x_1,x_2,\ldots,x_k\}$ and $Y'=F_k\cup\{y'\}$, a one point extension of $F_k$. We claim there exists a quasi-metric $d_{Y'}$ on $Y'$ satisfying
\begin{enumerate}[(a)]
\item $d_{Y'}\vert F_k = d_X\vert F_k$,
\item $d_{Y'}(f,y') = d_Y(f,y)$,
\item $d_{Y'}(y',f) = d_Y(y,f)$, and
\item $d_{Y'}(y',x_k) = d_{Y'}(x_k, y') = \delta 2^{-k}$.
\end{enumerate}
It clear that the condition (a) defines a quasi-metric on $F_k$. We will show that the conditions (a), (b), (c) and (d) together also define a quasi-metric $d_{F'}$ on $F'=F\cup\{x_k,y'\}$. 

Denote by $\Delta(u,v,w)$ the triangle inequality $d_{F'}(u,w)\leq d_{F'}(u,v)+d_{F'}(v,w)$ for some points $u,v,w\in F'$. The inequalities $\Delta(y',f_1,f_2)$, $\Delta(f_1,y',f_2)$ and $\Delta(f_1,f_2,y')$ where $f_1,f_2\in F$ follow from our assumption of $Y$ being a quasi-metric space while the inequalities $\Delta(y',x_k,f)$, $\Delta(f,y',x_k)$, $\Delta(y',x_k,f)$,\\ $\Delta(x_k,y',f)$ and $\Delta(f,x_k,y')$ where $f\in F$ clearly follow by (i) and (ii). The remaining two inequalities, $\Delta(y',f,x_k)$ and $\Delta(x_k,f,y')$ follow directly from (iv) (we have $d_{F'}(f,x_k)\geq 3\delta 2^{-k}\geq \delta 2^{-k}=d_{F'}(y',x_k)$ and $d_{F'}(x_k,f)\geq 3\delta 2^{-k}\geq \delta 2^{-k}=d_{F'}(x_k,y')$).
 
Therefore, $d_{F'}$ is a quasi-metric on $F'=F\cup\{x_k,y'\}$ agreeing with the induced quasi-metric on $F_k=F\cup\{x_1,x_2,\ldots,x_k\}$ on the intersection $F_k\cap F'=F\cup\{x_k\}$. Hence, there exists a quasi-metric on the union $Y'=F_k\cup F'$ satisfying the properties (a) -- (d) (this is easily shown by taking the distance between any two points not in the intersection to be the shortest path through the intersection).

By the Lemma \ref{lemma:approx_univ}, there exists a point $x_{k+1}\in X$ such that for each $f'\in F_k$, \[\abs{d_X(x_{k+1},f')-d_{Y'}(y',f')}\leq\delta 2^{-k-1}\] and \[\abs{d_X(f',x_{k+1})-d_{Y'}(f',y')}\leq\delta 2^{-k-1}\] and thus, by (a) and (b), it follows that for all $f\in F$, \[\abs{d_X(x_{k+1},f)-d_{Y}(y,f)}\leq\delta 2^{-(k+1)}\] and \[\abs{d_X(f,x_{k+1})-d_{Y}(f,y)}\leq\delta 2^{-(k+1)}.\] Furthermore, by (d), \[d_X(x_{k+1},x_k)\leq  \delta2^{-k-1} + d_{Y'}(y',x_k)\leq\delta2^{-k+1}\] and
\[d_X(x_k,x_{k+1})\leq  \delta2^{-k-1} + d_{Y'}(y',x_k)\leq\delta2^{-k+1},\] implying $\qam{d}_X(x_k,x_{k+1})\leq \delta 2^{-k+1}$. Finally, for all $f\in F$,
\begin{align*}
d_X(f,x_{k+1}) & \geq d_{Y'}(f,y) - \delta 2^{-k-1}\\
& \geq d_X(f,x_k) - d_{Y'}(y',x_k) - \delta 2^{-k-1}\\
& \geq 3\delta 2^{-(k+1)}.
\end{align*}
Similarly, $d_X(x_{k+1},f)\geq 3\delta 2^{-(k+1)}$. 

We conclude by induction that there exists an infinite sequence $x_1,x_2,\ldots$ satisfying (i) -- (iv). By (iii), this sequence is $\qam{d}_X$-Cauchy and hence convergent since $X$ is bicomplete. It converges to the required $x$ by (i) and (ii). 
\end{proof}
\end{lemma}

\begin{corol}
There exists a $U$-universal bicomplete separable quasi-metric space $\V$.
\begin{proof}
The required space $\V=\tilde{\V^\Q}$, the bicompletion of the universal countable rational quasi-metric space $\V^\Q$.
\end{proof}
\end{corol}
\index{Quasi-metric!space|)}
\index{Quasi-metric|)}


\chapter{Sequences and Similarities}\label{ch:bioseq_qm}

Pairwise sequence comparison is undoubtedly one of the core areas of bioinformatics. The most well known tool (actually a set of tools) is NCBI BLAST (Basic Local Alignment Search Tool) \cite{altschul97gapped} which, given a DNA or protein sequence of interest, retrieves all similar sequences from a sequence database. The similarity measure according to which sequences are compared is based on extension of a similarity measure on the set of nucleotides in the case of DNA, or the set of amino acids in the case of proteins to DNA or protein sequences, using a procedure known as \emph{alignment}. Two types of (pairwise) alignments are usually distinguished: \emph{global}, between whole sequences and \emph{local}, between fragments of sequences. Similarity scores on nucleotides or amino acids, as well as the penalties for `gaps' introduced into sequences while aligning them, usually have statistical interpretation.

The objective of this chapter is to establish the link between similarity measures on biological sequences and quasi-metrics. While the connections of global similarities to (quasi-) metrics have been known for long \cite{SWF81}, the novel result is that local similarities can also be converted to quasi-metrics while preserving the neighbourhood structure. The assumptions required for such conversion are satisfied by the similarity measures most widely used for searching DNA and protein databases. We develop this result in the context of free semigroups, which correspond to sets of strings from a finite alphabet and use the string and semigroup terminology interchangeably. The use of semigroup terminology may point to generalisations and extensions of our results to other areas. 

\section{Free semigroups and monoids}\label{sec:semigroupdefns}

Recall that the \emph{free monoid} on a nonempty set $\Sigma$, denoted $\Sigma^*$, is the monoid whose elements, called \emph{words} or \emph{strings}, are all finite sequences of zero or more elements from $\Sigma$, with the binary operation of concatenation. The unique sequence of zero letters (empty string), which we shall denote $e$, is the identity element. The \emph{free semigroup} on $\Sigma$, denoted $\Sigma^+$ is the subset of $\Sigma^*$ containing all elements except the identity.

The length of a word $w\in\Sigma^*$, denoted $\abs{w}$, is the number of occurrences of members of $\Sigma$ in it. For $w=\sigma_1\sigma_2\ldots\sigma_n$, where $\sigma_i\in\Sigma$, $\abs{w}=n$ and we set $\abs{e}=0$.

For two words $u,v\in\Sigma^+$, $u$ is a \emph{factor} or \emph{substring} of $v$ if $v=xuy$ for some $x,y\in\Sigma^*$; $u$ is a \emph{prefix} of $v$ if $v=uw$ for some $w\in\Sigma^*$; $u$ is a \emph{suffix} of $v$ if $v=wu$ for some $w\in\Sigma^*$; $u$ is a \emph{subsequence} or \emph{subword} of $v$ if $v=w^*_1u^*_1w^*_2u^*_2\ldots w^*_nu^*_nw^*_{n+1}$, where $u=u^*_1u^*_2\ldots u^*_n$, $u^*_i\in\Sigma^*$ and $w^*_i\in\Sigma^*$.  For any $x\in\Sigma^*$, we use $\mathfrak{F}(x)$ to denote the set of all factors of $x$.

We call a semigroup (monoid) $(X,\star)$ \emph{free} if it is isomorphic to the free semigroup (monoid) on some set $\Sigma$. The unique set of elements of $X$ mapping to $\Sigma$ under the isomorphism is called the set of \emph{free generators}.

As a convention, for any word $u\in\Sigma^*$, the notation $u=u_1u_2\ldots u_n$, where $n=\abs{u}$ shall mean that $u_i\in\Sigma$ while the notation $u=u^*_1u^*_2\ldots u^*_m$ shall imply that $u^*_i\in\Sigma^*$. For all $1\leq k\leq \abs{u}$ we shall use $\bar{u}_k$ to denote the word $u_1u_2\ldots u_k$ and set $\bar{u}_0=e$. 

The motivating examples of free semigroups for this chapter are biological sequences and structures related to them. It is quite natural that those macromolecules which are linear polymers of a limited number of small molecules and whose properties strongly depend on the sequence of their constituent building blocks can be represented in this way. For example, a DNA molecule can be represented as a word in the free semigroup generated by the four-letter nucleotide alphabet $\Sigma=\{A,T,C,G\}$ while an RNA molecule is a word in the free semigroup generated by the alphabet $\Sigma=\{A,U,C,G\}$. A protein can be thought of as a word in the free semigroup generated by the amino acid alphabet (Table \ref{tbl:amino_acids}).

A further example from biological sequence analysis is provided by \emph{profiles} \cite{Gribskov:1987,YoLe01}. Let $\Sigma$ be a set and denote by $\mathcal{M}(\Sigma)$ the set of all probability measures supported on $\Sigma$. We shall call the elements of the free monoid $\mathcal{M}(\Sigma)^*$ \emph{profiles} over $\Sigma^*$. Profiles arise as models of sets of structurally related biological sequences where $\Sigma$ is the DNA or protein alphabet.


\section{Generalised Hamming Distance}\label{sec:genhamming}

A simplest way to extend a distance from generators to words of equal length is to use what we call a \emph{generalised Hamming distance}, a special case of the \emph{$\ell_1$-type sum} mentioned in the Example \ref{ex:prodqm}. 

\begin{defin}
Let $\Sigma$ be a set and let $\Sigma^n=\{w\in\Sigma^+:\abs{w}=n\}$, the set of words in the free semigroup generated by $\Sigma$ of length $n$. Let $d_\Sigma:\Sigma\times\Sigma\to\R$ be a distance on $\Sigma$. The \emph{generalised Hamming distance} on $\Sigma^n$ is a function $d:\Sigma^n\times\Sigma^n$ where \[d(u,v)=\sum_{i=1}^n d_\Sigma(u_i,v_i).\]
\end{defin}

As mentioned in the Example \ref{ex:hammingdist}, the \emph{Hamming distance} is a special case where $d_\Sigma$ is the discrete metric. If the distance on the set of generators $\Sigma$ is a quasi-metric, the same holds for the generalised Hamming distance on $\Sigma^n$ (Example \ref{ex:prodqm}). Obviously, similarity measures on the generators can be extended in the same way.

The generalised Hamming distance has an advantage that it can be computed in linear time. It can be interpreted as the total cost of substitutions necessary to transform one word into another. It is worth noting that it is permutation invariant -- permuting both words with a same permutation does not change their distance.

The main practical disadvantage of the generalised Hamming distance is that it is restricted to the words of the same size and that it does not consider any other type of transformation but substitution. Hence it is only suitable for modelling the sets of words of the same length where insertions or deletions of factors (i.e. single characters or segments) are unlikely.

\section{String Edit Distances}

The term \emph{string edit distances} shall be used to refer to all distances between words defined as the smallest weight of a sequence of permitted weighted transformations transforming one word into another. In a stricter sense, the string edit distance denotes the smallest number of permitted edit operations required to transform one string into another where the permitted edit operations are substitutions of one character for another, insertions of one character into the first string and deletions of one character from the first string. It was first mentioned in the paper by V. Levenstein \cite{Lev66} and is often referred to as the \emph{Levenstein} distance. In their 1976 paper \cite{WSB76}, Waterman, Smith and Beyer introduced the most general form of the string edit distance and proposed an algorithm to compute it in some important cases. Below, we outline their construction of the so-called \emph{$\tau$-(quasi-) metric} which we shall refer to as the \emph{W-S-B distance}.

\subsection{W-S-B distance}

\begin{defin}\label{defn:tauset}
Let $\Sigma$ be a set and $\Sigma^*$ a free monoid over $\Sigma$ with the identity element $e$. Suppose $\tau=\{T:\mathscr{D}(T)\to\Sigma^*\ |\ \mathscr{D}(T)\subseteq\Sigma^*\}$ is a finite set of transformations defined on subsets $\Sigma^*$ such that the identity transformation $I$ is in $\tau$. Let $w:\tau\to\R_+$ be a function such that $w(T)=0\iff T=I$. We call the pair $(\tau,w)$ a \emph{set of weighted edit operations on $\Sigma^*$}.
\end{defin}

\begin{defin}\label{defn:tautransf}
Let $\Sigma$ be a set and $(\tau,w)$ a (finite) set of weighted edit operations on $\Sigma^*$. Let $u=u_1u_2\ldots u_n\in\Sigma^*$, where $u_i\in\Sigma$ and let $T\in\tau$. Fix $1\leq j\leq n$ and suppose $u_ju_{j+1}\ldots u_n\in\mathscr{D}(T)$. Then $T^j$ is defined by \[T^j(u)=u_1u_2\ldots u_{j-1}T(u_ju_{j+1}\ldots u_n).\] If $e\in\mathscr{D}(T)$, then $T^{n+1}$ is defined by $T^{n+1}(u)=uT(e)$. 

For any $u,v\in\Sigma^*$ define \[\{u\to v\}_\tau = \{T_{i_m}^{j_m}, T_{i_{m-1}}^{j_{m-1}},\ldots ,  T_{i_1}^{j_1}:T_{i_m}^{j_m}T_{i_{m-1}}^{j_{m-1}}\ldots T_{i_1}^{j_1}(u)=v\},\] 
where $T_{i_k}\in\tau$, that is, $\{u\to v\}_\tau$ is the set of all finite sequences of transformations from $\tau$ such that ordered composition of such transformation maps $u$ into $v$. The members of $\{u\to v\}_\tau$ are called \emph{edit scripts}. Also, if $\{u\to v\}_\tau\neq\emptyset$, for any $\zeta=T_{i_m}^{j_m},T_{i_{m-1}}^{j_{m-1}},\ldots , T_{i_1}^{j_1}\in\{u\to v\}_\tau$, define \[w(\zeta)=\sum_{k=1}^m w(T_{i_k}).\]
\end{defin}

\begin{remark}
In theory, $\tau$ can be allowed to be an infinite set. In that case, the minimum in the Definition \ref{def:taudist} of the $\tau$-distance below must be replaced by infimum and many proofs become very awkward. So far there have been no interesting examples involving infinite sets of transformations.
\end{remark}

\begin{defin}\label{def:taudist}
Let $\Sigma$ be a set and $(\tau,w)$ a (finite) set of weighted edit operations on $\Sigma^*$. For any $u,v\in\Sigma^*$, define the \emph{$\tau$-distance} $\rho_{\tau,w}:\Sigma^*\to\Sigma^*$ by \[\rho_{\tau,w}(u,v)=\min_{\zeta\in\{u\to v\}_\tau} w(\zeta),\] if $\{u\to v\}_\tau\neq\emptyset$ and $\rho_{\tau,w}(u,v)=\infty$ if $\{u\to v\}_\tau=\emptyset$. 
\end{defin}

Hence, the $\tau$-distance between two words is the smallest weight of an edit script of operations in $\tau$ transforming (in the sense of ordered composition) one word into another.

The relation $\rho_{\tau,w}(u,v)<\infty$ is an equivalence relation and partitions $\Sigma^*$ into equivalence classes $\{\Sigma^*_i\}$ where the value of $\rho_{\tau,w}$ between any two members of $\Sigma^*_i$ is finite. We have the following simple fact:

\begin{thm}[\cite{WSB76}]
Let $\Sigma$ be a set and $(\tau,w)$ a set of weighted edit operations on $\Sigma^*$. For each equivalence class $\Sigma^*_i$ of $\Sigma^*$, $\rho_{\tau,w}\vert\Sigma^*_i$ is a quasi-metric. \qed
\end{thm}

The $\tau$-metric is defined on each $\Sigma^*_i$ as the associated metric $\qam{\rho}_{\tau,w}$. Note that the requirement that $w(T)>0$ for each $T\in\tau$ such that $T\neq I$ implies that $\rho_{\tau,w}$ is a $T_1$-quasi-metric. 

\begin{remark}
It is easy to observe that the $\tau$-quasi-metric is equivalent to the path quasi-metric on the connected components of a weighted directed multigraph (two vertices can be joined by more than one directed edge) where the vertices are words in $\Sigma^*$ and two words $u$ and $v$ are joined with an edge if there is a transformation $T\in\tau$ such that for some $j$, $T^j(u)=v$. The weight of each edge is the weight of the corresponding transformation and an edit script is a path in the multigraph. Section \ref{sec:wght_dur_graph} presents the development of path quasi-metric on a weighted directed graph and the same technique can be trivially extended to multigraphs.
\end{remark}

We now present the terminology and notation for the most biologically relevant sets of weighted edit operations.

\begin{defin}\label{def:subsindels}
Let $\Sigma$ be a set and $\Sigma^*$ a free monoid over $\Sigma$ with the identity element $e$. Define the following transformations of elements of $\Sigma^*$:
\begin{itemize}
\item $T_{u-}:uv\mapsto v$, where $u\in\Sigma^+$, $v\in\Sigma^*$,
\item $T_{u+}:v\mapsto uv$, where $u\in\Sigma^+$, $v\in\Sigma^*$, and
\item $T_{(a,b)}:au\mapsto bu$, where $a,b\in\Sigma$ and $u\in\Sigma^*$.
\end{itemize}
The transformations of the type $T_{(a,b)}$ are called \emph{substitutions} or \emph{mutations}, of the type $T_{u+}$ are called \emph{insertions} and of the type $T_{u-}$ are called \emph{deletions}. Insertions and deletions are collectively called \emph{indels}.

Define \[\tau_0= \{T_{a-}:a\in\Sigma\}\cup\{T_{a+}:a\in\Sigma\}\cup\{T_{(a,b)}:a,b\in\Sigma\}\] and  \[\tau_\lambda= \{T_{u-}:u\in\Sigma^+\}\cup\{T_{u+}:u\in\Sigma^+\}\cup\{T_{(a,b)}:a,b\in\Sigma\}.\] 
\end{defin}

Note that $\tau_0$ and $\tau_\lambda$ implicitly contain the identity transformation $I=T_{(a,a)}$ for any $a\in\Sigma$.

\begin{example}\label{ex:lev}
For a set of letters $\Sigma$, the Levenstein distance is realised as $\rho_{\tau_0,w}$ where $w(T)=1$ for all $T\in\tau_0$ such that $T\neq I$.
\end{example}

While providing an easily interpretable example, the Levenstein distance is too simplistic for comparison of biological sequences and more general distances must be used. From an evolutionary point of view, each transformation should correspond to a mutational event and the resulting distance to the `evolutionary distance' between two sequences. In practice, not all transformations of biological sequences are equally likely. For example, substitutions are generally more likely than indels, while some substitutions may be more likely than others. This is certainly the case in proteins where one observes for example, that substitutions of I for V are more common than substitutions of I for K. It was also argued \cite{SWF81} that indels are more likely to take place by segments than character-by-character and hence that indels of arbitrary segments should take weights smaller than the sum of the weights of indels of single characters comprising each segment.
%

\begin{example}\label{ex:sellers}
The Sellers (or $s$-) distance, introduced by Sellers in 1974 \cite{Se74}, is a metric obtained by extension of a metric $\rho$ on the set $\Sigma^\dagger=\Sigma\cup\{e\}$, the set of generators plus the identity element, to the free monoid $\Sigma^*$. The value of $\rho(\sigma,\tau)$ for $\sigma,\tau\in\Sigma$ represents the cost of substitution of $\sigma$ for $\tau$ in a word in $\Sigma^+$ while $\rho(\sigma,e)$ is the cost of insertion or deletion of a character $\sigma$.

The $s$-metric can be considered as a special case of the W-S-B metric by using $\tau_0$ as the set of transformations. Suppose $w(T_{a-})=d(a,e)$, $w(T_{a+})=d(e,a)$ and $w(T_{(a,b)})=d(a,b)$. Waterman, Smith and Beyer \cite{WSB76} showed that the necessary and sufficient condition for the $\tau$-metric induced by the above weights to coincide with an $s$-metric is that $d$ be a metric on $\Sigma^\dagger$.

In fact, the construction of Sellers has long been known in the theory of topological groups \cite{Pe99b}. The $s$-metric on $\Sigma^+$ is equivalent to the Graev pseudo-metric \cite{Graev1948,Graev1951} on the free group $F(\Sigma)$ (i.e. the free group generated by $\Sigma$), restricted to $\Sigma^+$. The Graev pseudo-metric, can be described as the maximal bi-invariant pseudo-metric $\bar{\rho}$ on $F(\Sigma)$ such that $\bar{\rho}\vert X^\dagger=\rho$. 
\end{example}

\begin{example}\label{ex:lcs}
Let $\Sigma$ be a set and for $u,v\in\Sigma^*$ denote by $LCS(u,v)$ the \emph{longest common subsequence} of $u$ and $v$. Define \[\rho_{LCS}(u,v)=\abs{u}+\abs{v}-2\abs{LCS(u,v)}.\] It can be easily shown that $\rho_{LCS}$ is a metric on $\Sigma^*$ and that $\rho_{LCS}=\rho_{\tau_0,w}$ where $w(T_{a+})=w(T_{a-})=1$ and $w(T_{(a,b)})\geq 2$ for all $a,b\in\Sigma$ (i.e. optimal sequences of edit operations only involve indels). The LCS metric provides a special case of string edit distance (more specifically of Sellers distance) which has been extensively studied in computer science \cite{apostolico97string}.   
\end{example}

\begin{example}
Let $\Sigma$ be a set and suppose $\tau$ consists only of the transformations of the type $T_{(a,b)}$, where $a,b\in\Sigma$. Suppose $w(T_{(a,b)})=d_\Sigma(a,b)$ where $d_\Sigma$ is a function $\Sigma\times\Sigma\to\R_+$ such that $d(a,a)=0$ for all $a\in\Sigma$ and $d(a,b)>0$ for all $a\neq b$.. It is clear that $\rho_{\tau,w}(u,v)=\infty$ if and only if $\abs{u}\neq\abs{v}$ and therefore the partitions of the equivalence relation $\rho_{\tau,w}(u,v)<\infty$ are the sets $\Sigma^n$ for all $n\in\N_+$ plus the set $\{e\}$. It is easy to verify that on each $\Sigma^n$, $\rho_{\tau,w}$ coincides with the generalised Hamming distance $d$ if and only if $d$ satisfies the triangle inequality (i.e. $d$ is a quasi-metric).
\end{example}

\subsection{Alignments}

In biology, one is usually interested not only in the distance between two words, but also in the edit script realising it. A standard way of representing an edit script mapping one sequence into another is called a (pairwise) \emph{alignment}.
 
\begin{defin}\label{def:alignment}
Let $\Sigma$ be a set, $u,v\in\Sigma^+$ and suppose $(\tau_\lambda,w)$ is a set of weighted edit operations on $\Sigma^*$. A \emph{global alignment} between $u$ and $v$ is a finite sequence of pairs $(u^*_i,v^*_i)$ such that $u^*_i,v^*_i\in\Sigma^*$ for all $i$ and 
\begin{enumerate}[(i)]
\item $u=u^*_1u^*_2\ldots u^*_m$,
\item $v=v^*_1v^*_2\ldots v^*_m$, 
\item $u^*_i\neq e\vee v^*_i\neq e$ for all $i$,\ and 
\item there exists $T\in\tau_\lambda$ such that $v^*_i=T(u^*_i)$. 
\end{enumerate}
The \emph{weight} or \emph{score} of the alignment $\langle(u^*_i,v^*_i)\rangle_i$ is the sum $\sum_i w(T_i)$ where $T_i\in\tau_\lambda$ and $v^*_i=T_i(u^*_i)$. 
\end{defin}
 
The axiom (iii) in the Definition \ref{def:alignment} above ensures that a sequence that is a global alignment is finite.

\begin{defin}
A \emph{local alignment} between $u,v\in\Sigma^*$ is a global alignment between $u'$ and $v'$ where $u'$ is a factor of $u$ and $v'$ a factor of $v$.
\end{defin}

Alignments are usually displayed by first inserting chosen spaces (or dashes), either into or at the ends of $u$ and $v$, and then placing the two resulting strings one above the other so that every character or space in either string is opposite a unique character of a unique space in the other string \cite{Gusfield97}.

It is obvious that every (global) alignment can be associated with an edit script of the same weight. The converse is not true in general as the Example \ref{ex:twoonsamesite} attests. Recall that $\tau_\lambda$ consists of substitutions, insertions and deletions (Definition \ref{def:subsindels}) and that a superscript on a transformation $T$ denotes the start of the fragment being acted on by $T$ (Definition \ref{defn:tautransf}).

\begin{example}\label{ex:twoonsamesite}
Let $\Sigma=\{a,b,c\}$ and consider $(\tau_\lambda,w)$,the set of weighted edit operations on $\Sigma^*$ where $w(T_{(a,b)})=w(T_{(b,c)})=1$, $w(T_{(a,c)})=3$ and for each $u\in\Sigma^*$, $w(T_{u+})=w(T_{u-})=5$.

Suppose $u=aa$ and $v=ac$. Then, it is clear that $\zeta=T_{(b,c)}^2, T_{(a,b)}^2\in\{u\to v\}_{\tau_\lambda}$ and that $w(\zeta)=2$. However, the alignment of smallest weight, $A=(a,a), (a,c)$, has weight $3$. It is easy to see that all other possible alignments have an even greater weight. 
\end{example}

\begin{defin}
Let $u,v\in\Sigma^+$. An edit script $T_{i_m}^{j_m},T_{i_{m-1}}^{j_{m-1}},\ldots ,  T_{i_1}^{j_1}\in\{u\to v\}_{\tau_\lambda}$ \emph{admits an alignment} if there exists a sequence $\langle u^*_i\rangle_{i=1}^m$ where  $u^*_i\in\Sigma^*$ such that $u=u^*_m u^*_{m-1}\ldots u^*_1$ and $v=T_{i_m}(u^*_m)T_{i_{m-1}}(u^*_{m-1})\ldots T_{i_1}(u^*_1)$.
\end{defin}

The following Lemma provides a straightforward characterisation of the above definition. 

\begin{lemma}\label{lemma:alignseq}
Let $x,y\in\Sigma^+$. An edit script $T_{i_m}^{j_m},T_{i_{m-1}}^{j_{m-1}}, \ldots , T_{i_1}^{j_1}\in\{x\to y\}_{\tau_\lambda}$, where $j_m\leq j_{m-1}\ldots\leq j_1$, admits an alignment if $j_m=1$ and
\begin{enumerate}[(i)]
\item $j_1=\abs{x}$ \quad if $T_{i_1}=T_{(a,b)}$ for some $a,b\in\Sigma$,
\item $j_1=\abs{x}+1$ \quad if $T_{i_1}=T_{u+}$ for some $u\in\Sigma^+$,
\item $j_1=\abs{x}-\abs{u}+1$ \quad if $T_{i_1}=T_{u-}$ for some $u\in\Sigma^+$,
\end{enumerate}
and for all $1<k\leq m$,
\begin{enumerate}[(i)]
\setcounter{enumi}{\value{enumi}+3}
\item $j_k = j_{k-1}-1$ \quad if $T_{i_k}=T_{(a,b)}$ for some $a,b\in\Sigma$;
\item $j_k = j_{k-1}$ \quad if $T_{i_k}=T_{u+}$ for some $u\in\Sigma^+$;
\item $j_k = j_{k-1}-\abs{u}$ \quad if $T_{i_k}=T_{u-}$ for some $u\in\Sigma^+$;
\end{enumerate}
\begin{proof}
For each $k=1,2\ldots m$ set
\[ x^*_k =
\begin{cases}
a, &\text{if}\ T_{i_k}=T_{(a,b)}\ \text{for some}\ a,b\in\Sigma\\
e, &\text{if}\ T_{i_k}=T_{u+}\ \text{for some}\ u\in\Sigma^+,\\ 
u, &\text{if}\ T_{i_k}=T_{u-}\ \text{for some}\ u\in\Sigma^+. 
\end{cases}
\]
We claim that $x=x^*_mx^*_{m-1}\ldots x^*_1$ and $y=T_{i_m}(x^*_m)T_{i_{m-1}}(x^*_{m-1})\ldots T_1(x^*_1)$. The first claim is proven by showing by induction that for all $k=1,2\ldots m$, \[x_{j_k}x_{j_k+1}\ldots x_{\abs{x}}e=x^*_kx^*_{k-1}\ldots x^*_1.\] Indeed, the conditions (i), (ii) and (iii) directly imply the base step while the conditions (iv), (v) and (vi) imply the inductive step. Since $j_m=1$, it follows that $x=x^*_mx^*_{m-1}\ldots x^*_1$.

Similarly, the second claim is proven by showing by induction that for all $k=1,2\ldots m$, 
\[T_{i_k}^{j_k}T_{i_{k-1}}^{j_{k-1}}\ldots T_{i_1}^{j_1}(x)=  \bar{x}_{j_{k}-1} T_{i_k}(x^*_k)T_{i_{k-1}}(x^*_{k-1})\ldots T_1(x^*_1).\] The base step in this case follows from the definition of $T^j$ while the inductive step follows easily from the conditions (iv), (v) and (vi).
\end{proof}
\end{lemma}

The following simple result was first observed by Smith, Waterman and Fitch \cite{SWF81}.

\begin{lemma}[\cite{SWF81}]\label{lemma:aligneq}
Let $\Sigma$ be a set, $u,v\in\Sigma^*$ and suppose $\langle(u^*_i,v^*_i)\rangle_i$ is a global alignment between $u$ and $v$. Then
\begin{equation}
\abs{u}+\abs{v} = 2\sum_{a\in\Sigma}\sum_{b\in\Sigma} M_{a,b} +\sum_{k}kI_k +\sum_{k}kD_k
\end{equation}
where $M_{a,b}=\abs{\{i:u^*_i=a\wedge v^*_i=b\ |\ a,b\in\Sigma\}}$, $I_k=\abs{\{i:u^*_i=e\wedge \abs{v^*_i}=k\}}$ and $D_k=\abs{\{i:v^*_i=e\wedge \abs{u^*_i}=k\}}$. \qed
\end{lemma}

String edits and alignments are best illustrated by examples. For simplicity we use the Levenstein distance.

\begin{example}\label{ex:steditdist}
Let $\Sigma$ be the English alphabet, let $u=\text{\tt COMPLEXITY}$ and $v=\text{\tt FLEXIBILITY}$. It is easy to see that the Levenstein distance between $u$ and $v$ is $8$. Indeed, if we align $u$ and $v$ in the following way, 
\begin{verbatim}
COMPLEXI----TY
---FLEXIBILITY
\end{verbatim}
we note that seven indels and one substitutions are necessary to convert $u$ into $v$ and vice versa. One can also easily see that this is the smallest number of transformations necessary (more formally, this fact would be a simple corollary of the Theorem \ref{thm:condMindel} to be stated and proven later).
\end{example}

The string edit distances may, in some cases, be more suitable for comparison of strings of the same length than the (generalised) Hamming distance.

\begin{example}
Consider the words $u=\text{\tt ABCDEF}$ and $v=\text{\tt FABCDE}$ of length 6. The Hamming distance between $u$ and $v$ is 6 while the Levenstein distance is 2.
\end{example}

\subsection{Dynamic programming algorithms}

While the $\tau$-metric (and quasi-metric) can be generated from any sets of transformations of $\Sigma^*$, the main motivation of Waterman, Smith and Beyer in \cite{WSB76} was to extend the construction of Sellers \cite{Se74} so that indels of multiple characters with weights less than the sum of the weights of indels of individual characters can be permitted. The algorithm they proposed for computing such distances is based on \emph{dynamic programming} technique, introduced by Bellman \cite{BHK59} in the general context and first applied to biological sequence comparison by Needleman and Wunsch \cite{NW70} using similarities and by Sellers \cite{Se74} using distances. Dynamic programming remains the foundation of all pairwise biological sequence alignment algorithms and we here briefly present it in relation to the W-S-B algorithm. 

The three essential components of the dynamic programming approach are \emph{recurrence relation}, \emph{tabular computation} and the \emph{traceback}. 

\subsubsection{Recurrence Relations}

We now outline the recurrence relations used for computation of the W-S-B metric which takes into account indels of multiple characters. 

\begin{defin}
Let $\Sigma$ be a set. The set of weighted edit operations $(\tau_\lambda,w)$ on $\Sigma^*$ satisfies the condition {\bf M} if for all $x,y\in\Sigma^+$ and for each sequence of edit operations $\zeta\in\{x\to y\}_{\tau_\lambda}$ there exists $\eta\in\{x\to y\}_{\tau_\lambda}$ which admits an alignment and $w(\eta)\leq w(\zeta)$.
\end{defin}

The condition {\bf M} was introduced in \cite{WSB76} in a slightly different but essentially equivalent form. It implies that the W-S-B distance between any two points is determined solely from edit scripts admitting an alignment and leads to the following theorem. Recall that for all $u\in\Sigma^*$ and for any $1\leq k\leq \abs{u}$,  $\bar{u}_k$ denotes the word $u_1u_2\ldots u_k$ and that $\bar{u}_0=e$.

\begin{thm}[\cite{WSB76}]\label{thm:WSBrecurrence}
Let $\Sigma$ be a set, $x,y\in\Sigma^*$ and suppose $(\tau_\lambda,w)$ is a set of weighted edit operations on $\Sigma^*$ satisfying the condition {\bf M}. Then, for all $0\leq i\leq\abs{x}$, $0\leq j\leq\abs{y}$ such that $i+j\neq 0$,
\begin{equation*}
\begin{split}
\rho_{\tau_\lambda,w}(\bar{x}_i,\bar{y}_j) =\min\bigg\lbrace & \rho_{\tau_\lambda,w}(\bar{x}_{i-1},\bar{y}_{j-1})+w(T_{(x_i,y_i)}),\\
&\min_{1\leq k\leq j}\left\{\rho_{\tau_\lambda,w}(\bar{x}_{i},\bar{y}_{j-k})+ w(T_{y_{j-k+1}y_{j-k+2}\ldots y_j+})\right\},\\ 
&\min_{1\leq k\leq i}\left\{\rho_{\tau_\lambda,w}(\bar{x}_{i-k},\bar{y}_{j})+ w(T_{x_{i-k+1}x_{i-k+2}\ldots x_i-})\right\}\bigg\rbrace,\\
\end{split}
\end{equation*}
where $\rho_{\tau_\lambda,w}(\bar{x}_{p},\bar{y}_{q})$ is ignored if $p$ or $q$ are negative.
\begin{proof}
Obviously $\rho(\bar{x}_0,\bar{y}_0)=0$. Fix $0\leq i\leq\abs{x}$ and  $0\leq j\leq\abs{y}$ such that $i+j\neq 0$. Since $(\tau_\lambda,w)$ satisfies the condition {\bf M}, there exists an edit script $T_{i_m}^{j_m}, T_{i_{m-1}}^{j_{m-1}},\ldots , T_{i_1}^{j_1}\in \{\bar{x}_i\to\bar{y}_j\}_{\tau_\lambda}$ that admits an alignment and $\rho_{\tau_\lambda,w}(\bar{x}_i,\bar{y}_j)=\sum_{k=1}^m w(T_{i_k})$. Since $T_{i_m}^{j_m}, T_{i_{m-1}}^{j_{m-1}}, \ldots , T_{i_1}^{j_1}$ admits an alignment, it follows that
$T_{i_m}^{j_m}, T_{i_{m-1}}^{j_{m-1}}, \ldots, T_{i_2}^{j_2}\in \{\bar{x}_{i'}\to\bar{y}_{j'}\}_{\tau_\lambda}$ for some $i'<i$, $j'<j$ and that $\rho_{\tau_\lambda,w}(\bar{x}_{i'},\bar{y}_{j'})=\sum_{k=2}^m w(T_{i_k})$ (otherwise the assumption $\rho_{\tau_\lambda,w}(\bar{x}_i,\bar{y}_j)=\sum_{k=1}^m w(T_{i_k})$ would be violated). The proof is completed by considering all possibilities for $T_{i_1}$.
\end{proof}
\end{thm}

%

\begin{remark}
Under the conditions of the Theorem \ref{thm:WSBrecurrence} it is clear that $\rho_{\tau_\lambda,w}$ is invariant (in the sense of the Definition \ref{def:qmsemigroup}) with respect to the string concatenation, that is, for all $x,y,z\in\Sigma^*$,
\[\rho_{\tau_\lambda,w}(xz,yz)\leq \rho_{\tau_\lambda,w}(x,y) \quad \text{and}\quad \rho_{\tau_\lambda,w}(zx,zy)\leq \rho_{\tau_\lambda,w}(x,y).\] Hence, the triple $(\Sigma^*,\rho_{\tau_\lambda,w},\star)$ where $\star$ is the string concatenation operation is a quasi-metric semigroup (Definition \ref{def:qmsemigroup}). 
\end{remark}

\begin{defin}
Let $\Sigma$ be a set. A map $f:\Sigma^+\to\R$ is called \emph{increasing} if for any $u\in\Sigma^+$ and any $v\in\mathfrak{F}(u)\setminus\{e\}$, $f(v)\leq f(u)$.
\end{defin}

\begin{defin}
Let $\Sigma$ be a set. The set of weighted edit operations $(\tau_\lambda,w)$ on $\Sigma^*$ satisfies the condition {\bf N} if 
\begin{enumerate}[(i)]
\item $w(T_{(a,b)})=d(a,b)$ for all $a,b\in\Sigma$,
\item $w(T_{u+})=g(\abs{u})+ \sum_{k=1}^{\abs{u}}s(u_i)$ for all $u\in\Sigma^+$, and
\item $w(T_{u-})=h(\abs{u})+ \sum_{k=1}^{\abs{u}}t(u_i)$ for all $u\in\Sigma^+$.
\end{enumerate}
where $d$ is a quasi-metric on $\Sigma$, $g,h$ are non-decreasing positive functions $\N\to\R_+$, and $s,t$ are non-negative functions $\Sigma\to\R_+$ such that for all $a,b\in\Sigma$, $s(b)-s(a)\leq d(a,b)$ ($s$ is right 1-Lipschitz) and $t(a)-t(b)\leq d(a,b)$ ($t$ is left 1-Lipschitz).
\end{defin}

We now show that the condition {\bf N} implies the condition {\bf M}.

\begin{lemma}\label{lemma:Nrepl}
Let $\Sigma$ be a set and $(\tau_\lambda,w)$ a set of weighted edit operations on $\Sigma^*$ satisfying the condition {\bf N}. Suppose $x=x_1x_2\ldots x_m\in\Sigma^*$, $1\leq j_2< j_1\leq m+1$ and let $T_1,T_2\in\tau$ such that $T_{1}^{j_1}T_{2}^{j_2}(x)$ is well-defined. Denote $x'=T_{1}^{j_1}T_{2}^{j_2}(u)$ and $\zeta=T_{1}^{j_1},T_{2}^{j_2}\in\{x\to x'\}_{\tau_\lambda}$. Then, there exists an edit script $\eta=T_{3}^{j_2},T_{4}^{l} \in\{x\to x'\}_{\tau_\lambda}$ such that $j_2\leq l$ and $w(\eta)\leq w(\zeta)$.
\begin{proof}
There are nine principal cases corresponding to all combinations of transformation types in $\zeta$.

If $T_2=T_{(a,b)}$ for some $a,b\in\Sigma$ (the transformation acting on the position $j_2$ is substitution), it is easy to see that $T_1^{j_1} T_2^{j_2}= T_2^{j_2} T_1^{j_1}$, whatever $T_1$ might be. Similarly, if $T_2=T_{v-}$ for some $v\in\Sigma^+$ (the transformation acting on the position $j_2$ is deletion), we have $T_1^{j_1} T_2^{j_2}= T_2^{j_2} T_1^{l}$, where $l=j_1+\abs{v}$, again whatever $T_{i_{k+1}}$ might be. This covers six cases.

Now consider the three cases where $T_2=T_{u+}$ (the transformation acting on the position $j_2$ is insertion). If $j_1\geq\abs{u}+j_2$, then, whatever $T_2$ might be, $T_1^{j_1} T_2^{j_2}=T_2^{j_2} T_1^{l}$, where $l=j_1-\abs{u}$ and the statement is satisfied. Hence, assume without loss of generality that $j_1<\abs{u}+j_2$. 

If $T_1=T_{v+}$ for some $v\in\Sigma^+$, we have a situation where $u=yz$ and
\begin{equation}\label{eq:insins}
x^*_1x^*_2\overset{T_2}{\longmapsto} x^*_1yzx^*_2 \overset{T_1}{\longmapsto} x^*_1yvzx^*_2,
\end{equation}
for some $x^*_1,x^*_2\in\Sigma^*$ and $y,z\in\Sigma^+$ and where $w(\zeta)= g(\abs{yz})+g(\abs{v}) +\sum_{k=1}^{\abs{y}} s(y_k) +\sum_{k=1}^{\abs{z}} s(z_k) +\sum_{k=1}^{\abs{v}} s(v_k)$. Since the weight of $\zeta$ depends solely on composition and length of inserted fragments and not on the order of generators within them, we can set $\eta=T_{u'+}^{j_2}, T_{v'+}^{{j_2}+\abs{u'}}$ where $u'v'=yvz$ and $\abs{u'}=\abs{yz}$. Clearly, $\abs{v'}=\abs{yvz}-\abs{yz}=\abs{v}$ and hence $w(\eta)=w(\zeta)$.

If $T_1=T_{(a,b)}$ for some $a,b\in\Sigma$, we have a situation where $u=yaz$ and
\begin{equation}\label{eq:inssubs}
x^*_1x^*_2\overset{T_2}{\longmapsto} x^*_1yazx^*_2 \overset{T_1}{\longmapsto} x^*_1ybzx^*_2,
\end{equation}
for some $x^*_1,x^*_2,y,z\in\Sigma^*$ and $w(\zeta)= g(\abs{yaz})+ \sum_{k=1}^{\abs{y}} s(y_k) +\sum_{k=1}^{\abs{z}} s(z_k) + s(a) + d(a,b)$. In this case, we can set $\eta=T_{ybz+}^{j_2},I^{j_2}$, where $w(\eta)=g(\abs{ybz})+ \sum_{k=1}^{\abs{y}} s(y_k) +\sum_{k=1}^{\abs{z}} s(z_k) + s(b)$. As $s$ is right 1-Lipschitz ($s(b)-s(a)\leq d(a,b)$), it follows that $w(\eta)\leq w(\zeta)$. The identity transformation $I^{j_2}=T^{j_2}_{(x_{j_2},x_{j_2})}$ is there so that the form of $\eta$ exactly satisfies the statement of the Lemma.

If $T_1=T_{v-}$ for some $v\in\Sigma^+$, we have a situation where $u=yvz$ and
\begin{equation}\label{eq:insdel}
x^*_1x^*_2\overset{T_2}{\longmapsto} x^*_1yvzx^*_2 \overset{T_1}{\longmapsto} x^*_1yzx^*_2,
\end{equation}
for some $x^*_1,x^*_2,y,z\in\Sigma^*$ such that $yz\in\Sigma^+$, and $w(\zeta)=g(\abs{yvz})+ \sum_{k=1}^{\abs{y}} s(y_k) +\sum_{k=1}^{\abs{z}} s(z_k) +\sum_{k=1}^{\abs{v}} s(v_k) + h(\abs{v}) + \sum_{k=1}^{\abs{v}} t(v_k)$. Set $\eta=T_{yz+}^{j_2},I^{j_2}$ so that $w(\eta)=g(\abs{yz})+ \sum_{k=1}^{\abs{y}} s(y_k) +\sum_{k=1}^{\abs{z}} s(z_k)$. Since $h,s$ and $t$ are non-negative functions and $g$ is a non-decreasing function, we have $w(\eta)\leq w(\zeta)$.
\end{proof}
\end{lemma}

\begin{lemma}\label{lemma:Nrepl1}
Let $\Sigma$ be a set and $(\tau_\lambda,w)$ a set of weighted edit operations on $\Sigma^*$ satisfying the condition {\bf N}. Then, for any $x,y\in\Sigma^*$ and any edit script $\zeta\in\{x\to y\}_{\tau_\lambda}$, there exists an edit script $\eta=T_{i'_n}^{j'_n}, T_{i'_{n-1}}^{j'_{n-1}},\ldots, T_{i'_1}^{j'_1} \in\{x\to y\}_{\tau_\lambda}$ such that $j'_n\leq j'_{n-1}\ldots\leq j'_1$ and $w(\eta)\leq w(\zeta)$.

\begin{proof}
Let $x,y\in\Sigma^+$ and let $\zeta=T_{i_m}^{j_m},T_{i_{m-1}}^{j_{m-1}},\ldots, T_{i_1}^{j_1} \in\{x\to y\}_{\tau_\lambda}$. We construct the required edit script $\eta$ by using the Lemma \ref{lemma:Nrepl} recursively on pairs of transformations from $\zeta$.

Set $\eta^1_0=\zeta$ and find the largest $k$ such that $j_k$ is the smallest superscript in $\eta^0$. If $k=m$, set $\eta^1_1=\eta^1_0$ and proceed to the next step. Otherwise, produce a new edit script $\eta^1_1\in\{x\to y\}_{\tau_\lambda}$ such that $w(\eta^1_1)\leq w(\zeta)$, by replacing the pair of terms $T_{i_{k+1}}^{j_{k+1}}, T_{i_k}^{j_k}$ in $\eta^1_0$ by the pair $T_{i_k}^{j_k}, T_{i_{k+1}}^{l}$ where $l\geq j_k$. By the Lemma \ref{lemma:Nrepl}, this is always possible.

After this step, $j_k$ will remain the smallest superscript in $\eta^1_1$. Apply the same procedure to $\eta^1_1$ to produce $\eta^1_2$ and so on. After at most $m$ steps we get an edit script $\eta^1=T_{i^1_m}^{j^1_m},T_{i^1_{m-1}}^{j^1_{m-1}},\ldots, T_{i^1_1}^{j^1_1}$, with the same number of terms as $\zeta$, such that $j^1_m$ is the smallest superscript.

To get from $\eta^p$ to $\eta^{p+1}$, $1\leq p\leq m-1$, repeat the above procedure to the edit script $T_{i^p_{m-p}}^{j^p_{m-p}},T_{i^p_{m-p-1}}^{j^p_{m-p-1}},\ldots, T_{i^p_1}^{j^p_1}$ to obtain the edit script $T_{i^{p+1}_{m-p}}^{j^{p+1}_{m-p}},T_{i^{p+1}_{m-p-1}}^{j^{p+1}_{m-p-1}},\ldots, T_{i^{p+1}_1}^{j^{p+1}_1}$ and then set $\eta^{p+1}=T_{i^1_m}^{j^1_m}, T_{i^2_{m-1}}^{j^2_{m-1}},\ldots, T_{i^p_{m-p+1}}^{j^p_{m-p+1}}, T_{i^{p+1}_{m-p}}^{j^{p+1}_{m-p}}, T_{i^{p+1}_{m-p-1}}^{j^{p+1}_{m-p-1}},\ldots, T_{i^{p+1}_1}^{j^{p+1}_1}$. After $m$ such steps we get $\eta=\eta^m=T_{i^1_m}^{j^1_m}, T_{i^2_{m-1}}^{j^2_{m-1}},\ldots, T_{i^{m}_1}^{j^{m}_1}$ where $j^1_m\leq j^2_{m-1}\leq \ldots \leq j^{m}_1$. Since the weight did not increase at any step, it follows that $w(\eta)\leq w(\zeta)$.
\end{proof}
\end{lemma}

\begin{thm}\label{thm:condMindel}
Let $\Sigma$ be a set and $(\tau_\lambda,w)$ a set of weighted edit operations on $\Sigma^*$ satisfying the condition {\bf N}. Then, for any $x,y\in\Sigma^*$ and any edit script $\zeta\in\{x\to y\}_{\tau_\lambda}$ there exists an edit script $\theta\in\{x\to y\}_{\tau_\lambda}$ such that $\theta$ admits an alignment and $w(\theta)\leq w(\zeta)$.
\begin{proof}
Let $x,y\in\Sigma^+$ and let $\zeta=T_{i_m}^{j_m},T_{i_{m-1}}^{j_{m-1}},\ldots, T_{i_1}^{j_1} \in\{x\to y\}_{\tau_\lambda}$. If $\zeta$ already admits an alignment, there is nothing to prove. Otherwise, due to the Lemma \ref{lemma:Nrepl1}, we can assume without loss of generality that $j_m\leq j_{m-1}\ldots\leq j_1$. Using a recursive process starting from $\zeta$, we construct an edit script $\theta \in\{x\to y\}_{\tau_\lambda}$ that satisfies the requirements of the Lemma \ref{lemma:alignseq} and hence admits an alignment. We will use the notation $\theta_p=T_{i^{p}_{m_p}}^{j^{p}_{m_p}}, T_{i^{p}_{m_p-1}}^{j^{p}_{m_p-1}}, \ldots, T_{i^{p}_1}^{j^{p}_1}$, where $p=0,1, \ldots, N$ to denote the edit script at each step of the recursion. 

If $j_m > 1$, set $\theta_0=T_{(x_1,x_1)}^{1}, T_{i_m}^{j_m}, T_{i_{m-1}}^{j_{m-1}},\ldots, T_{i_1}^{j_1}$, otherwise set $\theta_0=\zeta$. For each $p$, let $k_p$ denote the largest index such that one of the conditions (iv), (v) or (vi) of the Lemma \ref{lemma:alignseq}
is not satisfied (which one of the three is violated depends on the type of $T_{i_{k_p}}$).

If $T_{i^p_{k_p}}=T_{(b,c)}$ for some $b,c\in\Sigma$, the condition (iv) of the Lemma \ref{lemma:alignseq} requires that $j_{k_p}=j_{{k_p}-1}-1$. Since the condition (iv) is violated, it must follow that either $j_{k_p}<j_{{k_p}-1}-1$ or $j_{k_p}=j_{{k_p}-1}$. In the former case, set $\theta_{p+1}=T_{i^{p}_{m_p}}^{j^{p}_{m_p}}, T_{i^{p}_{m_p-1}}^{j^{p}_{m_p-1}}, \ldots, T_{i^{p}_{k_p}}^{j^{p}_{k_p}}, T_{(x_l,x_l)}^l, T_{i^{p}_{k_p-1}}^{j^{p}_{k_p-1}}, \ldots, T_{i^{p}_1}^{j^{p}_1}$ where $l=j^{p}_{k_p}+1$. Since the inserted transformation is the identity transformation, the weight does not change.

In the former case there are three possibilities. If $T_{i^p_{k_p-1}}=T_{(a,b)}$ for some $a,b\in\Sigma$, construct $\theta_{p+1}$ by replacing the terms $T_{(b,c)}^{j^{p}_{k_p}}, T_{(a,b)}^{j^{p}_{k_p-1}}$ in $\theta_p$, of total weight $d(b,c)+d(a,b)$, with a single transformation $T_{(a,c)}^{j^{p}_{k_p}}$, of weight $d(a,c)$, and leaving the rest of $\theta_p$ unchanged. Clearly, since $d$ satisfies the triangle inequality, $w(\theta_{p+1})\leq w(\theta_p)$. If $T_{i^p_{k_p-1}}=T_{u+}$ for some $u=bv\in\Sigma^+$, construct $\theta_{p+1}$ by replacing the terms $T_{(b,c)}^{j^{p}_{k_p}}, T_{bv+}^{j^{p}_{k_p-1}}$ in $\theta_p$, of total weight $d(b,c)+s(b)+\sum_{i}^{}s(v_i)$ with a single transformation $T_{cv+}^{j^{p}_{k_p}}$, of weight $s(c)+\sum_{i}^{}s(v_i)$. Again, $w(\theta_{p+1})\leq w(\theta_p)$ because of the right Lipschitz assumption on $s$. If $T_{i^p_{k_p-1}}=T_{u+}$ for some $u=bv\in\Sigma^+$, construct $\theta_{p+1}$ by replacing the $T_{(b,c)}^{j^{p}_{k_p}}, T_{u-}^{j^{p}_{k_p-1}}$ in $\theta_p$ with $ T_{u-}^{j^{p}_{k_p}}, T_{(b,c)}^{j^{p}_{k_p}+\abs{u}}$ without changing the weight.

If $T_{i^p_{k_p}}=T_{u+}$ for some $u\in\Sigma^+$, the condition (v) of the Lemma \ref{lemma:alignseq} requires that $j_{k_p}=j_{{k_p}-1}$. Since we assume it is violated, it follows that $j_{k_p}<j_{{k_p}-1}$. Set $\theta_{p+1}=T_{i^{p}_{m_p}}^{j^{p}_{m_p}}, T_{i^{p}_{m_p-1}}^{j^{p}_{m_p-1}}, \ldots, T_{i^{p}_{k_p}}^{j^{p}_{k_p}}, T_{(x_l,x_l)}^l, T_{i^{p}_{k_p-1}}^{j^{p}_{k_p-1}}, \ldots, T_{i^{p}_1}^{j^{p}_1}$ where $l=j^{p}_{k_p}$. Since the inserted transformation is the identity transformation, the weight does not change.

Finally, if $T_{i^p_{k_p}}=T_{u-}$ for some $u\in\Sigma^+$, the condition (vi) of the Lemma \ref{lemma:alignseq} requires that $j_{k_p}=j_{{k_p}-1}-\abs{u}$. If $j_{k_p}< j_{{k_p}-1}-\abs{u}$, set, without changing the weight, $\theta_{p+1}=T_{i^{p}_{m_p}}^{j^{p}_{m_p}}, T_{i^{p}_{m_p-1}}^{j^{p}_{m_p-1}}, \ldots, T_{i^{p}_{k_p}}^{j^{p}_{k_p}}, T_{(x_l,x_l)}^l, T_{i^{p}_{k_p-1}}^{j^{p}_{k_p-1}}, \ldots, T_{i^{p}_1}^{j^{p}_1}$ where $l=j^{p}_{k_p}+\abs{u}$.

If $j_{{k_p}-1}-\abs{u}<j_{k_p}\leq j_{{k_p}-1}$ and $T_{i^p_{{k_p}-1}}=T_{v-}$ for some $v\in\Sigma^*$, we have a situation where $u=yz$ and 
\begin{equation}\label{eq:deldel}
x^*_1yvzx^*_2\overset{T_{i^p_{{k_p}-1}}}{\longmapsto} x^*_1yzx^*_2 \overset{T_{i^p_{k_p}}}{\longmapsto} x^*_1x^*_2,
\end{equation}
for some $x^*_1,x^*_2\in\Sigma^*$ and $y,z\in\Sigma^+$. Construct $\theta_{p+1}$ by replacing the terms $T_{yz-}^{j_{k_p}}, T_{v-}^{j_{{k_p}-1}}$ in $\theta_p$ with $T_{u'-}^{j_{k_p}}, T_{v'-}^{j_{k_p}+\abs{u'}}$ such that $u'v'=yvz$ and $\abs{u'}=\abs{yz}$. Clearly, this case is analogous to (\ref{eq:insins}) of the Lemma \ref{lemma:Nrepl} and, since the weight of a deletion also depends only on composition and length of deleted fragments, $\theta_{p+1}$ will have the same weight as $\theta_p$.

If $j_{{k_p}-1}-\abs{u}<j_{k_p}\leq j_{{k_p}-1}$ and $T_{i^p_{{k_p}-1}}=T_{(a,b)}$ for some $a,b\in\Sigma$, we have a situation where $u=ybz$ and
\begin{equation}\label{eq:subsdel}
x^*_1yazx^*_2\overset{T_{i^p_{{k_p}-1}}}{\longmapsto} x^*_1ybzx^*_2 \overset{T_{i^p_{k_p}}}{\longmapsto} x^*_1x^*_2,
\end{equation}
for some $x^*_1,x^*_2,y,z\in\Sigma^*$. Construct $\theta_{p+1}$ by replacing the terms $T_{ybz-}^{j_{k_p}}, T_{(a,b)}^{j_{{k_p}-1}}$ in $\theta_p$ by a single transformation $T_{yaz-}^{j_{k_p}}$. This case is analogous to (\ref{eq:inssubs}) of the Lemma \ref{lemma:Nrepl} and hence, by the left 1-Lipschitz assumption on $t$, $w(\theta_{p+1})\leq w(\theta_p)$.

If $j_{{k_p}-1}-\abs{u}<j_{k_p}\leq j_{{k_p}-1}$ and $T_{i^p_{{k_p}-1}}=T_{v+}$ for some $v\in\Sigma^*$, we have a situation where  $u=yvz$ and
\begin{equation}\label{eq:delins}
x^*_1yzx^*_2\overset{T_{i^p_{{k_p}-1}}}{\longmapsto} x^*_1yvzx^*_2 \overset{T_{i^p_{k_p}}}{\longmapsto} x^*_1x^*_2,
\end{equation}
for some $x^*_1,x^*_2,y,z\in\Sigma^*$. Construct $\theta_{p+1}$ by replacing the terms $T_{yvz-}^{j_{k_p}}, T_{v+}^{j_{{k_p}-1}}$ in $\theta_p$ by a single transformation $T_{yz-}^{j_{k_p}}$. This case is analogous to (\ref{eq:insdel}) of the Lemma \ref{lemma:Nrepl} and, by a similar argument, $\theta_{p+1}$ will have the same weight as $\theta_p$.

Hence, in all cases where one of the conditions (iv), (v) or (vi) of the Lemma \ref{lemma:alignseq} is violated, we construct a new edit script of no greater weight where all transformations up to and including the previously violating transformation now fully satisfy the conditions. Depending on the particular type of violation, the number of transformations in the new edit script either decreases by one, remains the same or increases by one. The only way it can increase is by inserting an identity transformation and clearly, there can be finitely many such insertions. Thus, the recursion terminates after finitely many steps. It remains to satisfy the conditions (i), (ii) and (iii) of the the Lemma \ref{lemma:alignseq} concerning the first edit operation. This can be achieved by inserting as many of the identity transformations as necessary.
\end{proof}
\end{thm}

\begin{remark}
The Theorem \ref{thm:condMindel} is also valid in the case where $g\equiv 0$ and $h\equiv 0$, but in that case, in order to satisfy the Definition \ref{defn:tauset} of $(\tau, w)$, $s$ and $t$ must be strictly positive. 
\end{remark}

The Theorem \ref{thm:condMindel} is a generalisation of the Theorem 4 of \cite{WSB76}, which assumes $w(T_{(a,b)})=\lambda$, $w(T_{u+})=g(\abs{u})$ and $w(T_{u-})=h(\abs{u})$, where $\lambda>0$ and $g,h$ are positive increasing functions. The functions $g$ and $h$ giving the weights of indels are called \emph{gap penalties}. The most widely used gap penalties are \emph{linear}, of the form $g(k)=ak$ and \emph{affine}, of the form $g(k)=a+bk$, where $k$ is the length of a gap and $a,b$ are constants. Both linear and affine gap penalties are examples of \emph{concave functions}, satisfying $g(k+l)\leq g(k)+g(l)$. Gap penalties of the form $g(k)=a+b\log(k)$ have also been proposed \cite{Benner:1993}.

The complexity of dynamic programming algorithms depends on the gap penalty. In general, Waterman, Smith and Beyer \cite{WSB76} obtained the $O(m^2n+mn^2)$ average and worst case running time, where $m=\abs{x}$ and $n=\abs{y}$. If $g$ and $h$ are linear, this can be reduced to $O(nm)$. The same bounds hold for affine gap penalties using the algorithm of Gotoh \cite{Go82}.

\subsubsection{Tabular computation}

The Theorem \ref{thm:WSBrecurrence} can be used directly to compute $\rho_{\tau_\lambda,w}(x,y)$ for any $x,y\in\Sigma^*$. Let $m=\abs{x}$ and $n=\abs{y}$ and let $D$ be an $(m+1)\times(n+1)$ matrix with rows and columns indexed from $0$. Suppose $w(T_{(a,b)})=d(a,b)$, $w(T_{u+})=g(\abs{u})$ and $w(T_{u-})=h(\abs{u})$ where $d$ is a quasi-metric and $g,h$ are positive increasing functions. Clearly, $(\tau_\lambda,w)$ satisfies the condition {\bf N} and hence, by the Theorem \ref{thm:condMindel}, condition {\bf M}. 

Set $D_{0,0}=0$, $D_{i,0}=\min_{1\leq k\leq i}\left\{D_{i-k,0}+h(k)\right\}$, \\$D_{0,j}=\min_{1\leq k\leq j}\left\{D_{0,j-k}+g(k)\right\}$ and for all $i=1,2\ldots m$ and $j=1,2\ldots n$, 
\begin{equation*}
\begin{split}
D_{i,j}=\min\bigg\lbrace & D_{i-1,j-1}+d(x_i,y_j),\\ & \min_{1\leq k\leq j}\left\{D_{i,j-k}+g(k)\right\},\\& \min_{1\leq k\leq i}\left\{D_{i-k,j}+h(k)\right\}  \bigg\rbrace.
\end{split}
\end{equation*}
The form of the recurrence above is the same as in the Theorem \ref{thm:WSBrecurrence} and hence $\rho{(\tau_\lambda,w)}(x,y)=D_{m,n}$. The tabular computation approach involves computation of $D_{m,n}$ bottom-up: the values of $D_{i,j}$ for all $1\leq i\leq m$ and $1\leq j\leq n$ are computed in an increasing row (or column) order. The Example \ref{ex:dynprog1} provides an illustration.

\begin{example}\label{ex:dynprog1}
Let $\Sigma$ be the English alphabet, let $u=\text{\tt COMPLEXITY}$ and $v=\text{\tt FLEXIBILITY}$ as in the Example \ref{ex:steditdist}. For all $a,b\in\Sigma$, set $d(a,b)=0$ if $a=b$ and $d(a,b)=4$ if $a\neq b$ and let $g(k)=h(k)=9+k$. The matrix (or table) $D$ used for computation of the W-S-B distance $\rho_{\tau_\lambda,w}$ is given in the Table \ref{tbl:dynprgdist} -- observe that $\rho_{\tau_\lambda,w}(u,v)=D_{10,11}=29$.

\newcolumntype{Y}{>{\centering\arraybackslash}X}
\begin{table}[!ht]
{\footnotesize\tt
\begin{tabularx}{\linewidth-5mm}{|YY||Y|Y|Y|Y|Y|Y|Y|Y|Y|Y|Y|Y|}
\hline
   &   &  0&  1&  2&  3&  4&  5&  6&  7&  8&  9& 10& 11\\ 
   &   &   &  F&  L&  E&  X&  I&  B&  I&  L&  I&  T&  Y\\ \hline\hline 
  0&   & {\bf 0}& 10& 11& 12& 13& 14& 15& 16& 17& 18& 19& 20\\ \hline
  1&  C& 10&  4& 14& 15& 16& 17& 18& 19& 20& 21& 22& 23\\ \hline
  2&  O& 11& 14&  8& 18& 19& 20& 21& 22& 23& 24& 25& 26\\ \hline
  3&  M& $\uparrow${\bf 12}& 15& 18& 12& 22& 23& 24& 25& 26& 27& 28& 29\\ \hline
  4&  P& 13& $\nwarrow${\bf 16}& 19& 22& 16& 26& 27& 28& 29& 30& 31& 32\\ \hline
  5&  L& 14& 17& $\nwarrow${\bf 16}& 23& 26& 20& 29& 30& 28& 32& 33& 34\\ \hline
  6&  E& 15& 18& 21& $\nwarrow${\bf 16}& 26& 27& 24& 29& 30& 31& 32& 33\\ \hline
  7&  X& 16& 19& 22& 25& $\nwarrow${\bf 16}& 26& 27& 28& $\leftarrow${\bf 29}& 30& 31& 32\\ \hline
  8&  I& 17& 20& 23& 26& 26& 16& 26& 27& 28& $\nwarrow${\bf 29}& 30& 31\\ \hline
  9&  T& 18& 21& 24& 27& 27& 26& 20& 30& 31& 32& $\nwarrow${\bf 29}& 34\\ \hline
 10&  Y& 19& 22& 25& 28& 28& 27& 30& 24& 34& 35& 36& $\nwarrow${\bf 29}\\ \hline
\end{tabularx}
}
\caption[An example of a dynamic programming table for computation of W-S-B distance between two strings.]{The dynamic programming table used to compute the W-S-B distance between the strings {\tt COMPLEXITY} and {\tt FLEXIBILITY}. The cells on an optimal path between $(0,0)$ and $(m,n)$ are shown in bold.}\label{tbl:dynprgdist}
\end{table}
\end{example}

\subsubsection{Traceback}

Computation using a dynamic programming table provides the value of distance but often, especially in biological applications, an optimal edit script (need not be unique) and the corresponding alignment need to be retrieved. This is most easily achieved (at least conceptually) by keeping one or more pointers at each entry $(i,j)$ of the dynamic programming table $D$ apart from $(0,0)$, pointing to the entries $(i_0,j_0)$ such that $D_{i,j}$ is obtained by summing $D_{i_0,j_0}$ and the weight of the corresponding transformation. An optimal edit script is obtained by following any path of pointers from $(m,n)$ to $(0,0)$ and accumulating the transformations corresponding to each pointer. This procedure is known as \emph{traceback}. It is clear that there exists a 1-1 correspondence between alignments and paths between $(0,0)$ and $(m,n)$.

\begin{example}
The path shown in bold in the Table \ref{tbl:dynprgdist} corresponds to the following alignment:

{\centering \parbox[c]{0.5\linewidth}{{\tt COMPLEX---ITY \\
---FLEXBILITY.}}}

Note that there exists a second optimal path in this case -- it corresponds to the alignment in the Example \ref{ex:steditdist}.
\end{example}

The correspondence between alignments and paths in the dynamic programming table suggests an alternative definition of a distance. Let $u,v\in\Sigma^+$ and suppose $d$ is a non-negative function $\Sigma\times\Sigma\to\R_+$ such that $d(a,a)=0$ and $g,h$ are positive functions. Define \[\rho(u,v) = \min_{\text{alignments of $u$ and $v$}} \sum_{a\in\Sigma}\sum_{b\in\Sigma} M_{a,b}\cdot d(a,b)+\sum_{k}I_k\cdot g(k) + \sum_{k}D_k\cdot h(k),\]
where, as in the Lemma \ref{lemma:aligneq}, $M_{a,b}=\abs{\{i:u_i=a\wedge v_i=b\}}$,\\ $I_k=\abs{\{i:u_i=e\wedge \abs{v_i}=k\}}$ and $D_k=\abs{\{i:v_i=e\wedge \abs{u_i}=k\}}$. The condition {\bf N} is the sufficient condition for $\rho$ to be a quasi-metric.

\section{Global Similarity}

An alternative approach to sequence comparison is maximise similarities instead of minimising distances. In this case a \emph{similarity measure} on $\Sigma$ and gap penalties are used to define the \emph{global similarity} between two sequences in $\Sigma^*$. The computation is handled using the Needleman-Wunsch dynamic programming algorithm \cite{NW70}  which is very similar to the W-S-B algorithm for computation of distances. We define global similarity using a dynamic programming matrix.

\begin{defin}
Let $\Sigma$ be a set, $x,y\in\Sigma^*$, $s:\Sigma\times\Sigma\to\R$ and $g,h:\N^+\to\R_+$. Let $x,y\in\Sigma^*$ and let $m=\abs{x}$ and $n=\abs{y}$. The \emph{Needleman-Wunsch} dynamic programming matrix, denoted $\NW(x,y,s,g,h)$, is an $(m+1)\times(n+1)$ matrix $S$ with rows and columns indexed from $0$ such that $S_{0,0}=0$, $S_{i,0}= \max_{1\leq k\leq i}\left\{S_{i-k,0}-h(k)\right\}$, $S_{0,j}=\max_{1\leq k\leq j}\left\{S_{0,j-k}-g(k)\right\}$ and for all $i=1,2\ldots m$ and $j=1,2\ldots n$ \[S_{i,j}=\max\left\{S_{i-1,j-1}+s(x_i,y_j), \max_{1\leq k\leq i}\left\{S_{i-k,j}-h(k)\right\}, \max_{1\leq k\leq j}\left\{S_{i,j-k}-g(k)\right\}  \right\}.\]
We define the \emph{global similarity} between the sequences $x$ and $y$ (given $s$, $g$, and $h$), denoted $\mathcal{S}(x,y)$, to be the value $S_{m,n}$.
\end{defin}

\begin{remark}\label{rem:simalign}
In terms of alignments, we have \[\mathcal{S}(x,y) = \max_{\text{alignments of $x$ and $y$}} \sum_{a\in\Sigma}\sum_{b\in\Sigma} M_{a,b}\cdot s(a,b)-\sum_{k}I_k\cdot g(k) - \sum_{k}D_k\cdot h(k),\] where, as before, $M_{a,b}=\abs{\{i:u_i=a\wedge v_i=b\}}$, $I_k=\abs{\{i:u_i=e\wedge \abs{v_i}=k\}}$ and $D_k=\abs{\{i:v_i=e\wedge \abs{u_i}=k\}}$. The term global is used because the alignments in question are global -- in the next section we will examine \emph{local similarities} which involve local alignments. 
\end{remark}

\begin{remark}
Traditionally the gap penalty is a positive function in the case of both distances and similarities, being added in one case and subtracted in the other. The running times of dynamic programming algorithms still depend on the types of gap penalties, as discussed in the section about distances.
\end{remark}

It is also possible to interpret similarities by considering the sets of weighted transformations similar to those used to define the W-S-B distance. In this case, the set $\tau$ still consists of weighted transformations of the elements of $\Sigma^*$ but the requirement that $W(T)=0\iff T=I$ is dropped. In particular, this means that each transformation of the form $T_{(a,a)}$, where $a\in\Sigma$, does not need to have weight $0$ and that the weights of $T_{(a,a)}$ and $T_{(b,b)}$ may be different for different $a,b\in\Sigma$. It may be desirable to impose as an additional condition that $W(T_{(a,a)})> W(T_{(a,b)})$ for all $a\neq b$. The definition of $\{u\to v\}_\tau$ remains as before and the similarity $\mathcal{S}$ of two words $u$ and $v$ is defined to be \[\mathcal{S}(u,v)=\max_{\{u\to v\}_\tau}\sum_{k=1}^m w(T_{i_k}).\] For this definition to be equivalent to the one obtained from the Needleman-Wunsch algorithm, it is necessary that a condition similar to the condition {\bf M} is fulfilled: there must be at least one optimal sequence of transformations which corresponds to a sequence of transformations considered by the Needleman-Wunsch algorithm. This is not always the case in practice (see Section \ref{sec:matrices} below) and one then needs to assume in addition that only those transformations acting on each alignment position only once are allowed.

\subsection{Correspondence to distances}

The following observation allows conversion of similarity scores to quasi-metrics. 
\begin{lemma}[\cite{AS2004}]\label{lemma:sim2qm}
Let $X$ be a set and  $s:X\times X\to\R$ a map such that
\begin{enumerate}[(i)]
\item $s(x,x)> 0\quad \forall x\in X$,
\item $s(x,x) \geq s(x,y)\quad \forall x,y\in X$,
\item $s(x,y) = s(x,x) \wedge  s(y,x) = s(y,y)\implies x=y\quad \forall x,y\in X$,
\item $s(x,y)+s(y,z) \leq s(x,z) + s(y,y)\quad \forall x,y,z\in X$.
\end{enumerate}
Then $d:X\times X\to\R$ where $(x,y) \mapsto s(x,x) - s(x,y)$
is a quasi-metric. Furthermore, if $s$ is symmetric, that is, $s(x,y)=s(y,x)$ for all $x,y\in X$, $(X,d)$ is a co-weighted quasi-metric space with the co-weight $w: x \mapsto s(x,x)$.
\begin{proof}
Positivity of $d$ is equivalent to (ii), separation of points is equivalent to (iii) while the triangle inequality is equivalent to (iv). If $s(x,y)=s(y,x)$ then $\cj{d}(x,y)+s(x,x) = s(y,y)-s(x,y)+s(x,x) =s(x,x)-s(x,y)+s(y,y)= \cj{d}(y,x)+s(y,y)$ and since $s(x,x)> 0$ it follows that $w: x \mapsto s(x,x)$ is a co-weight.
\end{proof}
\end{lemma}

Obviously, if $s$ satisfies all the requirements of the Lemma \ref{lemma:sim2qm} and is symmetric, then $-s$ is a partial metric (Subsection \ref{subsec:partmetr}) and the Lemma \ref{lemma:sim2qm} is equivalent to the Theorem \ref{thm:partmetr2qm}.

\begin{lemma}\label{lemma:selfsim1}
Let $\Sigma$ be a set and $x\in\Sigma^*$. If $s:\Sigma\times\Sigma\to\R$ is a map satisfying the conditions (i) and (ii) of the Lemma \ref{lemma:sim2qm}, $g$ and $h$ are functions $\N^+\to\R_+$ and $S=\NW(x,x,s,g,h)$, then for all $i=0,1,\ldots,\abs{x}$ and for all $j\leq i$, \[S_{i,i}> S_{i,j}\qquad \text{and}\qquad S_{i,i}> S_{j,i}.\]
\begin{proof}
We prove our claim by induction. Let $\preceq$ denote a partial order on $\N\times\N$ where $(i_0,j_0)\preceq (i,j)$ if $i_0<i$ or $i_0=i$ and $j_0\leq j$ (lexicographic order). The relation $\preceq$ is well--founded of order type $\omega^2$ (but of course the induction is finite) and our claim is trivially true for $(0,0)$. Assume it is true for all $(i',j')\prec (i,j)$. 

If $i>0$ and $j=0$, we have for some $1\leq k\leq i$, $S_{i,0}=S_{i-k,0}-h(k) < S_{i,i}$ since $S_{i-k,0} < S_{i,i}$ by the induction hypothesis and $h$ is non-negative. In a similar way, it follows that $S_{i,i}> S_{0,i}$ since $g$ is non-negative. 

We now consider the case where $i>0$ and $0<j\leq i$ and show that $S_{i,i}> S_{i,j}$. If $S_{i,j}=S_{i-1,j-1}+s(x_i,x_j)$ we have $S_{i-1,j-1}< S_{i-1,i-1}$ by the induction hypothesis and $s(x_i,x_j)\leq s(x_i,x_i)$ by the condition (ii), and therefore $S_{i,i}> S_{i,j}$. If $S_{i,j}=S_{i-k,j}-h(k)$ for some $1\leq k\leq j$, the result follows since $g$ is a non-negative function and $S_{i-k,j}< S_{i,i}$ by the induction hypothesis. If $S_{i,j}=S_{i,j-k}-h(k)$, the same result follows by the induction hypothesis and non-negativity of $h$. The inequality $S_{i,i}> S_{j,i}$ follows by the same argument.
\end{proof}
\end{lemma}

\begin{corol}\label{corol:selfsim}
Suppose $s:\Sigma\times\Sigma\to\R$ is a function satisfying the conditions (i) and (ii) of the Lemma \ref{lemma:sim2qm}, $g$ and $h$ are functions $\N^+\to\R_+$ and $\mathcal{S}$ the global similarity on $\Sigma^*$ with respect to $s,g$ and $h$. Then, for all $x\in\Sigma^*$,  \[\mathcal{S}(x,x) = \sum_{i=1}^{\abs{x}} s(x_i,x_i).\]
\begin{proof}

Let $x\in\Sigma^*$. If $x=e$, by definition $\mathcal{S}(x,x)=0$, coinciding with a sum over an empty set. For $x\in\Sigma^+$, the Lemma \ref{lemma:selfsim1} directly implies the required result.
\end{proof}
\end{corol}

\begin{thm}\label{thm:simdist}
Suppose $s:\Sigma\times\Sigma\to\R$ is a map satisfying the conditions of the Lemma \ref{lemma:sim2qm} and let $g,h$ be increasing functions $\N^+\to\R$. Then, the formula 
\[\rho(x,y) = \mathcal{S}(x,x) - \mathcal{S}(x,y),\] where $x,y\in\Sigma^*$ and $\mathcal{S}$ is the global similarity (given  $s,g$ and $h$), defines a $\tau$-quasi-metric $\rho$ on $\Sigma^*$.
\begin{proof}
Set $d(a,b)=s(a,a)-s(a,b)$. By the Lemma \ref{lemma:sim2qm}, $d$ is co-weightable quasi-metric with co-weight $s(a,a)$. The Lemma \ref{lemma:weightLip} implies that a co-weight function is left 1-Lipschitz. Consider the set $(\tau_\lambda,w)$ of edit operations over $\Sigma^*$ where $w(T_{(a,b)})=d(a,b)$, $w(T_{v+})=g(v)$ and $w(T_{v-})=h(v)+\mathcal{S}(v,v)=h(v)+\sum_{i=1}^{\abs{v}} s(v_i,v_i)$. Let $\rho=\rho_{\tau_\lambda,w}$. By our assumptions, $(\tau_\lambda,w)$ satisfies the condition {\bf N} and hence, by the Theorem \ref{thm:condMindel}, the condition {\bf M}. By the Theorem \ref{thm:WSBrecurrence}, we have $\rho(\bar{x}_0,\bar{y}_0)=0$, $\rho(\bar{x}_0,\bar{y}_j)=
\min_{1\leq k\leq j}\left\{\rho(\bar{x}_{0},\bar{y}_{j-k})+ g(k)\right\}$, $\rho(\bar{x}_i,\bar{y}_0)= \min_{1\leq k\leq i}\left\{\rho(\bar{x}_{i-k},\bar{y}_{0})+ h(k) + \mathcal{S}(x_{i-k+1}\ldots x_i,x_{i-k+1}\ldots x_i)\right\}$, and for all $1\leq i\leq\abs{x}$, $1\leq j\leq\abs{y}$,
\begin{equation*}
\begin{split}
\rho(\bar{x}_i,\bar{y}_j) =\min\Bigg\lbrace & \rho(\bar{x}_{i-1},\bar{y}_{j-1})+s(x_i,x_i)-s(x_i,y_j),\\
&\min_{1\leq k\leq j}\left\{\rho(\bar{x}_{i},\bar{y}_{j-k})+ g(k)\right\},\\ 
&\min_{1\leq k\leq i}\left\{\rho(\bar{x}_{i-k},\bar{y}_{j})+ h(k) + \mathcal{S}(x_{i-k+1}\ldots x_i,x_{i-k+1}\ldots x_i)
\right\}\Bigg\rbrace.\\
\end{split}
\end{equation*}
We claim that for all $0\leq i\leq\abs{x}$, $0\leq j\leq\abs{y}$, $\displaystyle\rho(\bar{x}_i,\bar{y}_j) = \mathcal{S}(\bar{x}_i,\bar{x}_i) - S_{i,j}$, where $S=\NW(x,y,s,g,h)$.

It is clear that $\rho(\bar{x}_0,\bar{y}_0)=S_{0,0}$ and that $\rho(\bar{x}_i,\bar{y}_0)=\mathcal{S}(\bar{x}_i,\bar{x}_i)-S_{i,0}$. By the Lemma \ref{corol:selfsim}, $\mathcal{S}(\bar{x}_0,\bar{x}_0)=\mathcal{S}(e,e)=0$ and hence $\rho(\bar{x}_0,\bar{y}_j)=\mathcal{S}(\bar{x}_0,\bar{x}_0)-S_{0,j}$. Let $0\leq i'\leq m$, $0\leq j'\leq n$ and assume $\rho(\bar{x}_{i},\bar{y}_j) = \mathcal{S}(\bar{x}_{i},\bar{x}_{i}) - S_{i,j}$ for all 
$(i,j)$ such that $0\leq i\leq i'$ and $0\leq j\leq j'$ but excluding $(i',j')$. Then, 

\begin{align*}
\begin{split}
&\rho(\bar{x}_{i'},\bar{y}_{j'}) =\min\Bigg\lbrace 
\mathcal{S}(\bar{x}_{i'-1},\bar{x}_{i'-1}) - S_{i'-1,j'-1} +s(x_{i'},x_{i'}) -s(x_{i'},y_{j'}),\\ 
&\quad\min_{1\leq k\leq j'}\left\{\mathcal{S}(\bar{x}_{i'},\bar{x}_{i'}) - S_{i',j'-k}+ g(k)\right\}\\
&\quad\min_{1\leq k\leq i'}\left\{\mathcal{S}(\bar{x}_{i'-k},\bar{x}_{i'-k}) - S_{i'-k,j'} + h(k) +\mathcal{S}(x_{i'-k+1}\ldots x_{i'},x_{i'-k+1}\ldots x_{i'})\right\}\Bigg\rbrace\\
\end{split} \displaybreak[0] \\
\begin{split} 
\phantom{\rho(\bar{x}_{i'},\bar{y}_{j'})} =\min\Bigg\lbrace &
\mathcal{S}(\bar{x}_{i'},\bar{x}_{i'}) - S_{i'-1,j'-1} -s(x_{i'},y_{j'}),\\ 
&\min_{1\leq k\leq j'}\left\{\mathcal{S}(\bar{x}_{i'},\bar{x}_{i'}) - S_{i',j'-k}+ g(k)\right\},\\
&\min_{1\leq k\leq i'}\left\{\mathcal{S}(\bar{x}_{i'},\bar{x}_{i'}) - S_{i'-k,j'} + h(k) \right\}\Bigg\rbrace\\
\end{split} \displaybreak[0] \\
\begin{split} 
\phantom{\rho(\bar{x}_{i'},\bar{y}_{j'})} =\mathcal{S}(\bar{x}_{i'},\bar{x}_{i'}) -\max\Bigg\lbrace &
S_{i'-1,j'-1} +s(x_{i'},y_{j'}),\\ 
&\max_{1\leq k\leq j'}\left\{S_{i',j'-k}- g(k)\right\},\\
&\max_{1\leq k\leq i'}\left\{S_{i'-k,j'}- h(k) \right\}\Bigg\rbrace\\
\end{split} \displaybreak[0] \\
\begin{split} 
\phantom{\rho(\bar{x}_{i'},\bar{y}_{j'})} =\mathcal{S}(\bar{x}_{i'},\bar{x}_{i'}) - S_{i',j'}, &\\
\end{split} \displaybreak[0] \\
\end{align*}
and our claim follows by induction. In particular, $\rho(x,y)= \mathcal{S}(\bar{x}_{m},\bar{x}_{m}) - S_{m,n} = \mathcal{S}(x,x) - \mathcal{S}(x,y)$ as required.
\end{proof}
\end{thm}

\begin{example}
It is well known \cite{Gusfield97} that the longest common subsequence problem can be approached using similarities rather than distances. Let $\Sigma$ be a set and set for all $a,b\in\Sigma$, $s(a,a)=1$ and $s(a,b)=0$ if $a\neq b$. Let $g(k)=h(k)=0$ for all $k\in\N^+$. It is easy to confirm that for $x,y\in\Sigma^*$, $\mathcal{S}(x,y)=\abs{LCS(x,y)}$. 

By the Theorem \ref{thm:simdist}, $d(x,y)=\mathcal{S}(x,x)-\mathcal{S}(x,y)= \abs{x}-\abs{LCS(x,y)}$ gives a co-weightable quasi-metric with co-weight $\abs{\cdot}$. The metric $\qsum{d}$ is the metric $\rho_{LCS}$ from the Example \ref{ex:lcs}. The associated order $\leq_d$ is clearly the subsequence order: \[x\leq_d y\iff\text{$x$ is a subsequence of $y$},\] and $(\Sigma^*,\leq_d)$ forms a meet semilattice where $x\sqcap y=LCS(x,y)$.

The partial order $(\Sigma^*,\leq_d)$ is an example of an invariant meet semilattice (Definition \ref{def:invmeetsemilattice}) since \[d(x\sqcap z,y\sqcap z)=\abs{x\sqcap z} -\abs{x\sqcap y\sqcap z}\leq d(x\sqcap z,x) +d(x,y)=d(x,y).\] By the Theorem \ref{thm:semilattice_weighted}, the map $f=\abs{\cdot}$ is a meet valuation and $d(x,y)=f(x)-f(x\sqcap y)$.
\end{example}

\section{Local Similarity}

Presently, most biological sequence comparison is done using local rather than global similarity measures. The principal reason is that elements of biological function whose detection is desired are usually restricted to discrete fragments of sequences and the strong similarity of fragments of two sequences may not extend to similarity of full sequences. For example, the structure of a protein consists of discrete structural domains interspersed with random coils linking them and variation is much higher in the parts not directly related to the function. Thus, even relatively closely related protein sequences may show little similarity outside the functionally important regions and their global similarity may not be significant.

The similar phenomenon occurs in DNA sequences, where events other than point mutations and insertions and deletions, such as inversions or translocations, may occur between very closely related sequences. Therefore, local similarity measures, and the associated local alignments between two sequences are most appropriate for general comparison of biological sequences. A dynamic programming algorithm for computation of local similarities, of the same complexity as the Needleman-Wunsch algorithm was proposed by Smith and Waterman in 1981 \cite{SW81}. While its cubic (quadratic if gap penalties are affine) complexity renders it not very suitable for sequential searches of large datasets, it remains the canonical yardstick with which the accuracy of any heuristic algorithms is assessed. We therefore follow the precedent of the previous section and define local similarity between two sequences using a dynamic programming matrix. 

\begin{defin}
Let $\Sigma$ be a set, $x,y\in\Sigma^*$, $s:\Sigma\times\Sigma\to\R$ and $g,h:\N^+\to\R_+$. Let $x,y\in\Sigma^*$ and let $m=\abs{x}$ and $n=\abs{y}$. The \emph{Smith-Waterman} dynamic programming matrix, denoted $\SW(x,y,s,g,h)$, is an $(m+1)\times(n+1)$ matrix $H$ with rows and columns indexed from $0$ such that $H_{0,0}=H_{i,0}=H_{0,j}=0$ and for all $i=1,2\ldots m$ and $j=1,2\ldots n$ 
\begin{equation*}
\begin{split}
H_{i,j}=\max\bigg\lbrace & 0, H_{i-1,j-1}+s(x_i,y_j),\\ & \max_{1\leq k\leq i}\left\{H_{i-k,j}-h(k)\right\}, \max_{1\leq k\leq j}\left\{H_{i,j-k}-g(k)\right\} \bigg \rbrace.
\end{split}
\end{equation*}
We define the \emph{local similarity} between the sequences $x$ and $y$ (given $s$, $g$, and $h$), denoted $\mathcal{H}(x,y)$, to be the largest entry of $H$, that is, $\mathcal{H}(x,y)= \max_{i,j} H_{i,j}$.
\end{defin}

An optimal edit script and a corresponding alignment is retrieved from $H$ by a slightly modified traceback procedure: the traceback starts at $(i,j)$ such that $H_{i,j}$ is maximal and ends at an entry of $H$ with a value of $0$ (Example \ref{ex:dynprog3}). Clearly, no traceback is possible if $H\equiv 0$.

Two additional requirements are usually associated with the Smith-Waterman algorithm: the expected value of $s$ must be negative and at least for some $a,b\in\Sigma$, $s(a,b)$ must be positive. The first requirement obviously requires a probability measure on $\Sigma$ and exists to ensure that the alignments retrieved are indeed local rather than global or close to global. The second requirement ensures that pairs of sequences with a positive local similarity score exist.

\begin{example}\label{ex:dynprog3}
Consider the English words $u=\text{\tt COMPLEXITY}$ and\\ $v=\text{\tt FLEXIBILITY}$ from the Example \ref{ex:dynprog3}. Suppose $s(a,a)=3$, $s(a,b)=-1$ if $a\neq b$ and let $g(k)=h(k)=9+k$. The matrix $H=\SW(u,v,s,g,h)$ is given in the Table \ref{tbl:SWsimscore}. The local similarity score is 12 -- the corresponding alignment is the exact match of the common substring {\tt LEXI}.
\newcolumntype{Y}{>{\centering\arraybackslash}X}
\begin{table}[!ht]
{\footnotesize\tt
\begin{tabularx}{\linewidth-5mm}{|YY||Y|Y|Y|Y|Y|Y|Y|Y|Y|Y|Y|Y|}
\hline
   &   &  0&  1&  2&  3&  4&  5&  6&  7&  8&  9& 10& 11\\ 
   &   &   &  F&  L&  E&  X&  I&  B&  I&  L&  I&  T&  Y\\ \hline\hline 
  0&   &  0&  0&  0&  0&  0&  0&  0&  0&  0&  0&  0&  0\\ \hline
  1&  C&  0&  0&  0&  0&  0&  0&  0&  0&  0&  0&  0&  0\\ \hline
  2&  O&  0&  0&  0&  0&  0&  0&  0&  0&  0&  0&  0&  0\\ \hline
  3&  M&  0&  0&  0&  0&  0&  0&  0&  0&  0&  0&  0&  0\\ \hline
  4&  P&  0&  0&  0&  0&  0&  0&  0&  0&  0&  0&  0&  0\\ \hline
  5&  L&  0&  0&  $\nwarrow${\bf 3}&  0&  0&  0&  0&  0&  3&  0&  0&  0\\ \hline
  6&  E&  0&  0&  0& $\nwarrow${\bf 6}&  0&  0&  0&  0&  0&  2&  0&  0\\ \hline
  7&  X&  0&  0&  0&  0&  $\nwarrow${\bf 9}&  0&  0&  0&  0&  0&  1&  0\\ \hline
  8&  I&  0&  0&  0&  0&  0& $\nwarrow${\bf 12}&  2&  3&  0&  3&  0&  0\\ \hline
  9&  T&  0&  0&  0&  0&  0&  2& 11&  1&  2&  0&  6&  0\\ \hline
 10&  Y&  0&  0&  0&  0&  0&  1&  1& 10&  0&  1&  0&  9\\ \hline
\end{tabularx}
}
\caption[An example of a dynamic programming table for computation of Smith-Waterman local similarity between two strings.]{The dynamic programming table used to compute the Smith-Waterman local similarity between the strings {\tt COMPLEXITY} and {\tt FLEXIBILITY}. The path recovering the optimal alignment is shown in bold.} \label{tbl:SWsimscore} 
\end{table}

\end{example}

The local similarity between two words as defined using the Smith-Waterman algorithm can be realised as a global similarity between some of their fragments (provided there exist two fragments with positive global similarity). Recall that we use $\mathfrak{F}(x)$ to denote the set of all factors (or fragments) of $x\in\Sigma^*$.

\begin{lemma}\label{lemma:globloc1}
Let $\Sigma$ be a set, $x,y\in\Sigma^*$, $s:\Sigma\times\Sigma\to\R$ and $g,h:\N^+\to\R_+$. Suppose $\mathcal{H}(x,y)>0$. Then there exist $x'\in\mathfrak{F}(x)$ and $y'\in\mathfrak{F}(y)$ such that $\mathcal{H}(x,y)=\mathcal{S}(x',y')$, where both global and local similarities are taken with respect to $s,g$ and $h$.
\begin{proof}
Since $\mathcal{H}(x,y)>0$, it follows that $x,y\in\Sigma^+$. We find $x'\in\mathfrak{F}(x),y'\in\mathfrak{F}(y)$ by traceback. Let $H=\SW(x,y,s,g,h)$. By definition of local similarity there exist $i_0,j_0$ such that $\mathcal{H}(x,y)=H_{i_0,j_0}>0$. We trace back the path of cells of the Smith-Waterman dynamic programming matrix from $(i_0,j_0)$ to a zero entry by constructing a sequence $\langle(i_k,j_k\rangle_{k=0}^m$ such that $H_{i_0,j_0}=\mathcal{H}(x,y)$, $H_{i_m,j_m}=0$ and $i_{k+1}\leq i_k$, $j_{k+1}\leq j_k$ in the following way. For each $k$, if $H_{i_k,j_k}=0$ stop. Otherwise, if $H_{i_k,j_k}=H_{i_k-1,j_k-1}+s(x_i,y_i)$, set $(i_{k+1},j_{k+1})=(i_k-1,j_k-1)$; if $H_{i_k,j_k}=H_{i_k,j_k-l}-g(l)$, set $(i_{k+1},j_{k+1})=(i_k,j_k-l)$; if $H_{i_k,j_k}=H_{i_k-l,j_k}-h(l)$, set $(i_{k+1},j_{k+1})=(i_k-l,j_k)$. Such sequence always exists since $H_{i_0,j_0}>0$. Furthermore, since $g$ and $h$ are non-negative, it follows that $i_m<i_0$ and $j_m<j_0$. Let $x'=x_{i_m+1}x_{i_m+2}\ldots x_{i_0}$, $y'=y_{j_m+1}y_{j_m+2}\ldots y_{j_0}$ and $S=\NW(x',y',s,g,h)$. Comparing the definitions of global and local similarities, it is easy to see that $S_{\abs{x'},\abs{y'}}=H_{i_0,j_0}$.
\end{proof}
\end{lemma}

\begin{corol}\label{corol:globloc2}
Let $\Sigma$ be a set, $x,y\in\Sigma^*$, $s:\Sigma\times\Sigma\to\R$ and $g,h:\N^+\to\R_+$. Then \[\mathcal{H}(x,y)=\max_{\substack{x'\in\mathfrak{F}(x) \\ y'\in\mathfrak{F}(y)}}\mathcal{S}(x',y')\vee 0.\]
\begin{proof}
Let $H=\SW(x,y,s,g,h)$ and $S=\NW(x,y,s,g,h)$. It can be easily verified from the definitions (for example by induction) that for all $i,j$, $H_{i,j}\geq S_{i,j}$ and therefore for all $x'\in\mathfrak{F}(x), y'\in\mathfrak{F}(y)$, $\mathcal{H}(x,y)\geq \mathcal{H}(x',y')\geq\mathcal{S}(x',y')$. If $\mathcal{H}(x,y)>0$, the Lemma \ref{lemma:globloc1} implies $\mathcal{H}(x,y)\leq\max\{\mathcal{S}(x',y')\ |\ x'\in\mathfrak{F}(x), y'\in\mathfrak{F}(y)\}$. 
\end{proof}
\end{corol}


We now present the main result of this chapter which gives the conditions for conversion of local similarity scores on a free semigroup to a quasi-metric. We first introduce a necessary technical condition.

\begin{thm}\label{thm:locsim2qm}
Let $\Sigma$ be a set and $f$ a strictly positive function $\Sigma\to\R$. Let $\rho$ be a metric on $\Sigma^*$ and let $\bar{f}$ be the canonical homomorphic extension of $f$ to the free semigroup $\Sigma^*$ given by $\bar{f}(x)=\sum_{i=1}^{\abs{x}} f(x_i)$ for all $x\in\Sigma^+$ and $\bar{f}(e)=0$. Suppose that for all $x,y\in\Sigma^*$, 
\begin{equation}\label{eq:locsim2qm2} 
\abs{\bar{f}(x)-\bar{f}(y)}\leq\rho(x,y)\leq \bar{f}(x)+\bar{f}(y), 
\end{equation}
and 
\begin{equation}\label{eq:locsim2qm3} 
\bar{f}(x)-\bar{f}(y) = \rho(x,y)\iff\quad y\in\mathfrak{F}(x), 
\end{equation}
then $d:\Sigma^*\times\Sigma^* \to\R$ defined by 
\begin{equation*}
d(x,y) = \bar{f}(x) - \frac12\max_{\substack{\tilde{x}\in\mathfrak{F}(x) \\ \tilde{y}\in\mathfrak{F}(y)}} \{\bar{f}(\tilde{x})+ \bar{f}(\tilde{y}) -\rho(\tilde{x},\tilde{y})\}
\end{equation*}
is a co-weightable quasi-metric with co-weight $\bar{f}$.
\begin{proof}
Let $x,y\in\Sigma^*$. Since $\bar{f}(x)\geq\bar{f}(\tilde{x})$ for any $\tilde{x}\in\mathfrak{F}(x)$ and since  (\ref{eq:locsim2qm2}) implies that $\bar{f}$ is 1-Lipschitz, it follows that $d(x,y)\geq 0$. It is also clear that $d(x,x)=0$. If $d(x,y)=0$, there exists $\tilde{x}\in\mathfrak{F}(x)$ and $\tilde{y}\in\mathfrak{F}(y)$ such that 
\begin{equation}\label{eq:a1}
\bar{f}(x)-\frac12 \left(\bar{f}(\tilde{x}) +\bar{f}(\tilde{y}) -\rho(\tilde{x},\tilde{y})\right)=0.
\end{equation}
Since $\tilde{x}\in\mathfrak{F}(x)$, there exist $u,v\in\Sigma^*$ such that $x=u\tilde{x}v$
and the Equation \ref{eq:a1} becomes \[\bar{f}(u)+\bar{f}(v)+\frac12 (\bar{f}(\tilde{x})- \bar{f}(\tilde{y})+ \rho(\tilde{x},\tilde{y})) =0.\] Since $\bar{f}(u)\geq 0$, $\bar{f}(v)\geq 0$ and $\bar{f}(\tilde{x}) -\bar{f}(\tilde{y})+ \rho(\tilde{x},\tilde{y})\geq 0$ ($\bar{f}$ is 1-Lipschitz), it must follow that $\bar{f}(u)=0$, $\bar{f}(v)=0$ and 
\begin{equation}\label{eq:a2}
\bar{f}(\tilde{x}) -\bar{f}(\tilde{y}) +\rho(\tilde{x},\tilde{y}) =0.
\end{equation}
From $\bar{f}(u)=0$ and $\bar{f}(v)=0$ we conclude that $u=e$, $v=e$ and $x=\tilde{x}$ while 
(\ref{eq:locsim2qm3})  implies that $x=\tilde{x}\in\mathfrak{F}(\tilde{y})$. Hence, since the maximum in the definition of $d(x,y)$ is invariant under permutation of $x$ and $y$, it follows that $d(x,y)=d(y,x)=0$ implies $x=\tilde{x}\in\mathfrak{F}(\tilde{y})$ and $y=\tilde{y}\in\mathfrak{F}(\tilde{x})$ and hence that $x=y$.

Now let $x,y,z\in\Sigma^*$ and suppose $d(x,y)=\bar{f}(x)-\frac12 \left(\bar{f}(\tilde{x}) +\bar{f}(\tilde{y}) -\rho(\tilde{x},\tilde{y})\right)$ and $d(y,z)=\bar{f}(y)-\frac12 \left(\bar{f}(\bar{y}) +\bar{f}(\bar{z}) -\rho(\bar{y},\bar{z})\right)$ for some $\tilde{x}\in\mathfrak{F}(x)$, $\tilde{y},\bar{y}\in\mathfrak{F}(y)$ and $\bar{z}\in\mathfrak{F}(z)$. Write out $\tilde{y}=y_iy_{i+1}\ldots y_{i+m-1}$, $\bar{y}=y_jy_{j+1}\ldots y_{j+n-1}$ where $m=\abs{\tilde{y}}$, $n=\abs{\bar{y}}$, $1 \leq i\leq i+m-1\leq \abs{y}$ and $1 \leq j\leq j+n-1\leq \abs{y}$.
 
If $\tilde{y}$ and $\bar{y}$ overlap, that is, if $i \leq j\leq m$ or $j \leq i\leq n$, let $y'$ denote the whole overlapping fragment (for example, if $i \leq j\leq i+m-1\leq i+n-1$, $y'=y_jy_{j+1}\ldots y_{i+m-1}$). If $\tilde{y}$ and $\bar{y}$ do not overlap or either $\tilde{y}$ or $\bar{y}$ is identity, let $y'=e$. Since $y'\in\mathfrak{F}(\tilde{y})$ and $y'\in\mathfrak{F}(\bar{y})$, by the triangle inequality on $\rho$ and 
by (\ref{eq:locsim2qm3}), we have
\begin{align*}
\rho(\tilde{x},\tilde{y})& \geq \rho(\tilde{x},y') - \rho(\tilde{y},y') =\rho(\tilde{x},y') + \bar{f}(y') -\bar{f}(\tilde{y}) \qquad \text{and}\\
\rho(\bar{y},\bar{z})& \geq \rho(y',\bar{z}) - \rho(y',\bar{y}) =\rho(y',\bar{z}) +\bar{f}(y') -\bar{f}(\bar{y}).
\end{align*}
Since $y'$ denotes the full extent of overlap of $\tilde{y}$ and $\bar{y}$, it follows that \[\bar{f}(y) +\bar{f}(y') -\bar{f}(\tilde{y}) -\bar{f}(\bar{y})\geq 0\] and therefore
\begin{align*}
\begin{split} 
d(x,y)+d(y,z)= &\phantom{+\enspace }\bar{f}(x)-\frac12 \left(\bar{f}(\tilde{x}) +\bar{f}(\tilde{y}) -\rho(\tilde{x},\tilde{y})\right) \\
&+ \bar{f}(y)-\frac12 \left(\bar{f}(\bar{y}) +\bar{f}(\bar{z}) -\rho(\bar{y},\bar{z})\right) \\
\end{split} \displaybreak[0] \\
\begin{split} 
\phantom{d(x,y)+d(y,z)}\geq &\phantom{+\enspace }\bar{f}(x)-\frac12 \left(\bar{f}(\tilde{x}) +2\bar{f}(\tilde{y}) -\bar{f}(y') -\rho(\tilde{x},y')\right) \\
&+ \bar{f}(y)-\frac12 \left(2\bar{f}(\bar{y}) +\bar{f}(\bar{z}) -\bar{f}(y') -\rho(y',\bar{z})\right) \\
\end{split} \displaybreak[0] \\
\begin{split} 
\phantom{d(x,y)+d(y,z)}\geq &\phantom{+}\bar{f}(x)-\frac12 \left(\bar{f}(\tilde{x}) +\bar{f}(\bar{z}) -\rho(\tilde{x},y') - \rho(y',\bar{z}) \right) \\
&+ \bar{f}(y) +\bar{f}(y') -\bar{f}(\tilde{y}) -\bar{f}(\bar{y}) \\
\end{split} \displaybreak[0] \\
\begin{split} 
\phantom{d(x,y)+d(y,z)}\geq &\phantom{+}\bar{f}(x)-\frac12 \left(\bar{f}(\tilde{x}) +\bar{f}(\bar{z}) -\rho(\tilde{x},\bar{z}) \right) \\
\end{split} \displaybreak[0] \\
\begin{split} 
\phantom{d(x,y)+d(y,z)}\geq &\phantom{+} d(x,z).\\
\end{split} \displaybreak[0]
\end{align*}
The fact that $d$ is co-weightable with co-weight $\bar{f}$ follows straight from the definition of $d$.
\end{proof}
\end{thm}

\begin{remark}
In general, the property (\ref{eq:locsim2qm2}) means that $\bar{f}$ can be interpreted as a distance from an abstract point $\star$ with respect to a metric on the set $\Sigma^*\cup\{\star\}$. Flood, in his PhD thesis \cite{Fl75} and a followup paper \cite{Fl84}, introduced the term \emph{norm pair} to denote the pair $(\rho,\bar{f})$ satisfying the property (\ref{eq:locsim2qm2}). However, in the context of the Theorem \ref{thm:locsim2qm}, it is clear that $\bar{f}(x)=\rho(x,e)$. Hence, the property (\ref{eq:locsim2qm2}) can be reformulated to state: for all $x\in\Sigma^*$, $\rho(x,e)$ is given by a canonical homomorphic extension of a strictly positive function on the set of generators.

The following Lemma \ref{lemma:metrweight} is a folklore result, see e.g. Flood's paper \cite{Fl84}, but we present the proof for the sake of completeness and because we could not find a reference that would be readily available for the reader.
\end{remark}

\begin{lemma}[\cite{Fl84}]\label{lemma:metrweight}
Let $(X,d)$ be a metric space and $f:X\to\R_+$ a positive 1-Lipschitz function. Then, the map $\rho:X\times X\to\R_+$ defined by \[\rho(x,y)=\min\{d(x,y),f(x)+f(y)\}\] is a metric.
\begin{proof}
Let $x,y,z\in X$. Clearly $\rho(x,x)=0$ and $\rho(x,y)=\rho(y,x)$. Since $f$ is positive, $\rho(x,y)=0\implies d(x,y)=0$ and hence $x=y$. For the triangle inequality we consider four cases. If $\rho(x,y)=d(x,y)$ and $\rho(y,z)=d(y,z)$, $\rho(x,y)+\rho(y,z)\geq \rho(x,z)$ by the triangle inequality of $d$. If $\rho(x,y)=d(x,y)$ and $\rho(y,z)=f(y)+f(z)$ we have $\rho(x,y)+\rho(y,z)\geq f(x)+f(z)\geq \rho(x,z)$. In the case where $\rho(x,y)=f(x)+f(y)$ and $\rho(y,z)=d(y,z)$ the result follows in the same way. Finally, if $\rho(x,y)=f(x)+f(y)$ and $\rho(y,z)=f(y)+f(z)$, we have $\rho(x,y)+\rho(y,z)\geq f(x)+f(z)+2f(y)\geq \rho(x,z)$ since $f$ is positive.
\end{proof}
\end{lemma}

\begin{corol}\label{corol:localsimdist}
Let $\Sigma$ be a set. Suppose $g$ is an increasing functions $\N^+\to\R$, $h=g$ and $s:\Sigma\times\Sigma\to\R$ is a map satisfying the conditions of the Lemma \ref{lemma:sim2qm} and being symmetric, that is $s(b,a)=s(a,b)$ for all $a,b\in\Sigma$. Let $\mathcal{H}$ be the local similarity with respect to $s,g$ and $h$. Then, a function $d:\Sigma^*\times\Sigma^*\to\R_+$ given by \[d(x,y) = \mathcal{H}(x,x) -\mathcal{H}(x,y)\] is a co-weightable quasi-metric with co-weight $x\mapsto\mathcal{H}(x,x)$ (equivalently, $-\mathcal{H}$ is a partial metric).

\begin{proof}
Let $\mathcal{S}$ be the global similarity with respect to $s,g$ and $h$. Clearly, $\mathcal{S}$ is symmetric since $s$ is symmetric and $g=h$. Let $\rho_0(x,y)=\mathcal{S}(x,x)-\mathcal{S}(x,y)$ for $x,y\in\Sigma^*$ and let $\mathcal{S}_0(x)=\mathcal{S}(x,x)=\sum_{k=1}^{\abs{x}} s(x_i,x_i)$ (Corollary \ref{corol:selfsim}). By the Theorem \ref{thm:simdist}, $\rho_0$ is a co-weighted quasi-metric with a co-weight $\mathcal{S}_0$ and therefore $\qsum{\rho}_0(x,y)=\mathcal{S}(x,x) +\mathcal{S}(y,y)-\mathcal{S}(x,y) -\mathcal{S}(y,x)$ is a metric and $\mathcal{S}_0$ is 1-Lipschitz with respect to $\qsum{\rho}_0$. By the Lemma \ref{lemma:metrweight}, $\rho(x,y)=\min\{\qsum{\rho}_0(x,y),\mathcal{S}_0(x)+ \mathcal{S}_0(y) \}$ gives a metric. 

It is easy to see that for all $x,y\in\Sigma^*$, \[\mathcal{S}(x,y)\vee 0=\frac12\left(  \mathcal{S}_0(x)+ \mathcal{S}_0(y)-\rho(x,y)\right),\] and hence, by the Corollary \ref{corol:globloc2}, \[\mathcal{H}(x,y)= \frac12\max_{\substack{\tilde{x}\in\mathfrak{F}(x) \\ \tilde{y}\in\mathfrak{F}(y)}} \{\mathcal{S}_0(\tilde{x})+ \mathcal{S}_0(\tilde{y}) -\rho(\tilde{x},\tilde{y})\}.\] Furthermore, $\mathcal{H}(x,x)=\mathcal{S}(x,x)$ since $s(a,a)>0$ for all $a\in\Sigma$. 

The main statement then follows from the Theorem \ref{thm:locsim2qm} and the remark of $-\mathcal{H}$ being a partial metric follows from the Theorem \ref{thm:partmetr2qm}. 
\end{proof}
\end{corol}

\begin{remark}
An alternative treatment of the same problem is given in the {\it Topology Proc.} paper by the thesis author. There however, a different definition of an alignment is given and the statement of the main theorem explicitly uses the properties of score matrices and gap penalties. Theorem \ref{thm:locsim2qm} is a more general statement of the same fact.
\end{remark}

It is clear from the proof of the Theorem \ref{thm:locsim2qm} that the partial order $\leq_d$ associated to the quasi-metric $d$ of Corollary \ref{corol:localsimdist} is a substring (factor) order: \[x\leq_d y \iff x\in\mathfrak{F}(y).\] The set $\Sigma^*$ with $\leq_d$ forms a meet semilattice. However, in general, $d$ is not invariant with respect to the concatenation or meet operation. For example, let $\Sigma=\{a,b,c\}$ and for all $\sigma,\tau\in\Sigma$ set 
\begin{equation*} 
s(\sigma,\tau)= 
\begin{cases}
1 & \text{if $\sigma=\tau$},\\
-5 & \text{otherwise}.
\end{cases}
\end{equation*}
Let $g(k)=h(k)=10+k$ and suppose $\mathcal{H}$ is a global similarity with respect to $s,g$ and $h$. If $x=aabb$, $y=bbbc$ and $z=aabc$, it is easy to verify that $x\sqcap z=aab$, $y\sqcap z=bc$, $d(x,y)=2$ and $d(x\sqcap z, y\sqcap z)=3>d(x,y)$, and hence $d$ is not invariant with respect to $\sqcap$. On the other hand if $x=aaab$, $y=aaa$ and $z=c$, we have $d(x,y)=1$ while $d(xz,yz)=2$ and therefore $d$ is not invariant with respect to string concatenation.

\section{Score Matrices}\label{sec:matrices}

The main result from the previous section indicates that, at least under some circumstances, free semigroups with local similarity measures can be considered as partial metric spaces, or equivalently, as co-weighted quasi-metric spaces. A consequence of the Theorem \ref{thm:partmetr2qm} of particular significance for biological applications is the fact that the transformation into quasi-metric preserves neighbourhoods with respect to similarity scores.

Let $x\in\Sigma^*$ and define for some $t>0$ \[\mathscr{N}_t(x)=\{y\in\Sigma^*:\mathcal{H}(x,y)\geq t\},\] that is, $\mathscr{N}_t(x)$ is the set of all points in $\Sigma^*$ whose local similarity with $x$ is not less than $t$. Retrieving points belonging to such neighbourhoods from datasets is the principal aim of similarity search, explored in detail in Chapter \ref{ch:3}. Corollary \ref{corol:localsimdist} implies that there exists a co-weightable quasi-metric $d$ with co-weight $w$ such that $\mathscr{N}_t(x)=\clball{x}{w(x)-t}$ (i.e. the neighbourhood system consisting of $\mathscr{N}_t(x)$ for all $x$ and $t$ form a base for a quasi-metrisable topology). Therefore, one can expect that existing and newly developed indexing techniques for similarity search in (weightable) quasi-metric spaces (see Chapter \ref{ch:3}) can be used to significantly speed-up sequence similarity searches without significant sacrifice in accuracy. Furthermore, the result makes it worthwhile to repeat the exploration of global geometry of proteins performed by Linial, Linial, Tishby and Yona \cite{LLTY97}, this time in the context of quasi-metrics. 

The current section explores the similarity measures (commonly called \emph{score matrices} for obvious reasons) on DNA and protein alphabets which satisfy the Lemma \ref{lemma:sim2qm} and which hence, with affine gap penalties, lead to local similarities corresponding to quasi-metrics. In particular, the most popular members of the BLOSUM \cite{Henikoff:1992} family of matrices satisfy all the requirements of the Lemma \ref{lemma:sim2qm}, unlike the members of the PAM family \cite{Dayhoff:1978}, which do not and which are therefore omitted from the discussion here.

\subsection{DNA score matrices}

The DNA alphabet consists of only 4 letters (nucleotides) and the frequently used similarity measures on it are very simple. The common feature of all general DNA matrices used in practice is that they are symmetric and that self-similarities of all nucleotides are equal. The consequence of this fact is that the distance $d$ resulting from the transformation $d(a,b)=s(a,a)-s(a,b)$ is always a metric and the co-weightable quasi-metric arising from local similarity on DNA sequences has co-weight proportional to the length of a sequence.

For example, the score matrix used by BLAST (more precisely, the {\it blastn} program for search of DNA database with DNA query sequence) is given by 
\begin{equation*}
s(a,b)=\begin{cases}
        5 & \text{if $a=b$}\\
        -4 & \text{if $a\neq b$}.\\
       \end{cases}
\end{equation*}
More complex score matrices, mostly distance-based and used in phylogenetics also exist.

\subsection{BLOSUM matrices}\label{subsec:BLOSUM}

As the protein alphabet consists of 20 amino acids of markedly different chemical properties and structural roles, it is to be expected that similarity measures on amino acids involved in protein sequence comparison are more complex. The BLOSUM family of matrices was constructed by Steven and Jorja Henikoff in 1992 \cite{Henikoff:1992} who also showed that one member of the family, the BLOSUM62 matrix, gave the best search performance amongst all score matrices used at the time. For that reason, BLOSUM62 matrix is the default matrix used by NCBI BLAST for searches of protein databases. 

The BLOSUM similarity scores are explicitly constructed as \emph{log-odds} ratios. Let $\Sigma$ be a (finite) set and let $p$ be a probability measure on $\Sigma$. The value of $p(a)$ is called the \emph{background frequency} of $a\in\Sigma$. Let $q$ be a probability measure on $\Sigma\times\Sigma$. The value of $q(a,b)$ is called the \emph{target frequency} of a match between $a$ and $b$, that is the likelihood that $a$ is aligned with $b$ in related sequences. For unrelated sequences, we expect that the probability of $a$ being aligned with $b$ would be $p(a)p(b)$. The similarity score $s(a,b)$ is defined (up to a scaling factor) by \[s(a,b)=\log\frac{q(a,b)}{p(a)p(b)}.\] Thus, $s(a,b)$ is positive if the target frequencies are greater than background frequencies, $0$ if they are equal and negative if background frequencies are greater. In this model, the condition (iv) of the Lemma \ref{lemma:sim2qm} (the triangle inequality of the corresponding quasi-metric) is equivalent to \[q(a,b)q(b,c)\leq q(a,c)q(b,b)\] for all $a,b,c\in\Sigma$ and can be interpreted as stating that a direct substitution of one letter to another on each site in the sequence is always preferred to two or more substitutions achieving the same transformation. It should be noted that according to Altschul \cite{Alt91}, who studied the statistics of scores of ungapped local alignments, any similarity score matrix can be interpreted as log-odds ratios (i.e. target frequencies can be derived from similarity scores given the background frequencies).

The target frequencies used to obtain the BLOSUM scores were derived from multiple alignments. A \emph{multiple alignment} between $n$ sequences can be defined in the similar way as a pairwise alignment between two sequences according to the Definition \ref{def:alignment}: it is only necessary to replace the sequence of pairs with a sequence of $n$-tuples and to adjust the remainder of the definition accordingly. The (ungapped) multiple alignments of related sequences (also called blocks) used to construct the BLOSUM similarities were obtained from the BLOCKS database of protein motifs of Henikoff and Henikoff \cite{HenikoffH91}.  

In order to reduce the contribution of too closely related members of blocks to target frequencies, members of blocks sharing at least $L\%$ identity were clustered together and considered as one sequence (for a block member to belong to a cluster, it was sufficient for it to share $L\%$ identity with one member of the cluster), resulting in a family of matrices. Thus, the matrix BLOSUM62 corresponds to $L=62$ (for BLOSUMN, no clustering was performed). After clustering, the target frequencies were obtained by counting the number of each pair of amino acids in each column in each block having more than one cluster and normalising by the total number of pairs. The background frequencies were obtained from the amino acid composition of the clustered blocks and log-odds ratios taken. The resulting score matrices are necessarily symmetric since the pair $(a,b)$ cannot be distinguished from $(b,a)$ in the multiple alignment.  

\begin{table}[!ht]
{\small
\begin{tabular}{|lr||l|r||lr|}
\hline
Matrix & Failures & Matrix  & Failures & Matrix  & Failures \\ \hline
BLOSUM30 & 44 & BLOSUM60 & 0  & BLOSUM80  & 0 \\ 
BLOSUM35 & 10 & BLOSUM62 & 0  & BLOSUM85  & 0 \\ 
BLOSUM40 & 6  & BLOSUM65 & 0  & BLOSUM90  & 0 \\ 
BLOSUM45 & 0  & BLOSUM70 & 2  & BLOSUM100 & 0 \\ 
BLOSUM50 & 0  & BLOSUM75 & 2  & BLOSUMN   & 0 \\ 
BLOSUM55 & 2  &          &    &           &   \\ \hline
\end{tabular}
}
\caption[Numbers of triples of amino acids failing the triangle inequality in the BLOSUM family of score matrices.]{Numbers of triples of amino acids failing the triangle inequality in the BLOSUM family of score matrices. Note that all BLOSUM matrices are symmetric and thus the number of independent triples is half the number reported. For BLOSUM55, BLOSUM70, and BLOSUM75, the one independent triple failing consists of amino acids I, V and A, that is, we have $s(I,V)+s(V,A)>s(I,A)+s(V,V)$. 
} \label{tbl:BLOSUMqm} 
\end{table}

Most BLOSUM matrices, when restricted to the standard amino acid alphabet satisfy the Lemma \ref{lemma:sim2qm} (Table \ref{tbl:BLOSUMqm}). In fact, the first three conditions are always satisfied and only the triangle inequality presents problems. Where it is not satisfied, it is either in very small number of cases or for small values of $L$ which correspond to alignments of distantly related proteins and where it is to be expected that a transformation from one amino acid to another can arise from more than one substitution. However, it should be stressed that BLOSUM50 and BLOSUM62, which are the most widely used score matrices for database searches, do satisfy the Lemma \ref{lemma:sim2qm}. 

This observation leads to a conclusion that the `near-metric' of Linial, Linial, Tishby and Yona \cite{LLTY97} derived from local similarities based on BLOSUM62 matrix and affine gap penalties by the formula $d(x,y)=\mathcal{H}(x,x)+\mathcal{H}(y,y)-2\mathcal{H}(x,y)$ is in fact a true metric and that the rare instances where the triangle inequality was observed to fail were solely due to non-standard letters such as B,Z and X which represent sets of amino acids (for example X stands for any amino acid) and whose similarity scores were derived by averaging over all represented letters.

\section{Profiles}\label{sec:profiles}

\subsection{Position specific score matrices}

From a biological point of view, \emph{profiles} are generalised sequences. They were originally introduced by Gribskov, McLachlan, and Eisenberg \cite{Gribskov:1987} in order to model the situations where similarity measures based on score matrices do not retrieve all biologically relevant neighbours. As mentioned in Chapter \ref{ch:intro}, the function of a protein depends on its structure which in turn depends on its amino acid sequence. The structure space is smaller than the sequence space \cite{Murzin:1995,Sander:1996} and hence similar structures can arise from quite distantly related (in the evolutionary sense) sequences that do not share sufficiently high similarity to be detected using score matrix based methods. However, even significantly different structurally related sequences often contain a few sites, usually associated with a particular biological role, that are strongly conserved across species. Hence the idea of using \emph{position specific scores} to model protein families and find their new members.

In the sense of Gribskov, McLachlan, and Eisenberg, the term profile can be used interchangibly with a term \emph{Position Specific Score Matrix} or \emph{PSSM}. A PSSM is an $n$-by-$\abs{\Sigma}$ matrix where $\Sigma$ is an appropriate finite alphabet (most often the set of 20 standard amino acids used in proteins -- in fact we will always assume this is the case and use `amino acid' and `letter' interchangeably). For any PSSM $M$, an entry $M_{i,a}$ where $1\leq i\leq n$ and $a\in\Sigma$ gives the score of the letter $a$ in position $i$. Obviously, entries of a PSSM can come from similarity score matrices, that is, from similarities on $\Sigma$. Let $x=x_1x_2\ldots x_n$ and let $s:\Sigma\times\Sigma\to\R$ be a similarity score function (or matrix since $\Sigma$ is assumed finite). Then, one can produce a PSSM by setting \[M_{i,a} = s(x_i,a). \] Of course, in this case, the PSSM is really not `position specific': the scores for the same amino acid at different positions are the same. To summarise, PSSMs are generalisations of similarity score matrices.

The score of a sequence with respect to a PSSM is calculated very similarly to the usual similarity scores. Let $x=x_1x_2\ldots x_m$ and let $M$ be an $n$-by-$\abs{\Sigma}$ PSSM. If $m=n$, one can write the score $M(x)$ as \[M(x) = \sum_{i=1}^m M_{i,x_i},\] that is, as an $\ell_1$-type sum. On the other hand, if $m\neq n$ and gapped local scores are desired, a modified Smith-Waterman algorithm can be used. 

Let $g,h$ be positive gap penalty functions $\N^+\to\R_+$ and let $H$ be an $n+1$-by-$m+1$ matrix indexed from $0$. Set $H_{0,0}=H_{i,0}=H_{0,j}=0$ and for all $i=1,2\ldots m$ and $j=1,2\ldots n$ \[H_{i,j}=\max\left\{H_{i-1,j-1}+M_{i,x_j}, \max_{1\leq k\leq i}\left\{H_{i-k,j}-h(k)\right\}, \max_{1\leq k\leq j}\left\{H_{i,j-k}-g(k)\right\},0  \right\}.\] The local similarity score of $x$ with respect to the PSSM $M$ , denoted $\mathcal{H}_M(x)$ is given by $\mathcal{H}_M(x)= \max_{i,j} H_{i,j}$. Global similarities can be produced using an appropriate modification of the Needleman-Wunsch algorithm.

\subsection{Profiles as distributions}

While we have seen that profiles may come from similarity score matrices, they are usually produced from collections of related sequences, that is, (putative) members of a protein family. Given a (finite) set of sequences\footnote{The index is in superscript rather than subscript in order to distinguish a sequence entry in $U$ ($u^i$) and a residue of $u$ at position $i$ ($u_i$).} $U=\{u^j\}_j$, we first produce a multiple alignment of all of them. For the sake of simplicity, assume that the multiple alignment is ungapped, that is, only letters are present\footnote{Profile hidden Markov models \cite{Eddy98} further generalise the profiles by modelling gaps as well as `matches'.}, and that all sequences have the same length. Clearly, the relative frequencies of letters at each position $i$ define a probability distribution $q_i$ where $q_{i}(a)$ is the probability of an amino acid $a$ occurring at the position $i$. Given a background amino acid distribution $p$, where $p(a)$ is the overall relative frequency of $a$, we can define a PSSM as a matrix of log odds ratios 
\begin{equation}\label{eq:proflogodds}
M_{i,a} = \log\frac{q_i(a)}{p(a)},
\end{equation}
exactly mirroring the definition of the BLOSUM matrices in Subsection \ref{subsec:BLOSUM}. 

This leads an alternative definition of profiles, used for example by Yona and Levitt \cite{YoLe01}. From this point of view, a profile is a sequence of probability distributions on $\Sigma$, that is, a member of a free semigroup generated by $\mathcal{M}(\Sigma)$, the set of all probability distributions over $\Sigma$. The two definitions are in fact closely related since, given a background distribution $p$, every sequence of distributions can be converted into a PSSM using the Equation (\ref{eq:proflogodds}), while it is also clear \cite{Alt91,KaA90} that scores at each position can be, after scaling, converted to probabilities. Note that the scaling factors need not be the same for each position and thus each scaling factor can be treated as a `weight' for the particular position. The log-odds scores and the scaling factors have information-theoretic interpretations \cite{Alt91,KaA90,Durbin:1998} that we will not discuss here.

The definition of profiles as members of $\mathcal{M}(\Sigma)^*$ opens interesting possibilities for introducing quasi-metrics for profile-profile comparison. Suppose we have a quasi-metric and a positive function on $\mathcal{M}(\Sigma)$. Then, we can extend them to obtain a weighted quasi-metric on $\mathcal{M}(\Sigma)^*$ using dynamic programming and the Theorem \ref{thm:locsim2qm}. The similarity scores and distances thus obtained would have a similar interpretation to the scores obtained from score matrices. Yona and Levitt \cite{YoLe01} produced a profile-profile comparison tool by using the same principles, that is, by extending a similarity score function on $\mathcal{M}(\Sigma)$ to $\mathcal{M}(\Sigma)^*$ using dynamic programming. However, it is unclear from their presentation if their score function can induce a quasi-metric.

\chapter{Quasi-metric Spaces with Measure}\label{ch:2}

The main object of this chapter study is the \emph{pq-space}, the quasi-metric space with Borel probability measure (or probability quasi-metric space) which we introduce here for the first time. As most of the theory of the measure concentration was developed within the framework of a metric space with measure, we will throughout this chapter state the definitions and results for the metric case first and then give the corresponding statements for the quasi-metric case. The proofs will be given only for the quasi-metric case (as they include the metric case) and where they are not available elsewhere. For an extensive review of the theory for the metric case the reader is referred to the excellent monograph by Ledoux \cite{Le01}, Chapter $3\frac12_+$ of the well-known Gromov's book \cite{Gr99} as well as the book by Milman and Schechtman \cite{MS86} which mainly concentrates on the normed spaces.

We aim to explore the phenomenon of concentration of measure in high dimensional structures in the case where the underlying structure is a quasi-metric space with measure. Many results and proofs can be transferred almost verbatim from the metric case. However, we also develop new results which have no metric analogues.

\section{Basic Measure Theory}

Let $\Omega$ be a set. A collection $\mathcal{A}$, of subsets of $\Omega$, is called a \emph{$\sigma$-algebra} if it satisfies
\begin{enumerate}[(i)]
\item $\Omega \in \mathcal{A}$,
\item if $A \in \mathcal{A}$ then $\Omega \setminus A \in A$,
\item if $A = \bigcup_{k=1}^{\infty}A_k$ with $A_k \in \mathcal{A}$
for all $k$, then $A \in \mathcal{A}$.
\end{enumerate}

Let $\mathcal{S}$ be a collection of subsets of $\Omega$. The $\sigma$-algebra \emph{generated by} $\mathcal{S}$, denoted $\sigma(\mathcal{S})$, is the smallest $\sigma$-algebra containing $\mathcal{S}$ (one $\sigma$-algebra containing $\mathcal{S}$ always exists: the power set $\PowE{\Omega}$).

A function $\mu: \mathcal{A}\to R^+$ such that $\mu(\emptyset)=0$ is a \emph{measure} on $\mathcal{A}$ if it is additive, that is if \[\mu(\bigcup_{k \geq 1} A_k)=\sum_{k \geq 1} \mu(A_k)\] for all pairwise disjoint sets $A_k \in \mathcal{A}$. A \emph{measure space} is a triple $(\Omega, \mathcal{A}, \mu)$ where $\Omega$ is a set, $\mathcal{A}$ is a $\sigma$-algebra and $\mu$ is a measure. A \emph{probability space} is a measure space with total measure
$\mu(\Omega)=1$.

Let $(\Omega,\mathcal{A},\mu)$ be a measure space. The measure $\mu$ is called \emph{$\sigma$-finite} if there exists a countable collection of sets $\{\Omega_i\}_{i=1}^\infty$ such that $\Omega=\bigcup_{i=1}^\infty\Omega_i$ and $\mu(\Omega_i)<\infty$ for each $i$.

The \emph{Borel $\sigma$-algebra} on a topological space $(X,\Top)$ is the smallest $\sigma$-algebra containing $\Top$. The existence and uniqueness of the Borel algebra is shown by noting that the intersection of all $\sigma$-algebras containing $\Top$ is itself a $\sigma$-algebra, so this intersection is the Borel algebra. The elements of the Borel $\sigma$-algebra are called \emph{Borel sets} while the measures on $\sigma$-algebras are called \emph{Borel measures}.

The Borel $\sigma$-algebra may alternatively and equivalently be defined as the smallest $\sigma$-algebra which contains all the closed subsets of $X$. A subset of $X$ is a Borel set if and only if it can be obtained from open (or closed) sets by using the set operations union, intersection and complement in countable number, more exactly via transfinite recursion in countable ordinals.

\section{pq-spaces}

\begin{defin}
A topological space $(X,\Top)$ is called \emph{Polish} if it is separable and metrisable by means of a complete metric.
\end{defin}

We recall the definition of a metric space with measure, as defined in \cite{GrMi83}.

\begin{defin}[\cite{GrMi83,Gr99,Gr03}]
An \emph{mm-space} is a triple $(X,d,\mu)$ where $(X,d)$ is a Polish metric space and $\mu$ a $\sigma$-finite Borel measure on $X$.

An mm-space where $\mu(X)=1$ is called a \emph{pm-space} . 
\end{defin}

We shall mostly be concerned with mm-spaces equipped with finite measures and will assume wherever possible that the measure has been normalised so that they become pm-spaces.

In order to define an analogue for a quasi-metric space $(X,d)$ we observe that it is not sufficient to use the Borel $\sigma$-algebra generated by $\Top(d)$ since we want to have the open and closed sets with respect to both $\Top(d)$ and $\Top(\cj{d})$ measurable. Hence, we use the Borel $\sigma$-algebra generated by $\Top(d)\cup\Top(\cj{d})$. It is easy to see that this structure is equivalent to the Borel $\sigma$-algebra generated by $\Top(\qam{d})$, the topology of the associated metric, by observing that $\ball{x}{\e}=\lball{x}{\e}\cap \rball{x}{\e}$ (Remark \ref{rem:qam_base}).

In order to make our definition fully analogous to the the definition of the mm-space, we additionally require that our quasi-metric be bicomplete, that is, that its associated metric be complete.

\begin{defin}
Let $(X,d)$ be a bicomplete separable quasi-metric space, and $\mu$ a $\sigma$-finite measure over $\mathcal{B}$, a Borel $\sigma$-algebra of measurable sets generated by $\Top(\qam{d})$ where $\qam{d}$ is the associated metric to $d$. We call the triple $(X,d,\mu)$ an \emph{mq-space}. If in addition $\mu(X)=1$ we call such triple a \emph{pq-space}.

Furthermore, we call the mq-space $(X,\cj{d},\mu)$ the \emph{conjugate} or \emph{dual mq-space} to $(X,d,\mu)$ and the mm-space $(X,\qam{d},\mu)$ the \emph{associated mm-space} to $(X,d,\mu)$.
\end{defin}

Henceforth, we shall always use the symbol $\mathcal{B}$ in the context of mq-spaces to denote the underlying Borel $\sigma$-algebra.

\begin{remark}
The fact that $(X,\qam{d},\mu)$, the associated mm-space to $(X,d,\mu)$, is an mm-space indeed is a direct consequence of having the Borel $\sigma$-algebra of measurable sets generated by $\Top(\qam{d})$. 
\end{remark}

In this work we shall only consider pq-spaces, that is, the quasi-metric spaces with finite measure. The definition of an mq-space was introduced in order to correspond to the definition of an mm-space as given by Gromov \cite{Gr99,Gr03}.

In order to illustrate one possible way of interaction between a quasi-metric and measure we give another example of Lipschitz functions.

\begin{lemma}
Let $(X,d,\mu)$ be a pq-space and $0\leq p\leq 1$. The function $\rho_p: X\to\R$, where $\rho_p(x) = \inf\{r>0: \mu(\lball{x}{r}) \geq p \}$, is left 1-Lipschitz, while $\cj{\rho_p}: X\to\R$, where $\cj{\rho_p}(x) := \inf\{r>0: \mu(\rball{x}{r}) \geq p \}$, is right 1-Lipschitz.
\end{lemma}

\begin{figure}[h!]
\begin{center}
  \input{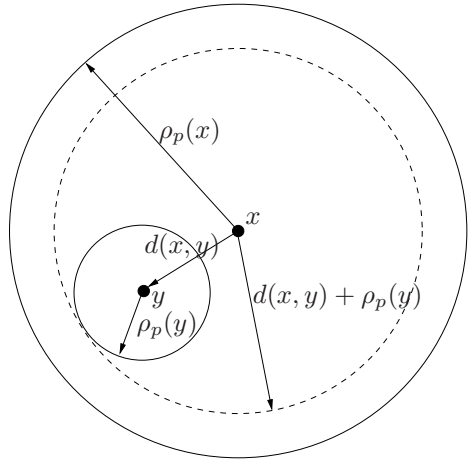}
  \caption{$\rho_p$ function.}
  \label{fig:rhop}
\end{center}
\end{figure}

\begin{proof}
Since $\lball{x}{d(x,y)+\rho_p(y)} \supseteq \lball{y}{\rho_p(y)}$ (Fig. \ref{fig:rhop}), one has \[\mu(\lball{x}{d(x,y)+\rho_p(y)}) \geq \mu(\lball{y}{\rho_p(y)}) \geq p\] and it follows that $\rho_p(x) \leq d(x,y) + \rho_p(y)$ and therefore $\rho_p(x) - \rho_p(y) \leq d(x,y)$.
The second statement follows in a similar manner.
\end{proof}

\section{Concentration Functions}

Recall the definition of the concentration function for an mm-space.
\begin{defin}
Let $(X,d,\mu)$ be an mm-space and $\mathcal{B}$ the Borel $\sigma$-algebra of $\mu$-measurable sets. The \emph{concentration function} $\alpha_{(X,d,\mu)}$, also denoted $\alpha$, is a function $\R_+\to [0,\frac{1}{2}]$ such that $\alpha_{(X,d,\mu)}(0)=\frac12$ and for all $\e>0$
\[ \alpha_{(X,d,\mu)}(\e) = \sup\left\{1-\mu(\nbhd{A}{\e});\ A\in \mathcal{B}, \ \mu(A) \geq \frac{1}{2}\right\}.\]
\end{defin}

The concentration function measures the maximum size of a complement (`cap') of a neighbourhood of a Borel set of a measure not less than $\frac12$. In a sense to be made more precise later, a space is `concentrated' if its concentration function is extremely small for small $\e$. 

As before with asymmetric structures, we introduce two concentration functions on a pq-space, left and right.

\begin{figure}[h!]
\begin{center}
  \input{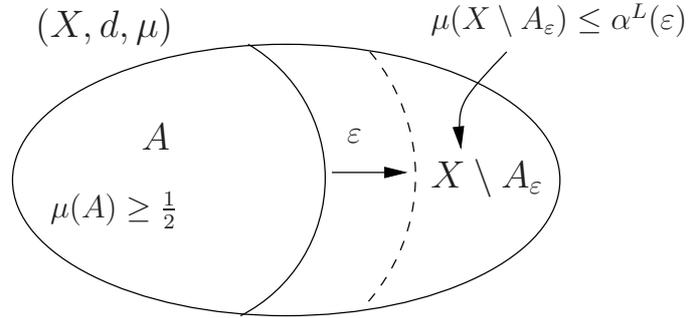}
  \caption{Left concentration function $\alpha^L$.}
  \label{fig:alphaconc}
\end{center}
\end{figure}

\begin{defin}
Let $(X,d,\mu)$ be a pq-space and $\mathcal{B}$ the Borel $\sigma$-algebra of $\mu$-measurable sets. The \emph{left concentration function} $\alpha^L_{(X,d,\mu)}$, also denoted $\alpha^L$, is a map $\R_+\to [0,\frac{1}{2}]$ such that $\alpha^L_{(X,d,\mu)}(0)=\frac12$ and for all $\e>0$
\[ \alpha^L_{(X,d,\mu)}(\e) = \sup\left\{1-\mu(\lnbhd{A}{\e});\ A\in \mathcal{B}, \ \mu(A) \geq \frac{1}{2}\right\}.\]

Similarly, the \emph{right concentration function} $\alpha^R_{(X,d,\mu)}$, also denoted $\alpha^R$, is a map $\R_+\to [0,\frac{1}{2}]$ such that $\alpha^R_{(X,d,\mu)}(0)=\frac12$ and for all $\e>0$
\[ \alpha^R_{(X,d,\mu)}(\e) = \sup\left\{1-\mu(\rnbhd{A}{\e});\ A \in \mathcal{B}, \ \mu(A) \geq \frac{1}{2}\right\}.\]
\end{defin}

\begin{remark}
For an mm-space $(X,d,\mu)$, $\alpha^L$ and $\alpha^R$ are equal and they coincide with the usual concentration function $\alpha_{(X,d,\mu)}$. It is also easy to observe that for a pq-space $(X,d,\mu)$, \[\alpha^L_{(X,d,\mu)}=\alpha^R_{(X,\cj{d},\mu)}.\]  
\end{remark}

%
%
%
%

The concentration functions $\alpha^L$ and $\alpha^R$ respectively measure the maximum size of the complement to any left and right neighbourhood of a Borel set of a measure not less than $\frac{1}{2}$ (Fig. \ref{fig:alphaconc}).

\begin{lemma}
For any pq-space $(X,d,\mu)$, the concentration functions  $\alpha^L_{(X,d,\mu)}$ and
$\alpha^R_{(X,d,\mu)}$ are decreasing and converge to $0$ as $\e\to\infty$. Furthermore, if $\diam(X)$ is finite, then for all $\e\geq\diam(X)$, $\alpha^L(\e)=\alpha^R(\e)=0$.

\begin{figure}[h!]
\begin{center}
  \input{concdecr.pstex_t}
  \caption{$\lnbhd{A}{\e}$ can take as much mass as required.}
  \label{fig:concdecr}
\end{center}
\end{figure}
\begin{proof}
We prove the statement for $\alpha^L$. It is obvious that $\alpha^L$ is bounded below by $0$ and decreasing since $\lnbhd{A}{\e_0}\subseteq\lnbhd{A}{\e_1}$ and hence $\mu(\lnbhd{A}{\e_0})\leq \mu(\lnbhd{A}{\e_1})$ for any Borel set $A$ and $0<\e_0\leq\e_1$. Thus the limit exists and is non-negative and we now show that $\lim_{\e\to\infty}\alpha^L(\e)=0$.

Take any $0<\delta\leq\frac12$. We need to show that there is some $\e_0>0$ such that for all $\e>\e_0$ and for any Borel set $A$ such that $\mu(A)\geq\frac12$ we have $\mu(A_\e)>1-\delta$ (this is trivially true for $\delta>\frac12$). Take any $x_0\in X$. We will show that there exist $\e'$ such that for all $\e>\e'$, $\mu(\ball{x_0}{\e})>1-\delta$. Indeed, taking the open balls $\ball{x_0}{n}$, $n\in\N_+$ \emph{with respect to the associated metric} $\qam{d}$ we have
\begin{align*}
\limsup_{n\to\infty}\mu(\ball{x_0}{n}) & = \lim_{n\to\infty}\left(\mu\left(\ball{x_0}{1}\right) + \sum_{i=1}^n \mu\left(\ball{x_0}{i+1}\setminus\ball{x_0}{i}\right)\right)\\
& = \mu\left(\ball{x_0}{1}\right) + \sum_{n=1}^\infty \mu\left(\ball{x_0}{i+1}\setminus\ball{x_0}{i}\right)\\
& = \mu(X) = 1
\end{align*}
by $\sigma$-additivity of measure. Thus there is some $n_0\in\N_+$ such that for all $n\geq  n_0$, $\mu\left(\ball{x_0}{n}\right)>1-\delta$. Now take any Borel set $A$ of measure greater than $\frac12$. $A$ must intersect $\ball{x_0}{n_0}$ (Figure \ref{fig:concdecr}) because if it would not, we would have $\mu(A)<\delta\leq\frac12$ leading to a contradiction. It now clear that for any $\e\geq\diam\left(\ball{x_0}{n_0}\right)=2n_0$ we have $\lnbhd{A}{\e} \supseteq\ball{x_0}{n_0}$. Indeed, let $a\in A$ and $b\in \ball{x_0}{n_0}$. Then by the triangle inequality 
\begin{align*}
d(a,b) &\leq  d(a,x_0)+d(x_0,b)\\
&\leq \qam{d}(a,x_0)+\qam{d}(x_0,b)\\
&< n_0 + n_0 = 2n_0.
\end{align*}
Therefore, for any $\e>2n_0$, $\mu\left(\lnbhd{A}{\e}\right)\geq \mu\left(\ball{x_0}{n_0}\right) >1-\delta$ as required. It is obvious that the same proof would work for $\alpha^R$ by substituting $\lnbhd{A}{\e}$ by $\rnbhd{A}{\e}$ above. 

It is also clear that if $\diam(X)<\infty$, then for any $\e>\diam(X)$ and any $A\subseteq X$, $X=\lnbhd{A}{\e}=\rnbhd{A}{\e}$ and hence $\alpha^L(\e)=\alpha^R(\e)=0$.
\end{proof}
\end{lemma}

The following lemmas show some relations between the various alpha functions.

\begin{lemma}\label{lem:qpconcfunc}
For any pq-space $(X,d,\mu)$, for each $\e\geq 0$,
\[
\max\{\alpha^L_{(X,d,\mu)}(\e), \alpha^R_{(X,d,\mu)}(\e)\} \leq \alpha_{(X,\qam{d},\mu)}(\e) \leq \alpha^L_{(X,d,\mu)}(\e)+\alpha^R_{(X,d,\mu)}(\e).
\]
\begin{proof}
Let $A \in \mathcal{B}$ be such that $\mu(A) \geq \frac{1}{2}$ and let $\e>0$.
Using $A_{\e}\subseteq \lnbhd{A}{\e} \cap \rnbhd{A}{\e}$,
\begin{align*}
1-\mu(\lnbhd{A}{\e}) & \leq 1 - \mu(A_{\e}) \leq \alpha(\e) \implies \alpha^L(\e) \leq \alpha(\e) \quad \text{and}\\
1-\mu(\rnbhd{A}{\e}) & \leq 1 - \mu(A_{\e}) \leq \alpha(\e) \implies \alpha^R(\e) \leq\alpha(\e),
\end{align*}
and it follows that $\max\{\alpha^L(\e), \alpha^R(\e)\} \leq\alpha_{(X,\qam{d},\mu)}(\e)$.

For the second inequality, use the fact that $A_{\e}\supseteq \lnbhd{A}{\e}\cap\rnbhd{A}{\e}$, and thus $X\setminus A_{\e} \subseteq \big(X\setminus \lnbhd{A}{\e} \big) \cup \big(X\setminus \rnbhd{A}{\e} \big)$, implying \[1-\mu(A_{\e}) \leq \big(1-\mu(\lnbhd{A}{\e}) \big)+\big(1-\mu(\rnbhd{A}{\e}) \big)\leq \alpha^L(\e) + \alpha^R(\e). \]
\end{proof}
\end{lemma}

It is easy to see that the above inequalities from the Lemma \ref{lem:qpconcfunc} are strict. Consider the following example.

\begin{figure}[h!]
\begin{center}
  \input{concineqex.pstex_t}
  \caption{Space where $\max\{\alpha^L(\e),\alpha^R(\e)\}<\alpha(\e)$.}
  \label{fig:concineqex}
\end{center}
\end{figure}

\begin{example}
Let $X=\{a,b,c\}$ where $d(a,b)=d(b,c)=1$, $d(c,b)=d(b,a)=2$, $d(a,c)=2$ and $d(c,a)=4$. Set an additive measure in the following way: $\mu(\{a\})=\mu(\{c\})=\frac18$ and $\mu(\{b\})=\frac34$ (Figure \ref{fig:concineqex}). It is clear that $(X,d,\mu)$ is a pq-space and that \[\alpha^L(\e)=\alpha^R(\e) = 
\begin{cases}
\frac12 & \text{if $\e$ = 0}\\
\frac14 & \text{if $0<\e<1$}\\
\frac18 & \text{if $1\leq\e<2$}\\
0       & \text{if $\e\geq2$}
\end{cases}\]
 On the other hand 
\[\alpha(\e)= 
\begin{cases}
\frac12 & \text{if $\e$ = 0}\\
\frac14 & \text{if $0<\e<2$}\\
0       & \text{if $\e\geq2$}
\end{cases}\]
Hence for $1\leq\e<2$ we have $\max\{\alpha^L(\e),\alpha^R(\e)\}<\alpha(\e)$.
\end{example}


The \emph{phenomenon of concentration of measure on high-dimensional structures} refers to the observation that in many metric spaces with measure which are, intuitively, ``high dimensional'', the concentration function decreases very sharply, that is, an $\e$-neighbourhood of any not vanishingly small set, even for very small $\e$, covers (in terms of the probability measure) nearly the whole space. Examples are numerous and come from many diverse branches of mathematics \cite{Maurey79,GrMi83,AlonMilman85,MS86,Gr99,Pe02,Tal96a}. Here we take a ``high dimensional'' pq-space to be a pq-space where both $\alpha^L$ and $\alpha^R$ decrease sharply.

\section{Deviation Inequalities}

\begin{defin}
Let $(X,\mathcal{B},\mu)$ be a probability space and $f$ a measurable real-valued function on $(X,d)$. A value $m_f$ is a \emph{median} or \emph{L\'evy mean} of $f$ for $\mu$ if
\[ \mu(\{f \leq m_f\}) \geq \frac{1}{2} \ \text{and} \ \mu(\{f \geq m_f)\} \geq \frac{1}{2}.\]
\end{defin}

A median need not be unique but it always exists. The following lemmas are generalisations of the results for mm-spaces.

\begin{lemma}\label{lemma:meddeviation}
Let $(X,d,\mu)$ be a pq-space, with left and right concentration functions $\alpha^L$ and $\alpha^R$ respectively and $f$ a left 1-Lipschitz function on $(X,d)$ with a median $m_f$. Then for any $\e > 0$
\begin{align*}
\mu(\{x\in X: f(x) \leq m_f - \e\}) \leq \alpha^L(\e) &\ \text{and}\\
\mu(\{x\in X: f(x) \geq m_f + \e\}) \leq \alpha^R(\e).
\end{align*}

Conversely, if for some non-negative functions $\alpha_0^L$ and $\alpha_0^R:\R_+\to\R$,
\begin{align*}
\mu(\{x\in X: f(x) \leq m_f - \e\}) \leq \alpha_0^L(\e) &\ \text{and}\\
\mu(\{x\in X: f(x) \geq m_f + \e\}) \leq \alpha_0^R(\e)
\end{align*}
for every left 1-Lipschitz function $f:X\to \R$ with median $m_f$ and every $\e>0$, then $\alpha^L\leq\alpha_0^L$ and $\alpha^R\leq\alpha_0^R$.

\begin{proof}
Set $A=\{x\in X: f(x) \geq m_f \}$. Take any $y\in X$ such that $f(y)\leq m_f -\e$. Then, for any $x\in A$, $d(x,y)\geq f(x)-f(y)\geq\e$ and hence $d(A,y)\geq\e$, implying $y\in X\setminus \lnbhd{A}{\e}$. Therefore, $\mu(\{x\in X: f(x) \leq m_f - \e\}) \leq 1 - \mu(\lnbhd{A}{\e})\leq\alpha^L(\e)$.

Now set $B=\{x\in X: f(x) \leq m_f \}$. Take any $y\in X$ such that $f(y)\geq m_f +\e$. Then, for any $x\in B$, $d(y,x)\geq f(y)-f(x)\geq\e$ and hence $d(y,B)\geq\e$, implying $y\in X\setminus \rnbhd{B}{\e}$. Thus, $\mu(\{x\in X: f(x) \geq m_f + \e\}) \leq 1 - \mu(\rnbhd{B}{\e})\leq\alpha^R(\e)$.

The converse is equivalent to finding for each Borel set $A\subseteq X$ such that $\mu(A)\geq \frac12$, left 1-Lipschitz functions $f$ and $g:X\to\R$ with medians $m_f$ and $m_g$ respectively, such that $1-\mu(\lnbhd{A}{\e}) \leq \mu(\{x\in X: f(x) \leq m_f - \e\})$ and $1-\mu(\rnbhd{A}{\e}) \leq \mu(\{x\in X: g(x) \geq m_g + \e\})$.

Let $A\subseteq X$ be such a set such that $\mu(A)\geq \frac12$ and set for each $y\in X$, $f(y)=-d(A,y)$ and $g(y)=d(y,A)$. It is easy to see that both $f$ and $g$ are left 1-Lipschitz and that $m_f=m_g=0$. If $y\in X\setminus \lnbhd{A}{\e}$, we have $d(A,y)\geq \e$ and thus $f(y)\leq -\e$. Similarly, if $y\in X\setminus \rnbhd{A}{\e}$, we have $d(y,A)\geq \e$ implying $g(y)\geq \e$ and the result follows.
\end{proof}
\end{lemma}

Hence, we can state the alternative definitions of $\alpha^L$ and $\alpha^R$:
\[\alpha^L(\e) = \sup \big\{\mu(\{x\in X: f(x) \leq m_f - \e\}):\ f \ \text{is left 1-Lipschitz}\big\}\]
and
\[\alpha^R(\e) = \sup \big\{\mu(\{x\in X: f(x) \geq m_f + \e\}): \ f \ \text{is right 1-Lipschitz} \big\}.\]

Similar results can be easily obtained for the right 1-Lipschitz functions by remembering that if $f$ is a right 1-Lipschitz, $-f$ is left 1-Lipschitz (Lemma \ref{lemma:flipLip}). It is also  straightforward to observe that the absolute value of deviation of a 1-Lipschitz function from a median thus depends on both $\alpha^L$ and $\alpha^R$.

\begin{corol}
For any pq-space $(X,d,\mu)$, a left 1-Lipschitz function $f$ with a median $m_f$ and $\e > 0$
\[\mu(\{\abs{f-m_f } \geq\e\}) \leq \alpha^L_{(X,d,\mu)}(\e)+\alpha^R_{(X,d,\mu)}(\e).\]
\end{corol}

This result reduces to the well-known inequality $\mu(\{\abs{f-m_f} \geq\e\}) \leq 2\alpha(\e)$ when $d$ is a metric. Deviations between the values of a left 1-Lipschitz function at any two points are also bound by both concentration functions.

\begin{lemma}\label{lemma:ptdeviation}
Let $(X,d,\mu)$ be a pq-space and $f\colon X\to\R$ a left (or right) 1-Lipschitz function. Then
\[(\mu\otimes\mu)(\{(x,y)\in X\times X:f(x)-f(y)\geq\e\}) \leq \alpha^L\left(\frac{\e}{2}\right) + \alpha^R\left(\frac{\e}{2}\right). \]
\begin{proof}
\begin{align*}
&\quad (\mu\otimes\mu) \left(\left\{(x,y)\in X\times X:f(x)-f(y) \geq \e \right\}\right)\\
&\leq (\mu\otimes\mu) \left(\left\{(x,y)\in X\times X:f(x)-m_f\geq\frac{\e}{2}\right\}\right)\\
& + (\mu\otimes\mu) \left(\left\{(x,y)\in X\times X: m_f - f(y) \geq \frac{\e}{2}\right\}\right)\\
&= \mu\left(\left\{x\in X:f(x)\geq m_f +\frac{\e}{2}\right\}\right) + \mu\left(\left\{x\in X:f(x)\leq m_f -\frac{\e}{2}\right\}\right)\\
&\leq \alpha^L\left(\frac{\e}{2}\right) + \alpha^R\left(\frac{\e}{2}\right).
\end{align*}
\end{proof}
\end{lemma}

\section{L\'evy Families}

\begin{defin}
A sequence of pq-spaces $\{(X_n,d_n,\mu_n)\}_{n=1}^\infty$ is called \emph{left L\'evy family} if the left concentration functions $\alpha^L_{(X_n,d_n,\mu_n)}$ converge to $0$ pointwise, that is
\[\forall\e>0, \quad \alpha^L_{(X_n,d_n,\mu_n)}(\e)\to 0 \: \text{as} \ n\to\infty.\]

Similarly, a sequence of pq-spaces $\{(X_n,d_n,\mu_n)\}_{n=1}^\infty$ is called \emph{right L\'evy family} if the right concentration functions $\alpha^R_{(X_n,d_n,\mu_n)}$ converge to $0$ pointwise, that is
\[\forall\e>0, \quad \alpha^R_{(X_n,d_n,\mu_n)}(\e)\to 0 \: \text{as} \ n\to\infty.\]

A sequence which is both left and right L\'evy family will be called a L\'evy family. Furthermore, if for some constants $C_1, C_2 >0$ one has $\alpha_n(\e) < C_1 \exp(C_2 \e^2 n)$, such sequence is called \emph{normal L\'evy family}.
\end{defin}

It is a straightforward corollary of Lemma \ref{lem:qpconcfunc} that a sequence of pq-spaces $\{(X_n,d_n,\mu_n)\}_{n=1}^\infty$ is a L\'evy family if and only if the sequence of associated mm-spaces $\{(X_n,\qam{d}_n,\mu_n)\}_{n=1}^\infty$ is a L\'evy family.

To illustrate existence of sequences of pq-spaces which are right but not left L\'evy families consider the following example.

\begin{example}
Let $X=\{a,b\}$ with $\mu(\{a\})=\frac23$ and $\mu(\{b\})=\frac13$. Set $d_n(a,b)=1$ and $d_n(b,a)=\frac1n$ where $n\in\N_+$.(Fig. \ref{fig:qpconc1}).
\begin{figure}[h!]
\begin{center}
\scalebox{1.0}{\input{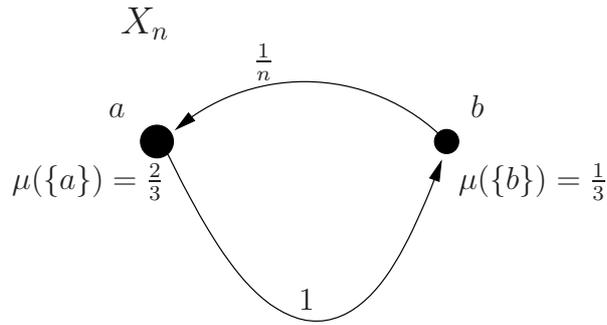}}
\caption{Spaces $X_n$ where $\alpha^R_n\to 0$ as $n\to\infty$ but $\alpha^L_n$ does not.}
\label{fig:qpconc1}
\end{center}
\end{figure}

It is clear that
\[ \alpha^L_n(\e) =
\begin{cases}
\frac12 , & \text{if} \ \e = 0\\
\frac13 , & \text{if} \ 0 < \e\leq 1\\
0, & \text{if} \ \e > 1,
\end{cases}
\quad \text{and} \quad
\alpha^R_n(\e) =
\begin{cases}
\frac12 , & \text{if} \ \e = 0\\
\frac13 , & \text{if} \ 0 < \e\leq \frac1n\\
0, & \text{if} \ \e > \frac1n .
\end{cases}
\]
Hence, $\alpha^R_n$ converges to $0$ pointwise while $\alpha^L_n$ does not. In this case $\alpha_n = \alpha^L_n$. 
\end{example}

Examples of L\'evy families of mm-spaces abound in many diverse areas of mathematics. We only mention a few.

\begin{example}[Maurey \cite{Maurey79}]
The sequence $\{(S_n,d_n,\mu_n)\}_{n=1}^\infty$ where $S_n$ is the group of permutations of rank $n$, $d_n$ is the normalised Hamming distance given by \[d_n(\sigma,\tau) = \frac1n\abs{i:\sigma(i) \neq\tau(i)},\] and $\mu_n$ is the normalised counting measure where \[\mu_n(A)=\frac{\abs{A}}{n!},\] forms a normal L\'evy family with the concentration functions satisfying \[\alpha_{S_n}(\e)\leq 2\exp(-\e^2n/64).\]
\end{example}

\begin{example}[L\'evy \cite{Levy51}]
The family of spheres $\S^{n}\subset\R^{n+1}$ with the geodesic metric and the rotation invariant measure forms a normal L\'evy family where \[\alpha_{\S^n}(\e)\leq \sqrt{\frac{\pi}{8}}\exp(-\e^2n/2).\]
\end{example}

\begin{example}[Gromov and Milman \cite{GrMi83}]
The special orthogonal group $SO(n)$ consists of all orthogonal $n\times n$ matrices having the determinant $1$. The family of these groups with the geodesic metric and the normalised Haar measure forms a normal L\'evy family where \[\alpha_{SO(n)}(\e)\leq \sqrt{\frac{\pi}{8}}\exp(-\e^2n/8).\]
\end{example}

The hamming cube, discussed in Subsection \ref{subsec:hammingcube} provides another example (Proposition \ref{prop:hammcubeconc}).

\section[High dimensional pq-spaces are very close to mm-spaces]{High dimensional pq-spaces are very close to\\ \mbox{mm-spaces}}\label{sec:qpclosemm}

Most of the above concepts and results are generalisations of mm-space results. However, we now develop some results which are trivial in the case of mm-spaces. The main result is that, if both left and right concentration functions drop off sharply, the \emph{asymmetry} at each pair of point is also very small and the quasi-metric is very close to a metric.

\begin{defin}
For a quasi-metric space $(X,d)$, the \emph{asymmetry} is a map $\Gamma:X\times X\to\R$ defined by $\Gamma(x,y)=\abs{d(x,y)-d(y,x)}$.
\end{defin}

Obviously, $\Gamma\equiv 0$ on a metric space. However, $\Gamma$ is also close to $0$ for high dimensional spaces, that is, those pq-spaces for which both $\alpha^L$ and $\alpha^R$ decrease sharply near zero.

\begin{thm}\label{thm:qpclosemm}
Let $(X,d,\mu)$ be a pq-space. For any $\e>0$,
\[(\mu\otimes\mu) (\{(x,y)\in X\times X:\Gamma(x,y)\geq\e\}) \leq \alpha^L\left(\frac{\e}{2}\right) + \alpha^R\left(\frac{\e}{2}\right).\]
\begin{proof}
Fix $a\in X$ and set for each $x\in X$, $\gamma_a(x)= d(x,a)-d(a,x)$. It is clear that $\gamma_a$ is a sum of two left 1-Lipschitz maps and therefore left 2-Lipschitz. Furthermore,
zero is its median since there is a measure-preserving bijection $(x,y)\mapsto (y,x)$ which maps the set $\{(x,y)\in X\times X: d(x,y) > d(y,x)\}$ onto the set $\{(x,y)\in X\times X: d(x,y) < d(y,x)\}$.
By the Lemma \ref{lemma:meddeviation}, $\mu(\{x\in X:\abs{\gamma_{a}(x)}\geq\e\}) \leq \alpha^L\left(\frac{\e}{2}\right)+\alpha^R\left(\frac{\e}{2}\right)$. Now, using Fubini's theorem,
\begin{align*}
& \quad(\mu\otimes\mu)(\{(x,y)\in X\times X:\abs{d(x,y)-d(y,x)}\geq\e\})\\
&=  \int_{x\in X}\int_{y\in X}
  \I_{\{\abs{\gamma_x(y)}\geq\e\}} d\mu(y) d\mu(x)\\
&\leq \left(\alpha^L\left(\frac{\e}{2}\right)+ \alpha^R\left(\frac{\e}{2}\right)\right)\int_{x\in X}d\mu(x)\\
&= \alpha^L\left(\frac{\e}{2}\right)+\alpha^R\left(\frac{\e}{2}\right).
\end{align*}
\end{proof}
\end{thm}

Thus, any pq-space where both $\alpha^L$ and $\alpha^R$ (equivalently, by the Lemma \ref{lem:qpconcfunc}, $\alpha$) sharply decrease are, apart from a set of very small size, very close to an mm-space.

If we restrict ourselves to longer ranges, that is, bound the distances $d(x,y)$ from below, then more precise bounds for the difference $d(x,y)-d(y,x)$ can be obtained.

\begin{corol}
Let $(X,d,\mu)$ be a pq-space and $0<\e\leq\delta<\infty$. Then, for any pair $(x,y)\in X\times X$ such that $\delta\leq d(x,y)$, apart from a set of ($\mu\otimes\mu$) measure at most $1 - \alpha^L(\frac{\e}{2})-\alpha^R(\frac{\e}{2})$, the values $d(x,y)$ and $d(y,x)$ differ by a factor of less than $1+\e/\delta$. More precisely,
\[ \big(1-\frac{\e}{\delta}\big) d(x,y) < d(y,x) < \big(1+\frac{\e}{\delta}\big) d(x,y).\]
\begin{proof}
By the previous theorem, for any $\e>0$, apart from a set of measure at most $1 - \alpha^L(\frac{\e}{2})- \alpha^R(\frac{\e}{2})$, the values of $d(x,y)$ and $d(y,x)$ differ by less than $\e$. The result now follows by rearrangement of the inequality $\abs{d(x,y)-d(y,x)}<\e$. Indeed, if $d(x,y) < d(y,x)$, we have $d(y,x) < \big(1+ \frac{\e}{d(x,y)}\big)d(x,y) \leq \big(1+ \frac{\e}{\delta}\big)d(x,y)$.  If $d(y,x) < d(x,y)$, then $d(y,x) > \big(1-\frac{\e}{d(x,y)}\big) d(x,y) \geq \big(1-\frac{\e}{\delta}\big)d(x,y)$.
\end{proof}
\end{corol}


\section{Product Spaces}

\subsection{Hamming cube}\label{subsec:hammingcube}

\begin{defin}
Let $n\in\N$ and $\Sigma=\{0,1\}$. The collection of all binary strings of length $n$, denoted $\Sigma^n$ is called the \emph{Hamming cube}.
\end{defin}
\begin{defin}
The \emph{Hamming distance (metric)} for any two strings $\sigma=\sigma_1\sigma_2\ldots\sigma_n$ and  $\tau=\tau_1\tau_2\ldots\tau_n \in \Sigma^n$ is given by \[d_n(\sigma, \tau) = \abs{\{i\in\N:\sigma_i\neq\tau_i\}}.\] The \emph{normalised Hamming distance} $\rho_n$ is given by \[\rho_n(\sigma,\tau) = \frac{d(\sigma,\tau)}{n} = \frac{\abs{\{i\in\N:\sigma_i\neq\tau_i\}}}{n}.\]
\end{defin}
\begin{defin}
The \emph{normalised counting measure} $\mu_n$, of any subset $A$ of a Hamming cube $\Sigma^n$ is given by \[\mu_n(A)=\frac{\abs{A}}{2^n}.\] 
\end{defin}

It is easy to see that the above definitions indeed give a set with a metric and a measure and that $(\Sigma^n, \rho_n, \mu_n)$  is a pm-space. One may wish to consider $\Sigma^n$ as a product space with $\rho_n$ as an $\ell_1$-type sum of discrete metrics on $\{0,1\}$ and $\mu_n$ an $n$-product of $\mu_1$, where $\mu_1(\{0\})=\mu_1(\{1\})=\frac12$.

The following bounds to the concentration function on the Hamming cube were stated in the book by Milman and Schechtman \cite{MS86} (Section 6.2):

\begin{prop}\label{prop:hammcubeconc}
For any Hamming cube $\Sigma^n$ with the normalised Hamming distance $\rho_n$ and the normalised counting measure $\mu_n$, we have \[\alpha_{(\Sigma^n, \rho_n, \mu_n)}(\e) \leq \frac12 \exp(-2\e^2n).\]\qed
\end{prop}

\subsubsection{Law of Large Numbers}
Hence a sequence $\{(\Sigma^n, \rho_n, \mu_n)\}_{i=1}^{\infty}$ is a normal L\'evy family. An easy consequence of the Proposition \ref{prop:hammcubeconc} is the well-known \emph{Law of large numbers}.
\begin{prop}\label{prop:lln}
Let $(\epsilon)_{i\leq N}$ be an independent sequence of Bernoulli random variables ($P(\epsilon=1)=P(\epsilon=-1)=\frac12$). Then for all $t\geq 0$
\[P\big(\abs{\sum_{i\leq N} \epsilon_i} \geq t\big) \leq 2 \exp \Bigg(-\frac{t^2}{2N}\Bigg).\]
Equivalently, if $B_N$ is the number of ones in the sequence $(\epsilon)_{i\leq N}$ then
\[P\big(\abs{B_N-\frac{N}{2}} \geq t\big) \leq 2 \exp \Bigg(-\frac{2t^2}{N}\Bigg).\]\qed
\end{prop}

\subsubsection{Asymmetric Hamming Cube}

We will now produce a pq-space based on the Hamming cube by replacing $\rho_n$ by a quasi-metric. The simplest way is to define $d_1:\Sigma\to\R$ by $d_1(0,1)=1$ and $d_1(1,0)=d_1(0,0)=d_1(1,1)=0$ and set $d_n(\sigma, \tau)=\frac1n\sum_{i=1}^n d_1(\sigma_i,\tau_i)$. The triple $(\Sigma^n, d_n,\mu_n)$ forms a pq-space. It would not add much to generality to replace $\mu_n$ by a product of copies of a different probability measure on $\Sigma$. One immediately observes that $\{(\Sigma^n, d_n, \mu_n)\}_{i=1}^{\infty}$ is also a normal L\'evy family.

Take two strings $\sigma$ and $\tau$ and let us consider the asymmetry $\Gamma_n(\sigma,\tau)$. It is easy to see that $\Gamma_n$ takes value between $0$ and $1$, being equal to the quantity \[\frac{1}{n}\Big| \abs{\{i:\sigma_i=0\wedge\tau_i=1\}} - \abs{\{i:\sigma_i=1\wedge\tau_i=0\}} \Big|.\]

Since our asymmetric Hamming cube is a product space, we can consider for each $i \leq  n$ the 
value $\delta_i=d(\sigma_i,\tau_i)-d(\tau_i,\sigma_i)$ as a random variable taking values of $0$, $-1$  and $1$ with $P(\delta_i=0)=\frac12$ and $P(\delta_i=-1)=P(\delta_i=1)=\frac14$ so that  $\Gamma_n(\sigma,\tau) = \frac{1}{n}\sum_{i\leq n}\abs{\delta_i}$. Now,
\begin{align*}
\mu_n\otimes\mu_n (\{(\sigma,\tau)\in \Sigma^n\times \Sigma^n:\Gamma_n(\sigma,\tau)\geq\e\}) & = P\big(\sum_{i\leq n}\frac{1}{n}\abs{\delta_i} \geq\e\big) \\
& \leq P\big(\sum_{i\leq n}\frac{1}{n}\abs{\epsilon_i}\geq\e\big) \\
& \leq 2 \exp \Bigg(-\frac{n\e^2}{2}\Bigg).
\end{align*}

This is obviously the same bound as would be obtain by application of the Theorem \ref{thm:qpclosemm} and the Proposition \ref{prop:hammcubeconc}. 

%

\subsection{General setting}

Product spaces assume great importance in the present investigation for two reasons. Firstly, the theory of concentration there is quite extensively developed, mostly due to the work of Michel
Talagrand \cite{Tal95,Tal96}. Many of his results are quite general, that is, not restricted to the products of metric spaces, and can be applied directly to the quasi-metric spaces. Secondly, the space of protein fragments, the main biological example of this thesis, can be modelled as a product space, although the measure on it is definitely not a product measure. However, the bounds on the concentration function thus obtained can be used as a worst case estimate which can be useful in indexing applications.

It should also be noted that the generality of the results means that they can even be applied to the similarity scores that do not transform into quasi-metrics (i.e. which do not satisfy the triangle inequality).

Talagrand \cite{Tal95} obtained the exponential bounds for product spaces endowed with a non-negative `penalty' function generalising the distance between two points. Penalties form a much wider class of distances than quasi-metrics but provide ready bounds for the concentration functions.

We will outline here just one of results from \cite{Tal95} and apply it to obtain bounds for concentration functions in product quasi-metric spaces with product measure.

Consider a probability space $(\Omega,\Sigma,\mu)$ and the product $(\Omega^N, \mu^N)$ where the product probability $\mu^N$ will be denoted by $P$. Consider a function $f:2^{\Omega^N}\times\Omega^N \to\R_+$ which will measure the distance between a set and a point in $\Omega^N$. More specifically, given a function $h: \Omega\times\Omega\to\R_+$ such that $h(\omega,\omega)=0$ for all $\omega\in\Omega$, set \[f(A,x) = \inf\left\{\sum_{i\leq N}h(x_i,y_i); y\in A\right\}.\]

\begin{thm}[\cite{Tal95}]
Assume that \[\norm{h}_\infty = \sup_{x,y\in\Omega}h(x,y)\] is finite and set
\[\norm{h}_2 = \left(\int\int_{\Omega^2}h^2(\omega, \omega')d\mu(\omega) d\mu(\omega')\right)^{1/2}.\]
Then
\[ P(\{f(A,\cdot)\geq u\}) \leq \frac{1}{P(A)} \exp\left(-\min\left(\frac{u^2}{8N\norm{h}_2^2}, \frac{u}{2\norm{h}_\infty}\right)\right).\] \qed
\end{thm}

If we take as $h$ above $d_\Omega$, a quasi-metric on $\Omega$, and endow $\Omega^N$ with the $\ell_1$-type quasi-metric $d$ so that $x,y\in\Omega^N$, $d(x,y) = \sum_{i\leq N} d_\Omega(x_i,y_i)$, we have $\norm{d_\Omega}_\infty=\diam(\Omega)$ and $f(A,x)=d(x,A)$. Hence, the following corollary is obtained.

\begin{corol}\label{cor:pqproduct}
Suppose $\diam(\Omega) < \infty$. Then \[ \alpha_{(\Omega^N,d,\mu^N)}(\e) \leq 2\exp\left(-\min\left(\frac{\e^2}{8N\norm{d_\Omega}_2^2}, \frac{\e}{2\diam(\Omega)}\right)\right).\]\qed
\end{corol}

Note that the bound applies to $\alpha$ and hence to both $\alpha^L$ and $\alpha^R$ because the norms referred to above are symmetric.

An advantage of an inequality of this sort in applications to the biological sequences is that $\norm{q_\Omega}_2$ can be easily calculated for a finite alphabet $\Omega$. On the other hand, it is remarked in \cite{Tal95} that the constants above are not sharp. 

\begin{example}
Consider the pq-space $X=(\Sigma^N,d,\mu^N)$ where $\Sigma$ is the amino acid alphabet, $d$ is the $\ell_1$-quasi-metric extended from the quasi-metric $d_\Sigma$ on $\Sigma$ and $\mu$ is a probability measure on amino acids. Then, the Corollary \ref{cor:pqproduct} provides explicit bounds for the concentration functions on $X$.

In particular, if $d_\Sigma$ is the quasi-metric obtained from the BLOSUM62 similarity scores and $\mu$ is obtained from the amino acid counts from a large protein dataset (they differ very little if the dataset is general enough; specifically take the counts from the NCBI nr dataset described in detail in Subsection \ref{subsec:protdatasets}), we have $\diam(\Sigma)=15$ and $\norm{d_\Sigma}_2^2=\sum_{\sigma\in\Sigma}
\sum_{\tau\in\Sigma} d_\Sigma^2(\sigma,\tau)\mu(\{\sigma\})\mu(\{\tau\}) = 45.0193$.

While the above would give an explicit formula for the bounds of the concentration functions on the space of peptide fragments $\Sigma^N$ under the assumption that the measure on $\Sigma^N$ is a product measure, one would ultimately wish to estimate the `true' concentration functions on $\Sigma^N$ -- this is something we do not yet know how to do. 
Indeed, were it to be attempted directly from the definition, by choosing a subset and computing the measure of its $\e$-neighbourhood one at a time, the computational complexity would be exponential in the size of the set.
\end{example}

\chapter{Indexing Schemes for Similarity Search}\label{ch:3}

\section{Introduction}

It would not be exaggerated to state that database search is one of the pillars of the modern information society. Datasets come in many forms, from simple flat-files to relational databases. Classical databases are structured around data points (\emph{records}) with \emph{keys} which may contain numeric, textual or categorical data, allowing comparison and search queries. The most fundamental type of search queries is \emph{exact match} -- all datapoints matching a given key are retrieved. If the type of the key is numeric, it is possible to perform \emph{range queries} where the set of points within a given range of the query key is retrieved. If the key is a string, a \emph{partial match} query can be asked: it retrieved those datapoints whose keys match the query key in part (for example, by sharing a common prefix). In all cases an additional structure such as for example linear order is imposed on data keys to facilitate retrieval of queries.  

Sometimes it is possible to assume that datapoints belong to an $n$-dimensional vector space with the coordinates corresponding to their \emph{features}. In this case, exact matches are often not sufficient: unless the underlying space is strictly limited in some way, the probability that there will be a datapoint exactly matching a query is close to $0$. On the other hand, before proceeding with range queries, it is necessary to define a \emph{similarity} or \emph{proximity} measure used to retrieve queries, a function of two variables that on input of the query and some other point returns their similarity (degree to which the points are similar) or distance (in this case it is commonly called a \emph{dissimilarity measure}). For $n$-dimensional vector spaces the obvious choice of a dissimilarity measure is an $\ell_p^n$ or Minkowski metric where $d(x,y)=\left(\sum_{i=1}^n \abs{y_i-x_i}^p \right)^{\frac1p}$ or its weighted modifications where each coordinate is assigned a weight. 

The approach of retrieving points according to a similarity measure can be applied to datasets which cannot be easily represented as vector spaces, for example sets of words from a finite alphabet, colour images, time series, audio and video streams etc. Such sets are often large, complex (both in the structure of data and the underlying similarity measure) and fast growing. One well known example is GenBank \cite{BKLOW04}, the database of all publicly available DNA sequences (Figure \ref{fig:genbank}). In this case, the size of queries is much smaller than database size and it is imperative to attempt to avoid scanning the whole dataset in order to retrieve a very small part of it. 

\begin{figure}[!ht] 
\begin{center}
\scalebox{1.0}{\includegraphics{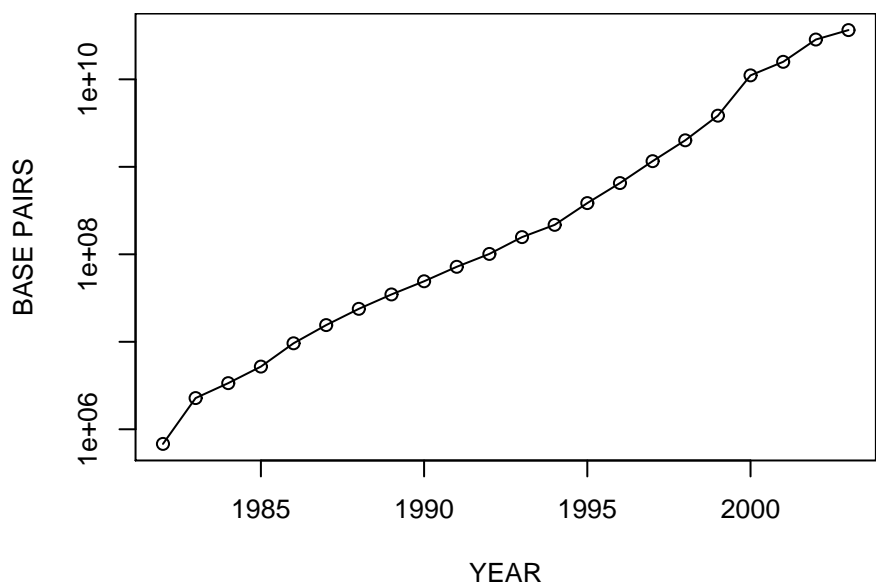}}
\caption[Growth of GenBank DNA sequence database.]{Growth of GenBank DNA sequence database (log scale). Data taken from \url{http://www.ncbi.nlm.nih.gov/Genbank/genbankstats.html}.}
\label{fig:genbank}
\end{center}
\end{figure}

Loosely speaking, \emph{indexing} denotes introduction of a structure, called \emph{indexing scheme}, to a dataset. This structure supports an \emph{access method} for fast retrieval of queries by enabling elimination of those parts of the dataset which can be certified not to contain any points of the query. There are numerous examples of indexing schemes and access methods, the best known being the B-Tree \cite{Co79} from the classical database theory. However, in order to design new and efficient indexing schemes, a fully developed mathematical paradigm of indexability that would incorporate the existing structures and possess a predictive power is needed. 

The master concept was introduced in the influential paper by Hellerstein, Koutsoupias and Papadimitriou \cite{H-K-P}: a \emph{workload}, $W$, is a triple consisting of a search domain $\Omega$, a dataset $X$, and a set of queries, $\mathcal Q$. An indexing scheme according to \cite{H-K-P} is just a collection of \emph{blocks} covering $X$. While this concept is fully adequate for many aspects of theory, we believe that analysis of indexing schemes for similarity search, which is the aim of this chapter, with its strong geometric flavour, requires a more structured approach. Hence, a concept of an indexing scheme as a system of blocks equipped with a tree-like search structure and decision functions at each step is put forward. This concept is a result of analysis of numerous concrete existing approaches to indexing. The notion of a \emph{consistent} indexing scheme, guaranteeing full retrieval of all queries, is stressed.

The notion of a \emph{reduction} of one workload to another, allowing creation of new access methods from the existing ones is also suggested. The final sections of the present chapter discuss how geometry of high dimensions (asymptotic geometric analysis) may offer a constructive insight into the performance of indexing schemes and, in particular, in the nature of the curse of dimensionality. 

Apart from \cite{H-K-P}, this work was influenced by the excellent reviews of similarity search in metric spaces by Chavez, Navarro, Baeza-Yates and Marroquin \cite{CNBYM} and by Hjaltason and Samet \cite{HjSa03}. While \cite{HjSa03} is mostly concerned with detailed descriptions of each of the existing methods, the main focus of the \cite{CNBYM} paper is on classification of indexing schemes and analysis of their performance, with particular emphasis on the \emph{curse of dimensionality}. Another good survey (in Italian) is Licia Capra's Masters thesis \cite{Capra00}. The conceptual framework and techniques for explaining the curse of dimensionality comes from the works of Pestov \cite{Pe00,Pe99} and this chapter can be thought of as an extension of the results presented therein. The paper of Ciaccia and Patella \cite{CiPa02}, while focusing only on one particular scheme, gives an important insight into cost models for similarity search. 

It should be noted that while the fundamental building blocks - similarity measures, data distributions, hierarchical tree index structures, and so forth - are in plain view, the only way they can be assembled together is by examining concrete datasets of importance and taking one step at a time. Generally, this thesis shares the philosophy espoused by Papadimitriou in \cite{Papa95} that theoretical developments and massive amounts of computational work must proceed in parallel. Indeed, it is our general impression that indexing schemes which are able to take into account the underlying structure of a domain often perform better than `generic' schemes.

As noted earlier, the main motivation comes from sequence-based biology, where similarity search already occupies a very prominent place and where high-speed access methods for biological sequence databases will be vital both for developing large-scale data mining projects \cite{Go02} and for testing the nascent mathematical conceptual models \cite{CaGr01}. 

As seen in Chapter \ref{ch:bioseq_qm}, the similarity measures used for biological sequence comparison often correspond to partial metrics or quasi-metrics. For that reason, a particular emphasis is placed on indexing schemes for quasi-metric workloads, which, while frequently mentioned as generalisations of metric workloads (e.g. in \cite{CiPa02}), have been so far been neglected as far the practical indexing schemes are concerned. The main technical result of this Chapter, the Theorem \ref{thm:rngconc} about the performance of range searches, is stated and proved in terms of the quasi-metric workloads.

An indexing scheme for short peptide fragments called FSIndex illustrates many of the concepts introduced in the present chapter, and is the main subject of the next chapter.

\newpage
\section{Basic Concepts}

\subsection{Workloads}

\begin{defin}[\cite{H-K-P,Pe00,PeSt02}]
A \emph{workload} is a triple $W=(\Omega,X,{\mathcal Q})$, where $\Omega$ is a set called the \emph{domain,} $X$ is a finite subset of the domain (\emph{dataset}, or \emph{instance}), and $\mathcal{Q}\subseteq\PowE{\Omega}$ is the set of \emph{queries,} that is, some specified subsets of $\Omega$. 

(Here, as in the Definition \ref{def:setsofsubsets}, $\PowE{\Omega}$ denotes the set of all subsets of $\Omega$ including $\emptyset$, the empty set.)

\emph{Answering a query} $Q\in{\mathcal Q}$ means listing all data points $x\in X\cap Q$. 
\end{defin}

The concept of workload was introduced in \cite{H-K-P} and the original definition is slightly extended here by having the queries as subsets of $\Omega$ rather than $X$. This is however an important distinction because it is often not directly known what the dataset contains and we may want to ask `questions' (queries) independently of possible `answers' (dataset points). For that reason empty queries are also allowed -- some processing is usually required in order to decide whether a query is in fact empty. There are also technical reasons which are discussed in Subsection \ref{subseq:probdist}.

The domain $\Omega$ can be a very large, even infinite set. It would be tempting at this stage to turn the domain with the set of queries into a topological space by requiring $\mathcal Q$ to satisfy the axioms of topology but there is no practical use for that. In the later sections, when we define similarity queries, the queries will become neighbourhoods of points according to some similarity measure (say a metric) and would thus form a base of a topology over $\Omega$. Even in that case, there is no need to require that finite intersections or infinite unions of families of queries are queries themselves. Indeed, since the dataset $X$ is finite, the finite unions would be sufficient for any practical purpose. The dataset itself with the topology induced from the domain would be topologically discrete and zero dimensional and thus trivial from the topological point of view.

Examples of workloads abound in database theory - we here focus on the most abstract versions that will be important further on.

\begin{example}
\label{ex:trivial}
The \emph{trivial workload}: $\Omega=X=\{\ast\}$ is a one-element set, with a sole possible non-empty query, $Q=\{\ast\}$. 
\end{example}

\begin{example} Let $X\subseteq\Omega$ be a dataset. The \emph{exact match queries} for $X$
are singletons, that is, sets $Q=\{\omega\}$, $\omega\in\Omega$.
\end{example}

\begin{example} Let $n\in\N$, $\Omega$ = $K\times Y_1\times Y_2\times\ldots\times Y_n$ and $X\subseteq\Omega$ be a dataset. Define the set of queries by $\mathcal Q=\{Q_k\ |\ k\in K\}$ where $Q_k=\{\omega\in\Omega: \omega|_K = k\}$. This is the most common type of a query in classical database theory where $\Omega$ is a \emph{table} with a \emph{key} $K$ and a query $Q_k$ retrieves all elements of $X$ whose key is equal to $k$.  
\end{example}

Here is the first way to create new workloads: by combining them as disjoint sums.

\begin{example} Let $W_i=(\Omega_i,X_i,{\mathcal Q}_i),i=1,2,\ldots,n$ be a finite collection of workloads. Their \emph{disjoint sum} is a workload $W= \sqcup_{i=1}^n W_i$, whose domain is the disjoint union $\Omega=\Omega_1\sqcup \Omega_2\sqcup\ldots\sqcup \Omega_n$, the dataset is the disjoint union $X=X_1\sqcup X_2\sqcup\ldots\sqcup X_n$, and the queries are of the form $Q_1\sqcup Q_2\sqcup\ldots\sqcup Q_n$, where $Q_i\in{\mathcal{Q}}_i$, $i=1,2,\ldots,n$. 
\end{example}

\begin{example} Let $W=(\Omega,X,{\mathcal Q})$ be a workload, and let $\Theta\subseteq\Omega$. The \emph{restriction} of $W$ to $\Theta$ is a workload $W\vert_\Theta$ with domain $\Theta$, dataset $X\vert_{\Theta}=X\cap\Theta$ and the set ${\mathcal{Q}}\vert_\Theta$ of queries of the form $Q\cap\Theta$, $Q\in{\mathcal{Q}}$. 
\end{example}

The main objects of this chapter are \emph{similarity workloads} where the queries are generated by \emph{similarity (or proximity) measures}.

\subsection{Similarity queries}

In general, a \emph{similarity measure} \cite{CPZ97,CPZ98a,HjSa03} on a set $\Omega$ is a function of two variables $s\colon \Omega\times\Omega\to\R$, often subject to additional restrictions. In a strict sense, such as in bioinformatics \cite{altschul97gapped}, the term \emph{similarity measure} (or \emph{similarity score}, or just \emph{similarity}) is used for a function $s$ such that the pairs of `close' points take a large and often positive value while the points which are `far' from each other take a small (often negative) value. 

Throughout this work we shall always consider \emph{dissimilarity} \cite{CPZ97,CPZ98a} or \emph{distance} measures, the similarity measures (in a wider sense) which measure how far apart two points are. We require that all the values are positive and add an additional requirement that the pair of identical points takes the value $0$ (this is different from Remark \ref{rem:dist} where we assume in addition that a distance satisfies the triangle inequality). The justification is that most commonly used (dis)similarity measures are metrics or at least quasi-metrics and that it is almost always possible to convert a similarity measure in a strict sense into a dissimilarity measure.

\begin{defin}
A \emph{dissimilarity measure} on a set $\Omega$ is a function $d\colon \Omega\times\Omega\to\R_+$ where for all $\omega\in\Omega$, $d(\omega,\omega)=0$.
\end{defin}

The three types of queries based on a dissimilarity measure of most interest \cite{CNBYM} are: a \emph{range query}, a \emph{nearest neighbour query} and a \emph{$k$-nearest neighbours} (or \emph{ kNN}) \emph{query}.

\begin{defin}
Let $\Omega$ be a set, $d$ a dissimilarity measure on $\Omega$, $X\subseteq\Omega$ a dataset and $r\in\R_+$. The \emph{($r$-) range similarity query centred at} $\omega\in\Omega$, denoted $Q^\text{rng}_d(\omega,r)$, is defined by
\[Q^\text{rng}_d(\omega,r) = \{x\in\Omega: d(\omega, x)\leq r\},\] that is, $Q^\text{rng}_d(\omega,r)$ consists of all $x\in\Omega$ that are within the distance $r$ of $\omega$. We will denote by $\mathcal{Q}^\text{rng}_d$ the set $\{Q^\text{rng}_d(\omega,r)\ |\ \omega\in\Omega,\ r\in\R_+\}$, of all possible range queries.

We call a workload $(\Omega,X,\mathcal{Q}^\text{rng}_d)$ a \emph{range (dis)similarity workload}.  
\end{defin}

If $d$ is a quasi-metric, the range query $Q^\text{rng}_d(\omega,r)$ corresponds exactly to the left closed ball $\clball{\omega}{r}$ and if $d$ is a metric then $Q^\text{rng}_d(\omega,r)=\cball{\omega}{r}$, the closed ball of radius $r$ about $\omega$.

\begin{defin}
Let $\Omega$ be a set, $d$ a dissimilarity measure on $\Omega$ and $X\subseteq\Omega$ a dataset. The \emph{nearest neighbour query centred at} $\omega\in\Omega$, denoted $Q^\text{NN}_d(\omega, X)$, is defined by \[Q^\text{NN}_d(\omega, X) = \{x\in X: d(\omega,x)\leq d(\omega,y)\ \text{for all}\ y\in X\},\] that is, it consists of members of $X$ closest to $\omega$. 

Denote by $d^\text{NN}_X(\omega)$ the distance to a nearest neighbour of $\omega$ in $X$.

We call a workload $(\Omega,X,\mathcal{Q}^\text{NN}_d)$ a \emph{nearest neighbour (dis)similarity workload}.
\end{defin}

\begin{defin}
Let $\Omega$ be a set, $d$ a dissimilarity measure on $\Omega$ and $X\subseteq\Omega$ a dataset and let \[r_k = \inf\{r\geq 0: \abs{Q^\text{rng}_d(\omega,r)\cap X}\geq k\}. \]
The \emph{$k$-nearest neighbour query centred at} $\omega\in\Omega$, also called a \emph{kNN query}, denoted $Q^\text{$k$NN}_d(\omega, X)$, is defined by \[Q^\text{$k$NN}_d(\omega, X) = Q^\text{rng}_d(\omega,r_k)\cap X.\] In other words, $Q^\text{$k$NN}_d(\omega, X)$ is a set of $k$ elements of $X$ closest to $\omega$ plus any other elements of $X$ at the same distance as the $k$-th nearest neighbour.

We call a workload $(\Omega,X,\mathcal{Q}^\text{$k$NN}_d)$ a \emph{kNN (dis)similarity workload}.
\end{defin}

The nearest neighbour and the $k$-nearest neighbours queries are jointly called \emph{NN-queries} \cite{CNBYM}. Unlike range queries, they directly depend on the dataset $X$. Note that our definition of $k$NN queries differs from the one commonly used in the literature \cite{CNBYM,HjSa03}, where any set of $k$ elements of $X$ closest to $\omega$ is sufficient to satisfy a $k$NN query. We chose the above definition for consistency -- every algorithm is guaranteed to return the same result and $Q^\text{$k$NN}_d(\omega, X)$ denotes a single set and not a family of sets. 

Our definition also makes the connection between NN-queries and range queries explicit: any NN-query can be expressed in terms of a range query. For example, for a nearest neighbour query, we have $Q^\text{NN}_d(\omega, X) = X\cap Q^\text{rng}_d(\omega, d^\text{NN}_X(\omega))$. Of course, in practical situations, $d^\text{NN}_X(\omega)$ is not known in advance. Nevertheless, we shall mostly concentrate on range similarity queries and workloads as the most fundamental of the three and easiest to process.

\begin{defin}
Let $\Omega$ be a domain and $d_1$ and $d_2$ dissimilarity measures. If $\mathcal{Q}^\text{rng}_{d_1} = \mathcal{Q}^\text{rng}_{d_2}$ we call $d_1$ and $d_2$ \emph{equivalent}.
\end{defin}

\begin{example}
Let $(\Omega, d_1)$ and $(\Omega, d_2)$ be metric spaces. Recall that two metrics $d_1$ and $d_2$ are \emph{equivalent} if and only if there exist strictly positive constants $a,b$ such that for all $x,y\in\Omega$, $ad_1(x,y)\leq d_2(x,y)\leq bd_1(x,y)$. The metric and dissimilarity measure notions of equivalency do not follow from each other. 

Take a set $\Omega=\{\frac1n:n\in\N_+\}\cup\{0\}$ with the metrics $d_1$ and $d_2$ where $d_1(x,y)=\abs{x-y}$ and $d_2(x,y)=\sqrt{\abs{x-y}}$. It is clear that $d_1$ and $d_2$ are equivalent as dissimilarity measures since they generate the same sets of balls while there is no strictly positive constant $a$ such that for all $x\in\Omega$, $\sqrt{x}\leq ax$ and thus $d_1$ and $d_2$ are not equivalent as metrics.

On the other hand, let $\Omega=\R^2$ where $d_1(x,y)=\sqrt{(x_1-y_1)^2+(x_2-y_2)^2}$ and $d_2(x,y)=\sqrt{(x_1-y_1)^2+2(x_2-y_2)^2}$. It is easy to see that $d_1$ and $d_2$ are equivalent metrics but not equivalent dissimilarity measures since $d_1$ generates the balls of circular shape (Euclidean balls) while $d_2$ generates elliptical balls. 

If $d_2$ is obtained from $d_1$ by a \emph{metric transform}, (i.e. $d_2(x,y)=F(d_1(x,y))$ where $F:[0,+\infty)\to [0,+\infty)$ is a concave monotone function with $F(0)=0$), then $d_1$ and $d_2$ are equivalent as similarity measures. One example of a metric transform is $d_2 = ad_1$ for some $a>0$, where $d_2$ is a multiple of $d_1$.
\end{example}

\subsection{Indexing schemes}

\begin{defin}
An \emph{access method} for a workload $W$ is an algorithm that on an input $Q\in{\mathcal Q}$ outputs all elements of $Q\cap X$.
\end{defin}

Typical access methods come from indexing schemes.

\begin{figure}[!ht] 
\begin{center}
\scalebox{1.0}{\input{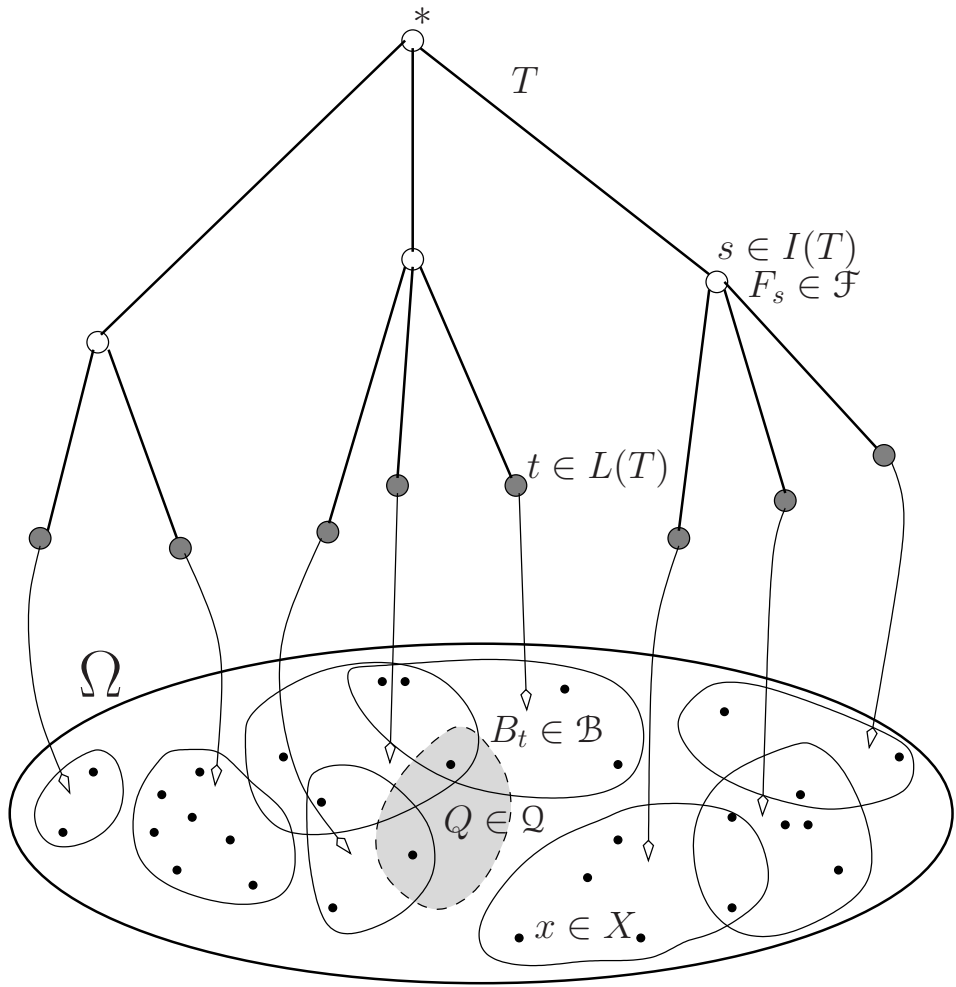}}
\caption{An indexing scheme ${\mathcal I}=(T,\mathscr{B},\F)$ on a workload$(\Omega,X,\mathcal{Q})$.}
\label{fig:indscheme}
\end{center}
\end{figure}

\begin{defin}
Let $T$ be a rooted finite tree. Denote by $L(T)$ the set of leaf nodes and by $I(T)$ the set of inner nodes of $T$.  The notation $t\in T$ means that $t$ is a node of $T$, and $C_t$ denotes the set of all children of a $t\in I(T)$. For any non root node $t$, the parent of $t$ is denoted $p(t)$.
\end{defin}

\begin{defin}\label{defn:indscheme}
Let $W=(\Omega,X,{\mathcal Q})$ be a workload. An \emph{indexing scheme} on $W$ is a triple ${\mathcal I}=(T,\mathscr{B},\F)$, where 
\begin{itemize}
\item $T$ is a rooted finite tree, with root node $\ast$,
\item ${\mathscr{B}}$ is a collection of subsets $B_t\subseteq\Omega$ (\emph{ blocks}, or \emph{bins}), where $t\in L(T)$, such that $X \subseteq\bigcup_{t\in L(T)}B_t$.
\item $\F=\{F_{t}\colon t\in I(T)\}$ is a collection of set-valued \emph{decision functions,} $F_t\colon {\mathcal Q} \to 2^{C_t}$, where each value $F_{t}(Q)\subseteq C_t$ is a subset of children of the node $t$.
\end{itemize}
\end{defin}

\begin{pseudocode}[ovalbox]{$W$.RetrieveIndexedQuery}{{\mathcal I},Q}\label{alg:main}
\COMMENT{Indexing scheme ${\mathcal I}=(T,\mathscr{B},\F)$ over $W=(\Omega,X,{\mathcal Q})$}\\
\COMMENT{Query $Q\in\mathcal{Q}$}\\
A_0\GETS\{\ast\}\\
R\GETS\emptyset\\
i\GETS 0\\
\WHILE A_i\neq\emptyset \DO
\BEGIN
  A_{i+1}\GETS\emptyset\\
  \FOREACH t\in A_i \DO
  \BEGIN
    \IF t\notin L(T) \THEN
       A_{i+1}\GETS A_{i+1}\cup F_{t}(Q)\\
    \ELSE
      \FOREACH x\in B_t \DO
      \BEGIN
        \IF x\in Q \THEN
          R\GETS R\cup\{x\}\\
      \END\\
   \END\\
   i\GETS i+1\\
\END\\
\RETURN{R}
\end{pseudocode}
\algocapt{Answering a query using an indexing scheme.}

Hence, an indexing scheme consists of a cover $\mathscr{B}$ of $X$ by blocks and a tree structure that determines the way in which a query is processed: for each query we traverse those nodes that have been selected at their parent nodes using the decision functions (Figure \ref{fig:indscheme}). Each of the bins associated with selected leaf nodes is sequentially scanned for elements of the dataset satisfying the query. The Algorithm \ref{alg:main} depicts a breadth-first traversal of the tree but any other equivalent algorithm can be used. We will only consider \emph{consistent indexing schemes}: those for which the above procedure retrieves all dataset elements belonging to any query, that is, no query points are missed. This is more formally expressed by the following definition:

\begin{defin}
An indexing scheme ${\mathcal I}=(T,\mathscr{B},\F)$ for a workload $W=(\Omega,X,{\mathcal Q})$ is \emph{consistent} if for every $Q\in\mathcal{Q}$ and for every $x\in Q\cap X$ there exists $t\in L(T)$ such that $x\in B_t$ and the path $s_0s_1\ldots s_m$, where $s_0=\ast$, $s_m=t$ and $s_{i}=p(s_{i+1})$, satisfies $s_{i+1}\in F_{s_i}(Q)$ for all $i=0,1\ldots m-1$.
\end{defin}

%

Clearly, for a consistent indexing scheme, any algorithm which, for any query, starting from the root, visits all branches returned by the decision functions at each node and scans all bins associated with the leaf nodes visited for the members of the query, is an access method. The Algorithm \ref{alg:main} provides one example.

Our definition of indexing scheme extends the definition of \cite{H-K-P} which considers only the set of blocks. The computational complexity of the decision functions $F_t(Q)$, as well as the amount of `branching' resulting from an application of Algorithm \ref{alg:main}, become major efficiency factors in case of similarity-based search, which is why we feel they should be brought into the picture.

Note that blocks may overlap in an indexing scheme, that is, a point $x\in X$ can belong to several blocks. There may even be different leaves pointing to the same block. This observation is at the heart of the concept of \emph{storage redundancy} developed in \cite{H-K-P} and \cite{HKMPS02} which will be examined later. 

We now present examples of indexing schemes related to some of the most fundamental algorithms of computer science, reformulating them within our proposed framework. We provide a very short description and a reference to the appropriate section of the Volume 3 (Sorting and Searching) of Knuth's `The Art of Computer Programming' (TAOCP) \cite{KnuthTAOCP}. It should be noted that while the discussion in TAOCP applies to exact searches, the ideas in many cases apply to more general cases with very few modifications.

\begin{example}\label{ex:seqscan} A simple linear scan (TAOCP, Vol. 3, Section 6.1) of a dataset $X$ corresponds to the indexing scheme where the tree $T=\{\ast,\star\}$ has a root $\ast$ and a single child $\star$, $\mathcal B$ consists of a single block $B_\star=\Omega$, and the decision function $F_\ast$ always outputs the same value $\{\star\}$.
\end{example}

\begin{example}\label{ex:hashing}
\emph{Hashing} (TAOCP, Vol. 3, Section 6.4) can be described in terms of the following indexing scheme for exact searches. The tree $T$ has depth one, with its leaves corresponding to bins, and the decision function $F_{\ast}$ is a hashing function: on input of a query object $Q$ it outputs the bin in which the elements of $X$ matching $Q$ are stored. If there are collisions (i.e. different objects mapping to the same bin), the retrieved bin needs to be further processed.

A related technique, which can be used in some cases, is to store the results of commonly used queries and retrieve them at search time using a hash function.
\end{example}

\begin{example}\label{lodomain}
If the domain $\Omega$ is linearly ordered and the set of queries consists of intervals $[a,b]$
then an efficient indexing structure is constructed using a generalisation of binary search trees (TAOCP, Vol. 3, Section 6.2). Each bin contains one element of the dataset and every node $t\in T$ is associated with an interval $[t_1,t_2]$ which, in the case of an inner node, covers the intervals associated with the children of $t$ and in the case a leaf node corresponds to the element of the dataset contained in the bin $B_t$ (Figure \ref{fig:lindomaintree}). Each decision function $F_t$ on an input $[a,b]$ outputs the set of all children nodes $s$ of $t$ such that $[s_1,s_2]\cap [a,b]\neq\emptyset$. 

Generalisations of this idea form the core of indexing schemes for similarity workloads (Sections \ref{sec:metrictrees} and \ref{sec:qmetrictrees}).

\begin{figure}[h!bt] 
\begin{center}
\scalebox{1.0}{\input{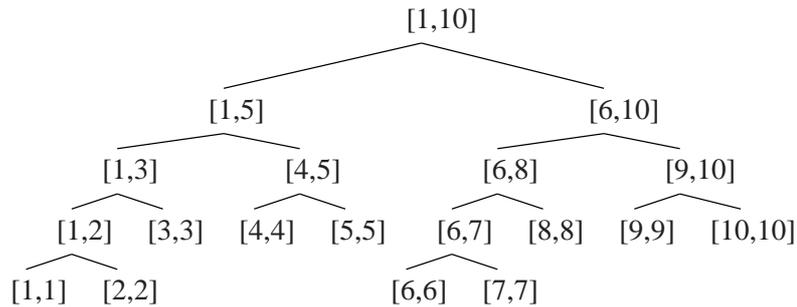}}
\caption[An indexing tree for range queries of a linearly ordered dataset.]{An indexing tree for range queries of a linearly ordered dataset of 10 elements.}
\label{fig:lindomaintree}
\end{center}
\end{figure}
\end{example}

\subsection{Inner and outer workloads}
\begin{defin}
A workload $W=(\Omega, X, \mathcal Q)$ is called \emph{inner} if $X=\Omega$ and \emph{outer} otherwise.
\end{defin}
 
Typically, for outer workloads $\abs{X}\ll\abs{\Omega}$. The difference between inner and outer workloads is particularly significant for similarity searches because inner similarity workloads can be thought of as directed weighted graphs where the dataset points are nodes and two nodes are connected with an edge with a weight corresponding to their similarity. In such case, it may be possible, depending on the characteristics of the graph and the types of queries, to use graph traversal algorithms as access methods.

In theory, every workload $W=(\Omega,X,{\mathcal Q})$ can be replaced with an inner workload $(X,X,{\mathcal{Q}}\vert_X)$, where the new set of queries ${\mathcal{Q}}\vert_X$ consists of sets
$Q\cap X$, $Q\in{\mathcal{Q}}$. However, in practical terms this reduction often makes little sense because while the complexity of storing and processing the query sets $Q\cap X$ remains essentially the same, and in addition to requiring the domain $\Omega$ to be implicitly present, we lose a geometric clarity of having the set $\Omega$ present explicitly.

\section{Metric trees}\label{sec:metrictrees}

Most existing indexing schemes for similarity search apply to metric similarity workloads, where a dissimilarity measure on the domain is a metric and the queries are balls of a given radius. Some indexing schemes apply only to a restricted class of metric spaces, such as vector spaces, others apply to any metric space. In most cases we encounter a hierarchical tree index structure where each node is associated with a set covering a portion of the dataset and a \emph{certification function} which certifies if the query ball does not intersect the covering set, in which case the node is not visited and the whole branch is \emph{pruned} (Figure \ref{fig:metrictree}). We show that for such indexing scheme to be consistent, that is, that no members of the dataset satisfying the query are missed, the certification functions need to be 1-Lipschitz.  The following concept of a \emph{metric tree} in its present precise form is new, and is based on our analysis of numerous existing approaches, which all turn out to be particular cases of our concept.

\begin{figure}[!ht] 
\begin{center}
\scalebox{1.0}{\input{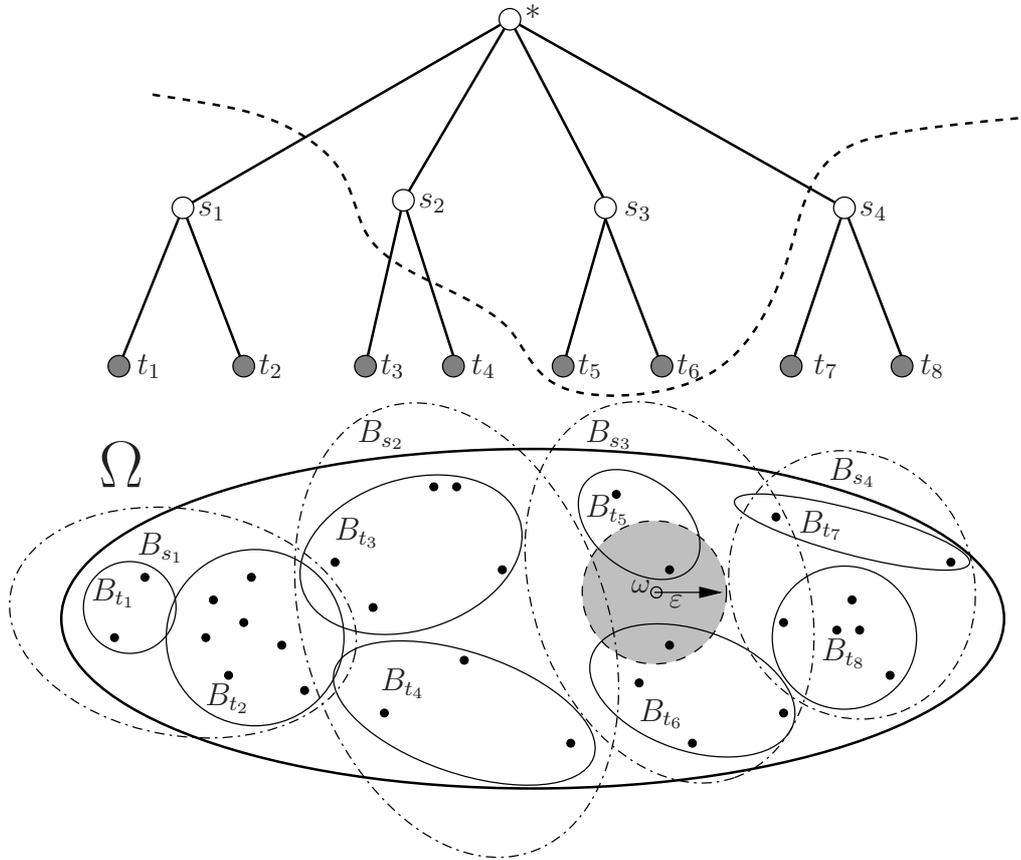}}
\caption[A metric tree indexing scheme.]{A metric tree indexing scheme. To retrieve the shaded range query the nodes above the dashed line must be scanned; the branches below can be pruned.}
\label{fig:metrictree}
\end{center}
\end{figure}

\begin{defin}
\label{def:mtree}
Let $(\Omega,X,\mathcal{Q}^\text{rng}_d)$ be a range dissimilarity workload, where $d$ is a metric. Let $T$ be a finite rooted tree with root $\ast$ and $\hat{\mathscr{B}}=\{B_t\ |\ t\in T\}$ a collection of subsets of $\Omega$ such that 
\begin{equation}\label{eq:leaf_cover}
X\subseteq\bigcup_{t\in L(T)}B_t\subseteq\Omega
\end{equation}
and for every inner node $t$, 
\begin{equation}\label{eq:inner_cover}
\bigcup_{s\in C_t}(B_s\cap X)\subseteq B_t.
\end{equation}
Also, let $\hat{\F}=\{f_t\colon\Omega\to\R \ |\ t\in T\setminus\{\ast\}\}$ be a collection of functions, called \emph{certification functions}, such that for each $t\in T\setminus\{\ast\}$,
\begin{itemize}
\item $f_t$ is 1-Lipschitz, and
\item For all $\omega\in B_t$, $f_t(\omega)\leq 0$.
\end{itemize}
We call the triple $(T,\hat{\mathscr{B}},\hat{\F})$ a \emph{metric tree} for the workload $(\Omega,X,\mathcal{Q}^\text{rng}_d)$. Let $\mathscr{B}=\{B_t\ |\ t\in L(T)\}$ and $\mathcal{F}=\{F_t\colon\mathcal{Q}\to 2^{C_t}\ |\ t\in I(T)\}$ where
\begin{equation}
F_t(\cball{\omega}{\e})=\{s\in C_t\colon f_s(\omega)\leq\e\}.
\label{certif}
\end{equation}
The indexing scheme $\mathcal{I}(T,\hat{\mathscr{B}},\hat{\F})=(T,\mathscr{B},\F)$ is called a \emph{metric tree indexing scheme}.
\end{defin}

%

The theoretical significance of the proposed concept is stressed by the following result.

\begin{thm}
\label{thm:mtree}
Let $W=(\Omega,X,\mathcal{Q}^\text{rng}_d)$ be a metric similarity workload and $(T,\hat{\mathscr{B}},\hat{\F})$ a metric tree. Then the metric indexing scheme $\mathcal{I}(T,\hat{\mathscr{B}},\hat{\F})$ is a consistent indexing scheme for $W$.
\begin{proof}
Let $Q=\cball{\omega}{\e}$ be a range query and let $x\in Q\cap X$, that is, $d(\omega, x)\leq\e$. By (\ref{eq:leaf_cover}), there exists a leaf node $t$ such that $x\in B_t$. Consider the path $s_0s_1\ldots s_m$ where $s_0=\ast$, $s_m=t$ and $s_{i}=p(s_{i+1})$, from root to $t$. By (\ref{eq:inner_cover}), for each $i=1,2\ldots m$, 
we have $(B_t\cap X)\subseteq (B_{s_{i}}\cap X)\subseteq  B_{s_{i-1}}$ and hence $x\in B_{s_i}$. It follows that $f_{s_i}(x)\leq 0$ and since $f_{s_i}$ is a 1-Lipschitz function, we have \[f_{s_i}(\omega)\leq \abs{f_{s_i}(\omega)-f_{s_i}(x)}\leq d(\omega,x)\leq\e.\] Therefore, $s_i\in F_{s_{i-1}}$ and hence $(T,\hat{\mathscr{B}},\hat{\F})$ is a consistent indexing scheme.
\end{proof}
\end{thm}

Once the collection $B_t,t\in T$ of blocks has been chosen, the certification functions always exist. 

\begin{thm}
Let $(\Omega,X,\mathcal{Q}^\text{rng}_d)$ be a range dissimilarity workload, where $d$ is a metric, $T$ be a finite rooted tree with root $\ast$ and $\hat{\mathscr{B}}=\{B_t\ |\ t\in T\}$ a collection of subsets of $\Omega$ satisfying (\ref{eq:leaf_cover}) and (\ref{eq:inner_cover}). Then, for each $t\in T$ where $t\neq\ast$, there exists a 1-Lipschitz function $f_t$ such that $f_t(\omega)\leq 0$ for all $\omega\in B_t$.
\begin{proof}
Put $f_t(\omega)=d(B_t,\omega)=\inf_{x\in B_t}d(x,\omega)$. By the Lemma \ref{lemma:setdists}, $f$ is 1-Lipschitz and clearly $f_t\vert B_t\equiv 0$.
\end{proof}
\end{thm}

However, the distances from sets are typically computationally very expensive. The art of constructing a metric tree consists in choosing computationally inexpensive certification functions that at the same time don't result in an excessive branching.

We now briefly review some of most prominent examples of metric trees. We concentrate on their overall structures in terms of the above general model and pay less attention to the details of algorithms and implementations, even though they significantly influence the performance. For many more examples and detailed descriptions the reader is directed to the original references as well as the excellent reviews \cite{CNBYM} and \cite{HjSa03}. The concept of a general metric tree equipped with 1-Lipschitz certification functions was first formulated in the present exact form in \cite{Pe00}.

\begin{figure}[!b] 
\begin{center}
\scalebox{1.0}{\input{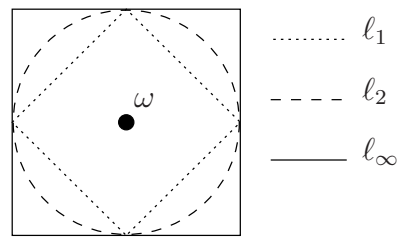}}
\caption{The shapes of the $\ell_1^2$, $\ell_2^2$ and $\ell_\infty^2$ unit balls.}
\label{fig:unitballs}
\end{center}
\end{figure}

\subsection{Vector space indexing schemes}

We first examine indexing schemes for `classical range searches', that is, for vector space workloads where the domain is $\R^n$ and the set of queries is given by the balls with respect to the $\ell_\infty^n$ metric, also called \emph{rectangles}. The rationale for this terminology is given by the shape of unit balls with respect to the $\ell_\infty^n$ norm in $\R^2$ -- the shapes of $\ell_1^2$, $\ell_2^2$ and $\ell_\infty^2$ balls are shown in Figure \ref{fig:unitballs}. Note also that this is the most general setting since for any $1\leq p<\infty$ an $\ell_p^n$ ball is contained in the $\ell_\infty^n$ ball with the same centre and radius and hence an access method for a $\ell_p^n$ workload can be obtained by what we call a \emph{projective reduction} (Subsection \ref{subsec:projred} below) to the $\ell_\infty^n$ workload. In practice, queries can be even more general, consisting of rectangles with sides of different lengths but this does not add anything to generality conceptually (if not in practical terms) since such queries can be represented, for example, as unions of (unit) balls.   

\begin{example}\label{ex:Rtree}
The \emph{R-tree} \cite{Gut84} is a dynamic structure for indexing points and rectangles in vector spaces. Many variants showing performance improvements exist, such as the $\text{R}^+$-tree \cite{SRF87} and the $\text{R}^*$-tree \cite{BKSS90}. The main feature of all variants is that bounding rectangles are used to enclose data points (at leaf nodes) or bounding rectangles of children nodes. 

\begin{figure}[!hb] 
\begin{center}
\input{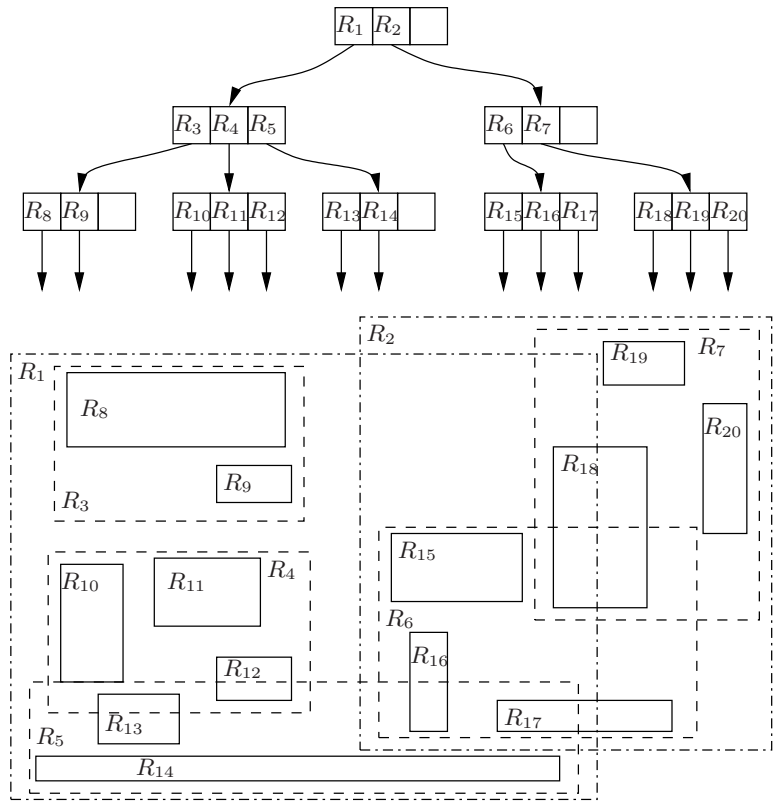}
\caption[An example of R-tree.]{An example of R-tree in two dimensions.}
\label{fig:Rtree}
\end{center}
\end{figure}

The R-trees are paged structures -- nodes are stored in secondary memory and retrieved as needed. Each non-root node of the tree $T$ has between $m$ and $M$ children with all leaves containing data points or rectangles appearing at the same level. The minimum bounding rectangle $R_t$ is associated to each node $t\in T$ (Figure \ref{fig:Rtree}). A node $t$ is visited if the query rectangle intersects $R_t$, that is, certification functions are $f_t:\omega\mapsto d(\omega,R_t)$, where $d$ is the $\ell_\infty$-metric. The structure is fully dynamic -- insertions and deletions can be intermixed with queries.

The main factor in performance of R-trees is organisation of bounding rectangles. The optimisations of the $\text{R}^*$-tree, which was shown to have the best performance of the above mentioned three variants, are based on reduction of volume and lengths of the edges of bounding rectangles at each node as well as on minimisation of overlap between rectangles associated with different nodes. 
\end{example}

\begin{figure}[!b] 
\begin{center}
\scalebox{1.0}{\input{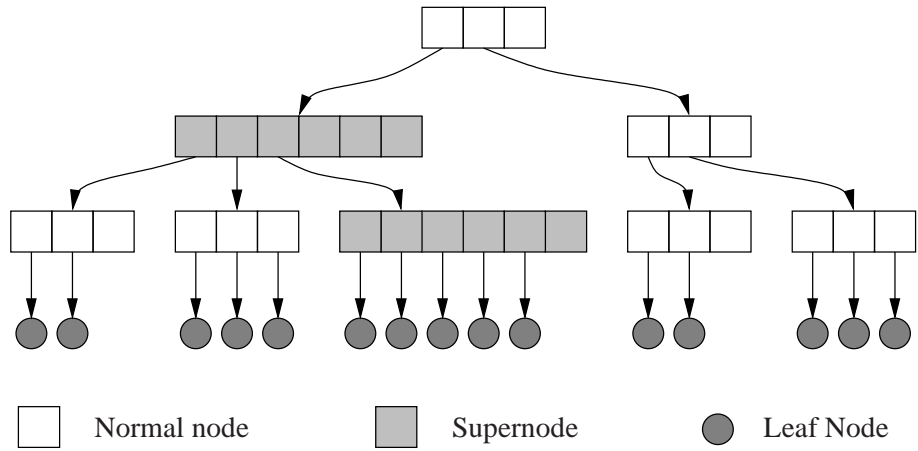}}
\caption{Structure of X-tree.}
\label{fig:Xtree}
\end{center}
\end{figure}

\begin{example}
The \emph{X-tree} \cite{BKK96} is a modification of the R-tree suitable for indexing high-dimensional vector space workloads. It is based on the observation (see Subsection \ref{subsec:dimcurse}) that high overlap between bounding rectangles of many children of R-tree nodes in high dimensions, leading to sequential scan of all them, is unavoidable. Hence the nodes whose bounding rectangles overlap to an excessively high degree are collapsed into \emph{supernodes} which are organised for linear scan (Figure \ref{fig:Xtree}). The X-tree uses the same certification functions as the R-tree: the distances to bounding rectangles. The authors report that X-tree outperforms the $\text{R}^*$-tree by as much as 8 times on high dimensional datasets.
\end{example}

\begin{example}
Consider the vector space workloads where the metric is the Euclidean ($\ell_2$) distance (more generally the weighted Euclidean distance where $w$ is a vector of weights and $d(x,y)=\sqrt{\sum_i w_i (x_i-y_i)^2})$. The \emph{SS-tree} \cite{WhJa96} is an indexing scheme where bounding spheres instead of bounding rectangles are used at each node (Figure \ref{fig:sstree}). More precisely, the region $B_t$ associated with each node $t$ is a ball centred at $x_t$, the centroid of all dataset points covered by $B_t$, with the \emph{covering radius} $r_t = \max\{d(x_t,y)\ |\ y\in X\cap B_t\}$. Hence, the certification functions are of the form $f_t(\omega)=d(\omega,x_t)-r_t$.
\begin{figure}[!hb] 
\begin{center}
\scalebox{1.0}{\input{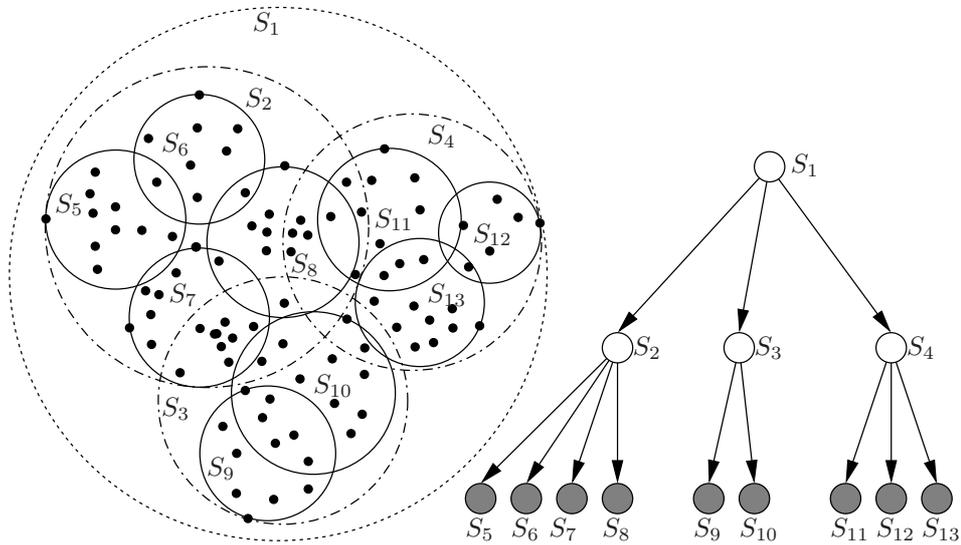}}
\caption{An example of SS-tree.}
\label{fig:sstree}
\end{center}
\end{figure}
\end{example}

\subsection{General metric space indexing schemes}

We now turn to the indexing schemes for general metric space workloads where no structure in addition to metric is assumed, that is, all that is available at creation time is the set of data points and a metric $d$.

\begin{example}\label{ex:vptree}
The \emph{vp-tree} \cite{Yianilos93} is an indexing scheme with a binary tree and certification functions of the form $f_{t_{\pm}}(\omega)= \pm\left(d(\omega,x_t)-M_t\right)$, where $x_t\in X$ is a \emph{vantage point} chosen for the non-leaf node $t$, $M_t$ is the median value for the function $\omega\mapsto d(\omega,x_t)$, and $t_{\pm}$ are two children of $t$. Thus, at each non-leaf node $t$, a part of the dataset covered by $B_t$ is partitioned into two equal halfs where $B_{t_+}= B_t\cap\ball{x_t}{M_t}$ and $B_{t_-}= B_t\setminus \ball{x_t}{M_t}$ (Figure \ref{fig:vptree}).

The $m$-ary versions, where the dataset is split in $m$-equal parts at each node, have also been proposed. 
\begin{figure}[!hb] 
\begin{center}
\scalebox{1.0}{\input{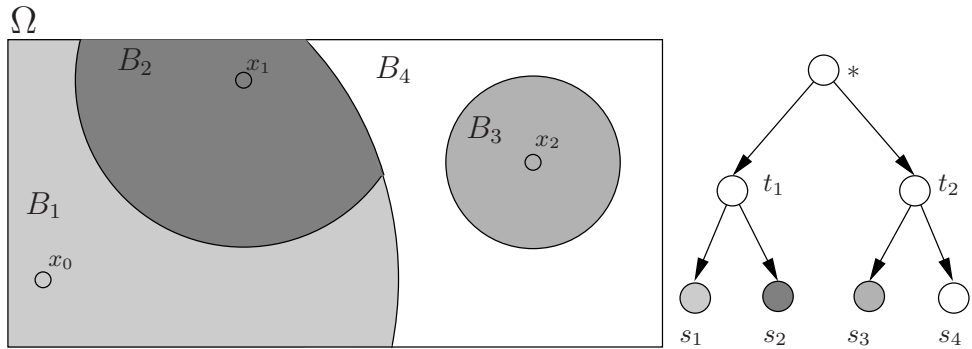}}
\caption[An example of a binary vp-tree.]{An example of a binary vp-tree with vantage points $x_0,x_1$ and $x_2$. The leaf nodes $s_1$ to $s_4$ correspond to regions $B_1$ to $B_4$.}
\label{fig:vptree}
\end{center}
\end{figure}
\end{example}

\begin{example}
The \emph{mvp-tree} \cite{BoOzs97} is a modification of the vp-tree which uses multiple vantage points at each node. In the binary case, for any node $t$, two vantage points, $x_1$ and $x_2$ are chosen and the part of the dataset covered by $B_t$ is split in four parts. 
\begin{figure}[!ht] 
\begin{center}
\scalebox{1.0}{\input{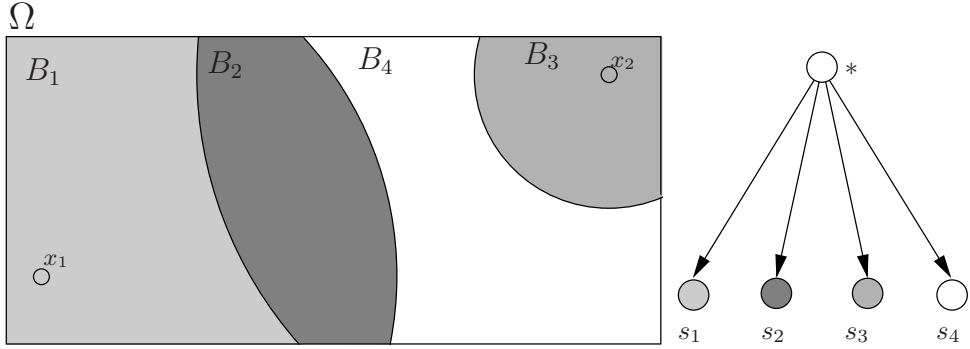}}
\caption[An example of an mvp-tree.]{An example of an mvp-tree with vantage points $x_1$ and $x_2$. The leaf nodes $s_1$ to $s_4$ correspond to regions $B_1$ to $B_4$.}
\label{fig:mvptree}
\end{center}
\end{figure}

Let $t$ be an inner node and $g_1$ and $g_2$ be the functions $\Omega\to\R$ where $g_1(\omega)=d(\omega,x_1)$ and $g_2(\omega)=d(\omega,x_2)$. Let $M_1$ be the median value for $g_1$ and $B_+= B_t\cap\ball{x_1}{M_1}$, $B_-= B_t\setminus \ball{x_1}{M_1}$. Let $M_{2+}$ be the median value for $g_2\vert B_+$ and $M_{2-}$ the median value for $g_2\vert B_-$. The certification functions for the children $t_1,t_2,t_3,t_4$ are
\begin{align*}
f_{t_1} &= \max\{d(\omega,x_1)-M_1,  d(\omega,x_2)-M_{2+}\},\\
f_{t_2} &= \max\{d(\omega,x_1)-M_1,  M_{2+}-d(\omega,x_2)\},\\
f_{t_3} &= \max\{M_1-d(\omega,x_1),  d(\omega,x_2)-M_{2-}\},\quad\text{and}\\
f_{t_4} &= \max\{M_1-d(\omega,x_1),  M_{2-}-d(\omega,x_2)\}.\\
\end{align*}
The maxima above are computed from left to right and the second value is not computed if the first exceeds the search radius. The main difference from the binary vp-tree is that two instead of three vantage points are used to divide a covering region into four regions, resulting in fewer distance computations.
\end{example}

\begin{example}\label{ex:gnat}
The \emph{GNAT} (Geometric Near-neighbour Access Tree) indexing scheme proposed by Sergey Brin \cite{Brin95}, one of the founders of Google, is based on splitting the domain $B_t$ at each node $t$ into $m$ regions $B_{t_i}$ based on proximity to the \emph{split points} $x_{t_1},x_{t_2},\ldots x_{t_m}\in X$, yielding an $m$-ary tree (Figure \ref{fig:gnat}). The sets $B_{t_i}$, called \emph{Dirichlet domains}, correspond to \emph{Voronoi cells} in $\R^n$. For each pair of split points $x_{t_i}, x_{t_j}$, the values $r_\text{lo}^{i,j}=\min\{d(x_{t_i}, y)\ |\ y\in B_{t_j}\cap X\}$ and $r_\text{hi}^{i,j}=\max\{d(x_{t_i}, y)\ |\ y\in B_{t_j}\cap X\}$ are stored. The certification functions are of the form \[f_{t_j}(\omega) = \max_{i\neq j}\max\{d(\omega,x_i)-r_\text{hi}^{i,j}, r_\text{lo}^{i,j} - d(\omega, x_i)\}.\]
\begin{figure}[!ht] 
\begin{center}
\scalebox{1.0}{\input{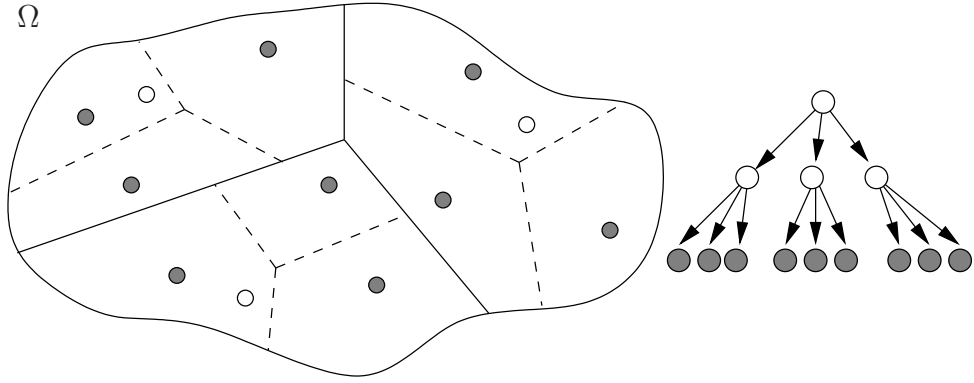}}
\caption{An example of GNAT.}
\label{fig:gnat}
\end{center}
\end{figure}
\end{example}

\begin{example}\label{ex:mtree}
Unlike the vp-tree and the GNAT but like the R-trees, the \emph{M-tree} \cite{CPZ97} is a dynamic and paged structure. The tree is binary and at each node $t$ a \emph{routing object} $x_t\in X$ is stored together with the covering radius $r_t=\max_{y\in B_t\cap X} d(x_t, y)$ and the distances to the routing objects of the children. The certification functions are of the form \[f_s(\omega)=\max\left\{\abs{d(\omega,x_{p(s)})- d(x_{p(s)},x_s)}-r_s,\quad d(\omega,x_s)-r_s \right\}.\] If the value $\abs{d(\omega,x_{p(s)})- d(x_{p(s)},x_s)}-r_s$ exceeds $\e$ the rest of $f_s$ need not be computed. This avoids potentially expensive computation of $d(\omega, x_s)$. The way the routing points are chosen and data points divided between them is determined by the user by  choosing one of many available \emph{split policies}. The best performing policy was found to be the generalised hyperplane decomposition where each data object is assigned to the routing object closest to it.

The \emph{QIC-M-tree} is a modification of the M-tree where instead of one, three distances on $\Omega$ are used: the \emph{index distance}, $d_I$, to construct the index, the \emph{comparison distance}, $d_C$, to be used in certification functions, and the \emph{query distance}, $d_Q$, according to which the actual result must be computed. The structure of the QIC-M-tree is the same as the structure of the M-tree except that the value of a certification function $f_s(\omega)$ is \[\max\left\{\abs{d_I(\omega,x_{p(s)})- d_I(x_{p(s)},x_s)}-r_s,\ d_C(\omega,x_s)-r_s,\ d_I(\omega,x_s)-r_s \right\},\]
where $x_s$ in the routing point of node $s$ and $r_s$ is the associated covering radius. As before, the evaluation is from left to right and is stopped as soon as one of the expressions exceeds the query radius. It is clear that for consistency of such indexing scheme it is necessary and sufficient that the identity maps $(\Omega, d_Q)\to (\Omega, d_I)$ and $(\Omega, d_Q)\to (\Omega, d_I)$ be 1-Lipschitz (Ciaccia and Patella allow for the scaling factors in the case this is not so). Any $d_Q$ finer than $d_C$ and $d_I$ can be used as a query distance. 

Modifications of the M-tree allowing for processing of complex queries have been proposed in \cite{CPZ98a}.  
\end{example}
 
\section{Quasi-metric trees}\label{sec:qmetrictrees}

Although often mentioned as possible generalisations of metric workloads (e.g. in \cite{CiPa02}), quasi-metric workloads have been so far neglected as far the practical indexing schemes are concerned. As our biological examples attest (Chapter \ref{ch:bioseq_qm}), quasi-metrics in fact often appear as similarity measures on datasets, even if they are not recognised as such. 

For a nearly symmetric quasi-metric $d$ on a set $\Omega$, where the asymmetry $\Gamma(x,y)=\abs{d(x,y)-d(y,x)}$ is small compared to the expected scale of the search, it may be possible to replace it by a suitable metric without significant loss of performance by the way of what we call a \emph{projective reduction} of a workload (Subsection \ref{subsec:projred}). We find a metric $\rho$ such that $\rho(x,y)\leq Kd(x,y)$ for all $x,y\in\Omega$ where $K$ is the smallest positive constant ensuring the above inequality ($K$ is in fact the Lipschitz constant of the map $(\Omega,d)\to(\Omega,\rho)$) and index the metric space $(\Omega,\rho/K)$. The QIC-M-tree \cite{CiPa02} provides exactly the framework to do so. Obvious choices for $\rho$ are $\qam{d}$ or $\qsum{d}$. In the next chapter we perform the analysis of this approach for a set of peptide fragments.

However, if the quasi-metric in question is highly asymmetric, significant loss of performance may result because the required Lipschitz constant may be very large (or even non-existent if $d$ is a $T_0$ quasi-metric) and the metric $\rho$ becomes a poor approximation to $d$. It is therefore desirable to develop a theory of indexability for quasi-metric spaces. 

We use left 1-Lipschitz functions as certification functions to establish the direct analogs of the Definition \ref{def:mtree} and the Theorem \ref{thm:mtree} (indeed, the advantage of our general model is that it allows the incorporation of the quasi-metric case with very few differences). Recall that a left 1-Lipschitz function $X\to\R$ from a quasi-metric space $(X,d)$ satisfies $f(x)-f(y)\leq d(x,y)$ for all $x,y\in X$ (Definition \ref{def:leftLip}).

\begin{defin}\label{def:qmtree}
Let $(\Omega,X,\mathcal{Q}^\text{rng}_d)$ be a range dissimilarity workload, where $d$ is a quasi-metric. Let $T$ be a finite rooted tree with root $\ast$ and let $\hat{\mathscr{B}}=\{B_t\ |\ t\in T\}$ be a collection of subsets of $\Omega$ such that 
\begin{equation}\label{eq:qleaf_cover}
X\subseteq\bigcup_{t\in L(T)}B_t\subseteq\Omega
\end{equation}
and for every inner node $t$, 
\begin{equation}\label{eq:qinner_cover}
\bigcup_{s\in C_t}(B_s\cap X)\subseteq B_t.
\end{equation}
Also, let $\hat{\F}=\{f_t\colon\Omega\to\R \ |\ t\in T\setminus\{\ast\}\}$ be a collection of certification functions such that for each $t\in T\setminus\{\ast\}$,
\begin{itemize}
\item $f_t$ is left 1-Lipschitz, and
\item For all $\omega\in B_t$, $f_t(\omega)\leq 0$.
\end{itemize}
We call the triple $(T,\hat{\mathscr{B}},\hat{\F})$ a \emph{quasi-metric tree} for the workload $(\Omega,X,\mathcal{Q}^\text{rng}_d)$. Let $\mathscr{B}=\{B_t\ |\ t\in L(T)\}$ and $\mathcal{F}=\{F_t\colon\mathcal{Q}\to 2^{C_t}\ |\ t\in I(T)\}$ where
\begin{equation}
F_t(\clball{\omega}{\e})=\{s\in C_t\colon f_s(\omega)\leq\e\}.
\label{qcertif}
\end{equation}
The indexing scheme $\mathcal{I}(T,\hat{\mathscr{B}},\hat{\F})=(T,\mathscr{B},\F)$ is called a \emph{quasi-metric tree indexing scheme}.
\end{defin}

\begin{thm}
\label{thm:qmtree}
Let $W=(\Omega,X,\mathcal{Q}^\text{rng}_d)$ be a quasi-metric similarity workload and $(T,\hat{\mathscr{B}},\hat{\F})$ a quasi-metric tree. Then the quasi-metric indexing scheme $\mathcal{I}(T,\hat{\mathscr{B}},\hat{\F})$ is a consistent indexing scheme for $W$.
\begin{proof}
Let $x\in\clball{\omega}{\e}\cap X$. By (\ref{eq:qleaf_cover}), there exists a leaf node $t$ such that $x\in B_t$. Consider the path $s_0s_1\ldots s_m$ where $s_0=\ast$, $s_m=t$ and $s_{i}=p(s_{i+1})$, from root to $t$. By (\ref{eq:qinner_cover}), for each $i=1,2\ldots m$, we have $(B_t\cap X)\subseteq (B_{s_{i}}\cap X)\subseteq  B_{s_{i-1}}$ and hence $x\in B_{s_i}$.
It follows that $f_{s_i}(x)\leq 0$ and since $f_{s_i}$ is a left 1-Lipschitz function, we have \[f_{s_i}(\omega)\leq f_{s_i}(\omega)-f_{s_i}(x)\leq d(\omega,x)\leq\e.\] Therefore, $s_i\in F_{s_{i-1}}$ and consistency follows.
\end{proof}
\end{thm}

As with metric trees, certification functions satisfying the above properties always exist -- they are provided by the distances from points to covering sets.

\begin{thm}
Let $(\Omega,X,\mathcal{Q}^\text{rng}_d)$ be a range dissimilarity workload, where $d$ is a quasi-metric, $T$ be a finite rooted tree with root $\ast$ and $\hat{\mathscr{B}}=\{B_t\ |\ t\in T\}$ a collection of subsets of $\Omega$ satisfying (\ref{eq:qleaf_cover}) and (\ref{eq:qinner_cover}). Then, for each $t\in T$ where $t\neq\ast$, there exists a left 1-Lipschitz function $f_t$ such that $f(\omega)\leq 0$ for all $\omega\in B_t$. \qed
\begin{proof}
Put $f_t(\omega)=d(B_t,\omega)$. By the Lemma \ref{lemma:setdists}, $f$ is left 1-Lipschitz and $f_t\vert B_t\equiv 0$.
\end{proof}
\end{thm}

No general quasi-metric tree indexing scheme has been produced as yet -- our indexing scheme for protein fragments (Chapter \ref{ch:4}) is an example of a quasi-metric tree but is not general. While it is possible to generalise existing indexing schemes to support quasi-metric queries, the resulting structure is usually more complex. For example, while the function $d_x:\omega\mapsto d(\omega,x)$ is left 1-Lipschitz (Lemma \ref{lem:ptLips}), $-d_x$ is right 1-Lipschitz but not necessarily left 1-Lipschitz and hence the generalisation of the vp-tree (Example \ref{ex:vptree}) certification functions as they are, just by replacing the metric with a quasi-metric, is not possible. If the distances from the same vantage point are desired to be used at each node, both the left and the right distance need to be computed and cutoff values chosen so that the whole dataset is covered and (if possible -- it may not be) that overlap is minimal. The same is true for the GNAT (Example \ref{ex:gnat}): certification functions need to be adjusted to be left 1-Lipschitz and for this it is necessary to compute both left and right distance to the split points. Hence, additional computation may be necessary at each node, adversely affecting the performance.

It appears that, out of all our examples of metric indexing schemes, the M-tree (Example \ref{ex:mtree}) is most suitable for adaptation for indexing quasi-metric workloads. The structure of a balanced binary tree should remain while the covering set at each node $s$ should be the \emph{right} closed ball $\crball{x_s}{r_s}$ of radius $r_s$ about the routing object $x_s$. The certification function $f_s$ should be set so that \[f_s(\omega)=\max\left\{d(\omega,x_{p(s)})- d(x_s, x_{p(s)})-r_s,\quad d(\omega,x_s)-r_s \right\}.\] The distances $d(x_s, x_{p(s)})$ from routing objects to their parents, as well as the covering radii $r_s=\max\{q(y,x_s)\ |\ y\in B_s\}$, can be, as is the case with M-tree, computed and stored at creation time.

The above proposal for turning the M-tree into a quasi-metric tree is, at present, only conceptual. Many challenges remain, for example in designing a good split policy to be used in the creation algorithm. If an attempt to develop a quasi-metric version of M-tree is made, it will be necessary to test it on a variety of actual quasi-metric datasets.  

\section{Valuation Workloads and Indexing Schemes}

Closely related to similarity workloads are what we call \emph{valuation workloads}.

\begin{defin}
Let $\Omega$ be a set, $X\subseteq\Omega$ a dataset and $f$ a function $\Omega\to\R$. For  
$r\in\R_+$ the \emph{($r$-) range valuation query}, denoted $Q^\text{rng}_f(r)$, is defined by
\[Q^\text{rng}_f(r) = \{x\in\Omega: f(x)\leq r\}.\] We denote by $\mathcal{Q}^\text{rng}_f$ the set $\{Q^\text{rng}_f(r)\ |\ r\in\R_+\}$ and call a workload $(\Omega,X,\mathcal{Q}^\text{rng}_f)$ a \emph{range valuation workload}.  
\end{defin}

\begin{defin}
Let $T$ be a rooted tree. A function $f:T\to\R$ is \emph{increasing on $T$} if for all $s\in T$, $t\in C_s$, $f(s)\leq f(t)$.
\end{defin}

\begin{defin}
Let $(\Omega,X,\mathcal{Q}^\text{rng}_f)$ be a range valuation workload and suppose $T$ is a finite rooted tree with root $\ast$ and $\mathscr{B}=\{B_t\ |\ t\in L(T)\}$ a collection of subsets of $\Omega$ such that $X\subseteq\bigcup_{t\in L(T)}B_t\subseteq\Omega$. Suppose $g:T\to\R$ is increasing on $T$ and for all $t\in L(T)$, \[g(t)\leq \inf_{x\in B_t} f(x). \] 
Let $\F_g=\{F_s\ |\ s\in I(T)\}$ where $F_s(Q^\text{rng}_f(r))=\{t\in C_s: g(s)\leq g(t)\}$. The indexing scheme $\mathcal{I}_g=(T,\mathscr{B},\F_g)$ is called a \emph{valuation indexing scheme}.
\end{defin}

\begin{thm}\label{thm:valuation}
Every valuation indexing scheme is consistent.
\begin{proof}
Let $\mathcal{I}_g=(T,\mathscr{B},\F_g)$ be a valuation indexing scheme over a range valuation workload $(\Omega,X,\mathcal{Q}^\text{rng}_f)$ and $Q\in\mathcal{Q}^\text{rng}_f$. Suppose $x\in Q\cap X$, that is $f(x)\leq r$ for some $r\geq 0$. Since $\mathscr{B}$ is a cover of $X$, there exists a leaf node $t$ such that $x\in B_t$. Consider the path $s_0s_1\ldots s_m$ where $s_0=\ast$, $s_m=t$ and $s_{i}=p(s_{i+1})$, from root to $t$. Since $g$ is increasing on $T$, we have $g(s_0)\leq g(s_1)\leq\ldots \leq g(t)\leq f(x)\leq r$ and therefore $s_i\in F_{s_{i-1}}$ for each $i=1,2\ldots m$.
\end{proof}
\end{thm}

Valuation workloads are perhaps not very interesting on their own but it should be noted that 
every workload can be decomposed as a union of valuation workloads having the same underlying domain and dataset (Subsection \ref{subsec:querypart}). If a tree structure is present, the Theorem \ref{thm:valuation} ensures that a consistent indexing scheme can be constructed. 

\section{New indexing schemes from old}\label{sec:newfromold}

Here we formulate in an abstract setting some constructions commonly used to generate new access methods from the existing ones. Our general approach makes these constructions amenable to analysis by means of theoretical computer science.

\subsection{Disjoint sums}
\label{subs:disjoint}
Any collection of access methods for workloads $W_1,W_2,\ldots,W_n$ leads to an access method for the disjoint sum  workload $\sqcup_{i=1}^n W_i$: to answer a query $Q=\sqcup_{i=1}^n Q_i$, it suffices to answer each query $Q_i$, $i=1,2,\ldots,n$, and then merge the outputs.

In particular, if each $W_i$ is equipped with an indexing scheme, ${\mathcal I}_i=(T_i,\mathscr{B}_i,{\mathcal F}_i)$, then a new indexing scheme for $\sqcup_{i=1}^n W_i$, denoted ${\mathcal I}=\sqcup_{i=1}^n {\mathcal I}_i$, is constructed as follows: the tree $T$ contains all $T_i$'s as branches beginning at the root node, while the families of bins and of decision functions for $\mathcal I$ are unions of the respective collections for all ${\mathcal I}_i$, $i=1,2,\ldots,n$.

This construction is often used coupled which an equivalence relation which partitions the domain, instance and each of the queries into smaller spaces, perhaps with a better structure which are then indexed separately (`subindexed'). A good illustration is our indexing scheme for weighted quasi-metric spaces.

\begin{example}\label{ex:FMtree}
Recall that a weighted quasi-metric (Section \ref{sec:weighted_qm}) over a domain $\Omega$ is a quasi-metric $d$ such that for some weight function $w$ and for all $x,y\in\Omega$, \[d(x,y)+w(x) = d(y,x)+w(y).\] The following Proposition shows that any weighted quasi-metric similarity workload $W=(\Omega,X,\mathcal{Q}^\text{rng}_d)$ can be indexed using the decomposition into a disjoint union of metric spaces or \emph{fibres}, one for each value that the weight function $w$ takes.

\begin{prop}
Let $(\Omega,d,w)$ be a weighted quasi-metric space and denote by $G_z$ the set $\{x\in\Omega: w(x)=z\}$, and by $\overline{\mathfrak{B}^\star}_{\e}(x)$ the closed ball of radius $\e$ centred at $x\in\Omega$ with respect to the metric $\rho$ where for each $x,y\in\Omega$, $\rho(x,y)=\frac12\left(d(x,y)+d(y,x)\right)=\frac12\qsum{d}(x,y)$. Then
\begin{enumerate}[(i)]
\item $\Omega=\bigsqcup_{z\in w(\Omega)}G_z$,
\item $\clball{x}{\e} = \bigsqcup_{z\in w(\Omega)}\clball{x}{\e}\vert {G_z}$\quad for all $x\in\Omega,\ \e>0$, and
\item $\clball{x}{\e}\vert {G_z} = \overline{\mathfrak{B}^\star}_{\e+\frac12(z-w(x))}(x)\vert {G_z}$\quad for all $x\in\Omega,\ \e>0$.
\end{enumerate}
\begin{proof}
The first two statements are obvious while the third claim follows directly from \[\rho(x,y) = \frac12\left(d(x,y)+d(y,x)\right) = d(x,y) + \frac12\left(w(x)-w(y)\right).\]
\end{proof}
\end{prop}
 
Therefore, provided that $w$ takes few values on the dataset (otherwise close fibres need to be merged), it is possible to index into $W$ by indexing data points for each fibre using one of the existing indexing schemes for metric spaces and then collecting the results. We call this scheme a \emph{FMTree} (Fibre Metric Tree). Some of our attempts to use this scheme to index into datasets of short protein fragments are described in the next chapter.
\end{example}

\subsection{Query partitions}\label{subsec:querypart}

A similar technique can be used where the set of queries over some domain is partitioned and separate indexing scheme exists for each partition.
 
Let $\Omega$ be a domain, $X\subset\Omega$ a dataset and $\mathcal{Q}_i$, $i=1,2,\ldots,n$ a pairwise disjoint family of queries over $\Omega$. A collection of access methods for the workloads $W_i=(\Omega,X,\mathcal{Q}_i)$ leads to an access method for the workload $W=(\Omega,X,\bigsqcup_{i=1}^n\mathcal{Q}_i)$: to answer a query $Q\in\bigsqcup_{i=1}^n\mathcal{Q}_i$, find $i$ such that $Q\in\mathcal{Q}_i$ and answer it using the access method for the workload $W_i$.

As in the disjoint sum case, if each $W_i$ is equipped with a consistent indexing scheme, ${\mathcal I}_i=(T_i,\mathscr{B}_i,{\mathcal F}_i)$, then a new consistent indexing scheme for $W$, denoted $\mathcal{I}$ is constructed as follows: the tree $T$ contains all $T_i$'s as branches beginning at the root node, while the families of bins and of decision functions for $\mathcal I$ contain the unions of the respective collections for all ${\mathcal I}_i$, $i=1,2,\ldots,n$. The decision function at the root for each query $Q\in\mathcal{Q}_i$ returns the set consisting of the branch $T_i$. We call such indexing scheme a \emph{query partitioning indexing scheme}.

A query partitioning indexing scheme can be considered to be highly redundant (see Subsection \ref{subsec:redundancy} for the precise definition of redundancy of indexing schemes) since each major branch contains the bins covering the whole dataset which, in many cases, may occupy considerable space. However, in some cases it may be possible for such indexing scheme to occupy the space much more efficiently. Our indexing scheme for protein fragment workloads, called FSindex, is a good example of the query partitioning approach with no redundancy  -- each data point is stored only once.

\subsection{Inductive reduction}\label{subsec:indred}
Let $W_i=(\Omega_i,X_i,{\mathcal Q}_i)$, $i=1,2$ be two workloads. An \emph{inductive reduction} of $W_1$ to $W_2$ is a pair of mappings $i\colon \Omega_2\to \Omega_1$, $i^{\twoheadleftarrow}\colon {\mathcal Q}_1\to{\mathcal Q}_2$, such that
\begin{itemize}
\item $i(X_2)\supseteq X_1$,
\item for each $Q\in{\mathcal Q}_1$,
$i^{-1}(Q)\subseteq i^{\twoheadleftarrow}(Q)$.
\end{itemize}
Notation: $W_2\stackrel{i}{\indred} W_1$.

An access method for $W_2$ leads to an access method for $W_1$, where a query $Q\in{\mathcal Q}_1$ is answered as in the Algorithm \ref{alg:indred}:

\begin{pseudocode}[ovalbox]{$W_1$.RetrieveQuery}{Q}\label{alg:indred}
\COMMENT{$W_2=(\Omega_2,X_2,{\mathcal Q}_2)\stackrel{i}{\indred} W_1=(\Omega_1,X_1,{\mathcal Q}_1)$, $Q\in {\mathcal Q}_1$}\\
R_1\GETS \emptyset\\
R_2\GETS \CALL{$W_2$.RetrieveQuery}{i^{\twoheadleftarrow}(Q)}\\
\COMMENT{$R_2 = X_2\cap i^{\twoheadleftarrow}(Q)$}\\
\FOREACH y\in R_2 \DO
\BEGIN
\IF i(y)\in Q \THEN R_1\GETS R_1\cup\{i(y)\}
\END\\
\RETURN{R_1}
\end{pseudocode}
\algocapt{Answering a query using inductive reduction of the workload.}

If ${\mathcal I}_2=(T_2,\mathscr{B}_2,{\mathcal F}_2)$ is a consistent indexing scheme for $W_2$, then a consistent indexing scheme ${\mathcal I}_1=r_\ast({\mathcal I}_1)$ for $W_1$ is constructed by taking $T_1=T_2$, $B^{(1)}_t=i(B^{(2)}_t)$, and $F^{(1)}_t(Q)= F^{(2)}_t(i^{\twoheadleftarrow}(Q))$ (the upper index $i=1,2$ refers to the two workloads). The bigger workload used for inductive reduction usually carries a structure that supports an efficient access method. 

\begin{example} Let $\Gamma$ be a finite graph of bounded degree, $k$. Associate to it a \emph{graph workload}, $W_\Gamma$, which is an inner workload with $X=V_\Gamma$, the set of vertices, and $\mathcal{Q}=\{Q^\text{$k$NN}_d(v, V_\Gamma)\ |\ v\in V_\Gamma\}$, the set of $k$NN queries where $d$ is the shortest path metric on $\Gamma$.

A \emph{linear forest} is a graph that is a disjoint union of paths. The \emph{linear arboricity}, $la(\Gamma)$, of a graph $\Gamma$ is the smallest number of linear forests whose union is $\Gamma$. This number is, in fact, fairly small: it does not exceed $\ceil{3D/5}$, where $D$ is the degree of $\Gamma$ \cite{Gu86,Al88}. The \emph{Linear Arboricity Conjecture} \cite{AEH80,AEH81}, which states that $la(\Gamma)\leq\ceil{\frac{D+1}{2}}$, was found to hold for numerous cases \cite{Al88}. Results for $k$-linear arboricity, the minimum number of forests whose connected components are paths of length at most $k$ are also available \cite{LiWo98}. This concept leads to an indexing scheme for the graph workload $W_\Gamma$, as follows. 

Let $F_i$, $i=1,\ldots,la(\Gamma)$ be linear forests. Denote $F=\sqcup_{i=1}^{la(\Gamma)}F_i$ and let $\phi\colon F\to \Gamma$ be a surjective map preserving the adjacency relation. Every linear forest can be ordered, and indexed into as in Ex. \ref{lodomain}. At the next step, index into the disjoint sum $F$  as in Subsection \ref{subs:disjoint}. Finally, index into $\Gamma$ using the inductive reduction $\phi\colon F\to \Gamma$. This indexing scheme outputs nearest neighbours of any vertex of $\Gamma$ in time $O(D\log n)$, requiring storage space $O(n)$, where $n$ is the number of vertices in $\Gamma$.
\end{example}

\subsection{Projective reduction}
\label{subsec:projred}
Let $W_i=(\Omega_i,X_i,{\mathcal Q}_i)$, $i=1,2$ be two workloads. A \emph{ projective reduction} of $W_1$ to $W_2$ is a pair of mappings $r\colon \Omega_1\to\Omega_2$, $r^{\twoheadrightarrow}\colon {\mathcal Q}_1\to{\mathcal Q}_2$, such that
\begin{itemize}
\item $r(X_1)\subseteq X_2$,
\item for each $Q\in{\mathcal Q}_1$, $r(Q)\subseteq r^{\twoheadrightarrow}(Q)$. 
\end{itemize}
Notation: $W_1 \stackrel{r}{\projred} W_2$.

An access method for $W_2$ leads to an access method for $W_1$, where a query $Q\in{\mathcal Q}_1$ is answered as follows:

\begin{pseudocode}[ovalbox]{$W_1$.RetrieveQuery}{Q}\label{alg:projred}
\COMMENT{$W_1=(\Omega_1,X_1,{\mathcal Q}_1) \stackrel{r}{\projred}W_2=(\Omega_2,X_2,{\mathcal Q}_2)$, $Q\in {\mathcal Q}_1$}\\
R_1\GETS \emptyset\\
R_2\GETS \CALL{$W_2$.RetrieveQuery}{r^{\twoheadrightarrow}(Q)}\\
\COMMENT{$R_2 = X_2\cap r^{\twoheadrightarrow}(Q)$}\\
\FOREACH y\in R_2 \DO
\BEGIN
  \FOREACH x\in r^{-1}(y) \DO
  \BEGIN
     \IF x\in Q \THEN R_1\GETS R_1\cup\{x\}    
  \END\\
\END\\
\RETURN{R_1}
\end{pseudocode}
\algocapt{Answering a query using projective reduction of the workload.}

Let ${\mathcal I}_2=(T_2,\mathscr{B}_2,{\mathcal F}_2)$ be a consistent indexing scheme for $W_2$. The projective reduction $W_1 \stackrel{r}{\projred} W_2$ canonically determines an indexing scheme ${\mathcal I}_1=r^\ast({\mathcal I}_2)$ as follows: $T_1=T_2$, $B^{(1)}_t= r^{-1}(B^{(2)}_t)$, and $f^{(1)}_t(Q)=f^{(2)}_t(r^{\twoheadrightarrow}(Q))$.

\begin{example} The linear scan of a dataset is a projective reduction to the trivial workload: $W\projred \{\ast\}$.
\end{example}

If $W=(\Omega,X,{\mathcal Q})$ is a workload and $\Omega^\prime$ is a domain, then every mapping $r\colon \Omega\to\Omega^\prime$ determines the \emph{direct image workload,} $r_\ast(W)=(\Omega^\prime,r(X),r({\mathcal Q}))$, where $r(X)$ is the image of $X$ under $r$ and $r({\mathcal Q})$ is the family of all queries $r(Q),Q\in{\mathcal Q}$.

\begin{example} \label{ex:blocks} Let $\mathcal B$ be a finite collection of \emph{blocks} partitioning $\Omega$. Define the \emph{discrete workload} $(\mathscr{B},\mathscr{B},2^\mathscr{B})$, and define the reduction by mapping each $w\in\Omega$ to the corresponding block and defining each $r^{\twoheadrightarrow}(Q)$ as the union of all blocks that meet $Q$. The corresponding reduction forms a basic building block of many indexing schemes \cite{CNBYM}.
\end{example}

\begin{example} 
\label{lipschitz}
Let $W_i$, $i=1,2$ be two metric range similarity workloads, that is, their query sets are generated by metrics $d_i$, $i=1,2$. In order for a mapping $f\colon \Omega_1\to\Omega_2$ with the property $f(X_1)\subseteq X_2$ to determine a projective reduction $f\colon W_1 \stackrel{r}{\projred} W_2$, it is necessary and sufficient that $f$ be 1-Lipschitz: indeed, in this case every ball $\ball{x}{\e}^X$ will be mapped inside of the ball $\ball{f(x)}{\e}^Y$ in $Y$.
\end{example}

\begin{example} More specifically, the following technique (described in detail in \cite{CNBYM}) is often used to map metric spaces into $\ell_\infty$ in order to use vector space indexing schemes such as the R-tree (Example \ref{ex:Rtree}). 

Let $(\Omega, d)$ be a metric space and choose $n$ 1-Lipschitz functions $f_1,f_2,\ldots f_n$. It is easy to see that the map $\omega\mapsto (f_1(\omega),f_2(\omega),\ldots,f_n(\omega))$ is a 1-Lipschitz map $\Omega\to\ell_\infty^n$ and thus induces a projective reduction to the vector space workload. The most common way of choosing the required 1-Lipschitz functions is to select $n$ \emph{pivots} $x_1,x_2,\ldots x_n$ and set $f_i(\omega)=d(x_i,\omega)$. 
\end{example}

\begin{example}
\label{ex:prefilter}
Pre-filtering is an often used instance of projective reduction. In the context of metric similarity workloads, this normally denotes a procedure whereby a metric $\rho$ is replaced with a coarser distance $d$ which is computationally cheaper. While the distance $d$ need not be a metric (in fact it need not even satisfy the triangle inequality), it is necessary and sufficient that $d(x,y)\leq\rho(x,y)$ for all $x,y\in\Omega$ for the identity map to induce a projective reduction.
The QIC-M-Tree \cite{CiPa02} provides an example of this approach.
\end{example}

\begin{example} A frequently used tool for dimensionality reduction of datasets is the famous Johnson--Lindenstrauss lemma \cite{JoLi84}. Let $\Omega=\R^N$ be an Euclidean space of high dimension, and let $X\subset\R^N$ be a dataset with $n$ points. If $\e>0$ and $p$ is a randomly chosen orthogonal projection of $\R^N$ onto a Euclidean subspace of dimension $k={O(\log n)/\e^2}$, then with overwhelming probability the mapping $\left(\sqrt{N/k}\right) p$ does not distort distances within $X$ by more than the factor of $1\pm\e$. More results of the same type, for embedding $n$-point datasets into lower dimensional linear (not necessarily Euclidean) spaces, were obtained in \cite{LiLoRa95}.

Such techniques do not extend with the same distortion to the entire domain $\Omega=\R^N$, meaning that they can be only applied to construct consistent indexing schemes for the \emph{inner workload} $(X,X,{\mathcal Q})$, and not the outer workload $(\Omega,X,{\mathcal Q})$. 
\end{example}

\section{Performance and Geometry}\label{sec:perfgeom}

In the preceding sections we were mostly concerned with the abstract foundations of indexing and similarity search and therefore have mostly ignored the issue of the performance. This is of course the key question: the rationale for indexing is exactly that it is supposed to speed up searches. Our definitions of similarity workload and indexing scheme clearly point towards a geometric setting for answering the questions about the performance. Here we attempt to examine some factors concerning the performance of indexing schemes, albeit at a purely conceptual level. This is indeed the only possible way without either a concrete dataset, or very detailed assumptions about the workload.

Our main result is yet another way of describing the \emph{Curse of Dimensionality} which is a general observation that indexing schemes for high dimensional spaces perform very badly -- often an optimised sequential scan performs better. The framework we use was first introduced in \cite{Pe00}: a metric similarity workload is identified with an mm-space where the measure reflects the distribution of query points. We use the techniques from \cite{Pe00} to derive the lower bounds on the number of blocks that must be processed in order to answer a range query of radius $\e$.

\subsection{Cost model for indexing schemes}

In estimating the performance of indexing schemes, as with other algorithms and data structures in computer science, we are primarily interested in two quantities: the \emph{space} occupied by the indexing structure and the \emph{time} required to process the query. As always there is a tradeoff between the two. For example, for an $n$-point dataset, sequential scan (Example \ref{ex:seqscan}) takes $\Omega(n)$ time with $\Omega(n)$ space (the space necessary to store all data points) while, if the workload is inner, hashing (Example \ref{ex:hashing}) takes $\Omega(1)$ time with $\Omega(\abs{\mathcal{Q}})$ space. Therefore, an investigation of performance of an indexing scheme has to take into account both the space and the query time complexity as well as the time required to build or update the structures.

The space complexity is of great importance in practice, especially with large datasets -- often we are constrained to take no more than $O(n)$ space. However, we shall concentrate mostly on the query time complexity since the space complexity can be easily estimated directly. At this stage we deliberately ignore the index creation complexity -- we always assume that an index is already constructed, that is, that all of $(T,\mathscr{B},\mathcal{F})$ are defined. 

The general goal of indexing is to produce access methods that have time complexity sublinear in the size of the dataset. Often, the authors of indexing schemes claim to achieve $O(\log n)$ time (see for example a summary of space and time complexities of existing metric indexing schemes in \cite{CNBYM}), but this claim usually only holds for `small' queries. Nevertheless, in practice, even a constant reduction of the number of data points to be scanned, say to $10\%$, if not accompanied with a too large overhead, is worthwhile pursuing.

\subsubsection{General time complexity}
In most general terms, the time required to process query $Q\in\mathcal{Q}$ using a consistent indexing scheme $\mathcal{I}=(T,\mathscr{B},\mathcal{F})$ on a workload $W=(\Omega,X,\mathcal{Q})$ is given by the 
\begin{equation}
\time(Q) = \time_T(Q) + \time_{\mathscr{B}}(Q) + \time_\mathcal{F}(Q)\\
\label{eq:timegen}
\end{equation}
where $\time(Q)$ is the total time required to process query $Q$, $\time_T(Q)$ is the time associated with traversing the nodes of $T$, $\time_\mathcal{F}(Q)$ is the total time spent evaluating decision functions at all visited inner nodes of $T$ and $\time_{\mathscr{B}}(Q)$ is the total time spent scanning the sets $B\cap X$ for each block $B\in\mathscr{B}$ associated with the leaf nodes visited.

The $\time_T(Q)$ is mostly associated with the data structures required for tree traversal. It includes the cost of retrieving the nodes from secondary memory (I/O costs) if it is used as well as the cost of any additional data structures used. For example, some algorithms for kNN similarity search \cite{HjSa03}, which are described in more detail in the context of our indexing scheme for peptide fragments in Chapter \ref{ch:4}, make use of priority queue for tree traversal. Under some circumstances, such as the large number of nearest neighbours required, both the space and the time costs of the priority queue are not negligible. On the other hand, if the whole structure is stored in primary memory and no expensive data structures are used, the $\time_T(Q)$ can be very small compared with the other two times and is often ignored \cite{CNBYM}. 

The equation \ref{eq:timegen} can be elaborated in the following way: let $S(Q)$ be the set of nodes of $T$ visited in order to retrieve a query $Q$. Denote by $I(Q)$ the set $I(T)\cap S(Q)$ and by $L(Q)$ the set $L(T)\cap S(Q)$. Then we have
\begin{equation}
\time(Q) = \time_T(Q) + \sum_{t\in L(Q)}\sum_{x\in B_t\cap X}\time(Q,x) + \sum_{t\in I(Q)}\time(Q,F_t)
\label{eq:timegen1}
\end{equation}
where $\time(Q,x)$ is the time required to check if $x\in Q$ and $\time(Q,F_t)$ is the time required to evaluate $F_t(Q)$.

Most frequently, we are not interested in the performance for a single query but in either the average or the worst case performance. However, in order to measure the average search time it is necessary to have a probability distribution on the set queries $\mathcal{Q}$. We shall return to this theme in Subsection \ref{subseq:probdist}. 

\begin{example}
In \cite{CNBYM} the general cost of a (range) query for a metric indexing scheme is measured by the number of distances evaluated. In this case the $\time(Q,x)$ is the time taken to evaluate the distance from the query centre $\omega$ to $x$ and it is assumed that each evaluation of a certification function is based on one or more distance evaluations. The I/O costs ($\time_T(Q)$) are ignored and it is assumed that other costs of the indexing structure are an order of magnitude less than costs of distance evaluations.
\end{example}

\begin{example}
A more elaborate cost model, consistent with the Equations \ref{eq:timegen} and \ref{eq:timegen1}, was proposed by Ciaccia and Patella \cite{CiPa02} in the context of the QIC-M-tree (Example \ref{ex:mtree}). Since the QIC-M-tree is a paged structure, the I/O costs are explicitly included. The $\time_{\mathscr{B}}(Q)$ depends only upon the comparison distance $d_C$ (it is exactly the time to evaluate query distances to all points retrieved from the leaf nodes) while the $\time_\mathcal{F}(Q)$ depends on the index distance $d_I$ as well as $d_C$. The authors note that the performance does not depend directly on the query distance $d_Q$ which is approximated by $d_I$ and $d_C$, give formulae for the average costs in terms of the distributions of $d_I$ and $d_C$ and develop ways to choose comparison distances so as to optimise performance.
\end{example}

\subsubsection{Redundancy and Access Overhead}\label{subsec:redundancy}

In their 1997 paper \cite{H-K-P} and its followup with additional coauthors  Miranker and Samoladas \cite{HKMPS02}, Hellerstein, Koutsoupias and Papadimitriou proposed two measures of performance of indexing schemes: \emph{redundancy} and \emph{access overhead} and showed that there is a tradeoff between the two. We present the adaptations of their concepts to our model.

\begin{defin}
Let $W=(\Omega,X,\mathcal{Q})$ be a workload and $\mathcal{I}=(T,\mathscr{B},\mathcal{F})$ an indexing scheme. The \emph{redundancy} $r(x)$ of $x\in X$ is the number of blocks that contain $x$, that is, \[r(x) = \abs{\{B\in\mathscr{B}:x\in B\}}.\]

The \emph{average redundancy} $r(\mathcal{I})$, of the indexing scheme $\mathcal{I}$, is the average of $r(x)$ over all data points: \[r(\mathcal{I}) = \frac{1}{\abs{X}}\sum_{x\in X}r(x).\]
\end{defin}

\begin{defin}\label{defin:access_overhead}
Let $W=(\Omega,X,\mathcal{Q})$ be a workload and $\mathcal{I}=(T,\mathscr{B},\mathcal{F})$ an indexing scheme. For a query $Q\in\mathcal{Q}$ denote, as before, by $L(Q)$ the set of leaf nodes visited to answer $Q$. The \emph{access overhead} $A(Q)$ of query $Q$ is defined as \[A(Q) = \frac{\sum_{t\in L(Q)}\abs{B_t\cap X}}{\max\{\abs{Q\cap X},1\}}.\] The (worst case) access overhead $A(\mathcal{I})$ for indexing scheme $\mathcal{I}$ is \[A(\mathcal{I}) = \sup\{A(Q)\ |\ Q\in\mathcal{Q}\}.\]
If furthermore all blocks $B_t\in\mathscr{B}$ contain $m$ data points, we define the \emph{block access overhead} $A_\mathscr{B}(Q)$ of query $Q$ by \[A_\mathscr{B}(Q)=\frac{\abs{L(Q)}}{\max\{\ceil{\abs{Q\cap X}/m},1\}},\] and of indexing scheme $\mathcal{I}$ by 
$A_\mathscr{B}(\mathcal{I}) = \sup\{A_\mathscr{B}(Q)\ |\ Q\in\mathcal{Q}\}$.

If $\mu$ is a probability measure on $Q$, we define the \emph{average access overhead} $\bar{A}(\mathcal{I})$ for the indexing scheme $\mathcal{I}$ by \[\bar{A}(\mathcal{I}) = \int_\mathcal{Q} A(Q)d\mu, \] and the \emph{average block access overhead} $\bar{A}_\mathscr{B}(\mathcal{I})$ by \[\bar{A}_\mathscr{B}(\mathcal{I}) = \int_\mathcal{Q} A_\mathscr{B}(Q)d\mu. \]
\end{defin}

The access overhead $A(Q)$ measures the cost of answering the query $Q$ using the set of blocks $\mathscr{B}$ (that is, the $\time_\mathscr{B}$ -- the costs associated with $T$ and $\mathcal{F}$ are ignored) normalised by the ideal cost and hence takes values in $[1,\infty)$. The block access overhead measures the same cost in terms of block accesses and corresponds to the original definition of access overhead in \cite{H-K-P}. Our new definition was chosen in order not to depend on block size which in some indexing schemes may vary considerably and to allow for empty queries which do take time to process.

The main result of \cite{HKMPS02} is the Redundancy Theorem which in a workload independent way gives a lower bound for the redundancy in terms of the block size and access overhead. 

\begin{thm}[\cite{HKMPS02}]\label{thm:redundancy}
Let $W=(\Omega,X,\mathcal{Q})$ be a workload and $\mathcal{I}=(T,\mathscr{B},\mathcal{F})$ an indexing scheme such that all blocks contain $m$ datapoints and $A_\mathscr{B}(\mathcal{I})\leq \sqrt{m}/4$. Let $Q_1,Q_2\ldots,Q_M$ be queries such that for every $i=1,2,\ldots,M$:
\begin{enumerate}[(i)]
\item $\abs{Q_i\cap X}\geq m/2$, and
\item $\abs{Q_i\cap Q_j\cap X}\leq m/16A_\mathscr{B}^2$, \quad for all $j=1,2,\ldots,M$ and $j\neq i$.
\end{enumerate}
Then, the average redundancy is bounded by $\displaystyle r(\mathcal{I})\geq \frac{1}{12\abs{X}}\sum_{i=1}^{M}\abs{Q_i\cap X}$.
\end{thm}

In most applications, due to space constraints, the redundancy of each datapoint $x$ is set to $1$, that is, there is only one block containing $x$. The Theorem \ref{thm:redundancy} then gives the lower bound for the block access overhead provided the queries do not pairwise intersect to a too great extent. If a better block access overhead is desired while block size stays the same, it is necessary to increase the (average) redundancy.

\subsection{Workloads and pq-spaces}\label{subseq:probdist}

In order to estimate the average performance it is necessary to have a probability distribution on the set of queries which is often not available in any useful form. This is true in particular for similarity workloads with range queries which depend both on the query centre $\omega\in\Omega$ and the radius $\e$. Subsequently, we shall assume that the radius is fixed and attempt to analyse the performance of indexing schemes with only $\omega$ as a parameter. 

Indeed, there are good reasons to consider performances of indexing schemes for different search radii separately. We show in Subsection \ref{subsec:dimcurse} that there are significant qualitative differences between performances at different scales. Furthermore, this approach corresponds with many real-life situations where the radius has a direct, problem-specific interpretation and is chosen in advance. One example is biological sequence search performed by BLAST \cite{altschul97gapped} -- in almost all practical cases the users do not change the default threshold which corresponds to the expected number of sequences to be retrieved according to a null model. The threshold is translated into a cutoff similarity score and thus into a quasi-metric radius (depending on the query centre only). 

Therefore, we shall assume that the domain $\Omega$ is equipped with a (Borel) probability measure $\mu$ reflecting the distribution of query centres. If the dissimilarity measure $d$ is a metric (respectively quasi-metric), it follows that the triple $(\Omega,d,\mu)$ is a pm- (respectively pq-) space. The measure $\mu$ can always be approximated from the dataset itself: for any $A\subseteq\Omega$ set $\displaystyle \mu(A) = \frac{\abs{A\cap X}}{\abs{X}}$. This would imply that the distribution of the query centres coincides with the distribution of the dataset and is the approach taken in \cite{CiPa02}.

A complementary way of looking at the measure $\mu$ on $\Omega$ is to treat it as a sort of an `ideal' measure and the dataset as an $n$-point sample according to $\mu$. One can consider a family of datasets from $\Omega$ distributed according to $\mu$ and attempt to construct an indexing scheme which would answer queries of all datasets efficiently. This was one of the reasons we defined the queries as subsets of $\Omega$ rather than $X$. 

One can go even further by having two measures on $\Omega$ -- one giving the dataset distribution as above and another, possibly very different, providing the distribution of the query centres. It has long been observed in the context of relational databases \cite{Christodoulakis84} that that it is necessary to consider non-uniform distributions of queries in order to well estimate the query performance and there is no reason to suppose that the same does not hold for similarity-based queries. However, the introduction of a second measure would present non-trivial technical challenges and we therefore leave it for subsequent work.

\subsection{The Curse of Dimensionality}\label{subsec:dimcurse}

It has long been known (c.f. for example \cite{BeWeYa80}) that exponential complexity might be inherent in any algorithm for answering near neighbour queries because a point in a high-dimensional space can have many `close' neighbours. In fact, this phenomenon is not only associated with similarity searches but with other data analysis related areas such as machine learning using neural networks \cite{Bishop95}, clustering \cite{HinKeim99}, function or density estimation \cite{Friedman97}, signal processing \cite{WaKc97} and many others. In all cases the procedures that perform well on two or three dimensional sets fail to do in higher dimensions. We take the paradigm of Pestov \cite{Pe00} that the curse of dimensionality is primarily a manifestation of the concentration phenomenon. It allows us to use the techniques developed in Chapter \ref{ch:2} to provide estimates of performance of indexing schemes with as few assumptions as possible regarding the nature of the dataset. We first outline the previous results for the nearest neighbour queries and then proceed to our contribution for range queries in quasi-metric workloads. 

\subsubsection{Nearest Neighbour Queries}

In their 1999 paper, Beyer et al. \cite{BeyerGRS99} investigated the effect of dimensionality to the nearest neighbour problem. Their main result states that under certain conditions every nearest neighbour query (in a metric space) is \emph{unstable}: the distance from any point to its nearest neighbour is very close to the distances to most other points. We outline here the contribution of Pestov \cite{Pe00} who both relaxed the assumptions of Beyer et al. and  obtained stronger conclusions using the techniques of the asymptotic geometric analysis, that is, the concentration phenomenon.

\begin{defin}[\cite{BeyerGRS99}]
Let $(\Omega,X,\mathcal{Q}_d^\text{NN})$ be a workload where $(\Omega,d)$ is a metric space and $\mathcal{Q}_d^\text{NN}$ is the set of nearest neighbour queries. A query $Q(\omega,X)\in \mathcal{Q}_d^\text{NN}$ is called \emph{$\e$-unstable} for an $\e>0$ if \[\abs{\left\lbrace x\in X:d(\omega,x)\leq(1+\e)d_X(\omega)\right\rbrace}>\frac{\abs{X}}{2}.\] 
\end{defin}

\begin{defin}
Let $(\Omega,d,\mu)$ be an pm-space and $X\subseteq\Omega$ a finite subset. For an $x\in X$ denote by $R_x=\sup\{r>0:\mu(\ball{x}{r})\leq\frac12\}$ the maximal radius of an open ball in $\Omega$ centred at $x$ of measure not more that $\frac12$. For a $\delta>0$ we say that $X$ is \emph{weakly $\delta$-homogeneous} in $\Omega$ if all radii $R_x,\ x\in X$ belong to an interval of length less than $\delta$.
\end{defin}

\begin{thm}[\cite{Pe00}]
Let $(\Omega,d,\mu)$ be an pm-space and $X\subseteq\Omega$ a finite subset. Denote by $M$ a median value of $d_X$, the distance from a point in $\Omega$ to its nearest neighbour in $X$. Let $0<\e<1$ and assume that $X$ is weakly $(M\e/6)$-homogeneous in $\Omega$.

Then for all points $\omega\in\Omega$, apart from a set of total measure at most $3\alpha(M\e/6)$, the open ball of radius $(1+\e)d_X(\omega)$ centred at $\omega$ contains at least \[\min\left\{\abs{X},\ceil{\frac{1}{2\sqrt{\alpha(M\e/6)}}}\right\}\] elements of $x$. 
\end{thm}

Hence, provided that $X$ is weakly $(M\e/6)$-homogeneous in $\Omega$ (which it is, as remarked in \cite{Pe00}, with probability not less than $1-2\abs{X}\alpha(M\e/12)$ if $X$ is sampled randomly with regard to $\mu$) and that $(\Omega,d,\mu)$ has concentration property, with very high probability every nearest neighbour query is $\e$-unstable. 

The point of all this is that in the case of query instability there is little information to be gained by the nearest neighbour search -- the quality of results is such that they can not be well interpreted. Hinnenburg et al. \cite{HinneburgAK00} proposed a solution to a generalised nearest neighbour problem by dimensionality reduction and weighting of the dimensions according to the query point. This amounts to a redefinition of a metric to be used. In all cases, it is not hard to see that the performance of any indexing scheme is poor if almost the whole dataset is to be retrieved.

\subsubsection{Range Queries}

Turning to range queries in quasi-metric spaces we adopt the paradigm outlined in Subsection \ref{subseq:probdist}. The radius is fixed while the query centres are distributed according to a measure $\mu$ on $\Omega$. We are interested in the number of blocks that need to be processed in order to answer the query $\clball{\omega}{\e}$ which would give us an estimate on the $\time_{\mathscr{B}}$ and the access overhead. Since metric and quasi-metric trees are built hierarchically so that at each level and at each node we have a set covering a portion of the dataset, the same result can be used to give an estimate for the $\time_{\mathcal{F}}$. 

\begin{lemma}\label{lemma:deltainc}
Let $(X,d)$ be a quasi-metric space, $A\subseteq X$ and $0<\delta<\e$. Then $\rnbhd{\left(\rnbhd{A}{\delta}\right)}{\delta'}\subseteq\rnbhd{A}{\e}$, where $\delta'=\e-\delta$.
\begin{proof}
Suppose $x\in\rnbhd{\left(\rnbhd{A}{\delta}\right)}{\delta'}$. Then there exists $y\in\rnbhd{A}{\delta}$ such that $d(y,x) < \e$. By the Lemma \ref{lemma:set_tr_eq}, $d(x,A)\leq d(x,y)+d(y,A)< \delta'+\delta=\e$.
\end{proof}
\end{lemma}

\begin{lemma}\label{lemma:alphaineq}
Let $(X,d,\mu)$ be a pq-space, $A$ a Borel subset of $X$, $\e>0$ and $\mu(A)>\alpha^L(\e)$. Then $\mu(\rnbhd{A}{\e})>\frac12$. 
\begin{proof}
Suppose that $\mu(A)>\alpha^L(\e)$ and $\mu(\rnbhd{A}{\e})\leq\frac12$. Let $B=X\setminus\rnbhd{A}{\e}$. Then $\mu(B)>\frac12$ and therefore $\mu(A)\leq\mu(X\setminus\lnbhd{B}{\e}) = 1-\mu(\lnbhd{B}{\e})\leq\alpha^L(\e)$, leading to a contradiction.
\end{proof}
\end{lemma}

The following is proved using a similar technique to the Lemma 4.2 of \cite{Pe00}. In addition to the worst case result similar to the one provided in \cite{Pe00}, we also give a bound for the average case performance which is arguably more important than the worst case.

\begin{thm}\label{thm:rngconc}
Let $(\Omega,d,\mu)$ be a pq-space, $\e>0$ and $\mathscr{B}$ a collection of subsets $B\subseteq\Omega$ such that $\mu\left( \bigcup\mathscr{B}\right)=1$ and for all $B\in\mathscr{B}$, $\mu(B)\leq \xi\leq\frac14$. Denote by $\delta=(\alpha^L)^{\leftarrow}(\xi) = \inf\{\e>0:\alpha^L(\e)\leq\xi\}$  the generalised inverse of $\alpha^L$ at $\xi$. Then, for any $\e>\delta$, 
\begin{enumerate}
\item There exists $\omega\in\Omega$ such that $\lball{\omega}{\e}$ meets  at least \[\min\left\{\ceil{\frac{1}{\xi}}, \ceil{\frac{1}{\alpha^R\left(\e-\delta\right)}-1}\right\}\]
elements of $\mathscr{B}$.
\item A left ball $\lball{\omega}{\e}$ around $\omega\in\Omega$ meets on average (in $\omega$ ) at least \[\min\left\{\ceil{\frac{1}{\xi}}, \ceil{\frac{1}{4\alpha^R\left(\e-\delta\right)}}\right\}\] elements of $\mathscr{B}$.
\end{enumerate}
\begin{proof}
By assumption on each $B\in\mathscr{B}$ and by the choice of $\delta$, $\mu(B)\leq \xi\leq\alpha^L(\delta)$. Decompose $\mathscr{B}$ into a collection of pairwise disjoint subfamilies $\mathscr{B}_i$, $i\in I$ in a such way that $\alpha^L(\delta)<\mu(A_i)\leq 2\alpha^L(\delta)$ for each $A_i=\bigcup\mathscr{B}_i$. Clearly, \[\frac{1}{2\alpha^L(\delta)}\leq\abs{I}<\frac{1}{\alpha^L(\delta)}\leq\frac1\xi.\] Let $\delta'=\e-\delta>0$. Then, by the Lemmas \ref{lemma:deltainc} and \ref{lemma:alphaineq},  \[\mu\left(\rnbhd{\left(A_i\right)}{\e}\right)\geq \mu\left(\rnbhd{\left(\rnbhd{\left(A_i\right)}{\delta}\right)}{\delta'}\right)\geq 1-\alpha^R(\delta'),\] and hence the probability that a random left ball of radius $\e$ does not intersect $A_i$ is less than \mbox{$\alpha^R(\e-\delta)$}. For any $J\subseteq I$, \[\mu\left( \bigcap_{i\in J}\rnbhd{\left(A_i\right)}{\e}\right)\geq 1-\abs{J}\alpha^R(\e-\delta).\] The first claim follows by choosing $J$ such that $\abs{J}=\min\left\{\abs{I}, \ceil{\frac{1}{\alpha^R(\e-\delta)}-1}\right\}= \min\left\{\ceil{\frac{1}{\xi}}, \ceil{\frac{1}{\alpha^R(\e-\delta)}-1}\right\}$ so that $\mu\left( \bigcap_{i\in J}\rnbhd{\left(A_i\right)}{\e}\right)>0$. To prove the second statement observe that the probability that a random ball of radius $\e$ meets at least $\ceil{\frac{1}{2\alpha^R\left(\e-\delta\right)}}$ elements is at least $\frac12$. Hence, the average number of subsets of $\mathscr{B}$ intersecting a ball of radius $\e$ is at least $\ceil{\frac{1}{4\alpha^R\left(\e-\delta\right)}}$.
\end{proof} 
\end{thm}

Our result directly leads to the following Corollary stated in terms of a range similarity workload (with fixed radius). Note that the open balls are replaced by the closed balls in order to be consistent with the definition of the range similarity workload.

\begin{corol}\label{cor:rngconc}
Let $\e>(\alpha^L)^{\leftarrow}(\xi)$ and $W=(\Omega,X,\mathcal{Q})$ be a workload where $\mathcal{Q}=\{\clball{\omega}{\e}\ |\ \omega\in\Omega\}$ (the left closed balls are taken with respect to a quasi-metric $d$ on $\Omega$). Suppose the dataset $X$ and the query centres are distributed according to the Borel probability measure $\mu$ on $\Omega$. Let $\mathscr{B}$ be a finite set of blocks such that $\mu(\bigcup\mathscr{B})=1$ and for any $B\in\mathscr{B}$, $\mu(B)\leq\xi\leq\frac14$. Then the number of blocks accessed to retrieve the query $\clball{\omega}{\e}$ is on average at least $\ceil{\frac{1}{4\alpha^R\left(\e-(\alpha^L)^{\leftarrow}(\xi)\right)}}$ and in the worst case at least $\ceil{\frac{1}{\alpha^R\left(\e-(\alpha^L)^{\leftarrow}(\xi)\right)}-1}$ or $\ceil{\frac{1}{\xi}}$, whichever is smaller. \qed
\end{corol}

As observed in Chapter \ref{ch:2}, for many metric spaces we have $\alpha(\e)\leq C_0e^{-C_1\e^2N}$ where $N$ is the dimension of the space. In this case it is easy to see that any indexing scheme, unless its blocks have all very small measure, will need to scan very many blocks in order to retrieve not only the worst case but also a typical range query. Even if the access overhead is not large, the sequential scan of the whole dataset might outperform an indexing scheme due to the overhead associated with the tree structure. The bounds from the Theorem \ref{thm:rngconc} while certainly not tight, give some indication on the number of blocks that can be expected to be retrieved.

Note that the Theorem \ref{thm:rngconc} holds only for $\e>\delta$ -- the value $\delta$ is the scale at which we observe such phenomenon. Obviously, at the scales smaller than $\delta$ the indexing scheme need not suffer in performance. Observe that both $\alpha^L$ and $\alpha^R$ are involved but their role is not the same. The left concentration function determines the scale at which the concentration effect take place while the $\alpha^R$ establishes the number of bins accessed. For `bad' performance it is necessary that the $\alpha^R$ decreases sharply near $0$. 

Since our metric and quasi-metric indexing schemes, as defined in Sections \ref{sec:metrictrees} and \ref{sec:qmetrictrees} involve covering sets at each level of the tree, it is straightforward to apply the Theorem \ref{thm:rngconc} to derive the bounds for the number of certification function evaluations at each level.

\subsection{Dimensionality estimation}

Unlike our approach above, which uses only geometric assumptions and where the performance is linked to the concentration functions, Pagel, Korn and Faloutsos \cite{PKF00} seek to estimate the performance of nearest neighbour query retrieval based on fractal (Hausdorff or correlation) dimensions of the dataset. This line of investigation stems from the observation that for real datasets embedded in vector spaces, features are often correlated and hence the estimates based on independence assumptions are too pessimistic. Hence the effort to find the `real' dimensionality of the datasets. 

Traina, Traina and Faloutsos \cite{TrainaTF99} introduced the \emph{distance exponent} which gives the intrinsic dimension of any metric space by assuming that (at least for small $\e$), the size of a ball $\ball{x}{\e}$ grows proportionally to $\e^N$ where $N$ is the dimension of the space. They claimed that performance of metric trees could be well approximated in terms of the distance exponent. As a part of his summer research assistantship at the Australian National University in summer 1999/2000, the thesis author performed some experiments to determine the ways of estimating the distance exponent from the datasets. These previously unpublished results are presented in the Appendix \ref{app:distexp}.

In \cite{CNBYM} another definition of the intrinsic dimensionality is given (again in terms of the distance distribution) and bounds on the number of distances to be evaluated by metric indexing schemes are derived.

\section{Discussion and Open problems}

So far we have provided a conceptual framework for similarity search and hinted that the Curse of Dimensionality is related to the concentration phenomenon. The Theorem \ref{thm:rngconc} extends the previous results to the case of range searches in quasi-metric spaces. We next outline possible directions for further investigation.

\subsection{Workload reductions}

Our definition of an indexing scheme (Definition \ref{defn:indscheme}) emphasises the three structures which are found in all examples known to us: the set of blocks that cover the dataset, the tree structure supporting an access method and the decision functions. While this setting allows us to directly identify the factors that influence the performance, access methods for similarity queries could be investigated through workload reductions as in Section \ref{sec:newfromold}, without the explicit reference to indexing schemes.

Consider a \emph{tree workload}, $W_T=(T,T,\mathcal{Q})$ where $T$ is a finite rooted directed weighted tree, such that every edge is assigned a zero weight in the direction towards the root and a positive weight in the opposite direction. The $\mathcal{Q}$ is the set of range similarity queries induced by the path quasi-metric (Section \ref{sec:wght_dur_graph}). There is an obvious access method associated with such workload: traverse the tree starting from the query point and retrieve all nodes closer than the cutoff value.

Observe that any metric or quasi-metric indexing scheme where the blocks are pairwise disjoint can be represented as a projective reduction of the original workload $W_0$ to a discrete workload mapping each point to its block, followed by an inductive reduction to a tree workload. In our notation, 
\[W_0\stackrel{r}{\projred}(\mathscr{B},\mathscr{B},2^\mathscr{B}) \stackrel{i}{\indred}W_T.\]
The requirement that the blocks are pairwise disjoint comes from $r$ being a function -- this is a limitation that may need to be overcome. 

While this approach is perhaps too abstract and limited at this stage, hiding the decision functions in the reduction maps, it opens new lines of investigation. In particular, one can ask if all access methods involve reductions to inner workloads and attempt to construct access methods involving inductive reductions to non-tree workloads.

Another topic for investigation would be to construct a hierarchy of all workloads (with measures on the sets of queries) according to their \emph{indexability}, a term introduced in \cite{H-K-P}. For example, a workload would be higher in the hierarchy if it is more difficult to index and one could decide indexability of any particular workload in reference to some canonical workloads. It is clear that the trivial workload should be on the top of the hierarchy as the most difficult to index.

For mm-spaces, one can hope to be able to use Gromov's relation $\succ$ between mm-spaces (\cite{Gr99}, Chapter $3{\frac12}$, pp. 133--140): for two mm-spaces $X$ and $Y$, $X$ \emph{(Lipschitz) dominates} $Y$, denoted $X\succ Y$, if there exists a 1-Lipschitz map $X\to Y$ pushing forward the measure $\mu_X$ to a measure $\nu$ on $Y$ proportional to $\mu_Y$. Obviously, a one point space $\{*\}$ (with any measure) is a minimal mm-space and the more concentrated a space is, the more it is dominated by other mm-spaces. This notion should be able to be generalised to quasi-metric spaces with measure. Going even further, one would wish to include the dataset in any resulting theory.

\subsection{Certification functions}

As we noted before, the bounds from the Corollary \ref{cor:rngconc} are not tight -- they usually indicate better than actual performance. Indeed, much closer estimates can be obtained if the distributions of the values of the certification functions are known, such as in \cite{CiPa02} where they correspond to the distance distributions. Ciaccia and Patella also emphasise that their model attests that the performance depends only on the distributions of the index and comparison distances (i.e. the certification functions) and not on the query distance. This is not contrary to our results -- our bounds are for a best possible indexing scheme and the performance in practice could be much worse.

Hence, there are reasons to believe that the main reason for the Curse of Dimensionality is not the inherent high-dimensionality of datasets, but a poor choice of certification functions. Efficient indexing schemes require usage of \emph{dissipating functions,} that is, 1-Lipschitz functions whose spread of values is more broad, and which are still computationally cheap. Such functions correspond to `tighter' covering sets with little overlap between them. This interplay between complexity and dissipation is, we believe, at the
very heart of the nature of dimensionality curse, at least in relation to the $\time_\mathcal{F}$. Requirements for blocks to contain certain number of points have a large contribution as well. 

Generic metric indexing schemes use only distances (from points) to construct their certification functions. While this ensures that they can be applied to any metric space, it may also be significant limitation if the distances are computationally expensive. More specific knowledge of the geometry of the domain is clearly necessary to produce computationally cheaper certification functions. The QIC-M-tree \cite{CiPa02} is a great step in this direction as it allows the user to specify three distances to be used. It should be possible to go even further by developing a structure which allows the user to specify classes of certification functions and an algorithm which fits them to a dataset and produces an indexing scheme. The insight gained by the approaches attempting to reduce overlap between the covering sets associated with the nodes of a metric tree, such as Slim-trees \cite{TrainaTSF00}, will no doubt play a role. 

%
%
%
%
%

\section{Conclusion}

Our proposed approach to indexing schemes used in similarity search allows for a unifying look at them and facilitates the task of transferring the existing expertise to more general similarity measures than metrics. In particular, we have extended the concepts associated to metric workloads to the quasi-metric workloads.

We hope that our concepts and constructions will meld with methods of geometry of high dimensions and lead to further insights on performance of indexing schemes.  While we have not yet reached the stage where asymptotic geometric analysis can give accurate predictions of performance as there exists no algorithm for estimating concentration functions from a dataset, at least it leads to some conceptual understanding of their behaviour. We have deliberately ignored non-consistent indexing schemes in our discourse -- while they may show much better performance, they do so at a price of losing some members of the query.

In the next Chapter we shall further illustrate our concepts on the concrete dataset of peptide fragments and point out some specific issues affecting performance of indexing schemes.


\chapter{Indexing Protein Fragment Datasets}\label{ch:4}

While the previous chapters emphasised the theory, laying the foundations and introducing the concepts, the present chapter and the one following focus on applications to actual protein sequence datasets. The present chapter has two principal aims: to illustrate the notions of Chapter \ref{ch:3} on the sets of biological sequences and to introduce an indexing scheme for datasets of short peptide fragments to be used for biological investigations of Chapter \ref{ch:5}. 

An additional reason for studying indexing schemes for short peptide fragments is that it has been frequently pointed in the literature \cite{BuKi01,NaBY00,Hu04,HuAtIr01,KaSi01,Buhler01,NaBYSuTa01,GiWaWaVo00} that algorithms for indexing short fragments could be used as subroutines of BLAST-like programs for searches of full sequences. It is hoped that as a part of the future work, the experience gained from indexing short fragment could be applied to the challenge of indexing datasets of full DNA and protein sequences.

\section{Protein Sequence Workloads}\label{sec:sequence_workloads}

Let $\Sigma$ denote the standard 20 amino acid alphabet. A \emph{full sequence workload} has the domain $\Sigma^*$ and the sets of queries consisting of range or kNN queries based on the quasi-metric corresponding to the local (Smith-Waterman) similarity scores based on BLOSUM matrices and affine gap penalties. The dataset in this case is any actual set of protein sequences.

A \emph{short fragment workload} has the domain $\Sigma^m$, the set of all amino acid sequences of length $m$ which will mostly range from 6 to 12. The set of queries consists of range or kNN queries based on an $\ell_1$-type quasi-metric extending a quasi-metric $d_{\Sigma}$ on $\Sigma$ (Section \ref{sec:genhamming}). The co-weightable quasi-metric $d_{\Sigma}$ is derived from a similarity score matrix $s$ from the BLOSUM family using the formula $d_{\Sigma}(x,y)=s(x,x)-s(x,y)$ while the dataset is obtained from a full sequence dataset by taking all fragments of length $m$ from all sequences.

Depending on the protein sequence dataset, there may exist cases where two short fragments have the same sequence (Subsection \ref{subseq:unique}). For the purpose of this thesis, a kNN query is defined with respect to the original fragment dataset (which is therefore a pseudo-quasi-metric space), not to the quotient set where points with identical sequence are merged into one point.

Most of the present chapter, as well as Chapter \ref{ch:5}, examines short fragment workloads with some ideas transferable to full sequence workloads. The remainder of the present section investigates some geometric aspects of sets of short peptide fragments. 

\subsection{Sequence datasets}\label{subsec:protdatasets}

Two protein sequence datasets were used for investigations of the present chapter: NCBI nr (non-redundant) \cite{WheelerNCBI04} and SwissProt \cite{Boeckmann2003}. 

The NCBI nr dataset is a comprehensive general protein sequence database, including entries from most other major protein sequence databases (such as \mbox{SwissProt}) as well as the translated coding sequences from GenBank entries (GenPept). Where multiple identical sequences exist, they are consolidated into one entry. The nr dataset is the main dataset searched by NCBI BLAST and the latest version can be downloaded from \url{ftp://ftp.ncbi.nlm.nih.gov/blast/db/} where other datasets searched by NCBI BLAST can be found as well. Since the full nr dataset is very large (the version from June 2004 contains 1,866,121 sequences consisting of 619,474,291 amino acids) smaller samples rather than the full dataset were used. It should be noted that many protein sequences belonging to GenPept and hence nr were translated from coding segments of GenBank sequences that were verified solely using computational techniques, that is, without experimental validation. Thus, nr may contain sequences which are not expressed in any organism. 

The SwissProt dataset, maintained at the Swiss Institute of Bioinformatics \url{http://www.expasy.org/sprot/}, is ``a curated protein sequence database which strives to provide a high level of annotation (such as the description of the function of a protein, its domains structure, post-translational modifications, variants, etc.), a minimal level of redundancy and high level of integration with other databases". Its entries contain, apart from the sequence information, extensive functional annotation, literature citations and links to other resources. Because of its moderate size, non-redundancy and high level of sequence characterisations, SwissProt (Release 43.2 of April 2004, containing 144,731 sequences consisting of 53,363,726 amino acid residues) was used as the main dataset for the experiments of this chapter. 

\subsection{Unique fragments}\label{subseq:unique}

SwissProt and nr are (almost -- there are few duplicate sequences in SwissProt) non-redundant. However, when short fragments are taken to form the fragment database, it often occurs that multiple instances of the same fragment exist (Figure \ref{fig:uniquefrags}). In other words, the underlying measure on $\Sigma^m$ where $m$ is small is not the counting measure.

\begin{figure}[!ht] 
\begin{center}
\scalebox{0.7}{\includegraphics{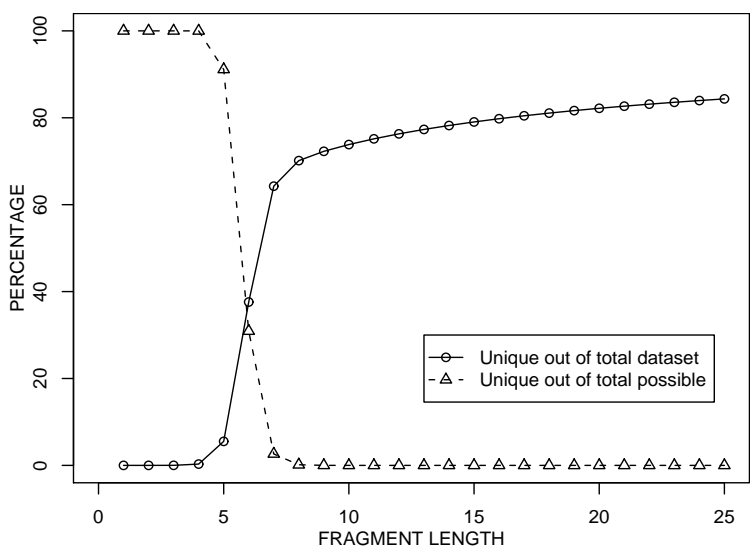}}
\caption[Percentage of unique SwissProt fragments of various lengths.]{Percentages of unique fragments of fixed length from the SwissProt dataset out of total fragments in the dataset and total possible fragments ($\abs{\Sigma}^m$). The fragments containing letters not belonging to the standard amino acid alphabet were ignored.}
\label{fig:uniquefrags}
\end{center}
\end{figure}

For similarity searches, this situation can be handled in two ways. If many duplicate fragments are present (very short fragment lengths), a preprocessing step is necessary to collect the identical fragments together, introducing some space overhead but significantly saving search time. If relatively few duplicates (longer fragment lengths) are present, they can be treated as separate points introducing an additional time cost for unnecessary distance evaluations but avoiding space overhead for collecting identical fragments. 

A further observation that can be made from the Figure \ref{fig:uniquefrags} is that for very short fragments, almost every possible sequence is represented in the dataset -- the workload is effectively inner, allowing the possibility of using combinatorial algorithms for indexing. This is definitely not true for longer fragments and full sequences where the workload is outer. For example, the number of potential fragments of length 10 is $20^{10}$ while there are only about 38.5 million (or 0.0004\%)) unique fragments in SwissProt.

\subsection{Random sequences}\label{subsec:randseq}

Most experiments of this chapter, investigating geometry of datasets and performance of indexing schemes, involve simulating a probability measure on the set of all possible protein fragments using generated random sequences. It is necessary to do so because the workloads (with the exception of sets of fragments of very short lengths) are outer and it is quite likely that a query sequence would be (slightly) different from all sequences existing in a dataset. Generally, the `true' distribution of protein sequences or fragments is unknown and the measure obtained by counting the points of an actual dataset is not appropriate because the full natural variation of protein sequences cannot be captured by any dataset, that is, one always expects to discover novel sequences. Hence, it is necessary to use theoretical models of sequence distributions and attempt to balance the practical issues, such as the ability to quickly generate sufficiently many random sequences, with accuracy.

The simplest way of generating random fragments of fixed length is to assume the underlying measure is the product measure based on background (overall) amino acid frequencies, that is, to generate each fragment by an independent, identically distributed process where the probability measure is given by the background frequencies. Such approach can be extended to sequences of arbitrary length by modelling sequence length according to some distribution (for example, discretised log-normal \cite{PHHC00}) and once the length is chosen, proceeding as above.

A more general model, actually used to generate testing datasets for the experiments of the current chapter, is based on \emph{Dirichlet mixtures} \cite{SKBHKMH96}. As in the previous case, the length of each sequence is taken from a discretised log-normal distribution and the amino acids of a sequence are generated by an independent, identically distributed process. However, the probabilities for that distribution are selected from a mixture of Dirichlet densities (for a description of Dirichlet distributions and mixtures see Chapter 11 of the Durbin {\it et.al.} book \cite{Durbin:1998}) instead from a single (background) distribution. 

The code and the data for generating random sequences according to Dirichlet mixtures were obtained from \url{http://www.cse.ucsc.edu/research/compbio/dirichlets/}. To obtain samples of fragments of fixed length to be used in experiments, for each desired length, 5000 non-overlapping fragments were sampled from full sequences generated according to the above method. The same testing datasets were used for all experiments ensuring that performances of different indexing schemes can be directly compared.

\subsection{Quasi-metric or metric?}

Chapter \ref{ch:bioseq_qm} has shown that most common distances on protein sequences are quasi-metrics. However, since the theory and practice of indexability of metric spaces is much better studied, it is worthwhile to investigate the overhead of replacing a quasi-metric by a metric.

\begin{figure}[hb!]
\begin{center}
\scalebox{0.8}{\includegraphics{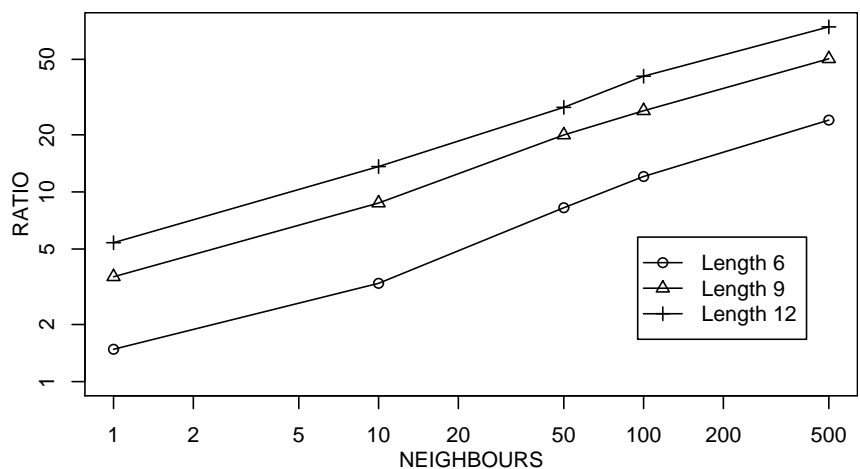}}
\caption[Ratios between sizes of metric and quasi-metric balls.]{Mean ratio between the sizes of smallest metric and quasi-metric balls containing $k$ nearest neighbours with respect to the BLOSUM62 quasi-metric. Each point is based on 5,000 searches of SwissProt fragment datasets using randomly generated fragments as ball centres.}\label{fig:qmballratio}
\end{center}
\end{figure}

From the point of view of performance, the best measure of the average overhead is the ratio between the sizes of the metric and the quasi-metric ball containing at least $k$ nearest neighbours with respect to the quasi-metric. If this ratio is close to 1, the metric and the quasi-metric have similar geometry and the replacement of the quasi-metric by a metric is feasible. The average sampled ratios for the fragment datasets of lengths 6, 9 and 12, using the associated metric (the smallest metric majorising the quasi-metric), are shown in the Figure \ref{fig:qmballratio}.

It is clear that replacement of quasi-metric by a metric would be very costly except for the nearest neighbour searches of very short fragments (length 6) and that it is indeed necessary to develop the theory and algorithms that would allow the use of the intrinsic quasi-metric. This observation was one of the principal motivations behind the development of the theory of quasi-metric trees in Chapter \ref{ch:3}.

\subsection{Neighbourhood of dataset}

A further way of assessing the way a dataset is embedded into its domain is by considering how far the closest point from the dataset is to any point in the domain, or alternatively, the smallest $\e$ such that the dataset forms an $\e$-net inside the domain. Even more information is revealed by the distribution of distances of points in the domain to the dataset; for example, it can be determined if there is a sizable amount of points significantly farther from the dataset than the rest. Note that such distribution function clearly depends on the underlying measure on the domain (query distribution).

While an overwhelming amount of computation would be necessary to obtain the exact distribution, it is possible to approximate it by resorting to simulation, that is, by generating points according to the assumed measure and finding for each generated point the distance to its nearest neighbour in the dataset. If an efficient indexing scheme is available, such approach is computationally inexpensive. Figure \ref{fig:nbhd} shows the results for SwissProt fragment datasets of lengths 6, 9 and 12 using the sample points generated according to Dirichlet mixtures (Subsection \ref{subsec:randseq}).

\begin{figure}[ht!]
\begin{center}
\scalebox{0.8}{\includegraphics{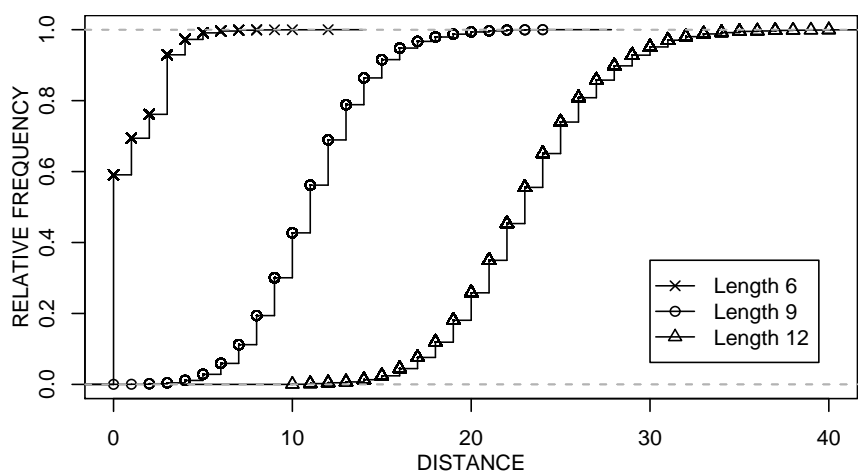}}
\caption[Distributions of distances from random fragments to the \mbox{SwissProt} fragment datasets.]{Distributions of BLOSUM62 distances from random fragments to the \mbox{SwissProt} fragment datasets. Based on 5000 random fragments generated according to Dirichlet mixtures.}\label{fig:nbhd}
\end{center}
\end{figure}

The estimated distribution for the fragments of length 6 supports the observations from Subsection \ref{subseq:unique} that the workloads based on sets of fragments of very short length are close to inner: almost 60$\%$ of random points are in the dataset (the BLOSUM62 quasi-metric (Figure \ref{fig:blosum62qd}) and hence its derived $\ell_1$ type distance on fragments is $T_1$ and therefore the distance of $0$ implies identical fragments) and most of the remainder are within one amino acid substitution from a dataset point (Figure \ref{fig:blosum62qd} shows the full BLOSUM62 quasi-metric). In fact, the number of random points belonging to the dataset is much greater than the proportion of the dataset in the domain from the Figure \ref{fig:uniquefrags} (about 30$\%$), which is essentially based on the counting measure on the domain. This (not surprisingly) indicates that the measure based on Dirichlet mixtures indeed approximates the dataset better than the counting measure. The distributions for the lengths 9 and 12 indicate that a neighbour is very likely to be found in the biologically significant ranges (20--35).

\subsection{Distance Exponent}\label{subsec:distexp}

Distance exponent (Appendix \ref{app:distexp}), measuring the rate of growth of balls in a metric space can be used to estimate the dimensionality and hence the complexity of workloads. The theory presently applies only to metric spaces (although the rationale is equally valid for quasi-metric spaces) and therefore the associated metric to the BLOSUM62 quasi-metric was used. Since the estimate of the dimensionality of the full domain, rather than just of the dataset was desired, the average size (in terms of points of the dataset) of a ball of given radius centred at a random point was computed and used to estimate the distance exponent. This approach is justified by the Remark \ref{rem:avgball}, provided the measure induced by the dataset is a good approximation to the measure used to generate the ball centres (i.e. the measure on the domain). The sizes of the balls of small radii for datasets of length 6 and 9 are shown in Figure \ref{fig:ball_size} (log-log scale).

\begin{figure}[ht!]
\begin{center}
\scalebox{0.8}{\includegraphics{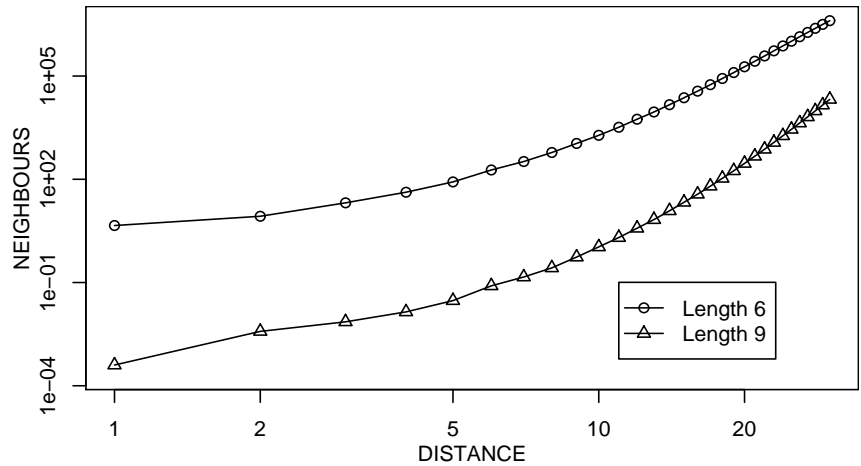}}
\caption[Growth of metric balls in SwissProt fragment datasets.]{Growth of balls centred at 5000 random fragments generated according to Dirichlet mixtures. The balls are taken with respect to the metric associated to the BLOSUM62 quasi-metric.}\label{fig:ball_size}
\end{center}
\end{figure}

It is apparent that the log-log graphs are not linear and therefore the method based on fitting a polynomial (Subsection \ref{subseq:polyfitting}) was used for distance exponent estimation. The estimated distance exponent is 7.6 for the fragments of length 6 and 10.6 for the fragments of length 9. Hence, in this context, the datasets are approximately equivalent to the cubes $[0,1]^{8}$ and $[0,1]^{11}$ respectively, with the $\ell_\infty$ metric (Subsection \ref{subseq:distexp_cube}). An interesting problem is to determine if `good' embeddings into cubes $\lambda[0,1]^n$ exist and if so, to index them as vector spaces, say using X-tree.

\subsection{Self-similarities}

As mentioned previously, in Chapter \ref{ch:bioseq_qm} as well as in the current chapter, protein sequence fragments with (some) BLOSUM similarity measures can be treated as co-weighted quasi-metric spaces with the co-weight of each point given by its self-similarity. Self-similarities are significant because they are the sole source of asymmetry of the quasi-metric: we have $\Gamma(x,y)=\abs{d(x,y)-d(y,x)}=\abs{s(x,x)-s(y,y)}$ where $\Gamma$ denotes the asymmetry function introduced in Section \ref{sec:qpclosemm}. Therefore, the distribution of self-similarties determines the `distance' of the quasi-metric space from its associated metric space. Furthermore, if self-similarities of dataset points take very few values, as is the case with short fragment datasets, the co-weighted quasi-metric space can be divided into metric fibres which can be indexed separately using an indexing scheme for metric workloads (FMtree -- Example \ref{ex:FMtree}). Figure \ref{fig:selfsim} shows the estimates of distributions of self-similarities of SwissProt fragment datasets of length 7 and 12 based on approximately 1,000,000 samples.

\begin{figure}[ht!]
\begin{center}
\begin{tabular}[t]{lr}
\multicolumn{1}{l}{\mbox{\bf (a)}} &
        \multicolumn{1}{l}{\mbox{\bf (b)}} \\ [0.01cm]
\scalebox{1.0}[1.0]{\includegraphics{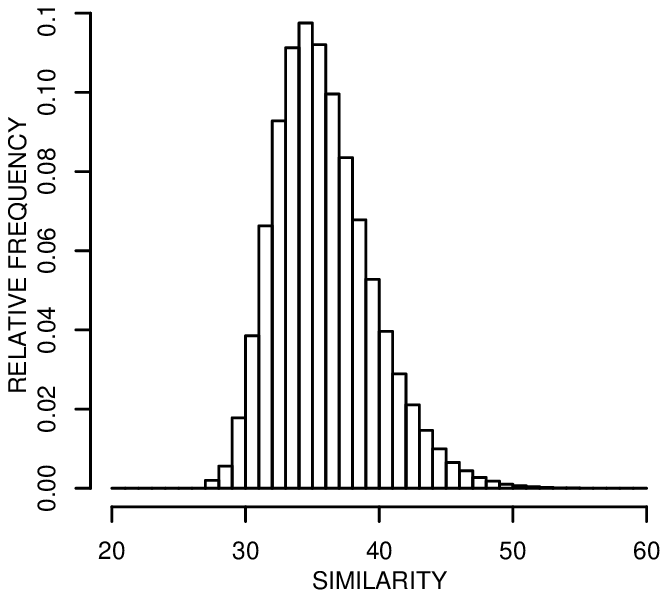}} &
\scalebox{1.0}[1.0]{\includegraphics{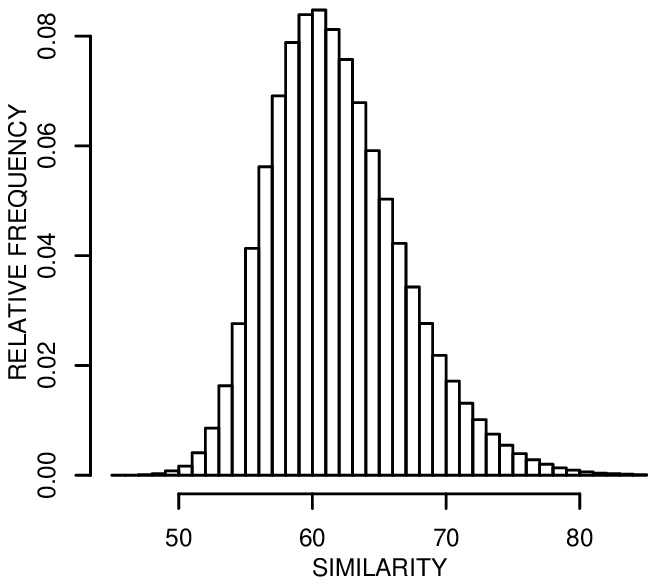}}\\ [-0.5cm]
\end{tabular}
\caption[Distributions of self-similarities of SwissProt fragment datasets.]{Distributions of self-similarities of SwissProt fragment datasets: {\bf(a)} Length 7; {\bf (b)} Length 12.}\label{fig:selfsim}
\end{center}
\end{figure}

It can be seen that both distributions are skewed to the right and that the distribution for the length 12 is more spread out, that is, less concentrated. However, if something is to be inferred about the measure concentration and hence indexability from self-similarities, it is necessary to take into account the scale. The median distance to the nearest neighbour for the length 12 workload is about 23 (Figure \ref{fig:nbhd}) while it clearly cannot be greater than 10 in length 7 case (the data for length 7 is not available in the Figure \ref{fig:nbhd} but it can be inferred from the data for lengths 6 and 9). Thus, if scaled in this way, the distribution for the length 7 would be indeed less concentrated.

\section{Tries, Suffix Trees and Suffix Arrays}

Trie, suffix tree and suffix array data structures form the basis of many of the established string search methods and provide an inspiration for some features of the FSIndex access method described in Section \ref{sec:FSIndex}.

\begin{figure}[b!]
\begin{center}
\scalebox{1.0}{\includegraphics{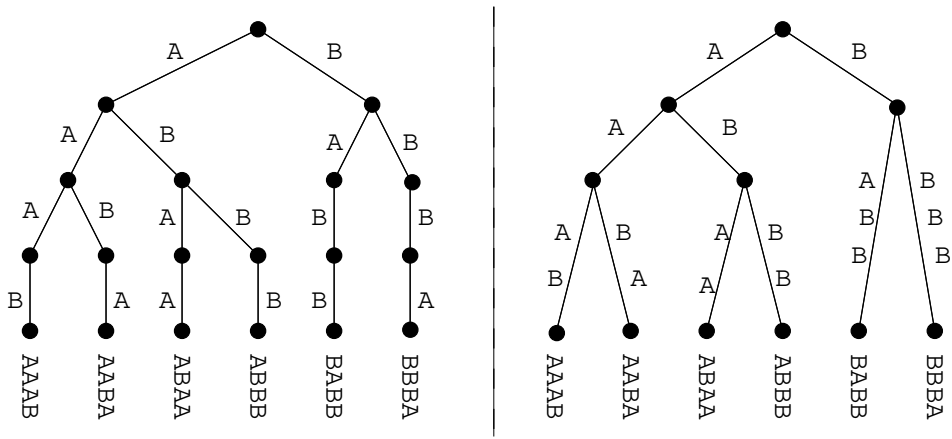}}
\caption[A trie and a PATRICIA tree.]{A trie (left) and a PATRICIA tree (right) for a set of six strings of length 4.}\label{fig:trie}
\end{center}
\end{figure}

Let $\Sigma$ be a finite alphabet and $X$ be a collection of $\Sigma$-strings (i.e. $X\subseteq\Sigma^*$). A \emph{trie} \cite{Fredkin60} is an ordered tree structure for storing strings having one node for every common prefix of two strings. The strings are stored in extra leaf nodes (Figure \ref{fig:trie}). A {PATRICIA tree} (Practical Algorithm to Retrieve Information Coded in Alphanumeric \cite{Morrison68}) is a compact representation of a trie where all nodes with one child are merged with their parent. Tries and PATRICIA trees can be easily used for string searches, that is, to find if a string $p$ belongs to $X$. Such searches take $O(n)$ time where $n=\abs{p}$.

Now consider a single (long) string $t\in X$ where $m=\abs{t}$. The \emph{suffix tree} \cite{Weiner73} for $t$ is the PATRICIA tree of the suffixes of $t$ and can be constructed in $O(m)$ time \cite{Weiner73,McCreight76,Ukkonen92}. Suffix trees, in their original form as well as generalised to suffixes of more than one string, can be used to solve a great variety of problems involving matching substrings of long strings (Gusfield, in his book \cite{Gusfield97} dedicates full five chapters exclusively to suffix trees and their applications).

\begin{figure}[hb!]
\begin{center}
\scalebox{1.0}{\includegraphics{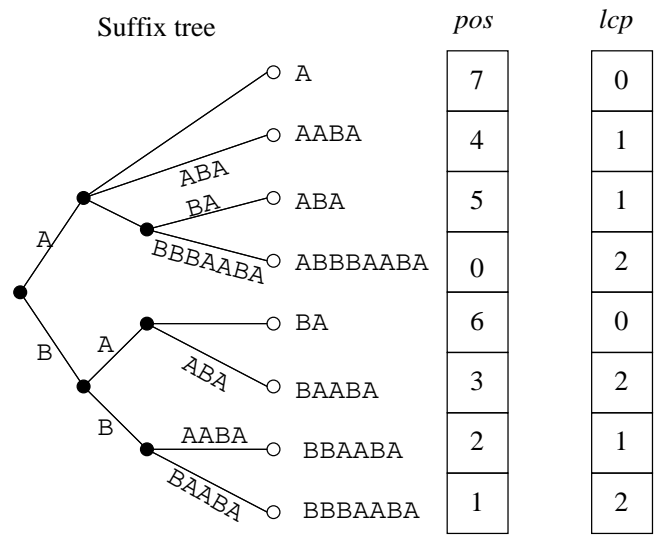}}
\caption[A suffix tree and a suffix array.]{A suffix tree and a suffix array for the word \texttt{ABBBAABA}.}\label{fig:suffixtree}
\end{center}
\end{figure}

One disadvantage of suffix trees is that they often occupy too much space -- up to $\Theta(m\abs{\Sigma})$ in many common cases \cite{Gusfield97}. The \emph{suffix array} data structure, first proposed by Manber and Myers \cite{Manber93}, is a compact representation of the suffix tree for $t$ consisting of the array $pos$, of integers in the range $0\ldots m-1$ specifying the lexicographic ordering of suffixes of $t$ (i.e. $pos[i]$ is the starting position of the $i$-th suffix of $t$ in lexicographic order), and the array $lcp$, where $lcp[i]$ contains the longest common prefix of the substrings starting at positions $pos[i-1]$ and $pos[i]$ (the first element of $lcp$ is $0$). Efficient $O(m)$ construction algorithms exist and using binary search on array $pos$ and the $lcp$ values, it is possible to search for occurrence of a string $p$ in $t$ in $O(n + \log m)$ time, where $n=\abs{p}$ \cite{Gusfield97}. Figure \ref{fig:suffixtree} shows an example of a suffix tree and a suffix array.

PATRICIA trees (and hence suffix trees and arrays), being compact representations of a set of strings, can be used to speed-up string comparisons and searches \cite{Gonnet:1992}. Indeed it is very easy to construct a quasi-metric tree for the short fragment similarity workload $(\Sigma^m, X, \mathcal{Q})$ (Section \ref{sec:sequence_workloads}) with a quasi-metric $d_\Sigma$. The tree is given by a trie or a PATRICIA tree for $X$ and each block is a set containing a single fragment associated with a leaf node. At each non-root node, a certification function calculates the distance between a prefix given by the path from the root to the node in question and a prefix of the query fragment of the same length, say $k$. In effect, a certification function calculates the distance from the query to the `cylindrical set' of fragments where the letters at first $k$ positions are fixed while varying arbitrarily at the remaining $m-k$ positions. 


\section{FSIndex}\label{sec:FSIndex}

\emph{FSIndex} is an access method for short peptide fragment workloads mainly based on two procedures: combinatorial generation and amino acid alphabet reduction.

For very short fragments (lengths 2-4), the number of all possible fragment instances is very small (for length 3, $20^3 = 8000$) and almost every fragment instance generated exists in the dataset. Hence, it is possible to enumerate all neighbours of a given point in a very efficient and straightforward manner using digital trees or even hashing. For larger lengths, the number of fragments in a dataset is generally much smaller than the number of all possible fragments (Figure \ref{fig:uniquefrags}) and generation of neighbours is not feasible. If it were to be attempted, most of the computation would be spent generating fragments that do not exist in the dataset. Hence the idea of mapping peptide fragment datasets to smaller, densely and, as much as possible, uniformly packed spaces where the neighbours of a query point can be efficiently generated using a combinatorial algorithm.

Partitions of amino acid alphabet provide the means to achieve the above. Amino acids can be classified by chemical structure and function into groups such as hydrophobic, polar, acidic, basic and aromatic (Table \ref{tbl:amino_acids}). Such classification appears in every undergraduate text in biochemistry and has been previously used in sequence pattern matching \cite{Smith90}. In general, substitutions between the members of the same group are more likely to be observed in closely related proteins than substitutions between amino acids of markedly different properties. The widely used similarity score matrices such as PAM \cite{Dayhoff:1978} or BLOSUM \cite{Henikoff:1992} are derived from target frequencies of substitutions and therefore capture these relationships more precisely.

The required mapping is constructed as following. Given a set of fragments of fixed fragment length $\Sigma^m$, an alphabet partition $\pi_i:\Sigma\to\Sigma_i$ is chosen for each position $i=0,1\ldots m-1$, where $\abs{\Sigma_i}<\abs{\Sigma}$. This induces the mapping $\pi:\Sigma^m\to\Sigma_0\times \Sigma_1\times\ldots \Sigma_{m-1}$ where $\pi(a_0a_1\ldots a_{m-1}) = \pi_0(a_0)\pi_1(a_1)\ldots \pi_{m-1}(a_{m-1})$. The members of $\Sigma_0\times \Sigma_1\times\ldots \Sigma_{m-1}$ are called \emph{bins} and the number of bins is denoted by $N$. The partitions $\pi_i$ are often equal for each $i$. An important consequence of such mapping is that distances to bins are easy to compute and can be used as certification functions. 

\begin{remark}
Positions in each fragment are zero based, that is, numbered from $0$ rather than from $1$, because the reference implementation of FSIndex is in the C programming language \cite{Kernighan88a} where arrays are indexed from $0$.
\end{remark}

\subsection{Data structure and construction}

The FSIndex data structure consists of three arrays: $frag$, $bin$ and $lcp$. The array $frag$ contains pointers to each fragment in the dataset and is sorted by bin. The array $bin$, of size $N+2$ is indexed by the rank of each bin and contains the offset of the start of each bin in $frag$ (the $N+1$-th entry gives the total number of fragments while the last entry is used solely for index creation). The bin ranking function $r:\Sigma_0\times \Sigma_1\times\ldots \Sigma_{m-1}\to\{0,1\ldots, K-1\}$ is defined as follows. For each $i=0,1,\ldots {m-1}$ let $r_i:\Sigma_i\to \{0,1,\ldots, \abs{\Sigma_i}-1\}$ be a ranking function of $\Sigma_i$ and define $\xi_i:\Sigma_i\to\N$ by
\begin{equation}\label{eq:FSrank1}
\xi_i(\sigma) = r_i(\sigma) \prod_{j=i}^{m-1}\abs{\Sigma_j}.
\end{equation}
In the case $i=m-1$ the empty product above is taken to be equal to $1$. Then,
\begin{equation}\label{eq:FSrank2}
r(x)= \sum_{i=0}^{m-1} \xi_i(x_i).
\end{equation}

In addition, each bin is sorted in lexicographic order and the value of $lcp[i]$ provides the length of the longest common prefix between $frag[i]$ and $frag[i-1]$. The value of $lcp[0]$ is set to $0$. Figure \ref{fig:FSIndexst} depicts an example of the full structure of an FSIndex.

\begin{figure}[!ht] 
\begin{center}
\scalebox{1.0}{\input{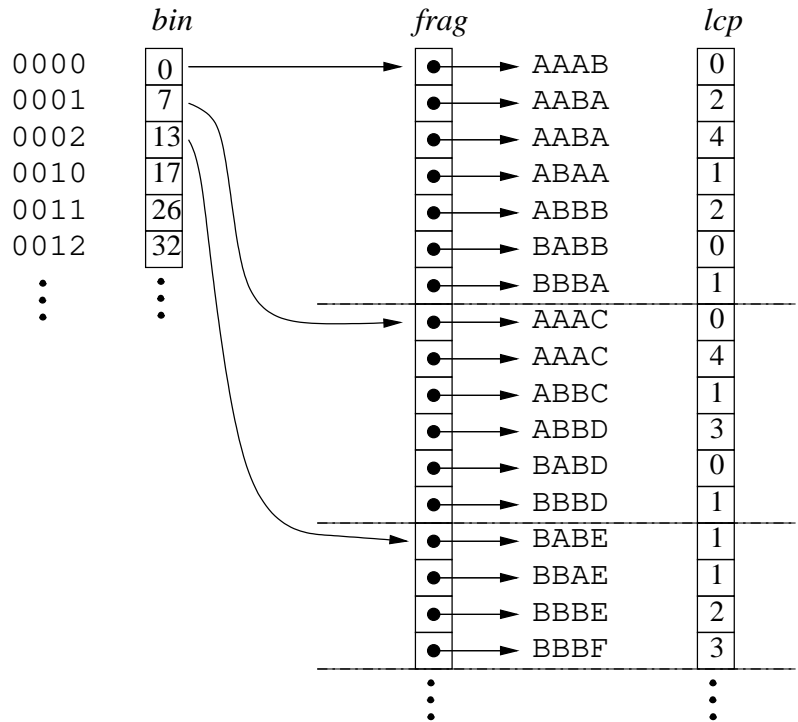}}
\caption[Structure of an FSIndex.]{Structure of an FSIndex of a dataset of fragments of length 4 from the alphabet $\Sigma=\{{\tt A,B,C,D,E,F}\}$. The same alphabet reduction is used at each position, mapping $\{{\tt A,B}\}$ to ${\tt 0}$, $\{{\tt C,D}\}$ to ${\tt 1}$ and $\{{\tt E,F}\}$ to ${\tt 2}$.}
\label{fig:FSIndexst}
\end{center}
\end{figure}

\begin{remark}
The arrays $frag$ and $lcp$ are inspired by suffix arrays but the order of offsets in $frag$ is different because $frag$ is first sorted by bin and then each bin is sorted in lexicographic order. Sorting $frag$ within each bin and constructing and storing the $lcp$ array is not strictly necessary and incurs a significant space and construction time penalty. The benefit is improved search performance for large bins, compensating for unbounded bin sizes. In effect, each bin is subindexed using a compact version of a PATRICIA tree.
\end{remark}

To construct the FSIndex data structure, any sorting algorithm can be used to produce the $frag$ array from which the $bin$ and $lcp$ arrays can be easily computed. Algorithm \ref{alg:FSconstruct} outlines the reference implementation. 

The space requirement of FSIndex is $\Theta(n+N)$. The exact space and time complexity of the
construction algorithm depends on the sorting algorithm used for sorting the $frag$ array. If the quicksort \cite{Hoare:1962} algorithm is used (the reference implementation), the space requirement is $\Theta(n+N)$ and the running time is $O(n + N + n\log n)$ on average and $O(n+N+n^2)$ in the worst case. Using radix sort \cite{Seward:1954}, the average and worst case running time can both be reduced to $O(n+N)$ with $O(n)$ (or $O(\log n)$) additional space overhead. Another alternative is to use heapsort \cite{Williams:1964} to sort the $frag$ array with the time complexity $O(n\log n + N)$ but no additional space overhead. 

\subsection{Search}

Search using FSIndex is based on traversal of implicit trees whose nodes are associated with reduced fragments (bins).

\begin{defin}
Let $u=u_0u_1\ldots u_{m-1} \in\Sigma_0\times \Sigma_1\times\ldots\times \Sigma_{m-1}$. For any $k=0,1,\ldots, m-1$ and $\sigma\in\Sigma_k$, denote by $u(k,\sigma)$ the sequence $u_0\ldots u_{k-1}\sigma u_{k+1} \ldots u_{m-1}$.

Let $i=0,1,\ldots,m-1$. Denote by $T_{u,i}$ the tree having the root $u$ connected to the subtrees $T_{u(k,\sigma),k+1}$ for all $k=i,i+1,\ldots, m-1$ and $\sigma\in\Sigma_k\setminus\{u_k\}$ and by $T_{u}$ the tree $T_{u,0}$.
\end{defin}

The trees $T_{u,i}$ are connected and unbalanced and can be shown to have depth $m-i$ while the root has the degree $\sum_{k=i}^{m-1}\abs{\Sigma_k}-1$. The tree topology is clearly independent of the choice of $u$. If $\abs{\Sigma_0}= \abs{\Sigma_1}=\ldots =\abs{\Sigma_{m-1}}=K$, $T_{u}$ is isomorphic to the \emph{multinomial tree} of order $(m,K)$. If $K=2$, such tree is called the \emph{binomial tree} of order $m$. An example is shown in the Figure \ref{fig:FStree}. 

The following Proposition is easily established.

\begin{prop}\label{prop:binmap}
Let $\Sigma_i$, $i=0,1,\ldots ,m-1$ be finite sets and $u\in \Sigma_0\times \Sigma_1\times\ldots\times \Sigma_{m-1}$. Then there exists a bijection between the nodes of $T_u$ and the set $\Sigma_0\times \Sigma_1\times\ldots\times \Sigma_{m-1}$. \qed
\end{prop}

\begin{figure}[!ht] 
\begin{center}
\scalebox{1.0}{\input{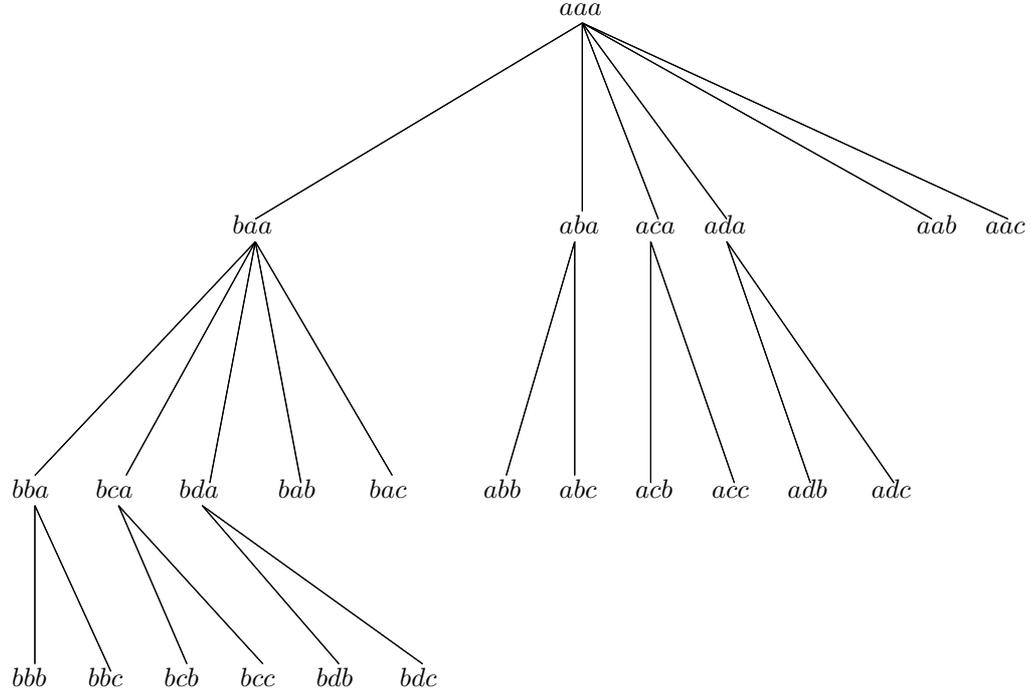}}
\caption[An example of $T_{\omega}$ (FSIndex implicit search tree).]{An example of $T_{\omega}$ where $\omega=aaa\in\Sigma_0\times\Sigma_1 \times\Sigma_2$, $\Sigma_0=\{a,b\}$, $\Sigma_1=\{a,b,c,d\}$, $\Sigma_2=\{a,b,c\}$.}
\label{fig:FStree}
\end{center}
\end{figure}

Retrieval of a quasi-metric range query $\cball{\omega}{\e}$ using the implicit tree structure is conceptually straightforward. Given a query point $\omega$ and the radius $\e$, map $\omega$ to its bin $\pi(\omega)$ and traverse the tree $T_{\pi(\omega)}$ from the root. At each node $u$, calculate the distance $d(\omega, u)$ and prune the subtree rooted at $u$ if $d(\omega, u)>\e$. For every visited node which is not pruned, calculate the distance to each fragment in the associated bin and collect all the fragments whose distance from $\omega$ is not greater than $\e$.

The indexing scheme providing the access method described above can be described as a query partitioning indexing scheme (Subsection \ref{subsec:querypart}) where the workload $(\Sigma^m,X,\mathcal{Q}^\text{rng}_d)$ is partitioned into a union of valuation workloads $(\Sigma^m,X,\mathcal{Q}^\text{rng}_{d_\omega})$ for each $\omega\in\Omega$, where $d_\omega(x)=d(\omega,x)$. Each valuation workload is associated with the valuation indexing scheme $\mathcal{I}_\omega$, defined as follows. The set of blocks is $\Sigma_0\times \Sigma_1\times\ldots\times \Sigma_{m-1}$ and the tree $T$ consists of the tree $T_{\pi(\omega)}$ where a leaf node corresponding to the same reduced sequence is attached to each node. The function $g:T\to\R$ increasing on $T$ is given by \footnote{This is a slight abuse of notation because the tree $T$ now has two distinct copies of each bin: one as an inner node and one as a leaf node attached to the inner node. The context should be clear nevertheless.} \[g(t)=d(\omega,t)= \min_{y\in t} d(\omega,y).\] 
 
It is clear that $\mathcal{I}_\omega$ is indeed a valuation indexing scheme. The proposition \ref{prop:binmap} ensures that the number of leaf nodes is $N$ while $g$ is increasing on $T$ because each child node is obtained by replacing one letter from the parent with another, different letter, an operation which increases the distance. Therefore, by the Theorem \ref{thm:valuation}, $\mathcal{I}_\omega$ is a consistent indexing scheme and it follows that the query partitioning indexing scheme over $(\Sigma^m,X,\mathcal{Q}^\text{rng}_d)$ is also consistent.

Unlike most published metric indexing schemes mentioned in Chapter \ref{ch:3}, FSIndex does not have a balanced tree. Therefore, the expected average and worst-case search time complexity is $O(n+K)$ -- the overhead is proportional to $K$, the number of inner nodes. So, based on these considerations, FSIndex is not scalable for queries of a fixed radius. However, the performance can be to a large extent controlled by the choice of alphabet partitions and hence some scalability can be achieved by using more partitions for larger datasets in order to reduce the scanning time while incurring some additional overhead. 

\subsection{Implementation}

Descriptions of FSIndex algorithms in this section are based on the reference implementation developed in the C programming language \cite{Kernighan88a} (some optimisations are omitted for clarity). Table \ref{tbl:FSvars} shows the descriptions of all global variables and functions used.

\begin{table}[!ht]
\begin{center}
\begin{tabular}{|l|l|}
\hline
 $X$ &  Fragment dataset \\
 $n$ &  Size of $X$ -- usually not known exactly beforehand\\
 $m$ &  Fragment length\\
 $\Sigma_j$ & Reduced alphabet at $j$-th position\\ 
 $\pi_j$ & Projection at $j$-th position\\
 $\xi_j$ & Integer value of a letter of reduced alphabet at $j$-th position\\
 $\pi$ & Projection function -- maps each fragment into its bin\\ 
 $N$ &  Total number of bins -- $N=\prod_{i=0}^{m-1}\abs{\Sigma_i}$\\ 
 $r$ &  Bin ranking function -- index into $bin$ array\\
 $u$ &  Index of a bin -- $u=r(x)$ where $x$ is a bin\\
 $\omega$ & Query fragment\\
 $d$ & Distance function\\
 $\e$ & Search radius\\
 $k$ & Number of nearest neighbours to retrieve\\
 $CD$ & Cumulative distance array of length $m+1$ used for processing each bin\\
 $HL$ & List of search results (hits)\\
 $PQ$ & Priority queue for kNN search\\
 \hline
\end{tabular}
\end{center}
\caption{Variables and functions of FSIndex creation and search algorithms.}\label{tbl:FSvars}
\end{table}

\subsubsection{Construction}

The construction algorithm (Algorithm \ref{alg:FSconstruct}) is closely related to counting sort \cite{Seward:1954}. It makes three passes over data fragments: to count the number of fragments in each bin, to insert the fragments into the $frag$ array and to compute the $lcp$ array. It allocates the memory for the arrays after counting.

The fragment dataset is in practice always obtained from a full sequence dataset by iterating over all subfragments of length $m$ from each sequence and it is often necessary to verify each fragment and reject those that contain non-standard letters such as `X', `B' or `Z' that do not represent actual amino acids and violate the triangle inequality for the score matrices. Therefore, the true number of data points is not known before the first pass through the dataset. 

\begin{algo}[p!]
\begin{pseudocode}[ovalbox]{CreateFSIndex}{X, m, N, \pi, r}\label{alg:FSconstruct}
bin\GETS \CALL{AllocateMemory}{N+2}\\
bin[0]\GETS0, bin[1]\GETS0\\
\COMMENT{Count bin sizes}\\
n\GETS 0\\
\FOREACH s\in X \DO
\BEGIN
  i\GETS r(\pi(s))\\
  bin[i+2]\GETS bin[i+2]+1\\
  n\GETS n+1\\
\END\\
\FOR i\GETS 2 \TO N+2 \DO
  bin[i]\GETS bin[i]+bin[i-1]\\
\COMMENT{Insert fragments into bins}\\
frag\GETS \CALL{AllocateMemory}{n}\\
\FOREACH s\in X \DO
\BEGIN
  i\GETS r(\pi(s))\\
  frag[bin[i+1]]\GETS s\\
  bin[i+1]\GETS bin[i+1]+1\\
\END\\
\COMMENT{Calculate longest common prefixes}\\
\FOR i\GETS 0 \TO N \DO \CALL{QuickSort}{frag[bin[i]:bin[i+1]]}\\
lcp\GETS \CALL{AllocateMemory}{n}\\
lcp[0]\GETS 0\\
\FOR j\GETS 1 \TO n-1 \DO
\BEGIN
  k\GETS 0, s\GETS frag[j-1], t\GETS frag[j]\\
  \WHILE s_k = t_k \DO k\GETS k+1\\
  lcp[j]\GETS k
\END\\
\RETURN{bin, frag, lcp}
\end{pseudocode}
\algocapt{FSIndex construction algorithm.}
\end{algo}

\subsubsection{Search}

Range search (Algorithm \ref{alg:FSrngsrch}) makes a recursive, depth-first traversal of the implicit tree implemented in the function \textsc{CheckNode} (Algorithm \ref{alg:FSchkbin}). The function \textsc{ProcessBin} (Algorithm \ref{alg:FSprocessbin}) scans each bin associated with an inner node not pruned using the $lcp$ array in order to reduce the number of computations necessary to calculate distances to each member of the bin.\footnote{Conceptually, Algorithm \ref{alg:FSprocessbin} is equivalent to depth-first traversal of a compact form of a PATRICIA tree for the set of fragments in the bin.} The function \textsc{InsertHit} (omitted in the case of range search) inserts the neighbour into the list of search results.

The search algorithm computes and stores the values of $d(\omega_k,\sigma)$,\\ $\min\big\{d(\omega_k,\sigma)\ |\ \sigma\in\Sigma_k \setminus\{\pi_k(\omega_k)\}\big\}$ and $\xi_k(\pi_k(\omega_k)) + \xi_k(\sigma)$ for all $k$ and all $\sigma$ before tree traversal so that the \textsc{CheckNode} function uses a table lookup.

\begin{algo}[ht!]
\begin{pseudocode}[ovalbox]{RangeSearch}{\omega, d, \e}\label{alg:FSrngsrch}
\COMMENT{Recursive tree traversal}\\
\GLOBAL{bin,frag,lcp,\xi_k,\pi,r,HL,CD}\\
\text{Initialise list of hits $HL$}\\
\text{Initialise cumulative distances $CD$, $CD[0]\GETS 0$}\\
u\GETS r(\pi(\omega))\\
\CALL{ProcessBin}{u}\\
\CALL{CheckNode}{u,0,0}\\
\RETURN{HL}
\end{pseudocode}
\algocapt{FSIndex range search algorithm.}
\end{algo}

\begin{algo}[ht!]
\begin{pseudocode}[ovalbox]{CheckNode}{u,D,i}\label{alg:FSchkbin}
\COMMENT{Recursive tree traversal}\\
\GLOBAL{d,\e,\xi_j,\pi_j}\\
\FOR j \GETS m-1 \DOWNTO i \DO
\BEGIN
  \IF D + \min\big\{d(\omega_j,\sigma)\ |\ \sigma\in\Sigma_j \setminus\{\pi_j(\omega_j)\}\big\} \leq\e \THEN
  \BEGIN
    \FOREACH \sigma\in\Sigma_j \setminus\{\pi_j(\omega_j)\} \DO
    \BEGIN
      E \GETS D + d(\omega_j,\sigma)\\
      \IF E\leq\e \THEN
      \BEGIN
        v \GETS u - \xi_k(\pi_j(\omega_j)) + \xi_j(\sigma)\\
        \CALL{ProcessBin}{v}\\
        \CALL{CheckNode}{v,E,j+1}\\
      \END
    \END   
  \END 
\END
\end{pseudocode}
\algocapt{FSIndex search tree traversal algorithm.}
\end{algo}

\begin{algo}[htb!]
\begin{pseudocode}[ovalbox]{ProcessBin}{u}\label{alg:FSprocessbin}
\COMMENT{Sequentially scan all entries.}\\
\GLOBAL{d,\e,HL,bin,frag,lcp,CD}\\
n\GETS bin[u+1]-bin[u]\\
\IF n>0 \THEN \RETURN{} \\
\FOR i\GETS 0 \TO n-1 \DO
\BEGIN
  s\GETS frag[u+i]\\
  \FOR j\GETS lcp[u+i] \TO lcp[u+i+1]-1 \DO
    CD[j+1] \GETS CD[j] + d(\omega_j,s_j)\\
  \IF CD[lcp[u+i+1]]\leq\e \THEN
  \BEGIN
    \FOR j\GETS lcp[u+i+1] \TO m-1 \DO CD[j+1] \GETS CD[j] + d(\omega_j,s_j)\\
    \IF CD[m]\leq\e \THEN \CALL{InsertHit}{HL,s,CD[m]} \\
  \END
\END
\end{pseudocode}
\algocapt{FSIndex bin processing algorithm.}
\end{algo}

The kNN search algorithms use \emph{branch-and-bound} \cite{CPZ97,HjSa03} traversal involving initially setting the radius $\e$ to a very large number ($+\infty$), inserting first $k$ data points encountered into the list of hits and then setting $\e$ to be the largest distance of a hit from a query. From then on, if a point closer to the query than the farthest hit is found, it is inserted in the list and the previous farthest hit is removed. Eventually, the current search radius is reduced to the exact radius necessary to retrieve $k$ nearest neighbours.

The branch-and-bound procedure is implemented using a priority queue (heap) which returns the farthest data point in the list of hits (Table \ref{tbl:FSPQ} outlines the operations on priority queue). Most of the code for range search can be reused: it is only necessary to use a different \textsc{InsertHit} function involving a priority queue (Algorithm \ref{alg:FSinsertPQ}) and to initialise the priority queue in the main search function (Algorithm \ref{alg:FSkNNsrch}). Algorithm \ref{alg:FSinsertPQ} uses the final list of results $HL$ as an auxiliary list to store those neighbours that have the same distance from the query as the farthest point in the priority queue. It copies the hits in the priority queue into $HL$ after finishing the tree traversal.

The performance of the branch-and-bound algorithm depends on the order of nodes visited -- it is to a great advantage if the nodes containing data points closest to the query are visited first so that the bounding radius becomes small early on. A frequently used solution \cite{CPZ97,HjSa03} is to traverse the tree breadth-first, keeping the nodes to be visited in a second priority queue, where the priority of a node is given by the upper bound of the distance of its covering set from the query. 

The second priority queue is not used for the FSIndex based kNN search. Since the implicit tree is heavily unbalanced, the branches with smallest depth are visited first with a similar effect without the overhead of the second priority queue. The visiting order of nodes is ensured in the outer loop of the \textsc{CheckNode} function where the index $j$ starts at $m-1$, decreasing to $i$ (Algorithm \ref{alg:FSchkbin}). Since the order does not affect the range search performance, the same code can be used for range search.

\begin{algo}[htb!]
\begin{pseudocode}[ovalbox]{KNNSearch}{\omega, d, k}\label{alg:FSkNNsrch}
\COMMENT{Recursive tree traversal}\\
\GLOBAL{\e, bin,frag,lcp,\xi_j,\pi,r,HL,CD}\\
\text{Initialise list of hits $HL$}\\
\text{Initialise cumulative distances $CD$, $CD[0]\GETS 0$}\\
\text{Initialise priority queue $PQ$}\\
u\GETS r(\pi(\omega))\\
\e \GETS \infty\\
\CALL{ProcessBin}{u}\\
\CALL{CheckNode}{u,0,0}\\
\text{Insert all hits from $PQ$ to $HL$}\\
\RETURN{HL}
\end{pseudocode}
\algocapt{FSIndex kNN search algorithm.}
\end{algo}

\begin{table}[!ht]
\begin{center}
\begin{tabular}{|l|p{9cm}|}
\hline
\textsc{PQ.Size()} & number of items in the priority queue $PQ$\\
\textsc{PQ.Insert($s,p$)} & inserts item $s$ with priority $p$\\
\textsc{PQ.Peek()} & retrieves the item with highest priority and its priority\\
\textsc{PQ.Remove()} & retrieves the item with highest priority and its priority and removes it from the queue\\
\hline
\end{tabular}
\end{center}
\caption{Priority queue operations.}\label{tbl:FSPQ}
\end{table}

\begin{algo}[htb!]
\begin{pseudocode}[ovalbox]{InsertHit}{HL,s,dist}\label{alg:FSinsertPQ}
\COMMENT{Hit insertion for kNN search.}\\
\GLOBAL{k, \e, PQ}\\
\IF \CALL{PQ.Size}{}<k \THEN 
\BEGIN
  \CALL{PQ.Insert}{s,dist}\\
  \IF \CALL{PQ.Size}{} = k \THEN 
  \BEGIN
    s1,dist1 \GETS \CALL{PQ.Peek}{}\\
    \e \GETS dist1
  \END
\END
\ELSEIF dist < \e \THEN
\BEGIN
  s1,dist1 \GETS \CALL{PQ.Remove}{}\\
  \CALL{PQ.Insert}{s,dist}\\
  s2,dist2 \GETS \CALL{PQ.Peek}{}\\
  \e \GETS dist2\\
  \IF dist1 = dist2 \THEN \CALL{HL.Insert}{s,dist}
  \ELSE \CALL{HL.Clear}{}
\END
\ELSE \CALL{HL.Insert}{s,dist}
\end{pseudocode}
\algocapt{FSIndex kNN hit list insertion algorithm.}
\end{algo}

\subsection{Extensions}\label{subsec:FSextensions}

FSIndex as described so far provides an access method for workloads of fragments of fixed length with quasi-metric similarity measures. However, with minor modifications it can be extended to fragment (suffix) datasets of arbitrary length and almost arbitrary similarity measures.

\subsubsection{Arbitrary fragment lengths}

In most practical situations, fragment datasets are datasets of suffixes of full sequences. The FSIndex structure as is can be used without modifications for answering queries longer than $m$, the original length: each fragment of length $m$ is a prefix of a suffix of length $m'$ where $m'\geq m$. To search with a query of length $m'$, traverse the search tree using the first $m$ positions and sequentially scan all the bins retrieved, using all $m'$ positions to calculate the distance. If $m'>m$, the few fragments of length $m$ at the end of each full sequence can be identified and ignored at the sequential scan step.

Similarly, FSIndex can be used to answer queries centered on fragments of length $m''$ where $m''<m$. At the construction step, insert all suffixes, including those of length less than $m$ into the index by mapping each fragment $x$ such that $\abs{x}=m''<m$, into the bin $\pi_1(x_1)\pi_2(x_2)\ldots \pi_{m''}(x_{m''})\sigma_{m''+1}\ldots \sigma_{m}$, where $\sigma_{m''+1},\ldots, \sigma_{m}$ are chosen so that $\xi_{m''+1}(\sigma_{m''+1})=\xi_{m''+2}(\sigma_{m''+2})=\ldots =\xi_{m}(\sigma_{m})=0$. 

To answer a query centered on $\omega$ such that $\abs{\omega}=m''$, traverse the search tree up to the depth $m''$ and sequentially scan all the bins attached to subtrees rooted at the accepted nodes using first $m''$ positions to calculate the distance. The ranking function given by the Equations \ref{eq:FSrank1} and \ref{eq:FSrank2} ensures that the bins that are the children of a given node are adjacent in the $frag$ array.

\subsubsection{Arbitrary similarity measures}

FSIndex does not directly depend on a quasi-metric: it is constructed solely from alphabet partitions. While index performance strongly depends on the way the distance agrees with partitions, the same index can be used for any distance which is an $\ell_1$-type sum. It is possible to make even further generalisations.

Let $i=0,1,\ldots, m-1$ and suppose $\Sigma_i$ are finite alphabets and $f_i$ are arbitrary functions $\Sigma_i\to\R$. Suppose $F:\Sigma_0\times\ldots\times \Sigma_{m-1}\to\R$ is given by $F(x)=\sum_{i=0}^{m-1} f_i(x_i)$. Let $\zeta_i=\min_{a\in\Sigma_i} f_i(a)$, $z_i=\argmin_{a\in\Sigma_i} f_i(a)$ and let $z$ denote the sequence $z_0z_1\ldots z_{m-1}\in\Sigma_0\times\ldots\times \Sigma_{m-1}$. It is clear that the function $F_0$ given by $F_0(x)=F(x)-\sum_{i=0}^{m-1} \zeta_i$ is increasing on the tree $T_z$ and therefore the FSIndex can be used to answer queries for any valuation workload or a union of valuation workloads. Important biological cases include PSSM or profile based similarities which are exactly $\ell_1$-type sums of real-valued functions at each position as well as any score matrix based similarity, whether or not the triangle inequality on the alphabet is satisfied. Note that the above statement applies only to consistency of the indexing scheme and not to the computational efficiency of query retrieval.

\section{Experimental Results}\label{sec:FSresults}

This section describes the experiments on actual fragment datasets carried out to evaluate the performance of FSIndex. Three main classes of tests were conducted investigating general performance, effects of similarity measures and scalability. The final set of experiments compares performance of FSIndex to performances of suffix arrays M-Tree and mvp-tree.

Each experiment consisted of 5000 searches using randomly generated queries (Subsection \ref{subsec:randseq}). The main measures of performance are the number of bins and dataset fragments scanned in order to retrieve $k$ nearest neighbours. The principal reason for expressing the results in terms of the number of nearest neighbours retrieved rather than the radius was that it allows comparison across different indexing schemes, datasets and similarity measures.  Furthermore, most existing protein datasets are strongly non-homogeneous and the number of points scanned in order to retrieve a range query for a fixed radius varies greatly compared to the number of points scanned in order to retrieve a fixed number of nearest neighbours. Nevertheless, most experiments involve range search algorithms, because they are generally more efficient and because in some cases no $k$NN implementation was available.

Other performance criteria were total running time (only shown where all experiments compared were performed on the same machine with similar loads) and the percentage of residues (letters) scanned out the total number of residues in all scanned fragments. The later statistic measures the effect of sub-indexing each bin using the suffix-array-like structure which involves `partially' scanning each fragment with a help of the $lcp$ array. The final statistic is access overhead, discussed in Section \ref{sec:perfgeom}.

The obvious reference algorithm, which was not run due to excessive running times for large datasets, is sequential scan of all fragments in a dataset. Most of the experiments were run on a Sun Fire[tm] 280R server (733 Mhz CPU).
 
\subsection{Datasets and indexes}

Experiments investigating general performance and effect of different similarity measures used overlapping protein fragment datasets derived from the SwissProt Release 43.2 of April 2004. Scalability experiments used, in addition to SwissProt, the datasets \texttt{nr018K}, \texttt{nr036K}, \texttt{nr072K}, and \texttt{nr288K}, obtained by randomly sampling 18, 36, 72 and 288 thousands of sequences respectively from the nr dataset (SwissProt fills the gap because it contains about 150,000 sequences). The experiments comparing FSindex to suffix arrays and mvp-tree used only the \texttt{nr018K} dataset.

Table \ref{tbl:FSinddata} describes the instances of FSIndex used in the evaluations. Two instances (\texttt{SPNA09} and \texttt{SPNB09}) were based on partitions that are not equal at all positions while the remainder had the same partitions at all positions.

\begin{table}[!ht]
\begin{center}
{\scriptsize\tt
\begin{tabular}{|l|l|l|r|r|}
\hline
Index & Dataset & Partitions & Fragments & Bins  \\
\hline\hline
SPEQ06 & SwissProt & T,SA,N,ILV,M,KR,DE,Q,WF,Y,H,G,P,C & 53486349 &  7529536 \\ \hline
SPEQ09 & SwissProt & TSAN,ILVM,KR,DEQ,WFYH,GPC & 53478888 & 10077696 \\ \hline
SPEQ12 & SwissProt & TSAN,ILVM,KRDEQ,WFYHGPC & 53472161 & 16777216 \\ \hline
\hline
nr01809 & nr018K & TSAN,ILVM,KR,DEQ,WFYH,GPC & 6005750 & 10077696 \\ \hline
nr03609 & nr036K & TSAN,ILVM,KR,DEQ,WFYH,GPC & 11911191 & 10077696 \\ \hline
nr07209 & nr072K & TSAN,ILVM,KR,DEQ,WFYH,GPC & 23878523 & 10077696 \\ \hline
nr28809 & nr288K & TSAN,ILVM,KR,DEQ,WFYH,GPC & 95593618 & 10077696 \\ \hline
\hline 
SPNA09 & SwissProt & KR,Q,E,D,N,T,SA,G,H,W,Y,F,P,C,ILV,M & 53478888 & 10483200\\
  & & KR,Q,ED,N,T,SA,G,HW,YF,P,C,ILV,M & &\\
  & & KR,QED,N,TSA,G,HW,YF,P,C,ILVM & &\\
  & & KR,QEDN,TSA,G,HWYF,PC,ILVM & &\\
  & & KR,QEDN,TSA,G,HWYFPC,ILVM & &\\
  & & KR,QEDN,TSAG,HWYFPC,ILVM & &\\
  & & KRQEDN,TSAG,HWYFPC,ILVM & &\\
  & & KRQEDN,TSAG,HWYFPCILVM & &\\
  & & KRQEDNTSAG,HWYFPCILVM & &\\ \hline
SPNB09 & SwissProt & KR,QEDN,TSA,G,HWYF,PC,ILVM & 53476582 & 8643600 \\
  & &  KR,QEDN,TSA,G,HWYF,PC,ILVM & &\\
  & &  KR,QEDN,TSA,G,HWYF,PC,ILVM & &\\
  & &  KR,QEDN,TSA,G,HWYF,PC,ILVM & &\\
  & &  KR,QEDN,TSA,G,HWYFPC,ILVM & &\\
  & &  KR,QEDN,TSAG,HWYFPC,ILVM & &\\
  & &  KR,QEDN,TSAG,HWYFPC,ILVM & &\\
  & &  KRQEDN,TSAG,HWYFPC,ILVM & &\\
  & &  KRQEDN,TSAG,HWYFPCILVM & &\\
  & &  KRQEDNTSAG,HWYFPCILVM & &\\ \hline
\hline
\end{tabular}
}
\end{center}
\caption[Instances of FSIndex used in experimental evaluations.]{Instances of FSIndex used in experimental evaluations. The last two digits of the index name denote the length of reduced fragments. The indexes \texttt{SPNA09} and \texttt{SPNB09} use non-equal partitions at different positions (all shown) while the remainder were constructed using one partition for all positions (only one shown).}\label{tbl:FSinddata}
\end{table}

The choice of amino acid alphabet partitions was mainly a result of practical considerations based on the BLOSUM62 quasi-metric (Figure \ref{fig:blosum62qd}). It was not possible to partition the alphabet in a way that all distances within partitions are smaller than distances between and hence the primary criterion was to have as high lower bound on distances from any possible query point to any partition but its own. The additional criterion was to balance to the greatest possible extent the sizes of bins and to avoid having too many empty bins which would introduce large overhead. Therefore, the number of partitions per residue was decreased with fragment length by amalgamating `close' partitions. Some amino acids having very small overall frequencies, such as tryptophan (`W') and cysteine (`C'), were in some cased clustered together in order to reduce the total number of partitions, even though their distances from and to any other amino acid are very large.  

\begin{figure}[htb!]
\centering
\scalebox{1}[1]{\tt\scriptsize
\begin{tabular}{l@{  }c@{}c@{}c@{}c@{}c@{}c@{}c@{}c@{}c@{}c@{}c@{}c@{}c@{}c@{}c@{}c@{}c@{}c@{}c@{}c@{}c@{}}
  & \bf T & \bf S & \bf A & \bf N & \bf I & \bf V & \bf L & \bf M & \bf K & \bf R & \bf D & \bf E & \bf Q & \bf W & \bf F & \bf Y & \bf H & \bf G & \bf P & \bf C \\
\bf T & \gh{0} & \gh{3} & \gh{4} & \gh{6} & 5 & 4 & 5 & 6 & 6 & 6 & 7 & 6 & 6 & 13 &8 & 9 & 10 & 8 & 8 & 10 \\
\bf S & \gh{4} & \gh{0} & \gh{3} & \gh{5} & 6 & 6 & 6 & 6 & 5 & 6 & 6 & 5 & 5 & 14 & 8 & 9 &  9 & 6 & 8 & 10 \\
\bf A & \gh{5} & \gh{3} & \gh{0} & \gh{8} & 5 & 4 & 5 & 6 & 6 & 6 & 8 & 6 & 6 & 14 & 8 & 9 & 10 & 6 & 8 & 9 \\
\bf N & \gh{5} & \gh{3} & \gh{6} & \gh{0} & 7 & 7 & 7 & 7 & 5 & 5 & 5 & 5 & 5 & 15 & 9 & 9 & 7 & 6 & 9 & 12 \\
\bf I & 6 & 6 & 5 & 9 & \gh{0} & \gh{1} & \gh{2} & \gh{4} & 8 & 8 & 9 & 8 & 8 & 14 & 6 & 8 & 11 & 10 & 10 & 10 \\
\bf V & 5 & 6 & 4 & 9 & \gh{1} & \gh{0} & \gh{3} & \gh{4} & 7 & 8 & 9 & 7 & 7 & 14 & 7 & 8 & 11 & 9  & 9  & 10 \\
\bf L & 6 & 6 & 5 & 9 & \gh{2} & \gh{3} & \gh{0} & \gh{3} & 7 & 7 & 10& 8 & 7 & 13 & 6 & 8 & 11 & 10 & 10& 10 \\
\bf M & 6 & 5 & 5 & 8 & \gh{3} & \gh{3} & \gh{2} & \gh{0} & 6 & 6 & 9 & 7 & 5 & 12 & 6 & 8 &10 & 9 & 9 &10 \\
\bf K & 6 & 4 & 5 & 6 & 7 & 6 & 6 & 6 & \gh{0} & \gh{3} & 7 & 4 & 4 &14&  9 & 9 & 9 & 8 & 8 &12 \\
\bf R & 6 & 5 & 5 & 6 & 7 & 7 & 6 & 6 & \gh{3} & \gh{0} & 8 & 5 & 4 &14&  9 & 9 & 8 & 8 & 9 &12 \\
\bf D & 6 & 4 & 6 & 5 & 7 & 7 & 8 & 8 & 6 & 7 & \gh{0} & \gh{3} & \gh{5} &15&  9 &10 & 9 & 7 & 8 &12 \\
\bf E & 6 & 4 & 5 & 6 & 7 & 6 & 7 & 7 & 4 & 5 & \gh{4} & \gh{0} & \gh{3} &14&  9 & 9 & 8 & 8 & 8 &13 \\
\bf Q & 6 & 4 & 5 & 6 & 7 & 6 & 6 & 5 & 4 & 4 & \gh{6} & \gh{3} & \gh{0} &13&  9 & 8 & 8 & 8 & 8 &12 \\
\bf W & 7 & 7 & 7 &10 & 7 & 7 & 6 & 6 & 8 & 8 &10 & 8 & 7 & \gh{0}&  \gh{5} & \gh{5} &\gh{10} & 8 &11 &11 \\
\bf F & 7 & 6 & 6 & 9 & 4 & 5 & 4 & 5 & 8 & 8 & 9 & 8 & 8 &\gh{10}&  \gh{0} & \gh{4} & \gh{9} & 9 &11 &11 \\
\bf Y & 7 & 6 & 6 & 8 & 5 & 5 & 5 & 6 & 7 & 7 & 9 & 7 & 6 & \gh{9}&  \gh{3} & \gh{0} & \gh{6} & 9 &10 &11 \\
\bf H & 7 & 5 & 6 & 5 & 7 & 7 & 7 & 7 & 6 & 5 & 7 & 5 & 5 &\gh{13}&  \gh{7} & \gh{5} & \gh{0} & 8 & 9 &12 \\
\bf G & 7 & 4 & 4 & 6 & 8 & 7 & 8 & 8 & 7 & 7 & 7 & 7 & 7 &13&  9 &10 &10 & \gh{0} & \gh{9} &\gh{12} \\
\bf P & 6 & 5 & 5 & 8 & 7 & 6 & 7 & 7 & 6 & 7 & 7 & 6 & 6 &15& 10 &10 &10 & \gh{8} & \gh{0} &\gh{12} \\
\bf C & 6 & 5 & 4 & 9 & 5 & 5 & 5 & 6 & 8 & 8 & 9 & 9 & 8 &13&  8 & 9 &11 & \gh{9} &\gh{10} & \gh{0} \\
\end{tabular}
}
\caption[BLOSUM62 quasi-metric.]{BLOSUM62 quasi-metric. Distances within members of an alphabet partition used for constructing an index for fragments of length 9 used in experiments are greyed.}
\label{fig:blosum62qd}
\end{figure}

The alphabet partitions from the Table \ref{tbl:FSinddata} agree with the `biochemical intuition' (i.e. the classification from the Table \ref{tbl:amino_acids} based on chemical properties of amino acids). For example, the clusters outlined in the Figure \ref{fig:blosum62qd} used for fragments of length 9 approximately correspond to polar uncharged, hydrophobic, basic, acidic, aromatic and `other' amino acids. The partition used for the fragments of length 12 is obtained by merging together acidic and basic as well as aromatic and `other' clusters. An interesting fact is that in this case each of the the four clusters has a relative frequency very close to $\frac14$. 

Despite efforts to balance bin sizes, the distributions of bin sizes were strongly skewed in favour of small sizes in all cases (Figure \ref{fig:binsizedist} shows one example) with many empty but also a few very large bins. Such distributions appear to follow the DGX distribution, a generalisation of Zipf-Mandelbrot law described by Bi, Faloutsos and Korn \cite{BFK01}.  

\begin{figure}[ht!]
\begin{center}
\scalebox{0.8}{\includegraphics{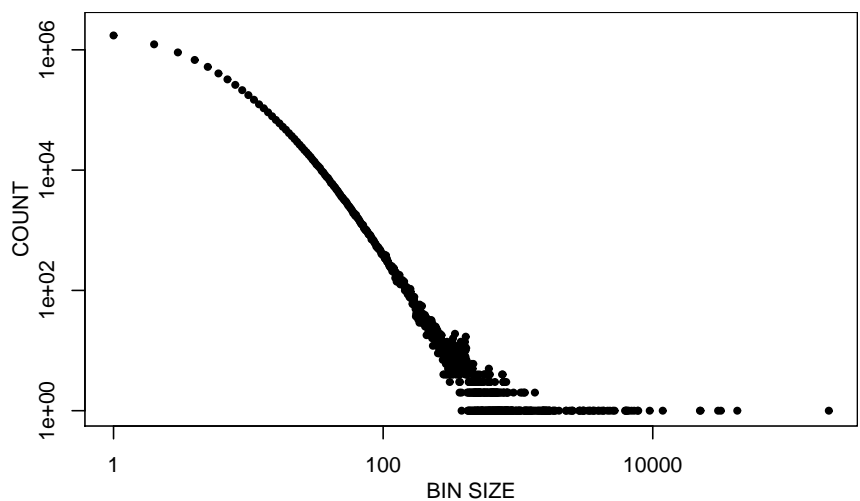}}
\caption[Distribution of \texttt{SPEQ09} bin sizes.]{Distribution of \texttt{SPEQ09} bin sizes (2,342,940 empty bins out of 10,077,696).}\label{fig:binsizedist}
\end{center}
\end{figure}

\subsection{General performance}\label{sec:FSperf}

Figures \ref{fig:FSperf06}, \ref{fig:FSperf09} and \ref{fig:FSperf12} present selected statistics of search experiments for fragment lengths 6,9 and 12 respectively, consisting in each case of range queries retrieving 1, 10, 50, 100, 500 and 1000 nearest neighbours with respect to the BLOSUM62-based $\ell_1$-type quasi-metric. For each length, $k$NN searches were performed prior to range searches using the index that was expected to be the fastest in order to determine the search ranges for each random query fragment. 


\begin{figure}[p]
\begin{center}
\begin{tabular}[t]{lr}
\multicolumn{1}{l}{\mbox{\bf (a)}} &
        \multicolumn{1}{l}{\mbox{\bf (b)}} \\ [0.01cm]
\scalebox{0.79}[0.79]{\includegraphics{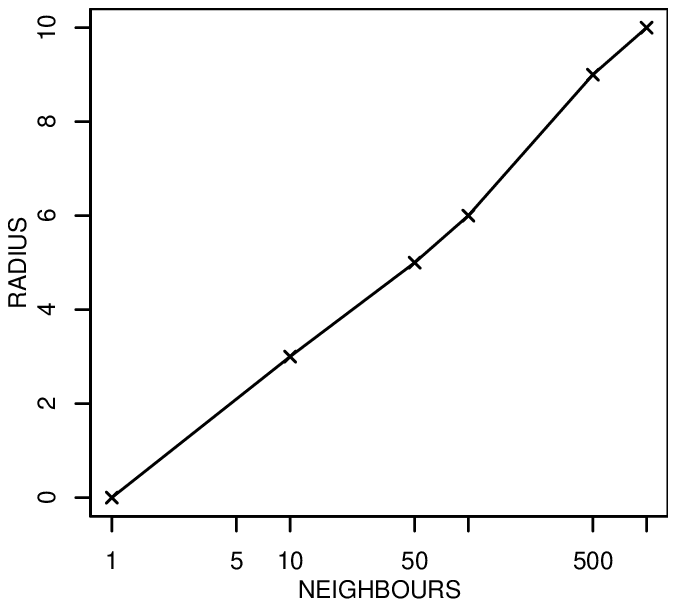}} &
\scalebox{0.79}[0.79]{\includegraphics{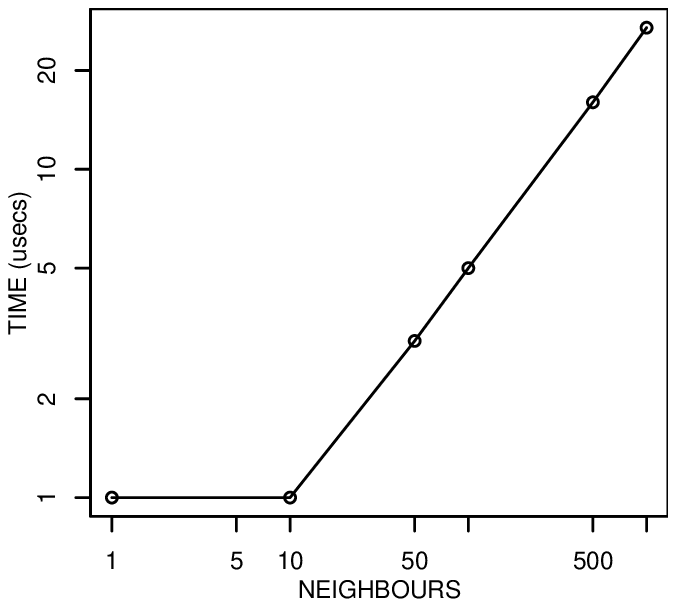}}\\ [-0.5cm]

\multicolumn{1}{l}{\mbox{\bf (c)}} &
        \multicolumn{1}{l}{\mbox{\bf (d)}} \\ [-0.1cm]
\scalebox{0.79}[0.79]{\includegraphics{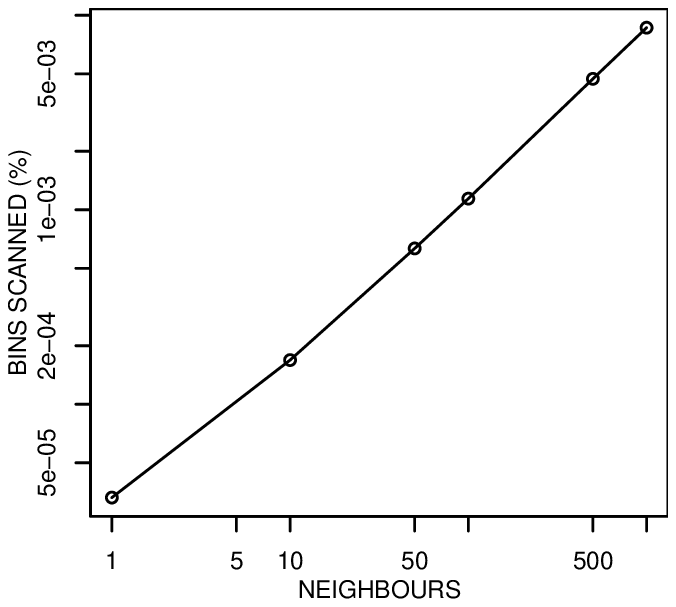}} &
\scalebox{0.79}[0.79]{\includegraphics{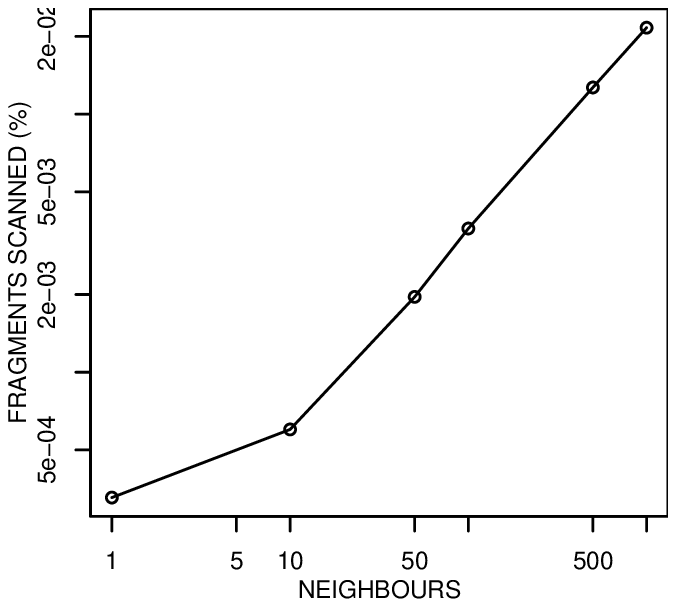}}\\ [-0.5cm]

\multicolumn{1}{l}{\mbox{\bf (e)}} &
        \multicolumn{1}{l}{\mbox{\bf (f)}} \\ [-0.1cm]
\scalebox{0.79}[0.79]{\includegraphics{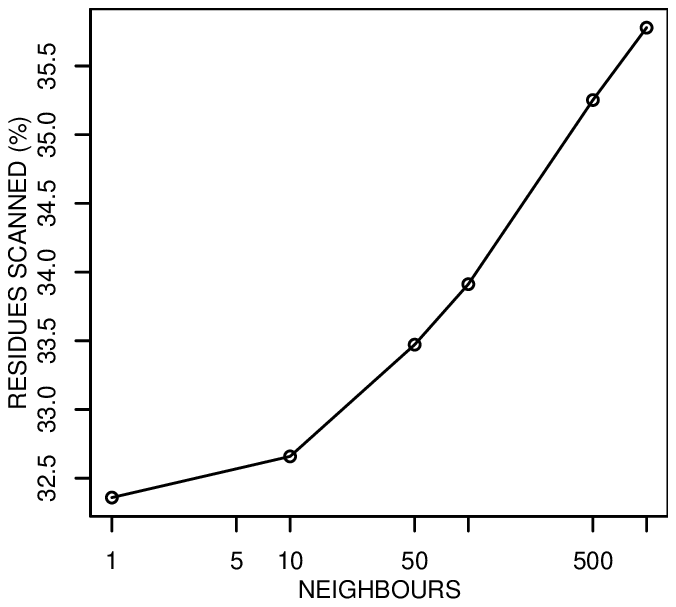}} &
\scalebox{0.79}[0.79]{\includegraphics{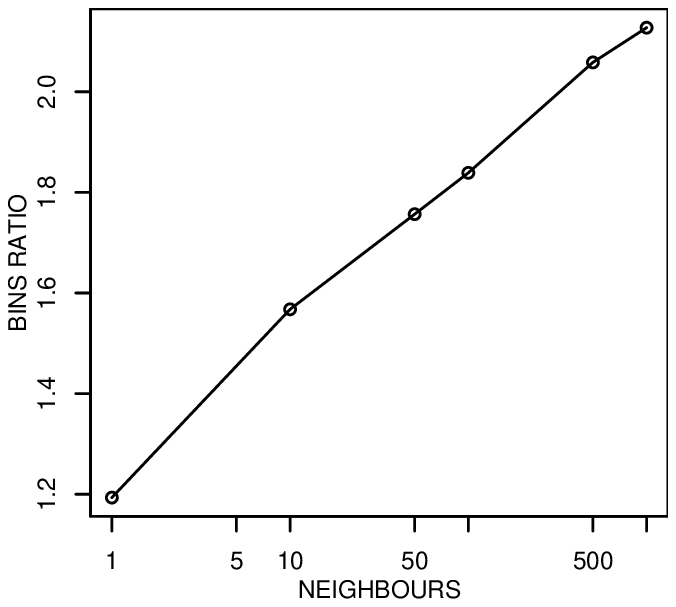}}\\
\end{tabular}
\caption[Performance of FSIndex for fragments of length 6.]{General performance of FSIndex for fragment dataset of length 6: {\bf(a)} Median radius of a ball containing $k$ nearest neighbours; {\bf (b)} Total running time for 5000 searches; {\bf (c)} Mean number of bins scanned; {\bf (d)} Mean number of fragments scanned; {\bf (e)} Percentage of residues scanned (out of total number of residues in fragments scanned); {\bf (f)} Mean ratio between the number of bins retrieved for kNN and range searches.}\label{fig:FSperf06}
\end{center}
\end{figure}

\begin{figure}[p]
\begin{center}
\begin{tabular}[t]{lr}
\multicolumn{1}{l}{\mbox{\bf (a)}} &
        \multicolumn{1}{l}{\mbox{\bf (b)}} \\ [0.01cm]
\scalebox{0.79}[0.79]{\includegraphics{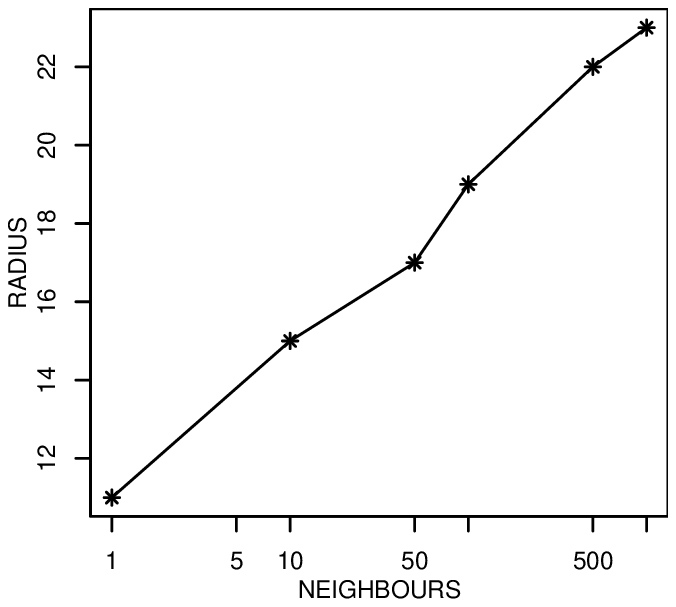}} &
\scalebox{0.79}[0.79]{\includegraphics{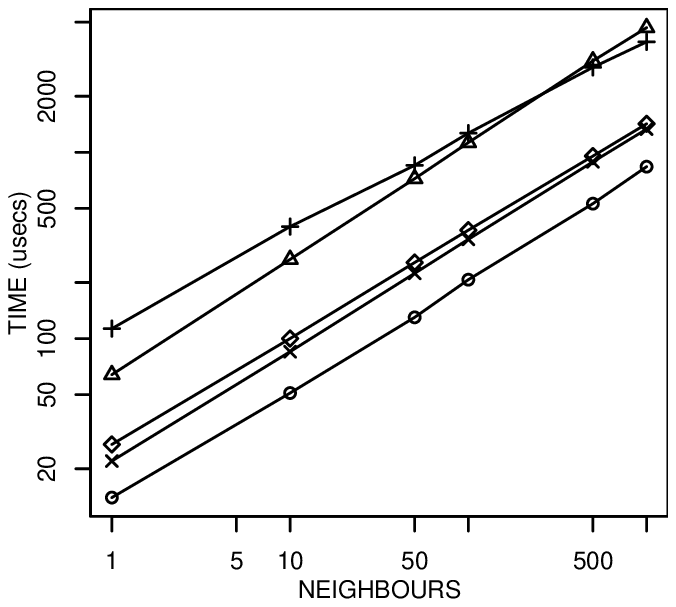}}\\ [-0.5cm]

\multicolumn{1}{l}{\mbox{\bf (c)}} &
        \multicolumn{1}{l}{\mbox{\bf (d)}} \\ [-0.1cm]
\scalebox{0.79}[0.79]{\includegraphics{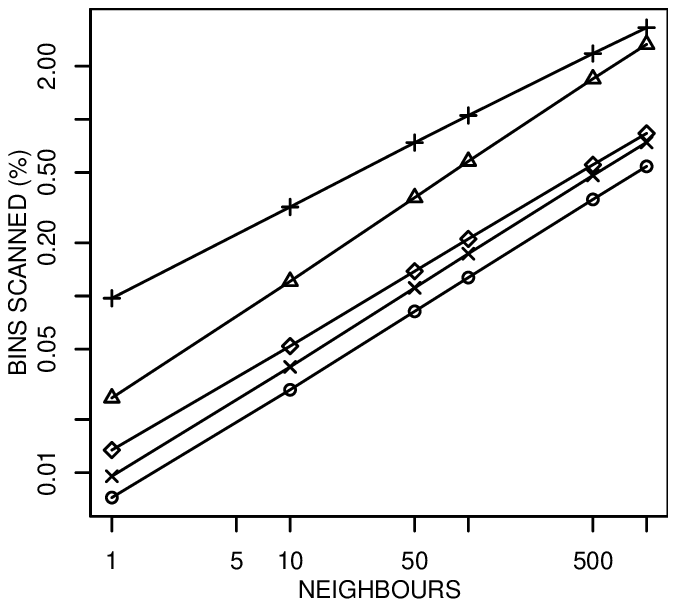}} &
\scalebox{0.79}[0.79]{\includegraphics{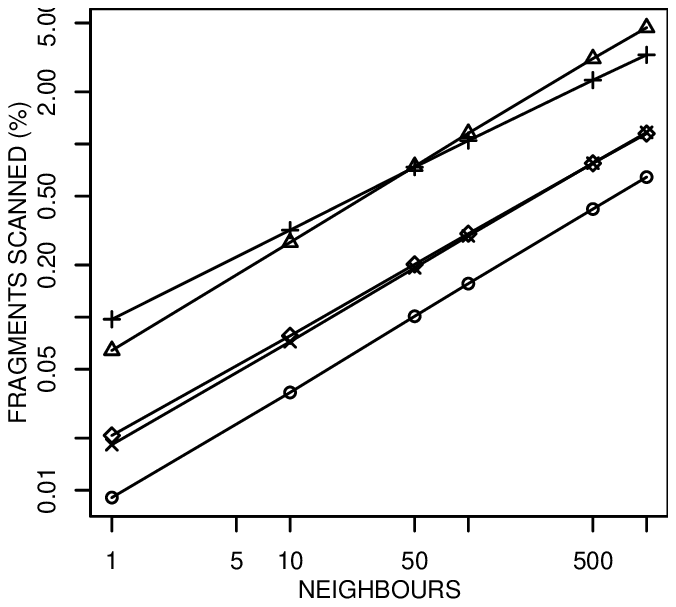}}\\ [-0.5cm]

\multicolumn{1}{l}{\mbox{\bf (e)}} &
        \multicolumn{1}{l}{\mbox{\bf (f)}} \\ [-0.1cm]
\scalebox{0.79}[0.79]{\includegraphics{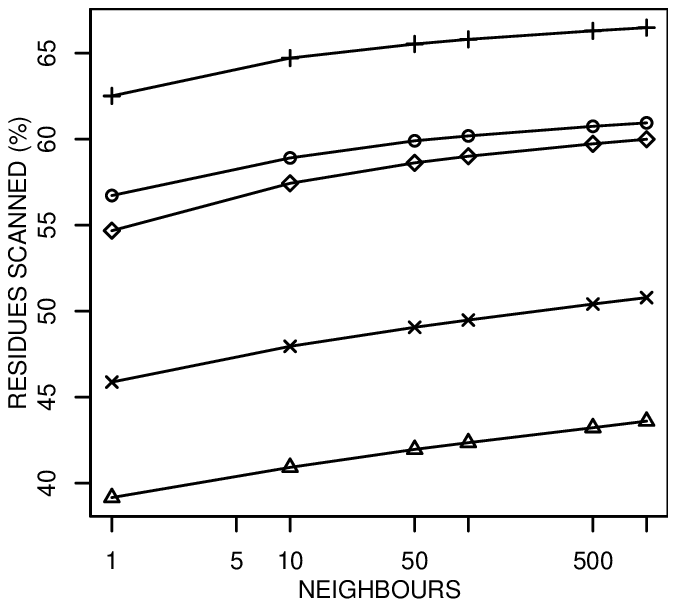}} &
\scalebox{0.79}[0.79]{\includegraphics{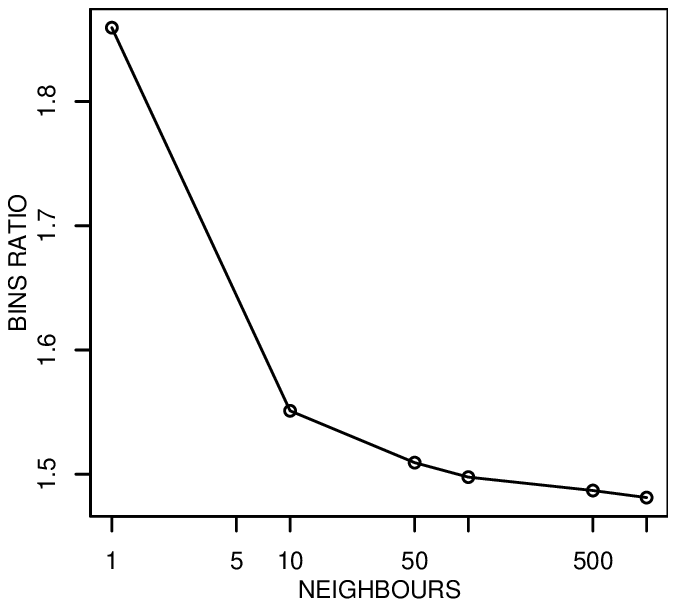}}\\
\end{tabular}
\begin{tabular}[c]{ll}
\mbox{
\begin{minipage}[c]{0.74\textwidth}
\caption[Performance of FSIndex for fragments of length 9.]{General performance of FSIndex for fragment dataset of length 9: {\bf(a)} Median radius of a ball containing $k$ nearest neighbours; {\bf (b)} Total running time for 5000 searches; {\bf (c)} Mean number of bins scanned; {\bf (d)} Mean number of fragments scanned; {\bf (e)} Percentage of residues scanned (out of total number of residues in fragments scanned); {\bf (f)} Mean ratio between the number of bins retrieved for kNN and range searches.}\label{fig:FSperf09}
\end{minipage}
} 
&
\mbox{
\begin{minipage}[c]{0.26\textwidth}
\scalebox{1.0}[1.0]{\includegraphics{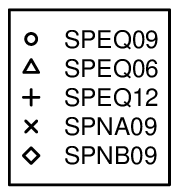}}\\
\end{minipage}
} \\
\end{tabular}
\end{center}
\end{figure}

\begin{figure}[p]
\begin{center}
\begin{tabular}[t]{lr}
\multicolumn{1}{l}{\mbox{\bf (a)}} &
        \multicolumn{1}{l}{\mbox{\bf (b)}} \\ [0.01cm]
\scalebox{0.79}[0.79]{\includegraphics{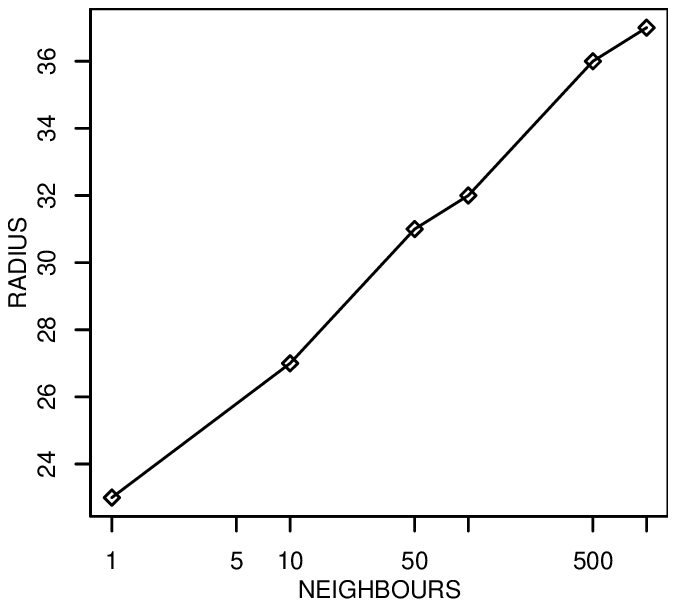}} &
\scalebox{0.79}[0.79]{\includegraphics{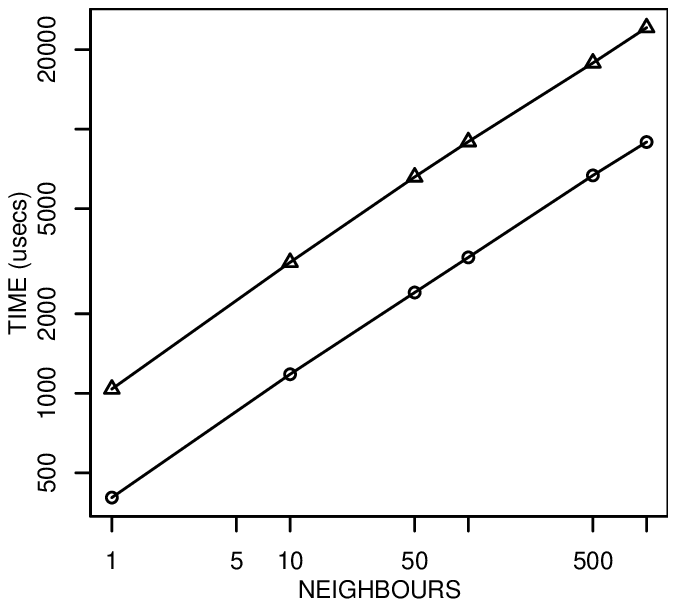}}\\ [-0.5cm]

\multicolumn{1}{l}{\mbox{\bf (c)}} &
        \multicolumn{1}{l}{\mbox{\bf (d)}} \\ [-0.1cm]
\scalebox{0.79}[0.79]{\includegraphics{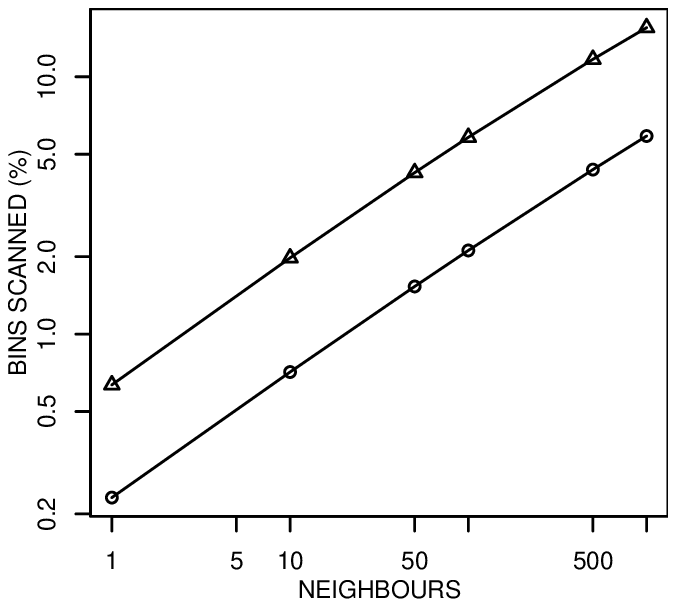}} &
\scalebox{0.79}[0.79]{\includegraphics{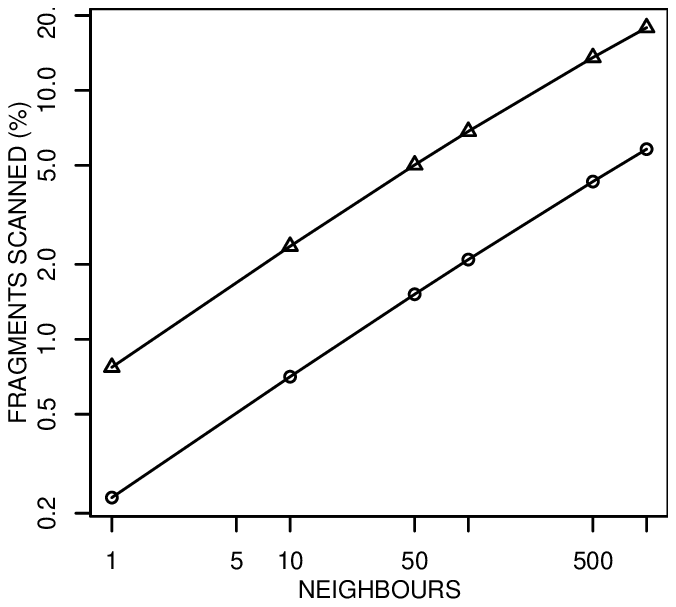}}\\ [-0.5cm]

\multicolumn{1}{l}{\mbox{\bf (e)}} &
        \multicolumn{1}{l}{\mbox{\bf (f)}} \\ [-0.1cm]
\scalebox{0.79}[0.79]{\includegraphics{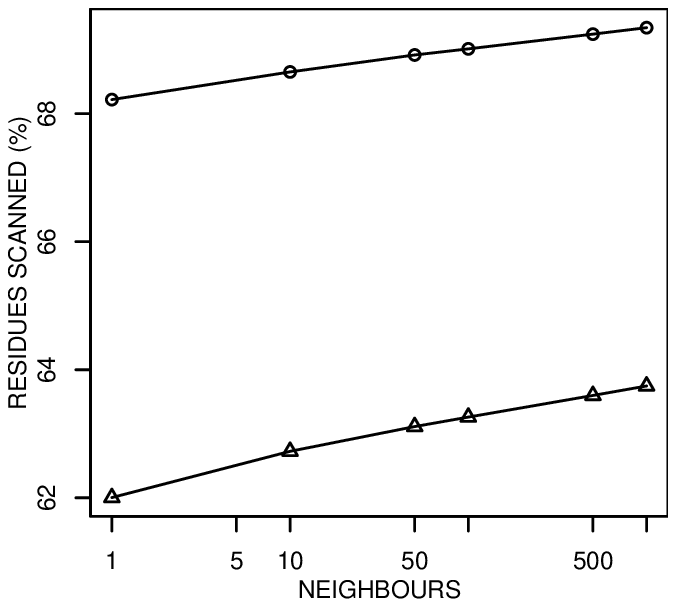}} &
\scalebox{0.79}[0.79]{\includegraphics{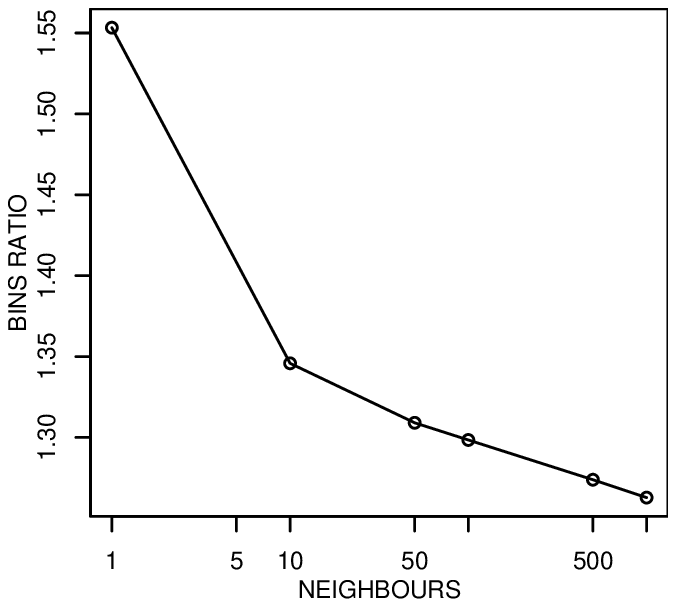}}\\
\end{tabular}
\begin{tabular}[c]{ll}
\mbox{
\begin{minipage}[c]{0.74\textwidth}
\caption[Performance of FSIndex for fragments of length 12.]{General performance of FSIndex for fragment dataset of length 12: {\bf(a)} Median radius of a ball containing $k$ nearest neighbours; {\bf (b)} Total running time for 5000 searches; {\bf (c)} Mean number of bins scanned; {\bf (d)} Mean number of fragments scanned; {\bf (e)} Percentage of residues scanned (out of total number of residues in fragments scanned); {\bf (f)} Mean ratio between the number of bins retrieved for kNN and range searches.}\label{fig:FSperf12}
\end{minipage}
} 
&
\mbox{
\begin{minipage}[c]{0.26\textwidth}
\scalebox{1.0}[1.0]{\includegraphics{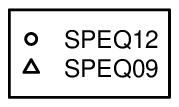}}\\
\end{minipage}
} \\
\end{tabular}
\end{center}
\end{figure}

\subsection{Dependence on similarity measures}

While queries based on more than one similarity measure can be used on a single FSIndex, it is to be expected that similarity measures different from the one originally used to determine the partitions would have worse performance. To investigate the difference in performance for different BLOSUM matrices, range queries needed to retrieve 100 nearest neighbours of testing fragments of length 9 were run using the index \texttt{SPEQ09} which was performing the best for the length 9 in the previous experiment (Figure \ref{fig:FSperf09}). In addition, searches were performed using the PSSMs (Section \ref{sec:profiles}) constructed for each test fragment from the results of a BLOSUM62-based 100 NN search in order to gain an insight in the actual search performance using the PSSM constructed from the results of a previous search that could be used to plan the biological experiments in Chapter \ref{ch:5}. Table \ref{tbl:FSmatdata} presents a summary of the results.

\begin{table}[!ht]
\begin{center}
{\small
\begin{tabular}{|l|r|r|r|r|}
\hline
Matrix & Bins (\%) & Fragments (\%) & Residues (\%) & kNN Ratio  \\
\hline\hline
BLOSUM45 & 0.1004 & 0.1230 & 60.8850 & 1.5004\\
BLOSUM50 & 0.0978 & 0.1146 & 61.0993 & 1.4807\\
BLOSUM62 & 0.0957 & 0.1194 & 60.9394 & 1.4689\\
BLOSUM80 & 0.1038 & 0.1306 & 61.1321 & 1.4771\\
BLOSUM90 & 0.1111 & 0.1539 & 61.1010 & 1.4733\\
PSSM & 0.0707 & 0.0869 & 58.1547 & 2.1805\\
\hline
\end{tabular}
}
\end{center}
\caption[Performance of the FSIndex with different similarity measures.]{Performance of the FSIndex \texttt{SPEQ09} with different similarity measures. The values shown are based on 100 NN queries of length 9. The columns denote the similarity measure (matrix), percentages of bins, fragments and residues (as before the percentage is out of the total number of residues in scanned fragments) scanned and the ratio between the number of bins retrieved for kNN and range searches.}\label{tbl:FSmatdata}
\end{table}

\subsection{Scalability}\label{sec:FSscale}

Figure \ref{fig:FSscale} shows the results of a set of experiments involving instances of FSIndex based on datasets of fragments of length 9 of different sizes (\texttt{nr018K}, \texttt{nr036K}, \texttt{nr072K}, \texttt{SwissProt} and \texttt{nr288K}). All indexes used the same alphabet partition (Table \ref{tbl:FSinddata}) and all queries were based on the BLOSUM62 $\ell_1$-type quasi-metric. Unlike the Figures \ref{fig:FSperf06}, \ref{fig:FSperf09} and \ref{fig:FSperf12}, Figure \ref{fig:FSscale} does not contain the total running time graph because the experiments were performed on different machines but instead includes a plot showing the total number of residues scanned against the database size. This graph indicates the dependence of the performance of (an example of) FSIndex on dataset size, that is, its scalability.

\begin{figure}[p]
\begin{center}
\begin{tabular}[t]{lr}
\multicolumn{1}{l}{\mbox{\bf (a)}} &
        \multicolumn{1}{l}{\mbox{\bf (b)}} \\ [0.01cm]
\scalebox{0.79}[0.79]{\includegraphics{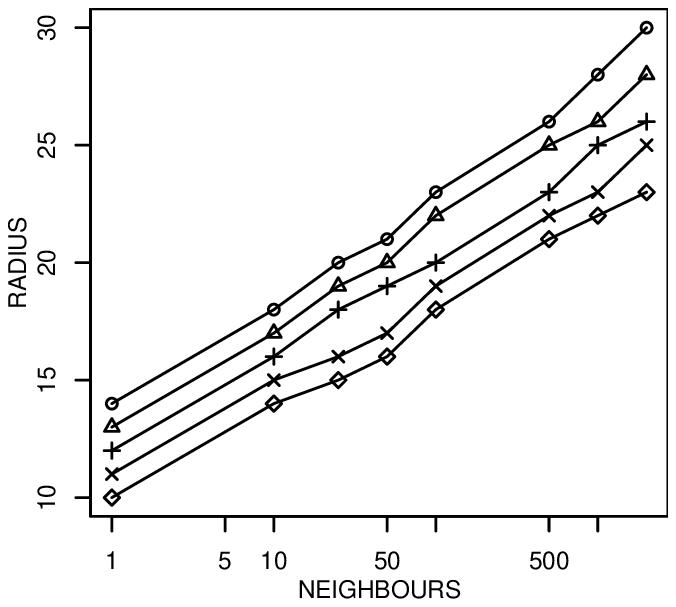}} &
\scalebox{0.79}[0.79]{\includegraphics{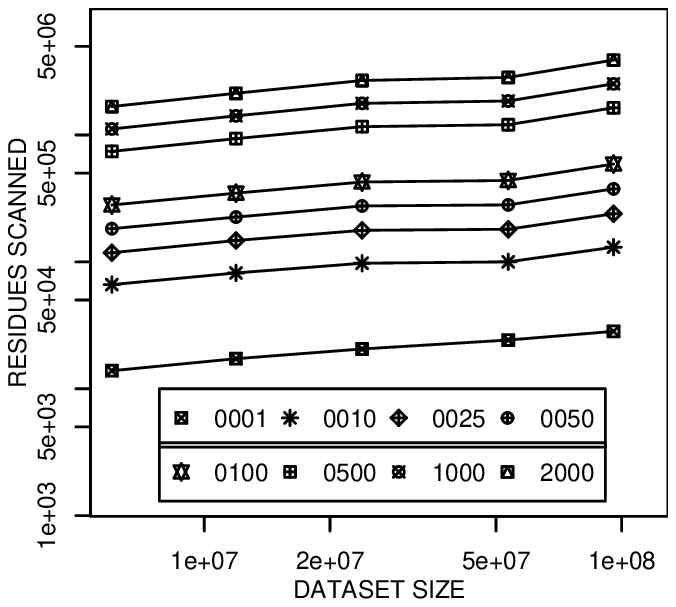}}\\ [-0.5cm]

\multicolumn{1}{l}{\mbox{\bf (c)}} &
        \multicolumn{1}{l}{\mbox{\bf (d)}} \\ [-0.1cm]
\scalebox{0.79}[0.79]{\includegraphics{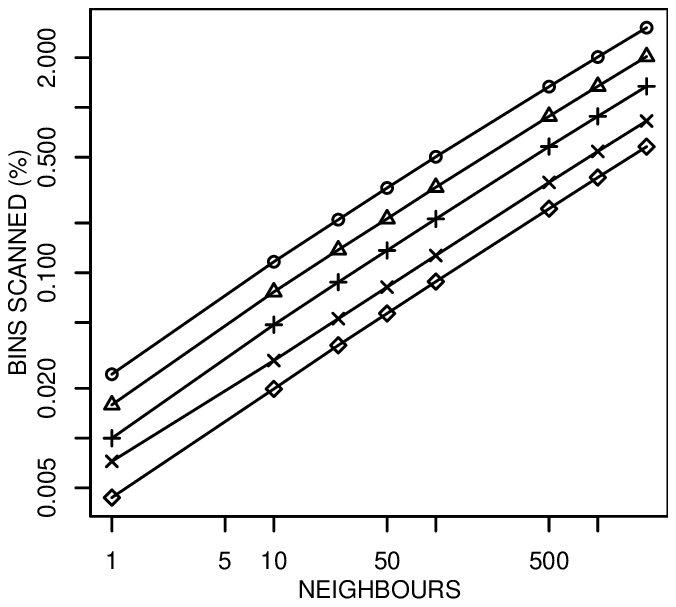}} &
\scalebox{0.79}[0.79]{\includegraphics{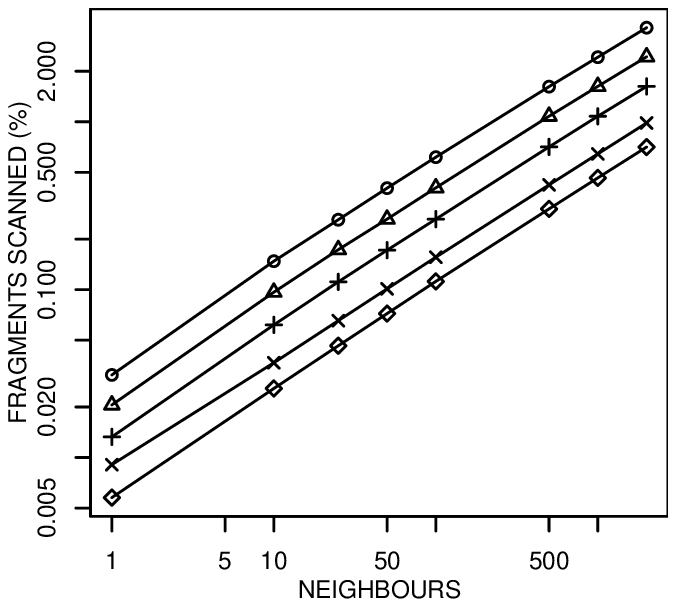}}\\ [-0.5cm]

\multicolumn{1}{l}{\mbox{\bf (e)}} &
        \multicolumn{1}{l}{\mbox{\bf (f)}} \\ [-0.1cm]
\scalebox{0.79}[0.79]{\includegraphics{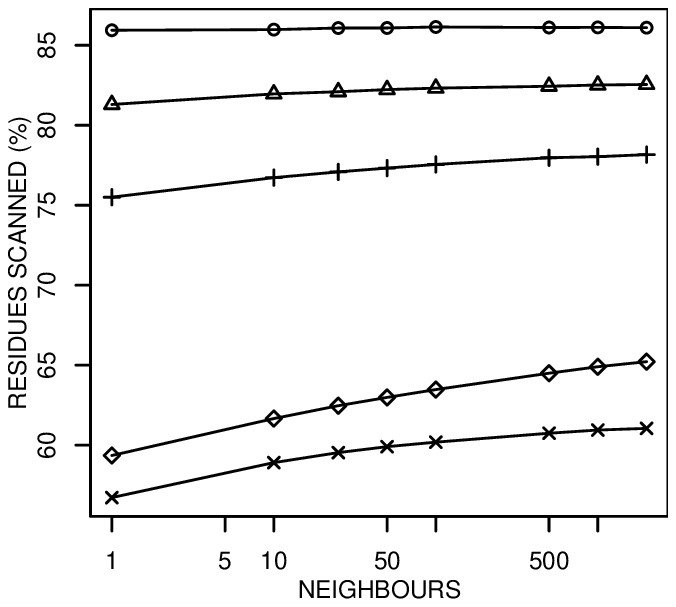}} &
\scalebox{0.79}[0.79]{\includegraphics{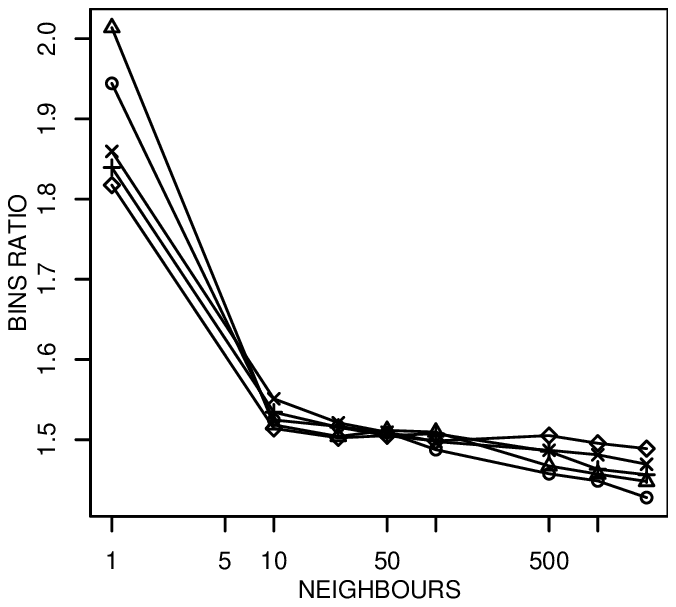}}\\
\end{tabular}
\begin{tabular}[c]{ll}
\mbox{
\begin{minipage}[c]{0.74\textwidth}
\vspace{-4mm}
\caption[Performance of FSIndex for fragments of length 9 (datasets of different sizes).]{Performance of FSIndex for fragment datasets of length 9 of different sizes: {\bf(a)} Median radius of a ball containing $k$ nearest neighbours; {\bf (b)} Scalability. Each line depicts a different number of nearest neighbours; {\bf (c)} Mean number of bins scanned; {\bf (d)} Mean number of fragments scanned; {\bf (e)} Percentage of residues scanned (out of total number of residues in fragments scanned); {\bf (f)} Mean ratio between the number of bins retrieved for kNN and range searches.}\label{fig:FSscale}
\end{minipage}
} 
&
\mbox{
\begin{minipage}[c]{0.26\textwidth}
\scalebox{1.0}[1.0]{\includegraphics{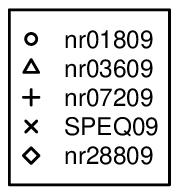}}\\
\end{minipage}
} \\
\end{tabular}
\end{center}
\end{figure}

\subsection{Access overhead}

Figure \ref{fig:FSaccess} summarises some of the results of Sections \ref{sec:FSperf} and \ref{sec:FSscale} by showing the average access overhead (Definition \ref{defin:access_overhead}), that is, the average ratio between the 
number of fragments scanned and the number of true neighbours retrieved, for all combinations of indexes and fragment lengths available. Range search algorithm and the BLOSUM62-based $\ell_1$-type quasi-metric were used in all cases.

\begin{figure}[ht!]
\begin{center}
\scalebox{0.7}{\includegraphics{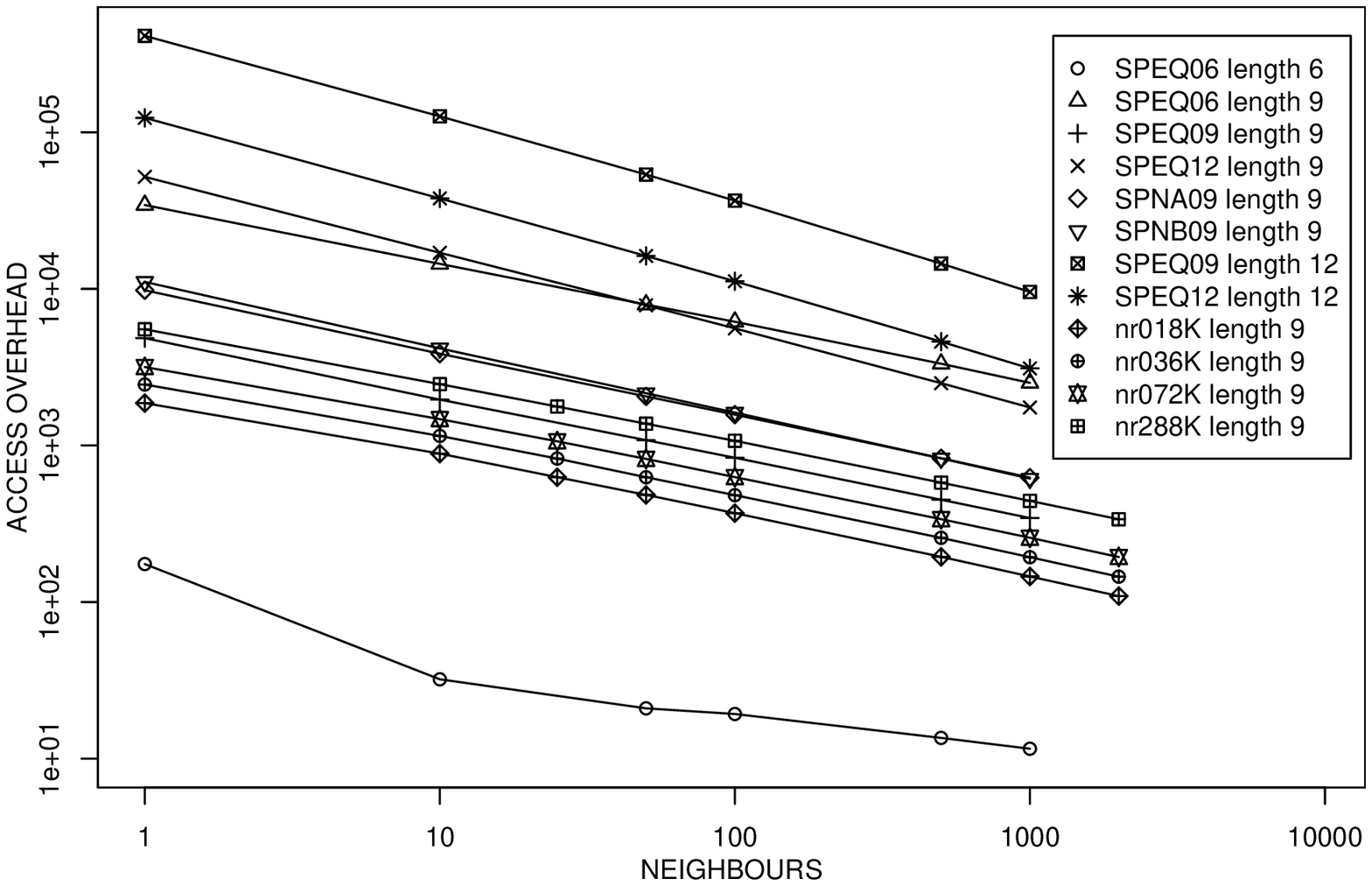}}
\caption{Average access overhead of searches using FSIndex.}\label{fig:FSaccess}
\end{center}
\end{figure}

\subsection{Comparisons with other access methods}\label{subsec:FScomp}

The final set of experiments compares FSIndex with M-tree, mvp-tree and suffix arrays. In general, other methods take significantly more space and time compared with FSIndex and it was therefore necessary to restrict the comparisons to small datasets and queries retrieving fewer neighbours.

\subsubsection{M-tree}

Recall that M-tree is a paged metric access method that stores the majority of the structure in secondary memory, usually on hard disk. This is in contrast with the implementations of FSIndex, mvp-tree and suffix arrays used here, which store the whole index structure in primary memory. Hence, although M-tree occupies large amounts of space, most of the costs are associated with the secondary memory, which is much less expensive. On the other hand, I/O costs, not considered here, can be quite large.

The experiments described below were performed earlier than the other experiments presented in the present Chapter, using the resources from the High Performance Computing Laboratory (HPCVL), a consortium of several Canadian universities that the thesis author had the fortune to access during his visits to University of Ottawa. M-tree was not tested directly but as a part of the FMTree structure (Example \ref{ex:FMtree}) that allows use of metric indexing schemes for retrieval of quasi-metric queries. 

The FMTree structure consisted of an array of M-trees with additional data describing the score matrix and the distribution self-similarities. FMTree was constructed by splitting the dataset into fibres and indexing each fibre separately using an instance of M-tree that was created using the BulkLoading algorithm of Ciaccia and Patella\cite{CP98}. To perform a range search, the FMTree range search algorithm queries all M-trees associated with fibres as described in the Example \ref{ex:FMtree} and collects the hits to produce the answer to the query. The M-tree implementation was obtained from its authors' site: \url{http://www-db.deis.unibo.it/Mtree/index.html}.

\begin{figure}[!htb]
\centerline{\scalebox{0.7}{\includegraphics{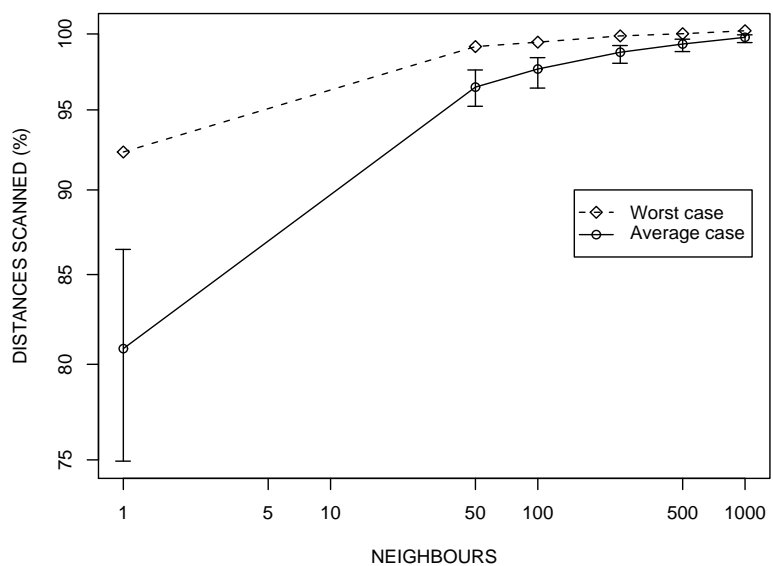}}}
\caption[Performance of FMTree based on M-tree on a dataset of fragments of length 10.]{Performance of FMTree based on M-tree on a dataset of fragments of length 10. Average (median) and worst case results for 100 random queries are shown. Error bars show the interquartile range.}\label{fig:FMTNNperf}
\end{figure}

The dataset in this experiment was the set of 1,753,832 unique fragments  fragments of length 10 obtained from a 5000 protein sequence random sample taken from SwissProt (Release 41.21). An FMTree was generated for BLOSUM62 $\ell_1$-type quasi-metric at a cost of 34,142,940 distance computations. Figure \ref{fig:FMTNNperf} shows the results based on 100 random queries (unfortunately, mostly due to I/O costs, each search took over 1 minute and it was necessary to use a smaller number of runs).

\subsubsection{Suffix arrays and mvp-tree}

Table \ref{tbl:FScomp} presents the results of comparisons between FSIndex ($k$NN and range search algorithm), suffix array and mvp-tree over the datasets of fragments of length 6 and 9 from \texttt{nr018K}. The similarity measure used was the associated metric to the BLOSUM62 $\ell_1$-type quasi-metric because mvp-tree is a metric access method and the performance of FSIndex does not much differ if a quasi-metric is replaced by its associated metric. If the mvp-tree showed good performance on metric workloads, the next step would be to split the datasets into fibres to create an FMTree for quasi-metric searches.

Instances of suffix array were constructed using the routines published at \url{http://www.cs.dartmouth.edu/~doug/sarray/}. The search algorithm was identical to the Algorithm \ref{alg:FSprocessbin} where the input is a single bin containing all fragments in the dataset. In order to construct an instance of mvp-tree, duplicate fragments in the datasets were collected together and the sets of unique fragments provided to the mvp-tree construction algorithm. The mvp-tree implementation, developed by the original authors of mvp-tree \cite{BoOzs97}, was kindly provided by Marco Patella and modified for use with protein fragments by the thesis author. The maximum size of a leaf node was set to be 5.

\begin{table}[!ht]
\begin{center}
{\small
\begin{tabular}{|c|r||>{\raggedleft}m{1.5cm}|>{\raggedleft}m{1.5cm}|r|r|}
\hline
Length & \centering Neighbours & \centering FSIndex ($k$NN) & \centering FSIndex (range) & \centering Suffix array & {\centering mvp-tree}  \\ \hline\hline
6  &  1 &   15.0 &    9.9 &  20130.7 &    7598.5 \\ \hline
6  & 10 &   12.1 &    7.1 &   3761.1 &    6229.5 \\ \hline
9  &  1 & 1869.7 & 1303.6 &  72351.1 & 1016181.1 \\ \hline
9  & 10 &  902.6 &  615.4 &  14827.2 &  214032.5 \\ \hline
\end{tabular}
}
\end{center}
\caption[Comparison of performance of FSIndex, suffix array and mvpt-tree.]{Comparison of performance of FSIndex, suffix array and mvpt-tree. The table shows the values of the effective access overhead, that is the number of characters (residues) accessed in order to retrieve a given number of nearest neighbours, normalised by the fragment length and the number of retrieved neighbours. The statistics are in terms of characters rather than data points because suffix array search algorithm passes by each point but only computes the distances if necessary.}\label{tbl:FScomp}
\end{table}

\section{Discussion}

While the experiments presented in Section \ref{sec:FSresults} covered very few datasets and a small proportion of possible parameters for FSIndex creation, it can still be observed that FSIndex performed well. Not only did it perform much better than the other indexing schemes tested but it has proven itself to be very usable in practice: it does not take too much space (5 bytes per residue in the original sequence dataset plus a fixed overhead of the $bin$ array), considerably accelerates common similarity queries and the same index can be used for multiple similarity measures without significant loss of performance. The remainder of the current section will examine some salient features of the experimental results.

\subsection{Power laws and dimensionality}

The most striking feature of the Figures \ref{fig:FSperf06}, \ref{fig:FSperf09} and \ref{fig:FSperf12} is the apparent power-law dependence of the total running time, the number of bins scanned and number of bins scanned on the number of actual neighbours retrieved, manifesting as straight lines on the corresponding graphs on log-log scale. For each index, the slopes of of the three graphs (i.e. running time, bins scanned and fragments scanned) are very close, implying that the same power law governs the dependence of all three variables on the number of neighbours retrieved. The exponents are 0.81 for length 6, between 0.57 and 0.63 for length 9, and about 0.45 for length 12. While a rigorous theory, especially in the context of quasi-metrics, is still missing, it is possible to offer an intuitive explanation for this phenomenon.

Clearly, the graphs in question show the average growth of a ball in the projection $\pi(\Sigma^m)$ against the growth of a ball same radius in the original space $\Sigma^m$. Denote by $k$ the number of true neighbours retrieved and by $V(k)$ the corresponding number of fragments scanned. The power relationship then can be written as $V(k) = O(k^{D_1})$. If we accept the reasoning behind the distance exponent (not obvious from the data and not justified except for very small radii -- see Appendix \ref{app:distexp}), that is that $k = O(r^{D_2})$ where $D_2$ is the `dimension' of the space, it follows that $V(r)=O(r^{D_1D_2})$. Using the same reasoning about the size of the ball in the projection (but note that the distance in the projection need not satisfy the triangle inequality), we conclude that the `dimension' of the projection is $D_1D_2$, that is, the original dimension $D_2$ is reduced by a factor $D_1$. Assuming that the values of the distance exponent do not depend on whether a quasi-metric or its associated metric is used and taking the values of distance exponent estimated in Subsection \ref{subsec:distexp}, the `dimension' of the projected space is close to 6.5 for both length 6 and length 9.

\subsection{Effect of subindexing of bins}

PATRICIA-like subindexing of bins was introduced in order to accelerate scanning of bins containing many duplicate or highly similar fragments. Figures \ref{fig:FSperf06}, \ref{fig:FSperf09}, \ref{fig:FSperf12} and \ref{fig:FSscale} (Subfigure (e) in each case) show that there are two main factors influencing the proportion of residues scanned out of the total number of residues in the fragments belonging to the bins needed to be scanned: the (average) size of bins and the number of alphabet partitions at starting positions. Instances of FSIndex having many partitions at first few positions perform well (\texttt{SPEQ06}, \texttt{SPNA09}), those that have few partitions with many letters per partition, less so. 

Clearly, if a bin has a single letter partition at its first position, the distance at that position need be only retrieved once, at the start of the scan, independently of the number of fragments the bin contains. The effects for the second and subsequent positions are less prominent, if only for the reason that using many partitions would result in many bins being empty. The actual composition of the dataset is also important, as Figure \ref{fig:FSscale} (e) attests: although same partitions are used and \texttt{nr0288K} is almost twice as large, \texttt{SPEQ09} scans fewer characters. The possible reason lies in the nature of SwissProt, which, as a human curated database, is biased towards the well-researched sequences which are more related among themselves while not necessarily being representative of the set of all known proteins. On the other hand, \texttt{nr0288K} is a random sample from the \texttt{nr} database which is exactly the non-redundant set of all known proteins. 

The actual proportion varies from 30\% (\texttt{SPEQ06}, length 6) to over 85\% (\texttt{nr018K}, length 9). The percentage of characters scanned grows slowly with increase of the number of neighbours retrieved -- most probably this is because the number of bins accessed also grows, requiring that at least one full sequence is scanned.

To summarise, subindexing of bins does produce some savings, the exact amount depending on the dataset and alphabet partitioning. However, and this is further attested by poor performance of pure suffix array compared to FSIndex (Table \ref{tbl:FScomp}), the good performance of FSIndex is mostly due to alphabet partitioning.

\subsection{Effect of similarity measures}

Table \ref{tbl:FSmatdata} indicates very little difference in performance of the same instance of FSIndex with respect to different similarity measures. This should not be a surprise because the BLOSUM matrices are indeed very similar, modelling the same phenomenon in slightly different ways but generally retaining the same groupings of amino acids. The PSSM-based searches also performed well, mainly because the PSSMs are usually constructed out of sets of sequences that are strongly conserved at least in one or two positions, and hence, in those positions, the `distances' to all other clusters are so large that many branches of the implicit search tree can be pruned. 

\subsection{Scalability}

Figure \ref{fig:FSscale} (b) indicates that FSIndex is scalable with respect to the number of nearest neighbours retrieved -- the number of residues needed to be scanned grows sublinearly with dataset size (in fact, the exponent is 0.25 to 0.3). The exponent for the growth of the number of scanned points (graphs not shown in any figure) is about 0.4, indicating that using PATRICIA-like structure improves scalability. The principal reason for sublinear growth of the number of items needed to be scanned is definitely that search radius decreases with dataset size (Figure \ref{fig:FSscale} (a)). Unfortunately, the results in terms of search radius are not available and it is not possible to examine the scalability with respect to a fixed radius although theoretical considerations imply that the growth would be linear. However, it may be that subindexing of bins would bring an appreciable sublinear behaviour in this case as well.

\subsection{Comparison with other indexing schemes}

Results of Subsection \ref{subsec:FScomp} indicate that FSIndex decisively outperforms all other indexing schemes considered. M-tree performed the worst, needing to scan 1.3 million fragments of length 10 in order to retrieve the nearest neighbour. The performance of mvp-tree is not much better, taking into account the dimensionality: it requires scanning about 1 million fragments of length 9 to retrieve the nearest neighbour. Suffix array was generally performing better than mvp-tree, except for retrieving the nearest neighbour of length 6.

In the case of suffix arrays, it is clear that large alphabet and relatively small dataset (Figure \ref{fig:uniquefrags}) are responsible for relatively poor performance. Also note that suffix trees (and hence suffix arrays) generally are not good approximations of the geometry with respect to $\ell_1$-type distances -- two fragments lacking a common prefix may have a small distance. It should be noted that performance of suffix array based scheme appears to improve with fragment length compared to FSIndex.

The poor performance of M-tree and mvp-tree is somewhat surprising because Mao, Xu, Singh and Miranker \cite{MXSM03} have recently proposed using exactly M-tree for fragment similarity searches. However, on closer inspection, several differences appear. First, Mao, Xu, Singh and Miranker use a different metric. More importantly, they use a significantly improved M-tree creation algorithms. Finally, if their results are compared with those from Figure \ref{fig:FMTNNperf} (this can be done at least approximately because the same fragment length was used and the size of the yeast proteome dataset used in \cite{MXSM03} was very close to the size of SwissProt sample used in our experiment), it appears that there is no more than 10-fold improvement. While this is quite significant, the total performance appears still worse than that of FSIndex. For more detailed comparisons it would be necessary to obtain the code of the improved M-tree from \cite{MXSM03} and run a full suite of comparison experiments.


\chapter{Biological Applications}\label{ch:5}

The present chapter introduces the prototype of the \emph{PFMFind} method for identifying potential short motifs within protein sequences. PFMFind uses the FSindex access method to query datasets of protein fragments.

\section{Introduction}\label{sec:biointro} 

Most of the widely used sequence-based techniques for protein motif detection depend on regular expressions (deterministic patterns) \cite{Smith90,BJEG98}, profiles (PSSMs) \cite{Gribskov:1987,altschul97gapped} or profile hidden Markov models \cite{KBMSH94,Eddy98}. As outlined in Chapter \ref{ch:bioseq_qm}, a PSSM is constructed by taking a set of protein fragments,\footnote{Fragments are usually used rather than full sequences because the motifs are associated with domains, which are by their nature local.} constructing a multiple alignment, estimating the positional distributions of amino acids and producing positional log-odds scores for each amino acid. A PSSM can then be used to search a sequence dataset in order to identify new sequences fitting the profile (that is, its underlying positional distribution). This procedure can be performed iteratively, using sequences retrieved in one iteration to construct a profile for the subsequent one. Profile hidden Markov models generalise profiles by also modelling the distributions of gaps found in the multiple alignments (see Chapter 5 of the book by Durbin \emph{et al.} \cite{Durbin:1998}).

The initial set of sequences consists of known examples of the motif in question. It can be obtained from results of laboratory investigations, from alignments of structures (for example using the SCOP database \cite{Murzin:2004}) or from results of sequence similarity searches. PSI-BLAST \cite{altschul97gapped} uses the latter approach: it searches a protein dataset using a score matrix such as BLOSUM62 and uses the results to construct a multiple alignment and produce a profile for the second iteration. Subsequent searches are based on profiles constructed from the results retrieved in the preceding iteration. Variations to this basic approach are possible, mostly involving the choice of dataset and weights of sequences used for profile construction \cite{RJLG00}. The performance of any particular technique is measured by its ability to retrieve relevant items from the database (sensitivity) and to retrieve only such items (selectivity).

The focus of the present investigation is short protein fragments of lengths 7--15 with the aim to develop new bioinformatic tools for discovery of relationships between protein fragments that cannot be necessarily found when considering longer fragments. Such relationships need not imply a common ancestor but could have arisen from convergence. The motifs discovered should correspond to a conserved function and should give an insight into a possible origin of such a function. 

Watt and Doyle \cite{WD05} recently observed that BLAST is not suitable for identifying shorter sequences with particular constraints and proposed a pattern search tool to find DNA or protein fragments matching exactly a given sequence or a pattern\footnote{A ``pattern" in the sense of Watt and Doyle is a group of ``target sequences", which are essentially regular expressions.} I propose here an alternative technique, named \emph{PFMFind} (PFM stands for Protein Fragment Motif) that involves the use of full similarity search with almost arbitrary scoring schemes and iterated searches closely resembling PSI-BLAST. It differs from PSI-BLAST in that it uses a global ungapped similarity measure over the fragments of fixed length (referred to as an $\ell_1$-type sum in the Chapter \ref{ch:bioseq_qm}) allowing use of FSindex as a subroutine. The similarity score being ungapped could affect sensitivity but one should note that gapped alignments of short fragments, at least of lengths not greater than 10, are often statistically insignificant if the usual gap penalties are used (for example, BLAST uses 11 as gap opening penalty, which is larger than the cost of any single substitution -- in fact two to three conservative substitutions can be usually had for that cost, depending on the exact score matrix). It is also possible to examine several fragment lengths thus compensating for the similarity being global rather than local. Of particular biological interest are cases where certain relationships can be found at a particular fragment length and not the others indicating a strongly conserved short motif that cannot be extended to a longer one.

The present chapter contains the description of the current PFMFind algorithm together with six case studies based on SwissProt \cite{Boeckmann2003} query sequences. The query sequences (SwissProt accessions in brackets) are: prion protein 1 precursor (PrP) (P10279), $\beta$-casein precursor (P02666), $\kappa$-casein precursor (P02668), $\beta$-lactoglobulin precursor (P02754), cytochrome P450 11A1  mitochondrial precursor (cholesterol side-chain cleavage enzyme) (P00189), and sensor-type histidine kinase prrB (Q10560). The first five sequences are bovine ({\it Bos taurus}) while the histidine kinase is from {\it Mycobacterium tuberculosis}.

The PrP protein is found in high quantity in the brain of humans and animals infected with transmissible spongiform encephalopathies (TSEs). These are degenerative neurological diseases such as kuru, Creutzfeldt-Jakob disease (CJD), Gerstmann-Straussler syndrome (GSS), scrapie, bovine spongiform encephalopathy (BSE) and transmissible mink encephalopathy (TME) \cite{Za04a,Za04b,PFSSSTG93,We04} that are caused by an infectious agent designated prion. While many aspects of the role of PrP in susceptibility to prions are known, its physiological role and the pathological mechanisms of neurodegeneration in prion diseases are still elusive \cite{FW04}. 

Caseins are major mammalian milk proteins involved in determination of the surface properties of the casein micelle which contain calcium and have major role in mammalian neonate nutrition \cite{MeBo99}. Bovine milk contains four different types of casein: $\alpha$-S1-, $\alpha$-S2-, $\beta$- and $\kappa$-. Caseins are expressed in mammary glands, secreted with milk and following digestion may give rise to bioactive peptides \cite{MeBo99}. 

$\beta$-Lactoglobulin is another major component of milk. It is the primary component of whey, binds retinol and unlike the caseins, has a well-defined conformation \cite{KHFEBG99} containing an eight-stranded continuous $\beta$-barrel and one major $\alpha$-helix.

Cytochromes P450s are a superfamily of heme-containing enzymes involved in metabolism of drugs, foreign chemicals, arachidonic acid, eicosanoids, and cholesterol, synthesis of bile-acid, steroids and vitamin D3, retinoic acid hydroxylation and many still unidentified cellular processes \cite{NeRu02}. The cytochrome P450 A11 is a mitochondrial,  enzyme coded by the CYP11A1 gene and catalyses a cholesterol side cleavage chain reaction \cite{HHGC04}.
 
Histidine kinases phosphorylate their substrates on histidine residues and have been well-characterised in bacteria, yeast and plants \cite{WTS02}, with a variety of functions including chemotaxis and quorum sensing in bacteria and hormone-dependent developmental processes in eukaryotes. They are also present in mammals \cite{BTA03}. Typically, histidine protein kinases are transmembrane receptors with an amino-terminal extracellular sensing domain and a carboxy-terminal cytosolic signaling domain and do not show significant similarity to serine/threonine or tyrosine protein kinases although they might be distantly related \cite{KLWRB00}. 

The query sequences were chosen mainly according to the interests of the author and his supervisors. For example, caseins have no known function apart from nutrition while being strongly conserved in mammals, leading to questions about their origins. Cytochromes P450 form a large and well-researched superfamily with many examples in SwissProt and TrEMBL, thus being particularly suitable for the PFMFind approach. Histidine kinases are a subset of the class of protein kinases while being very distantly related to the remainder of the class. PrPs are involved a well-publicised set of neurological diseases and have a relatively unusual structure of aromatic-glycine tandem repeats \cite{Ga02}.

\section{Methods}

\subsection{General overview}

PFMFind takes a full sequence of interest and divides it into all overlapping fragments of a given fixed length. For each fragment, it uses FSindex-based range search to find the set of statistically significant neighbours from a protein fragment dataset with respect to a general similarity scoring matrix such as BLOSUM62. All fragments that have fewer significant neighbours than a given threshold are excluded from further iterations. For each fragment where the number of significant results is sufficiently large, it constructs a PSSM from the results and proceeds with the next iteration. The procedure is repeated several times, each time using the results of one iteration, if their number is over the threshold, to construct the profile for the next search. 

As in PSI-BLAST, the measure of statistical significance is E-value, the expected number of fragments similar to a given query fragment under the assumption that amino acids in a protein fragment are independently and identically distributed. Subsection \ref{subsec:statsign} below describes the derivation and computation of the distribution of similarity scores with respect to a given query fragment and similarity measure. The E-value threshold decreases with iterations. This is because preliminary investigations have shown that too few results of the initial, general score matrix-based search, are significant under the model from Subsection \ref{subsec:statsign} at a level usually set in bioinformatics applications of a similar kind (for example, in PSI-BLAST, the inclusion threshold E-value is 0.005) while the hits having E-value up to 1.0 clearly belonged to the same protein (in a different species) as the query protein. In the iterations using profiles, more stringent significance levels have led to expected results.

\subsection{PSSM construction}

Since the fragment length is fixed, a collection of fragments directly corresponds to an ungapped multiple alignment. Therefore, the first nontrivial step is assigning a weight to each sequence in order to compensate the possible bias of the set of hits caused by over- and under- representation of a particular sequence. While each sequence is assigned a new weight, the total weight of the fragment set remains the original number of hits. The current version of PFMFind uses the weighting scheme proposed by Henikoff and Henikoff \cite{Henikoff:1994}, which gives smaller weight to well-represented sequences and is computationally simple. The second step involves obtaining the `observed' (given the weights) frequencies of amino acids at each position and combining them with mixtures of Dirichlet priors in a way described by Sj\"olander and others \cite{SKBHKMH96} (see also Chapter 5 of \cite{Durbin:1998}). The contribution of Dirichlet priors decreases with sample size, preventing overfitting the profile to a small sample while leaving the distribution derived from a large set essentially unchanged. Finally, the procedure calculates log-odds similarity scores to be used for searches. The scores are multiplied by two (that is, scaled to half-bit units) and converted to integers, enabling direct comparison with the BLOSUM62 scores which are also in half-bit units.

\subsection{Statistical significance of search results}\label{subsec:statsign}

To evaluate the statistical significance of a particular similarity score and therefore an alignment associated with it, we estimate how probable that score is given a null, or background hypothesis. In this case, we assume as a null hypothesis that fragments are generated by the independent, identically distributed process where the probability of each amino acid is given by its relative frequency in the dataset (Subsection \ref{subsec:randseq} discusses this and an alternative model of protein sequences). Let $m$ be the fragment length. For each $i=0,1,\ldots, m-1$, let $S_i:\Sigma\to\R$ be the score function at position $i$. If the similarity measure is given by a score matrix $s:\Sigma\times\Sigma\to\R$, we have $S_i(a)=s(\omega_i,a)$ where $\omega= \omega_0\omega_1\ldots\omega_{m-1}$ is the query fragment and $a\in\Sigma$,while in the case of a PSSM $S_i$ is the score function at its $i$-th position.

By our assumptions, it is clear that $\{S_i\}_{i=0}^{m-1}$ is a collection of independent random variables and that the similarity score $S$ of a fragment $x$ is given by the sum of the values $S_i(x_i)$ for each $i$. Hence, the density of $S$, denoted by $f_S$ is given by the convolution of the densities $f_{S_i}$ of the random variables $S_i$, that is \[f_S = f_{S_0}* f_{S_1}*\ldots f_{S_{m-1}}\] where \[(f*g)(t) = \int f(\tau)g(t-\tau)d\tau.\] By the well-known Convolution Theorem, the Fourier transform of the convolution of a collection of functions is a product of their Fourier transforms. Since the functions in questions are discrete, the efficient way of computing $f_S$ is to compute the discrete Fourier transforms of $f_{S_i}$ for each $i$, multiply them together and take the inverse discrete Fourier transform of the product, all using the FFT (Fast Fourier Transform) algorithm (the book by Smith \cite{Sm03a} provides a good reference about signals, convolutions and Fourier Transforms) and is freely available on the web).

Once the density of similarity scores is obtained, it is straightforward to compute the p-value of each score $T$, that is the probability that a random score $X$ is greater than $T$. The number of fragments in the dataset expected by chance to be equal to or exceed $T$, also known as E-value, is obtained by multiplying the p-value by the size of the dataset. The relationships represented by the search hits where the E-value of the similarity score is very low (usually $<< 1$) are considered unlikely to have arisen by chance and therefore statistically significant. The significance cutoff can be computed prior to search so that search by E-value reduces to range search.

\subsection{Implementation}

PFMFind is implemented in the Python programming language \cite{vanRossumDrake2003a}, accessing the FSindex library, which is written in the C programming language \cite{Kernighan88a}, through the SWIG \cite{Beazley1996} interface. The PFMFind code uses the routines from the Python standard library \cite{Lundh2001} as well as from the Biopython \cite{Biopython}, Numeric \cite{Numpy} and Transcendental \cite{Transcendental} packages. 

Architecturally, PFMFind system consists of a master server, several slave servers and at least one client, all communicating through TCP/IP sockets. The master server handles computation of searches and statistical significance by distributing the load to slave servers while the client is responsible for storage of results and computation of profiles.\footnote{It is planned to move the profile construction to the server side as well leaving only the storage and interface to the client.} Python programs making use of PFMFind create an instance of a client, connect to a master server and provide the parameters of desired searches. A graphical user interface, called FragToolbox, was written using the Tkinter module \cite{Grayson2000} from the Python standard library in order to facilitate the analysis of the results by displaying them in a human-usable format.

The above configuration is necessary in order to use large datasets which cannot fit into memory of a single machine. It also opens the possibility of parallelisation of most of computation, leaving only storage and display to clients.

\subsection{Experimental parameters}

\subsubsection{Dataset}

Preliminary investigations using SwissProt as the database have shown that in most cases too few sequences are available in order to be able to construct good profiles even if the initial E-value is relaxed. While SwissProt is manually annotated and therefore provides most confidence in functional annotation, it is also biased in favour of well-researched sequences. I therefore decided to use the full Uniprot \cite{BAWBB05} dataset consisting of SwissProt together with TrEMBL (translated EMBL DNA sequence dataset). Since the size of Uniprot is large (Release 3.5 that was used together with alternative splicing forms of some proteins had 556,628,177 amino acid residues in 1,737,387 sequences), it was necessary to divide it into 12 SwissProt-sized parts and to run a PFMFind slave server for each part on a different machine.

\subsubsection{Search and profile construction parameters}

The cutoff E-values were 1.0 for the first and second, 0.1 for the third and fourth and 0.01 for all subsequent iterations. As preliminary investigations indicated that at E-value thresholds of 1.0 or smaller most BLOSUM matrices produce similar results, my choice was to use BLOSUM62 in the first iteration. Profile construction algorithm used the Dirichlet mixture {\tt\small recode3.20comp} downloaded from the web site  \url{http://www.cse.ucsc.edu/research/compbio/dirichlets/} of some of the authors of \cite{SKBHKMH96}. They recommend the \\ {\tt\small recode3.20comp} mixture as the best to be used with close homologs. After several trials I set the number of hits necessary to proceed with the next iteration to 30 as a compromise between the need to have as large number of hits as possible in order to have a good profile and the average number of neighbours given the required statistical significance.

\section{Results}

The full PFMFind algorithm was run for the six test sequences. Fragment lengths 8 to 15 were considered for all test proteins except PrP where only fragments of length 8 were considered because of technical limitations: too many hits were encountered and the available memory was insufficient to store all but the length 8 results (there were usually more than 100 hits for each overlapping fragment, sometimes over 1000 hits). The hits were almost exclusively exact matches to fragments of the query sequence or other prion proteins, in the same or different species. PrP is glycine rich and contains several repeats which manifested as several hits to the same protein in a single fragment search. 

The running time for searches for all the examples was in the order of one to two hours, using 12 Intel$^\circledR$ Pentium$^\circledR$ IV 2.8 GHz machines running in parallel, with indices optimised for lengths 10 and 12. Running FSindex did not take more than half of that time, the remainder being taken by calculation of statistical significance, construction of profiles, communication between machines and I/O operations.

Table \ref{tbl:results} provides the summary of the results for all examples except PrP. The `Region' column denotes the region of the original query sequence where significant hits to database proteins were found and usually refers to the maximal extent of such region for the longest fragment length where hits were found. The `Feature' column contains the annotations of the region in question taken from SwissProt and InterPro \cite{MAABB05}, a database of protein families, domains and functional sites consisting of several member databases using a variety of motif-finding techniques. The last column includes the description of the major categories of proteins found in the hits. Some of the $\kappa$-casein hits are not included because they were difficult to characterise (no SwissProt entry present).

{\scriptsize
\begin{longtable}{|p{1.5cm}|p{1.2cm}|p{2.4cm}|p{7cm}|}
\caption*{{\bfseries $\beta$-casein precursor [{\itshape Bos taurus}] (P02666)}}\\ \hline
\textbf{Region} & \textbf{Lengths} & \textbf{Feature} & \textbf{Major classes of hits}\\ \hline\hline
1--18 & 8--15 & signal peptide & $\alpha$-S1-, $\alpha$-S2-, $\beta$-, $\gamma$-, $\epsilon$- casein, amelogenin (only 4--18) (all hits to signal peptide region); \\ \hline
3--15 & 11 & signal peptide (potential) & vitellogenin (signal peptide)\\ \hline
3--17 & 12--13, 15 & transmembrane (potential) & cation-, heavy metal- transporting ATPase \\ \hline
3--14 & 11--12  &  & cytochrome b \\ \hline
158--173, 182--200 & 12--15 & & proline, glutamine and alanine rich fragments from various proteins, repeats \\ \hline 
\multicolumn{4}{c}{} \\[0.2cm]

\caption*{{\bfseries $\kappa$-casein precursor [{\itshape Bos taurus}] (P02668) }}\\ \hline
\textbf{Region} & \textbf{Lengths} & \textbf{Feature} & \textbf{Major classes of hits}\\ \hline\hline
30--191 & 8--15 & full mature \mbox{protein} & $\kappa$- casein \\ \hline
110--133 & 13--15 &  & histidine rich fragments from various proteins \\ \hline
139--166 & 13--15 &  & threonine rich fragments from various proteins \\ \hline
32--46 & 14--15 & & self-incompatibility ribonucleases \\ \hline
31--45 & 15 &  & myosin \\ \hline
174--188 & 15 & & {\it Kluyveromyces lactis} strain NRRL Y-1140 chromosome E (apparently a repeat)\\ \hline
80--95 & 12--15 & part of \mbox{casoxin B} & bacterial aldehyde dehydrogenase \\ \hline
55--67 & 13--14 & includes \mbox{casoxin A} & Erythrocyte membrane protein ({\it Plasmodium falciparum})\\ \hline
51--63 & 13 & includes \mbox{casoxin A} & extracellular region of bacterial regulatory protein blaR1 \\ \hline
155--167 & 13 &  & bacterial sulfate adenylyltransferase\\ \hline
\multicolumn{4}{c}{} \\[2.8cm] 

\caption*{{\bfseries $\beta$-lactoglobulin precursor [{\itshape Bos taurus}] (P02754)}}\\ \hline
\textbf{Region} & \textbf{Lengths} & \textbf{Feature} & \textbf{Major classes of hits}\\ \hline\hline
25--39 & 12--15 & turn, helix, strand & $\beta$-lactoglobulin, outer membrane lipoproteins, plasma retinol-binding protein, \mbox{glycodelin}, recA, SbnH (length 12 only) \\ \hline
54--68 & 14--15 & turn, strand, turn & $\beta$-lactoglobulin, \mbox{glycodelin} \\ \hline
58--72 & 14--15 & strand, turn, strand (part) & glucose-1-phosphate thymidylyltransferases, $\beta$-lactoglobulin\\ \hline
110--124 & 14 & strand & $\beta$-lactoglobulin, glycodelin, bacterial DNA methylase\\ \hline
\multicolumn{4}{c}{} \\[0.2cm] 

\caption*{{\bfseries Cytochrome P450 A11 mitochondrial precursor [{\itshape Bos taurus}] (P00189)}} \\ \hline
\textbf{Region} & \textbf{Lengths} & \textbf{Feature} & \textbf{Major classes of hits}\\ \hline\hline
77--86 & 9--10 & turns & cytochrome P450 11A1, formyltetrahydrofolate synthetase \\ \hline
85--99 & 12,15 & turn, helix, turn, helix & various cytochromes P450 \\ \hline
119--135 & 13--15 & contains a turn & cytochrome P450 (11A1 and 11B2), serine/threonine-protein kinases Pim-2 and Pim-3 (kinase domain, length 14), transposase (lengths 13--14), various other proteins \\ \hline
260--273 & 12--14 & helix &  cytochromes P450 (mostly 11A1 and 11B2) \\ \hline
311--343 & 11,13--15 & helix, turn, helix & various cytochromes P450 (few hits at length 14) \\ \hline
343--356 & 14 & helix & cytochrome P450 11A1\\ \hline
370--396 & 9--15 & turn, helix, strand & various cytochromes P450 \\ \hline
398--442 & 9--15 & strand, turn, strand, turn, strand, helix, turn, turn & various cytochromes P450 (Note: only few fragments in this region have hits at shorter lengths) \\ \hline
448--483 & 9--15 & turn, turn, helix, turn, turn; heme binding site & various cytochromes P450 \\ \hline
\multicolumn{4}{c}{} \\[0.2cm] 

\caption*{{\bfseries Sensor-type histidine kinase prrB [{\itshape Mycobacterium tuberculosis}] (Q10560)}}\\ \hline
\textbf{Region} & \textbf{Lengths} & \textbf{Feature} & \textbf{Major classes of hits}\\ \hline\hline
230--257 & 9--15 & histidine kinase domain, contains phopshohistidine & various histidine kinases, sensory proteins, ethylene receptor\\ \hline 
373--398 & 11--15 & histidine kinase domain & various histidine kinases, DNA topoisomerase, gyrase, other proteins\\ \hline
400--425 & 10--15 & histidine kinase domain & various histidine kinases, ethylene receptor (cystein synthase and tripeptide permease appear in hits for one fragment of lengths 10--11 in this region)\\ \hline 
\caption{Significant hits to query fragments.}\label{tbl:results}
\end{longtable}
}

\section{Discussion}

Two kinds of hits can be observed in general: hits to the query protein itself and its very close homologs and hits to low-complexity regions of arbitrary proteins. There were also few hits to fragments of apparently unrelated proteins which were not low-complexity.

\subsection{Hits to close homologs}

Most commonly found hits, apart from the low-complexity fragments, were to the instances of the same protein in a variety of species and to its close homologs. The hits were concentrated in the regions where sufficiently many strongly conserved examples existed. In histidine kinases, the hits are found in the histidine kinase domain, more specifically, according to InterPro, in the His Kinase A (phosphoacceptor) subdomain (230--257) and the ATPase domain (373--398, 400--425). PFMFind identified DNA gyrase (a bacterial DNA repair enzyme) as being associated with the (373--398) region, which is also confirmed by InterPro. Hence, in the histidine kinase example, PFMFind retrieved strongly conserved, functionally important regions, agreeing with the established methods.

In the case of $\beta$-casein, PFMFind identified a single region corresponding to the signal peptide whose role is to target the protein to a particular cellular compartment or, as in this case, to be secreted. The hits were to signal sequences of other caseins and other secreted proteins (amelogenin, having a role in biomineralisation of teeth and vitellogenin, a major yolk protein). No hits were found in the mature protein segment (mature protein is the precursor from which the signal peptide and potentially other parts have been cleaved), mainly because the initial hits were only to the other $\beta$-casein instances of which there were not sufficiently many to proceed to the next iteration. Apart from these, there were also hits to low complexity and transmembrane regions of clearly unrelated proteins.

In the case of $\kappa$-casein, the majority of hits were to other $\kappa$-caseins, the remainder being to low complexity regions. The only difference from the $\beta$-casein case is that Uniprot apparently contains more $\kappa$-casein sequences (that is, more than the minimum number necessary to proceed to the next iteration) so that PFMFind obtained the hits over most of the length of the protein. In the $\beta$-lactoglobulin, PFMFind found hits to $\beta$-lactoglobulin itself and its close relatives (glycodelin, a pregnancy associated protein and other members of lipocalin family) as well as to some apparently unrelated proteins such as bacterial RecA (DNA recombination enzyme) and SbnH (polyamine biosynthesis). However, under closer scrutiny, it appears that at least the SbnH fragment has been identified to belong to the lipocalin domain (ProSite \cite{FPBPHSH02} reference PS00213) together with $\beta$-lactoglobulin and glycodelin. All regions in $\beta$-lactoglobulin corresponded to identified elements of secondary structure.

Cytochromes P450 are well represented both in SwissProt and in TrEMBL, providing sufficient amount of examples to produce good profiles. Unlike with $\kappa$-casein, it appears that only truly conserved regions were identified. Most hits were to the other cytochromes P450 (but not always to all members of superfamily -- sometimes only very closely related cytochromes are retrieved) with the exception of the regions associated with turns. 

\subsection{Low complexity regions and repeats}\label{subsect:lowcomplex}

Many of the significant hits retrieved by PFMFind were to low-complexity fragments, for example consisting all of proline or glutamine or histidine. Such fragments are much more common than would be expected from their amino acid compositions, at least in eukaryotes \cite{Golding99} and frequently present problems for similarity searches. It is important to note that whenever low complexity regions are hit, the profile `diverges' from the seed: the original sequence becomes no longer significant (or at least not most significant) and the profile describes a totally different target. This is mainly because of compositional bias of the results where there are too many `undesirable' hits which `take over' the profile for a subsequent iteration. Even though the algorithm uses Dirichlet mixtures to smooth the positional distributions, it can be swamped by the large amounts of apparently genuine hits. The same issue is evident where transmembrane domains, which are strongly hydrophobic and not associated with any specific function, are hit (for example, region 3--14 in $\beta$-casein). 

The problem with low-complexity segments has been recognised and several tools that identify and filter out such regions exist \cite{WoFe96,Wise01}. In BLAST, the default option is for all low-complexity segments to be masked prior to search. However, some  low-complexity regions may be biologically significant -- for example, some bioactive peptides could be classified as low-complexity. A different way to avoid the effect of compositional bias is to use Z-score statistic based on the distribution of scores of the fragments having the same composition as a given hit but different order of amino acids \cite{WeBa01}. While this approach is commonly taken where global alignments are used, it fails to give sufficiently many sufficiently significant fragments of short lengths (datasets are too large and $n!$ is too small for small $n$).  

Hence, it appears that selective filtering of low-complexity hits is necessary. Highly compositionally biased fragments of query sequences should be filtered prior to search. Other fragments should be filtered at profile construction time, if computationally feasible. The aim should be to retain as many of the results while ensuring that the profile does not diverge. One of the reasons for appearance of low-complexity fragments within the results is the relaxed significance requirements for the first few iterations but one should take care in that respect because genuine hits also have low significance at first.

The PrP searches have revealed a further weakness of the current PFMFind algorithm and implementation. Most of the PrP hits were to the sequence itself and its very close, almost identical homologs. While the numbers of such sequences are not too large, the structure of the PrP itself, containing many aromatic-glycine tandem repeats was responsible for very large result sets: every PrP homolog appeared several times (in a different region) as a hit for a single fragment. This made it impossible to proceed because the current implementation of PFMFind stores all results in main memory. The problem should be rectified by better filtering/weighting of hits and storage of results on disk, to be retrieved as needed. 

\subsection{Issues with algorithm and implementation}

A major issue that dominated all examples of PFMFind searches presented here was the non-homogeneity of the database. Some proteins are extremely well represented, containing instances from a variety of species, some are very rare while others have multiple instances from few species. Subsection \ref{subsect:lowcomplex} discussed the problems arising from low-complexity fragments. However, $\kappa$-casein case has shown that too many instances of the same protein can also present difficulties at least due to overfitting. Weighting of hits prior to profile construction is clearly a solution but it is necessary to use weighting that could lower the total weight instead of just redistributing it. An even better approach would be to use other information (structure, function, domains) contained in the databases as well as sequence information. However, the quality of annotations varies considerably and this would present an implementation challenge because it would require full access to annotated databases by the PFMFind algorithm. 

PFMFind would also benefit from access to biological information because of general low significance of short fragment hits under the current statistical model. A Bayesian model, including the prior information available as annotation, could be more appropriate, provided that sufficient data is available. One must note however, that any increase in complexity of profile construction algorithm would affect the running time. Already, except in rare cases, similarity search does not take the most of the running time of PFMFind. This can of course be attributed to the good performance of FSindex.

\section{Conclusion}

The six examples have shown that PFMFind is able to identify the regions in the query sequence that are strongly conserved and functionally important in the closely related proteins as well as in some apparently unrelated proteins. The results also indicated that some sort of filtering of low-complexity hits and repeats is desirable. Several improvements to the algorithms and implementation are necessary before large-scale experiments can be conducted.
\chapter{Conclusions}\label{ch:conclusion}

The motivation for this thesis comes from the biological objective of developing the methods for discovering the origin and function of short peptide fragments with conserved sequence. While most of the current approaches to protein sequence analysis consider either full sequences or longer domains, short fragments have significant biological importance on their own. For example, there are several peptide fragments in various milk proteins that are cleaved during digestion and have possible physiological activity. Other peptides, from completely unrelated organisms, may have the same activity. Hence, from a biological point of view, it would be very useful to have the tools to discover the relationships between short fragments that do not necessarily extend to whole proteins.

As in the analysis of the longer sequences, the primary technique used to relate the short fragments is similarity search: we find similar fragments to a given query fragment and associate the function of the search results of the known function to it. The existing methods such as BLAST proved inadequate, primarily for reasons concerning computational efficiency -- they were too slow for the large number of searches that were considered necessary. Hence the need to construct an efficient index for similarity search in short peptide fragments that would speed up the retrieval of queries.

Indexing a dataset in an efficient manner is only possible through a good understanding of the geometric properties of the similarity measure on it. While most existing indexing techniques assume that the similarity measure is given by a metric, that is, a distance function, this is not the case for biological sequences where the similarity measures are generally given by similarity scores. The principal reasons for using similarity scores in biology are that they have fewer constraints and have information-theoretic and statistical interpretations. For our work, as a similarity measure, we have chosen the one given by the ungapped global alignment between fragments of fixed length because we believe that gaps do not have major importance in the context of short fragments.

One of the important results of the thesis is the discovery that many of the widely used BLOSUM similarity score matrices, restricted to the standard amino acid alphabet, can be converted into weightable quasi-metrics (metrics without the symmetry axiom), which generate the same range queries as the original similarity scores.

This in turn lead to the following questions:
\begin{enumerate}[(i)]
\item What is known about the quasi-metrics and what are the principal examples?
\item Can the results from asymptotic geometric analysis be extended to quasi-metric spaces with measure and applied to the theory of indexing for similarity search?
\item Can some insights from the theory of quasi-metrics be used to build an efficient indexing scheme for short peptide fragments that can be applied towards answering the original biological problem?
\item Does the relationship between similarities and quasi-metrics on the alphabet extend to \emph{local} (Smith-Waterman) alignments between full sequences?
\end{enumerate}

Chapter \ref{ch:1} answers the first question above. Quasi-metrics generalise both metrics and partial orders and are well known in topology and theoretical computer science. The main motif that is encountered with quasi-metrics is duality: the interplay between the quasi-metric, its conjugate and their join, the associated metric. The novel contribution of the Chapter \ref{ch:1} is the construction of the universal bicomplete separable quasi-metric space $\V$. This space is an analog of the well-known Urysohn metric space and is universal, ultrahomogeneous and unique up to isometry. The main motivation for constructing such space was to provide a previously unknown example of a quasi-metric space and to lay foundations for future work. In particular, the universality property means that all bicomplete separable quasi-metric spaces can be studied as subspaces of $\V$.

The second question is considered in Chapters \ref{ch:2} and \ref{ch:3}. The main object introduced there is $pq$-space: a quasi-metric space with probability measure. The notion of concentration functions from asymptotic geometric analysis can be defined for $pq$-spaces in a way that emphasises duality -- instead of one concentration function, we have two: left and right. The main theoretical result of Chapter \ref{ch:2} is that a `high-dimensional' quasi-metric space is very close to being a metric space -- in other words, that asymmetry is being lost with concentration. In the context of the theory of similarity search, the thesis extends the theoretical framework for indexing metric spaces to quasi-metric spaces by introducing the concept of a quasi-metric tree. Furthermore, the developments from Chapter \ref{ch:2} are used to give bounds for performance of quasi-metric indexing schemes.

Chapters \ref{ch:4} and \ref{ch:5} give answer to the third question. FSIndex was developed as an indexing scheme for fragments of fixed length based on two principles: reduction of the amino acid alphabet based on biochemical properties of amino acids and combinatorial generation of neighbours in the space of reduced fragments. It uses distances to reduced sequences as certification functions and thus combines the insights from biochemistry and geometry, having significantly better performance than existing indexing schemes (by 1-2 orders of magnitude). In addition FSIndex can be also used for profile-based searches and as such provides the main component of PFMFind -- a system for retrieving short conserved motifs from protein sequences. The preliminary experimental results from Chapter \ref{ch:5} show that PFMFind is very good at identifying conserved regions but has some problems with fragments of low-complexity. FSIndex also offers useful insight into the nature of indexing in general.

The fourth question leads to what we consider as another important contribution of this thesis to bioinformatics and computational biology: the discovery of the relationships between local similarities and quasi-metrics in Chapter \ref{ch:bioseq_qm}, under the assumptions satisfied by the most widely used similarity score functions. The most significant aspect of this discovery is the triangle inequality property which could lead to novel applications to clustering and of course to indexing for similarity search.

\section{Directions for Future Work}

While the phenomenon of concentration of measure is well-researched for many classical objects of mathematics, the contribution of the Chapter \ref{ch:2} of this thesis and the corresponding paper in {\it Topology Proc.} \cite{AS2004} is only the beginning. Many non-trivial questions are opened by introducing asymmetry, that is, by replacing a metric by a quasi-metric. For example, it would be interesting to generalise Gromov's \cite{Gr99} metric between $mm$-spaces to $mq$-spaces and hence to obtain a framework for discussing convergence to an arbitrary $mq$-space, where concentration of measure is a particular case of convergence to a single point. Similarly, one would want to find out if Vershik's \cite{Ver01} relationships between $mm$-spaces, measures on sets of infinite matrices and Urysohn spaces, can be extended to $mq$-spaces. Finally, the task of constructing a universal quasi-metric space that is not bicomplete, as well as a universal quasi-metric space complete under different notions of completeness remains open.

Turning to indexing schemes for similarity search, while other factors play no doubt a significant role, the performance is principally determined by geometry. The main task ahead is to further adapt the concepts of abstract asymptotic geometric analysis to datasets, which are discrete but growing objects and to develop computational tools and techniques for predicting and improving performance. It is clear that due to the Curse of Dimensionality, indexing `high-dimensional' datasets gains nothing. However, it is a common perception that, in reality, useful datasets are never intrinsically high-dimensional. It remains a highly challenging geometric problem to formalise this perception, first in geometric terms, and subsequently algorithmic.

Unfortunately, many indexing schemes perform badly for datasets that cannot be said to be `high-dimensional' -- recall the performance of M-tree and mvp-tree for datasets of protein fragments -- and therefore, there is a lot of scope for improvements to existing algorithms and data structures. Another general observation, made apparent from experiences with FSIndex, is that additional knowledge of domain structure could be of significant help in developing an indexing scheme. 

FSIndex has shown its usability for searches of protein fragments. Another possible application that ought to be examined is as a subroutine of a full sequence search algorithm. The experiments using the preliminary versions of PFMFind have shown its significant potential for finding short conserved patterns in protein sequences. It remains however, to make further improvements in order to eliminate problems associated with low-complexity sequences.

The relationship between similarities and quasi-metrics also opens the possibility of characterising the global geometry of DNA or protein datasets directly, without resorting to projections or approximations. As quasi-metrics capture many important properties of biological sequences, it is an opinion of the thesis author that asymmetry should be cherished rather than avoided by symmetrisations.

A general conclusion from this work is that methods based on asymmetric distances and measures have a future in analysis of data, especially in bioinformatics and computational biology, and those applications, in turn, can provide directions for further mathematical research.

\appendix
\chapter{Distance Exponent}\label{app:distexp}

In this Appendix we outline some methods for estimating the dimensionality of datasets based on the distance exponent of Traina, Traina and Faloutsos \cite{TrainaTF99}. A more rigorous definition of distance exponent is introduced and the methods for estimating it are tested on some artificial datasets of known dimensions. 

\section{Basic Concepts}

We give a brief introduction to the Hausdorff and Minkowski fractal dimensions. All the definitions and results are from the book by Mattila \cite{Mattila95} and the reader should refer to it for more detailed treatment.

\begin{defin}
Let $X$ be a separable metric space. The $s$-dimensional \emph{Hausdorff measure}, denoted $\mathcal{H}^s$ is defined for any set $A\subset X$ by \[\mathcal{H}^s(A) = \lim_{\delta\downarrow 0}\mathcal{H}^s_\delta(A)\] where \[\mathcal{H}^s_\delta(A)=\inf\left\{\sum_i \diam(E_i)^s: A\subset\bigcup_i E_i,\diam(E_i)\leq\delta\right\}. \]
\end{defin}

It can be shown that $\mathcal{H}^s$ is a Borel regular measure. The measure $\mathcal{H}^0$ corresponds to the counting measure while $\mathcal{H}^1$ has an interpretation as a generalised length measure. In $\R^n$, $\mathcal{H}^n(\ball{x}{r})=(2r)^n$. 

\begin{defin}
The \emph{Hausdorff dimension} of a set $A\subset X$ is
\begin{equation*}
\begin{split}
\dim A =& \sup\{s:\mathcal{H}^s(A)>0\} = \sup\{s:\mathcal{H}^s(A)=\infty\}\\
= & \inf\{t:\mathcal{H}^t(A)<\infty\} = \inf\{t:\mathcal{H}^t(A)=0\}.
\end{split}
\end{equation*}
\end{defin}

The Hausdorff dimension has some desirable properties for the dimension namely:
\begin{itemize}
\item $\dim A\leq\dim B$ \qquad for all $A\subseteq B\subseteq X$,
\item $\dim \bigcup_{i=1}^\infty A_i = \sup_i \dim A_i$ \qquad for $A_i\subseteq X$, $i=1,2\ldots$, and
\item $\dim\R^n = n$.
\end{itemize}
Hence $0\leq\dim A\leq n$ for all $A\subseteq\R^n$.

\begin{defin}
Let $A$ be a non-empty bounded subset of $\R^n$. For $0<\e<\infty$, let $N(A,\e)$ be the smallest number of $\e$-balls needed to cover $A$: \[N(A,\e)=\min\left\{k:A\subseteq\bigcup_{i=1}^k \ball{x_i}{\e} \ \text{for some}\ x_i\in\R^n\right\}.\] The \emph{upper} and \emph{lower Minkowski dimensions} of $A$ are defined by \[\overline{\dim_M}A = \inf\{s:\limsup_{\e\downarrow 0} N(A,\e)\e^s=0\}\] and  \[\underline{\dim_M}A = \inf\{s:\liminf_{\e\downarrow 0} N(A,\e)\e^s=0\}.\] 
\end{defin}

It follows from the definitions that $\dim A\leq \underline{\dim_M}A\leq \overline{\dim_M}A\leq n$ and these inequalities can be strict. Equivalently,
\begin{align*}
\overline{\dim_M}A = & \limsup_{\e\downarrow 0}\frac{\log N(A,\e)}{\log (1/\e)},\\
\underline{\dim_M}A = & \liminf_{\e\downarrow 0}\frac{\log N(A,\e)}{\log (1/\e)}.
\end{align*}

The following theorem provides a motivation for considering the fractal dimension to be the exponent of the growth of the measure of a ball, at least in $\R^n$.

\begin{thm}[\cite{Sa91}]\label{thm:hausminkexp}
Let $A$ be a non-empty bounded subset of $\R^n$. Suppose there exists a Borel measure $\mu$ on
$\R^n$ and positive numbers $a$, $b$, $r_0$ and $s$ such that $0<\mu(A)\leq\mu(\R^n) < \infty$ and \[0<ar^s\leq \mu(B(x,r)) \leq br^s < \infty\] for all  $x\in A$ and $0<r\leq r_0$. Then $\dim A = \underline{\dim_M} A = \overline{\dim_M} A=s$, where $\dim A$ is the Hausdorff dimension and $\underline{\dim_M} A$ and $\overline{\dim_M}A$ are the lower and upper Minkowski dimensions of $A$. \qed 
\end{thm}

Traina, Traina and Faloutsos \cite{TrainaTF99} observed that the distributions of distances between points of many existing datasets follow a power law for small distances and proposed a concept of \emph{distance exponent} as an estimate of the fractal dimension of datasets. By their definition, the distance exponent is the slope of the linear part of the graph of the distance distribution function on the log-log scale. However, a more rigorous definition is necessary, because the power law is only an approximation and it is difficult to ascertain the exact bounds of the linear part. We define the distance exponent in the framework of pm-spaces.

\begin{defin}
Let $(\Omega,d,\mu)$ be a pm-space. Define $F:\R\to [0,1]$, the \emph{cumulative distance distribution function} of $(\Omega,d,\mu)$ by \[F(r) = \mu\otimes\mu(\{ (x,y)\in\Omega\times\Omega: d(x,y)\leq r\}).\]
\end{defin}

\begin{remark}\label{rem:avgball}
Clearly, $F(r)$ is the average measure of a closed ball of radius $r$. By Fubini's Theorem,
\begin{equation*}
\begin{split}
F(r) = & \mu\otimes\mu(\{ (x,y)\in\Omega\times\Omega: y\in\cball{x}{r}\})\\
= & \int_{x\in\Omega}\int_{y\in\cball{x}{r}} d\mu(y)d\mu(x)\\
= & \int_{x\in\Omega} \mu(\cball{x}{r})d\mu(x).\\
\end{split}
\end{equation*}
\end{remark}

\begin{defin}
Let $(\Omega,d,\mu)$ be a pm-space and $F$ its cumulative distance distribution function. The \emph{distance exponent}, denoted $\mathfrak{D}(\Omega,d,\mu)$, is defined by
\[\mathfrak{D}(\Omega,d,\mu) = \lim_{r\downarrow 0} \frac{\log F(r)}{\log r}.\]
\end{defin}

Note that the distance exponent need not be defined and that it makes sense only for the case where $\Omega$ is an infinite set and $\mu$ a continuous measure. Many existing workloads can be modelled in this way, with a domain a large infinite space and the dataset a finite sample according to some continuous measure (see the Section \ref{subseq:probdist}).

The exact relation between the distance exponent and fractal dimensions in general remains an open question -- indeed, our definition the Minkowski dimension applies only for $\R^n$. If a set $A\subset\R^n$ satisfies the conditions of the Theorem \ref{thm:hausminkexp}, then clearly  $0<ar^s\leq F(r)\leq br^s < \infty$ for $0<r\leq r_0$ and hence the distance exponent corresponds to the Hausdorff and Minkowski dimensions.

\section{Theoretical Examples}

Although it is usually difficult to derive a general distribution function of distances of points on a arbitrary manifold, it is sometimes possible to use the symmetry of specific objects and metrics to obtain the exact forms for their cumulative distance distribution functions.

Let $(M,\rho,P)$ be a pm-space where $M\subseteq\R^n$ and $f_X$ is the density function of the probability measure $P$. Suppose the metric $\rho$ on $M$ is induced by the norm $\norm{\cdot}$ on $\R^n$. Denote by $\B$ the unit ball with respect to $\norm{\cdot}$ (i.e. $\B=\{x\in\R^n:\norm{x}\leq 1\}$). Let $X$ and $Y$ be random variables taking values in $M$ according to $P$. Then the cumulative distance distribution function of $(M,\rho,P)$ is given by
\begin{equation}
\begin{split}
F(r) = & Pr(\norm{X-Y}\leq r) \\
= & Pr(X-Y \in r\B) \\
= & \int_{r\B} f_{X-Y} dP
\end{split}
\end{equation}
where $f_{X-Y}$ is the density function of differences $X-Y$. The integral above can be quite hard to evaluate in closed form but there are cases where this poses no problem. Two of such cases are provided for illustration.

\subsection{The cube $[0,1]^n$}\label{subseq:distexp_cube}

Consider the pm-space $(M,\rho,\mu)$ where $M$ is the unit cube $[0,1]^n$, $\rho$ is the $\ell_\infty$ metric (i.e. $\rho(x,y)=\max_{1\leq i\leq n} \abs{y_i-x_i}$) and $\mu$ is a uniform measure on $M$. The density function $f_X$ is given by
\begin{equation}
f_X(x) =
  \begin{cases}
  0& \text{if $x \notin [0,1]^n$}, \\
  1& \text{if $x  \in [0,1]^n$}.
  \end{cases}
\end{equation}
Observe that $f_X$ is a product of uniform distributions on $[0,1]$, that is:
\begin{equation}
\begin{aligned}
\label{eqn:proddist}
 f_X(x) = \prod_{i=1}^{p}
f_{X_i}(x_i), 
\end{aligned}
\qquad \text{where}\qquad 
\begin{aligned}
f_{X_i}(x_i) =
  \begin{cases}
  0& \text{if $x_i\notin [0,1]$}, \\
  1& \text{if $x_i\in [0,1]$}.
  \end{cases}
\end{aligned}
\end{equation}
Thus
\begin{equation*}
\begin{split}
f_{X-Y}(t) = & \prod_{i=1}^{n} f_{X_i-Y_i}(t_i) \\
= & \prod_{i=1}^{n} f_{X_i}*f_{-Y_i}(t_i) \\
= & \prod_{i=1}^{n} \int_{-\infty}^{\infty} f_{X_i}(\tau)f_{-Y_i}(t_i-\tau) d\tau\\
= & \prod_{i=1}^{n} \int_{-\infty}^{\infty}
f_{X_i}(\tau)f_{X_i}(\tau-t_i) d\tau
\quad \text{since} \quad f_{-Y_i}(y_i) = f_{Y_i}(-y_i)= f_{X_i}(-y_i) \\
= & \prod_{i=1}^{n} \int_{0}^{1} f_{X_i}(\tau-t_i) d\tau
\end{split}
\end{equation*}
Now if $g(u) = \int_0^1 f_{X_i}(\tau-u) d\tau$ then
$
g(u) =
  \begin{cases}
  1 + u& \text{if $u \in [-1,0]$},\\
  1 - u& \text{if $u \in [0,1]$},\\
  0&     \text{otherwise}.
  \end{cases}$\\
Remember that the unit ball with respect to the $\ell_\infty$ norm is $[-1,1]^n$ and therefore
\begin{equation*}
\begin{split}
F(r) = & Pr(\norm{X-Y}_{\infty} \leq r) \\
= & Pr(X-Y \in [-r,r]^n) \\
= & \int_{[-r,r]^n} f_{X-Y} dP \\
= & \int_{[-r,r]^n} \prod_{i=1}^n g(t_i) dt_i \\
= & \begin{cases}
      \prod_{i=1}^n 2\int_0^r (1-t_i) dt_i& \text{if $0 \leq r < 1$}, \\
      1& \text{if $r \geq 1$}.
      \end{cases} \\
= & \begin{cases}
      (2r-r^2)^n& \text{if $0 \leq r < 1$}, \\
      1& \text{if $r \geq 1$}.
      \end{cases} \\
\end{split}
\end{equation*}
It therefore follows that $\mathfrak{D}(\Omega,\rho,\mu) = n$ as expected.

\subsection{Multivariate normal distribution}

Now consider the pm-space $(M,\rho,\mu)$ where $M=\R^n$, $\rho$ is the $\ell_2$ metric (i.e. $\rho(x,y)=\sqrt{(y_i-x_i)^2}$) and $\mu$ is a multivariate Gaussian measure (normal distribution) on $\R^n$ with mean $0$ and variance 1 in all coordinate directions. The density function $f_X$ is given by
\begin{equation}
f_X(x) = \frac{1}{(\sqrt{2\pi})^p}\exp \left(-\frac{1}{2}\norm{x}^2 \right)\\
\end{equation}
Again, $f_X$ defines a product distribution as in the Equation (\ref{eqn:proddist}), where $f_{X_i}(x_i)=\prod_{i=1}^{n} \frac{1}{\sqrt{2\pi}}\exp \left(-\frac{x_i^2}{2} \right)$. Hence, we can use the fact that $f_{X_i}$ is an even function and a well-known result that the sum of two normal random variables is a normal random variable where the mean is the sum of means and the variance is the sum of variances of these random variables, to conclude that
\begin{align*}
f_{X-Y}(t) = & \prod_{i=1}^{n}\frac{1}{2\sqrt{\pi}}\exp\left(-\frac{t_i^2}{4}\right) \\
= & \frac{1}{(2\sqrt{\pi})^n}\exp \left(-\frac{1}{4}\norm{t}^2 \right) \\
\end{align*}
Let $g(t)=\frac{1}{(2\sqrt{\pi})} \exp \left(-\frac{t^2}{4}\right)$. Using the radial symmetry of $f_{X-Y}$ and the spherical coordinates,
\begin{equation*}
\begin{split}
F(r) = & P(\norm{X-Y}_2\leq r) \\
= & P(X-Y \in r\B^n) \quad \text{($\B^n$ is the Euclidean unit ball)} \\
= & \int_{r\B^n} f_{X-Y} dP \\
= & \int_{r\B^n}g(\norm{t}) dP \\
= & Vol(\B^n)\int_0^r t^{p-1}g(t) dt \\
= & \frac{2\pi^{n/2}}{\Gamma\left(\frac{n}{2}\right)} \int_0^r\frac{t^{n-1}}{(2\sqrt{\pi})^n}\exp \left(-\frac{t^2}{4} \right) dt\\
= & \frac{2}{\Gamma\left(\frac{n}{2}\right)} \int_0^{\frac{r}{2}} u^{n-1} \exp(-u^2) du
\end{split}
\end{equation*}
The above expression can be evaluated as power series. Let $H_n(r)=\int_0^r u^{n-1} \exp(-u^2) du$. Then
\begin{equation*}
\begin{split}
H_n(s) = & \left[ \frac{-u^{n-2}e^{-u^2}}{2}\right]_0^r + \frac{1}{2}\int_0^r (n-2)u^{n-3} \exp(-u^2) du \\
= & \frac{-r^{n-2}e^{-r^2}}{2} + \frac{n-2}{2} H_{n-2}(r)
\end{split}
\end{equation*}
The above recurrence relation can be solved for even and odd $n$ separately. If $n$ is even,
\begin{equation*}
\begin{split}
H_p(r) & = \frac{(n-2)(n-4)\ldots 4.2.H_2(r)}{2^{n/2-1}} \\
& \quad -\frac{1}{2}e^{-r^2}\left(r^{n-2} +\frac{n-2}{2}r^{n-4}+ \ldots        +\left(\frac{n}{2}-1\right)! \, r^2 \right) \\
& = \left(\frac{n}{2}-1\right)! \left(-\frac{1}{2}e^{-r^2}+\frac{1}{2}- \frac{1}{2}e^{-r^2}\sum_{k=1}^{n/2-1} \frac{r^{n-2k}}{\left(\frac{n}{2}-k\right)!} \right) \\
& = \frac{1}{2} e^{-r^2} \left(\frac{n}{2}-1\right)! \left(e^{r^2}- \sum_{k=0}^{n/2-1} \frac{r^{2k}}{k!} \right) \\
& = \frac{1}{2}\Gamma\left(\frac{n}{2}\right) e^{-r^2} \sum_{n/2}^{\infty}\frac{r^{2k}}{k!}.
\end{split}
\end{equation*}
If $n$ is odd,
\begin{equation*}
\begin{split}
H_p(r) & = \frac{(n-2)(n-4)\ldots 5.3.H_1(r)}{2^{\frac{n-1}{2}}} \\
& \quad - \frac{1}{2}e^{-r^2}\left(r^{n-2}+ \frac{n-2}{2}r^{n-4}+ \ldots + \frac{(n-2)(n-4)\ldots 3}{2^{\frac{n-1}{2}}}r \right) \\
& = \frac{1}{2}\Gamma\left(\frac{n}{2}\right) \left(\operatorname{erf}(r)-e^{-r^2} \sum_{k=1}^{\frac{n-1}{2}}\frac{r^{n-2k}}{\Gamma\left(\frac{n}{2}+1-k\right)}\right)\\
& = \frac{1}{2}\Gamma\left(\frac{n}{2}\right)e^{-r^2} \left(\sum_{k=1}^{\infty}\frac{r^{2k-1}} {\Gamma\left(k+\frac{1}{2}\right)}-\sum_{k=1}^{\frac{n-1}{2}}\frac{r^{2k-1}} {\Gamma\left(k+\frac{1}{2}\right)}\right) \\
& = \frac{1}{2}\Gamma\left(\frac{n}{2}\right)e^{-r^2} \sum_{k=\frac{n+1}{2}}^{\infty} \frac{r^{2k-1}}{\Gamma\left(k+\frac{1}{2}\right)}
\end{split}
\end{equation*}

Therefore,
\begin{equation}
\label{eq:normalF} F(r)=
  \begin{cases}
     e^{-r^2} \sum_{n/2}^{\infty}\frac{r^{2k}}{2^{2k}k!}& \text{if $n$ is even,}\\
     e^{-r^2} \sum_{k=\frac{n+1}{2}}^{\infty}\frac{r^{2k-1}}{2^{2k-1}\Gamma\left(k+\frac{1}{2}\right)}&
      \text{if $n$ is odd.}
  \end{cases}
\end{equation}
and hence it is not difficult to verify that $\mathfrak{D}(M,\rho,\mu) = n$.

\section{Estimation From Datasets}

Two algorithms were used to estimate the distance exponent from artificially generating datasets corresponding to geometric objects of known dimension. In each case an estimate $\hat{F}$ of $F$ was obtained by taking a random sample $X'\subseteq X\subset\Omega$ and calculating all distances between the points in $X'$. Therefore, \[\hat{F}(r) = \mu''(\{(x,y)\in X'\times X': d(x,y)\leq r\})\] where $\mu''$ is the normalised counting measure on $X'\times X'$. All computation was handled by the MATLAB package \cite{MathWorks:1992:MHPa}. In all cases (i.e. for all dimensions) the artificial datasets consisted of no more than 20000 points while approximately 200000 distances were sampled to obtain $\hat{F}$.

The main algorithms tested were based on calculation of the slope of the $\log\hat{F}(r)$ vs $\log r$ graph (original definition of Traina, Traina and Faloutsos \cite{TrainaTF99}) and the fitting of polynomial to $\hat{F}$, both for small values of $r$. A third method which was tried was based on estimation of derivatives but was not successful for the objects of dimensions greater than $3$.

The following artificial datasets were used to test the estimation algorithms:
\begin{itemize}
\item Euclidean spaces $\R^n$ with standard multivariate normal (Gaussian) distributions and $\ell_2$ metrics;
\item Cubes $[0,1]^n\subset\R^n$ with uniform distributions and $\ell_2$ metrics;
\item Spheres $\S^{n-1}\subset\R^n$ with uniform distributions and $\ell_2$ and geodesic metrics;
\item Parabolic through in $\R^n$ with $\ell_2$ metrics.
\end{itemize}

All objects were generated using the built-in MATLAB routines which provide random vectors in $\R^n$ 
according to the Gaussian or uniform distribution. These routines were used directly to generate the multivariate Gaussians and the cubes while additional transformations needed to be applied for the remaining spheres and parabolic throughs. 

Uniform distributions on the spheres were obtained by projecting multivariate Gaussian vectors in $\R^n$ onto the unit sphere $\S^{n-1}$. We define a \emph{parabolic through} $P$ to be a surface in $\R^n$ which is a Cartesian product of a parabola $(x,cx^2)$ where $x\in [a,b],\ a<0<b$, and a $n-2$ dimensional cube (Figure \ref{fig:parbfig}). In order to obtain the uniformly distributed points on $P$, it is sufficient to generate uniformly distributed points on the parabola and the cube separately. Uniform distribution on parabola was obtained by parameterising the parabola by arc-length, sampling from the uniform distribution on $[0,1]$ and mapping the sampled points to the parabola.

\begin{figure}[!hb]
 \begin{center}
 \scalebox{0.4}
 {\epsfig{file=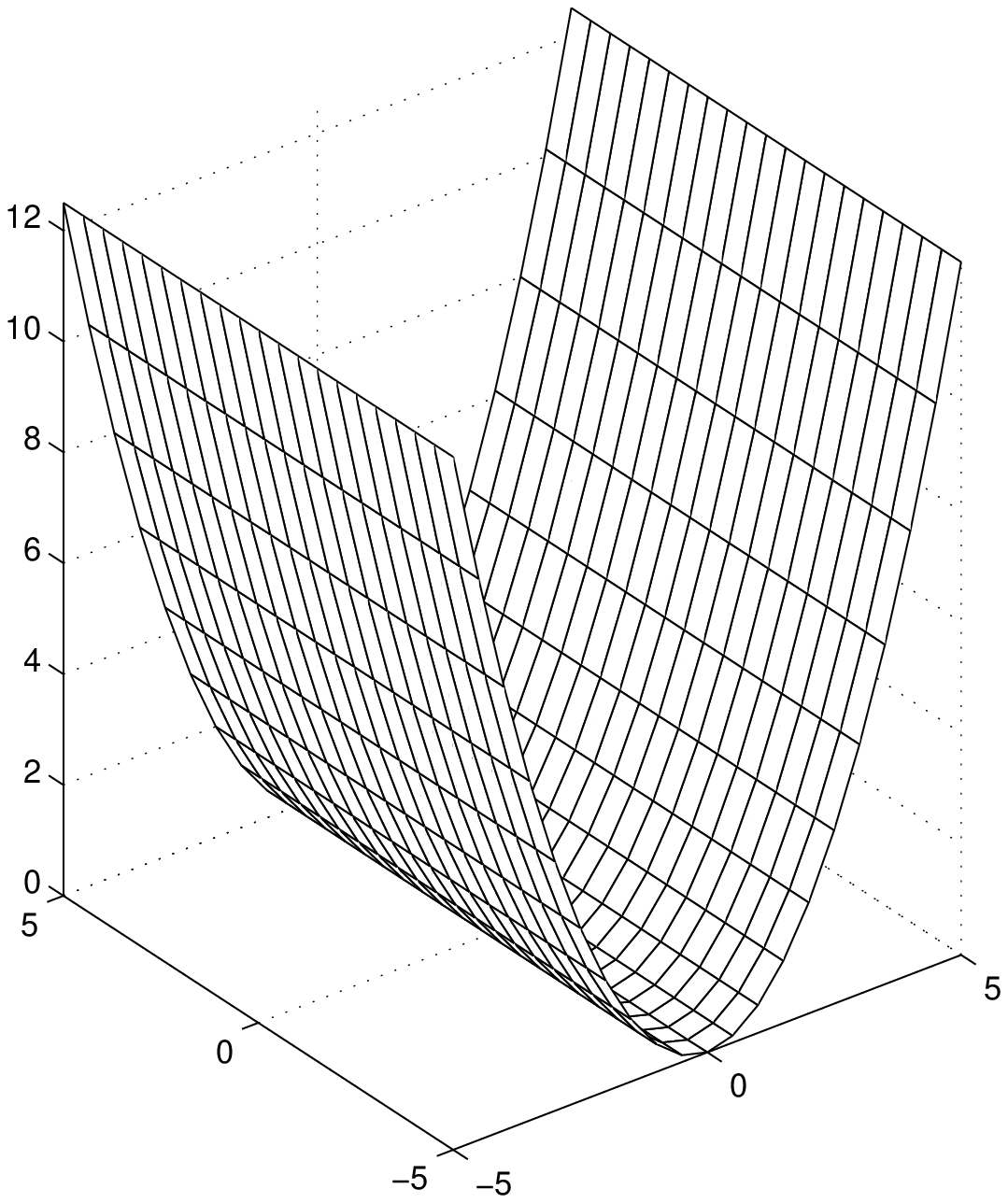}}
  \caption[A parabolic through in $\R^3$]
    {A parabolic through in $\R^3$}
  \label{fig:parbfig}
 \end{center}
\end{figure}

A typical example of the function $F$ and its sampling approximation $\hat{F}$ is shown in the Figure \ref{fig:typicalfig} below. 

\begin{figure}[!ht]
\begin{center}
\scalebox{0.7}{\epsfig{file=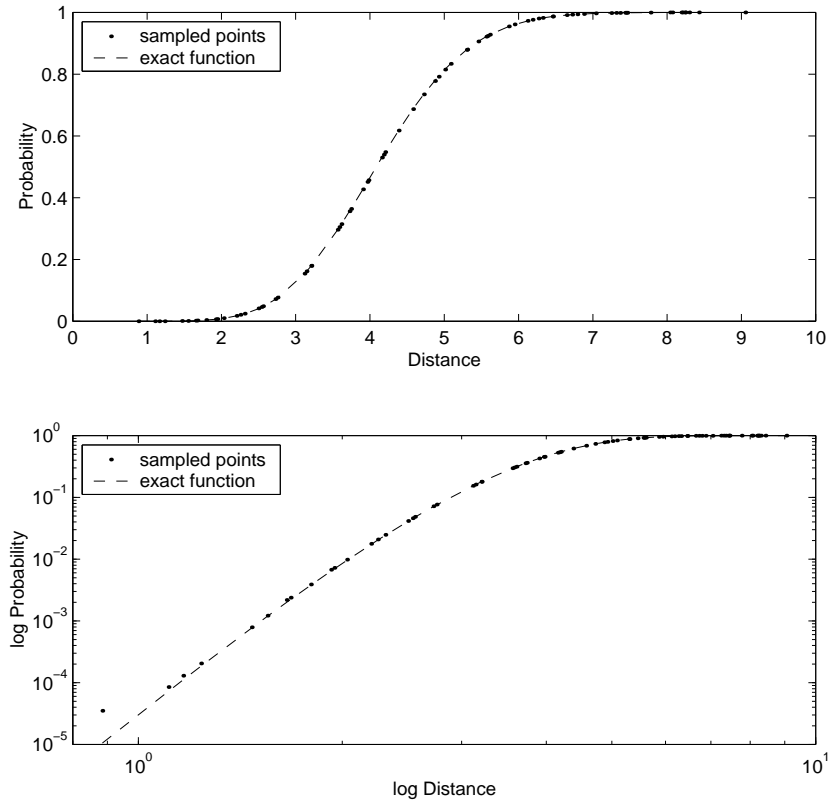}}
\caption[A typical distance distribution function and its approximation.]{The cumulative distance distribution function $F$ and its approximation $\hat{F}$ for the nine-dimensional multivariate Gaussian distribution. Top -- linear scale; bottom -- log-log scale.} \label{fig:typicalfig}
\end{center}
\end{figure}

\subsection{Estimation from log-log plots}

The definition of Traina, Traina and Faloutsos \cite{TrainaTF99} involves estimation of distance exponent from the slope of the `linear part' of the log-log plot of the cumulative distance distribution function $F$. Our implementation produced a least-squares estimation of the slope of $\log\hat{F}$ vs $\log r$ on a given interval $[a,b]$. The end-point of the interval was the fifth percentile (i.e. the smallest value $b$ such that $\hat{F}(b)\geq 0.05$) while the starting point was chosen so as to avoid the first few points corresponding to very small distances which were found not to be good estimates of the true distance distribution function $F$ (see the Figure \ref{fig:typicalfig}). The estimates of dimensions of some of the above mentioned objects using this method are shown in the Figure \ref{fig:loglog1}

\begin{figure}[!hb]
\begin{center}
\scalebox{0.8}{\epsfig{file=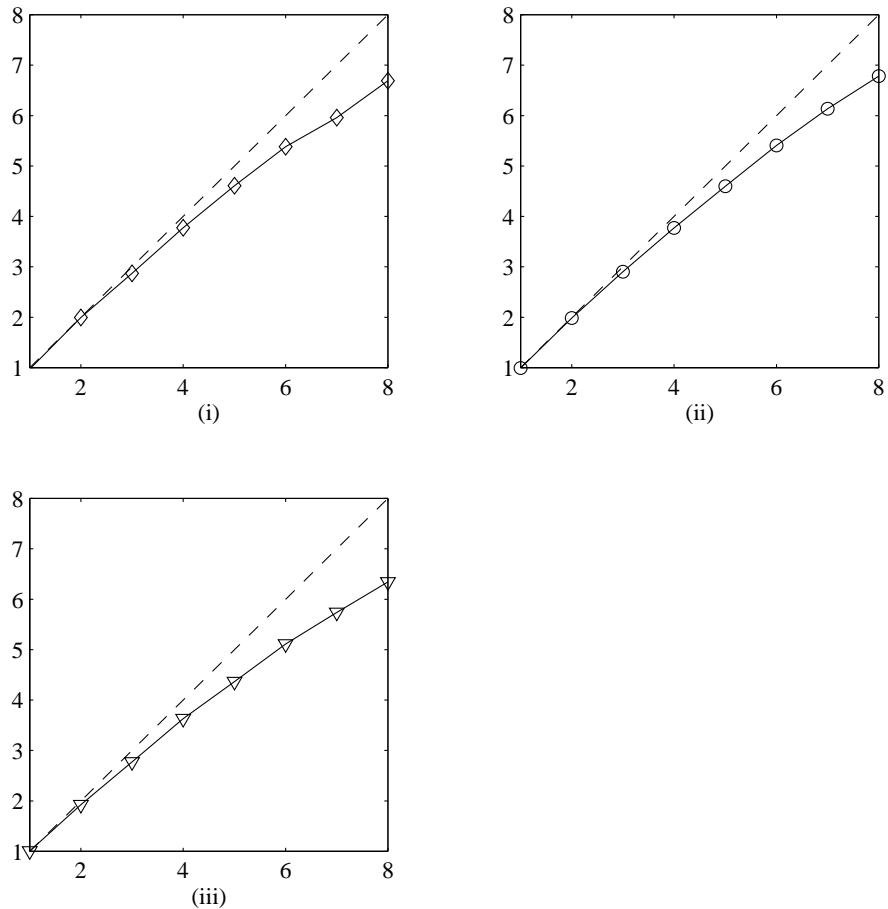}}
\caption[Approximations of distance exponent from the slope of log-log graph for a variety of datasets.] {Approximation of distance exponent from the slope of $\log\hat{F}$ vs $\log r$: estimated vs true dimension. Datasets: (i) multivariate Gaussian on $\R^n$ with $\ell_2$ distances; (ii) uniform distribution on the sphere with geodesic distances; (iii) uniform distribution on the parabolic through with $\ell_2$ distances.} \label{fig:loglog1}
\end{center}
\end{figure}

It is clear that our algorithm systematically underestimated the dimension of objects of `true' (i.e. expected) dimension greater than $3$. The distance exponent estimates for multivariate Gaussians and spheres did not differ to a significant extent while the dimension of parabolic throughs was underestimated to a greater degree than in the other two cases.

In order to find an explanation for our results we sampled the exact values of $F$ for the multivariate Gaussian on $\R^n$ (Equation (\ref{eq:normalF})) and applied our algorithm to them. The results are shown in the Figure \ref{fig:loglog2}.

\begin{figure}[!hb]
\begin{center}
\scalebox{0.6}{\epsfig{file=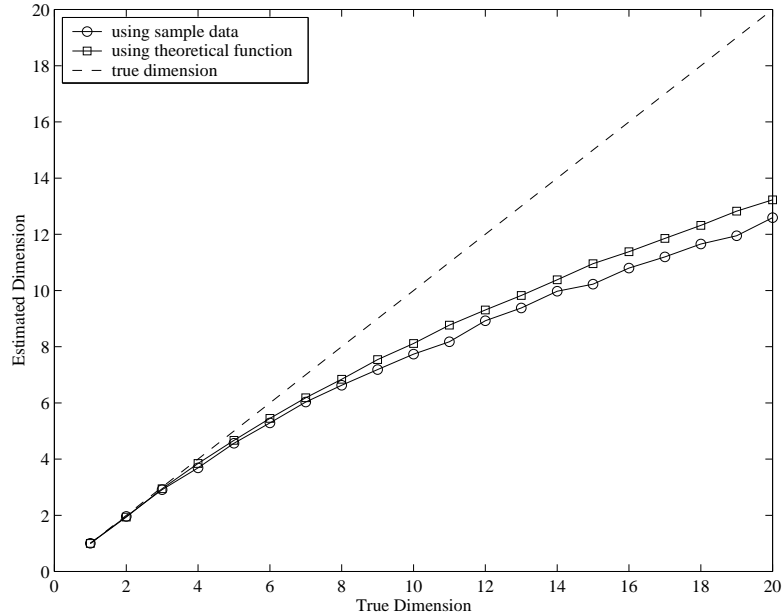}}
\caption[Approximations of distance exponent from the slope of log-log graph for multivariate Gaussian distributions (large interval).]
{Approximations of distance exponent for multivariate Gaussian distributions from the slope of $\log\hat{F}$ vs $\log r$ using $5\%$ of sampled points. Approximations using the exact values of $F$ in the same interval are also shown.} \label{fig:loglog2}
\end{center}
\end{figure}

It can be observed that the estimates of distance exponent obtained using the true values of $F$ (which has no variance due to sampling) are not significantly better than those obtained using the approximation $\hat{F}$. We conclude that most of the observed error is due to bias: $F$ (and therefore $\hat{F}$) is not linear in the region used for estimation of the distance exponent). A method based on weighted least squares, giving more weight to smaller distances (or equivalently reduction of the interval to include very few points, equally distributed along the `linear part') brought some improvement up to the dimension $7$ at a price of instability due to variance (Figure \ref{fig:loglog3}).

\begin{figure}[!hb]
\begin{center}
\scalebox{0.6}{\epsfig{file=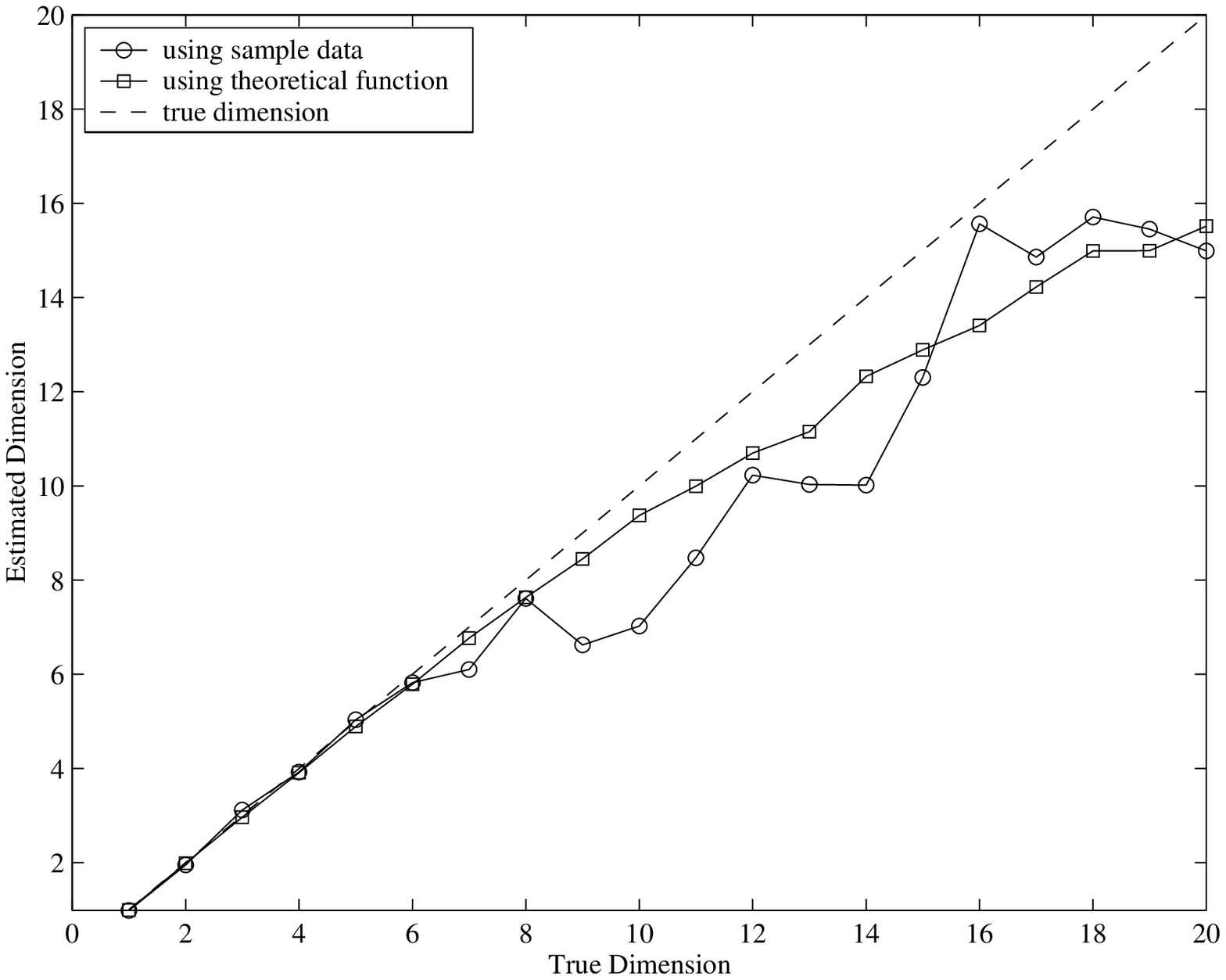}}
\caption[Approximations of distance exponent from the slope of log-log graph for multivariate Gaussian distributions (small interval).]
{Approximations of distance exponent for multivariate Gaussian distributions from the slope of $\log\hat{F}$ vs $\log r$ using only 15 sampled points. Approximations using the exact values of $F$ in the same interval are also shown.} \label{fig:loglog3}
\end{center}
\end{figure}

\subsection{Estimation by polynomial fitting}\label{subseq:polyfitting}

The second approach was based on the least squares approximation of $\hat{F}$ near zero by a polynomial $Q^n_p(x)=x^p\sum_{i=1}^n a_ix^{i-1}$. The estimation of distance exponent $\mathfrak{D}$ was based on the assumption that there exists $L$ such that for $x\in [0,L]$, $\hat{F}(x)\approx Q^n_{\mathfrak{D}(x)}$, and hence that the polynomial $Q^n_{\mathfrak{D}}$ would have the best fit to $\hat{F}$ among all other $Q^n_p$'s. The polynomials were in computed as follows.

Let $y_i=\hat{F}(x_i)$ for $i=1,2,\ldots,m$ where $x_m=L$. Given a possible dimension $p$, and the number of terms of the polynomial $n$, we want to find $Q^n_p$ which such that the $L_2$ norm of the differences between $Q^n_p$ and the sampled function $\hat{F}$ is minimal. Taking into account that $\hat{F}$ is a step function, we minimise
\begin{align*}
\int_0^L\left(\hat{F}(x)-x^p\sum_{i=1}^na_ix^{i-1}\right)^2 dx& =
\int_0^L\left(\hat{F}^2(x)-2\hat{F}(x)x^p\sum_{i=1}^na_ix^{i-1} \right)dx \\
 & \quad +\int_0^L
x^{2p}\left(\sum_{i=1}^na_ix^{i-1}\right)^2dx \\
& = C_0-2\sum_{j=1}^{m-1}\int_{x_j}^{x_{j+1}}y_j\sum_{i=1}^n
      a_ix^{p+i-1}dx \\
& \quad + \int_0^L x^{2p}\sum_{i=1}^n\sum_{k=1}^na_ia_kx^{i+k-2} dx \\
& = C_0-2\sum_{j=1}^{m-1}y_j\sum_{i=1}^n
      C_{ij}a_i + \sum_{i=1}^n\sum_{k=1}^nD_{ik}a_ia_k
\end{align*}
where
\begin{equation*}
\begin{aligned}
C_0 =&\sum_{j=1}^{m-1}y_j^2(x_{j+1}-x_j),
\end{aligned} 
\begin{aligned}
C_{ij}=&\frac{x_{j+1}^{p+i}-x_j^{p+i}}{p+i},
\end{aligned} \text{and}
\begin{aligned}
D_{ik}=&\frac{L^{2p+i+k-1}}{2p+i+k-1}.
\end{aligned}
\end{equation*}
Differentiating with respect to each $a_i$ we get for each $i=1,2\ldots n$, 
\begin{equation}\label{eq:polyterms}
\sum_{k=1}^{n}D_{ik}a_k = \sum_{j=1}^{m-1}C_{ij}y_j.
\end{equation}
Thus we have a system of linear equations $Da=b$ where $b_i=\sum_{j=1}^{m-1}y_jC_{ij}$ which can be solved numerically. For our computations only the one term polynomials were used and in that case the Equation \ref{eq:polyterms} is reduced to 
\begin{equation}
a_1 = \frac{2p+1}{(p+1)L^{2p+1}}\sum_{j=1}^{m-1}y_j(x^p_{j+1}-x^p_j).
\end{equation}

Given the value of $L$, the estimate of distance exponent was obtained by computing the errors for different values of $p$ and selecting the value of $p$ for which the $Q^1_p$ produced the smallest error. For our tests only the integral values of $p$ were tried since it was known that the datasets had the integral dimensions. In general, the optimal value of $p$ can be obtained by numerical optimisation. For the computations, the $\hat{F}$ data was divided into two equally sized sets: the `training' set was used to compute the coefficient of the polynomial and the `testing' set to compute the errors.

\begin{figure}[!hb]
\begin{center}
\scalebox{0.8}{\epsfig{file=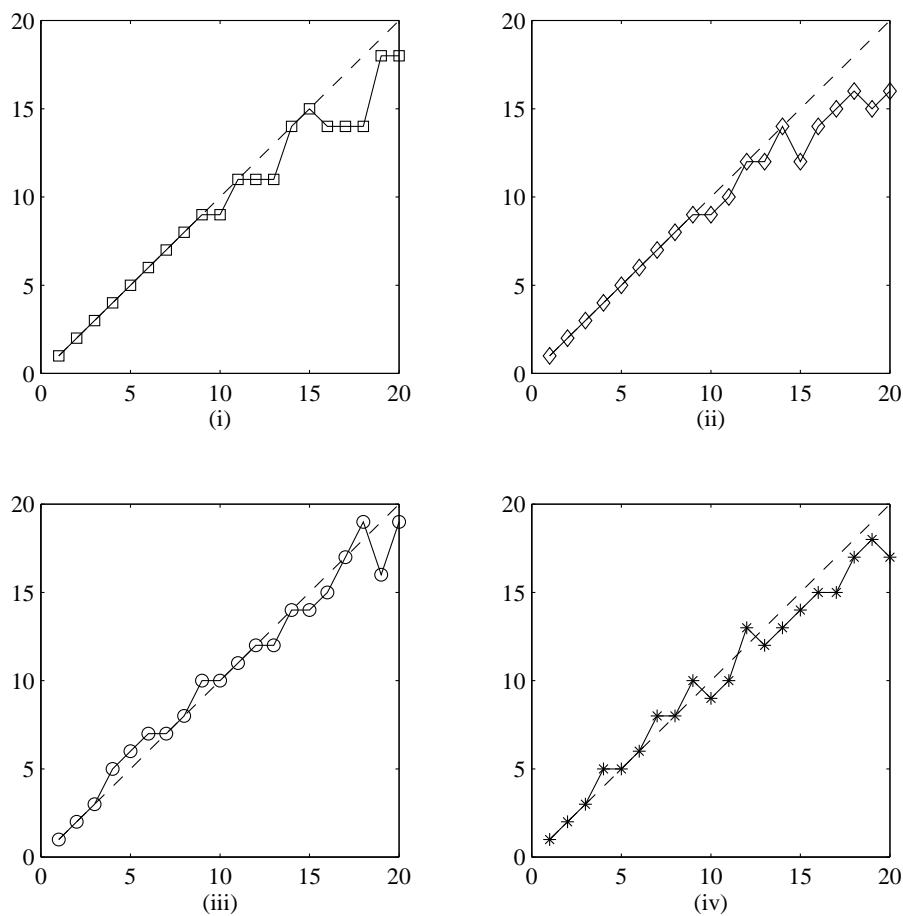}}
\caption[Approximation of distance exponent by fitting monomials.] {Approximation of distance exponent by fitting monomials $ax^p$: estimated vs true dimension. Datasets: (i) uniform distribution on cube with $\ell_2$ distances; (ii) multivariate Gaussian on $\R^n$ with $\ell_2$ distances; (iii) uniform distribution on sphere with geodesic distances; (iv) uniform distribution on sphere with $L_2$ distances.}\label{fig:poly}
\end{center}
\end{figure}

The problem of choosing $L$ (that is, the number of points) was solved by considering a variety of endpoints and picking the maximal value of estimated distance exponent among all of them. This approach was based on the observation that the value of $p$ for which $Q_p$ fits $\hat{F}$ the best has a maximum which is usually (for the low dimensions) the true dimension. The estimated dimension drops for $L$ close to zero because few points are used and a large variance component is present and also because the first few points of $\hat{F}$ usually overestimate $F$. On the other hand, if $L$ is large, the behaviour of $F$ is no longer dominated by $x^{\mathfrak{D}}$. 

The above heuristic method gave surprisingly good results for our simple objects (Figure \ref{fig:poly}). The approximations using the above heuristic method were much closer to the true dimension than those using the slope of $\log\hat{F}$ vs $\log r$.

While it was hoped that the polynomials with more than one term could be used, allowing us to use larger values of $L$, the approximations were not as accurate as those obtained by monomials and their interpretation was more difficult.

\section{General Observations}

It should be noted that estimation of the distance exponent appears to be an ill-posed problem because it is essentially equivalent to calculating derivatives of $F$ around zero (one can prove using l'H\^opital's rule that if distance exponent is $k$ then the first $k-1$ derivatives of $F$ at $0$ must be $0$). We met the variance against the bias problem in both proposed methods. A large interval in which $F$ is approximated by $\hat{F}$ was necessary in order to reduce the variance (since a small interval meant that fewer values of $\hat{F}$ were available) but it introduced the bias which lowered the estimate of the dimension (since the behaviour of $F$ was no longer dominated by $x^\mathfrak{D}$. In addition, in higher dimensions, most of distances at which the values of $\hat{F}$ were available were concentrated very close to the median. This was another manifestation of the Curse of Dimensionality. 
 
In our experiments, the polynomial fitting approach performed better in the higher dimensions than the estimation from log-log plots. It should be noted that all the datasets tested by Traina, Traina and Faloutsos \cite{TrainaTF99} had the dimension less than $7$ (in some cases only estimates were available) so that the underestimation we observed was not as pronounced as in higher dimensions. Our polynomial fitting algorithm can be improved by using numerical optimisation to find the optimal values of $p$ and $L$. 


\backmatter

\bibliographystyle{abbrv}
\addcontentsline{toc}{chapter}{Bibliography}
\bibliography{mathbib,compbib,biobib}

\end{document}